\documentclass[11pt,a4paper,oneside]{article}
\pdfoutput=1

\usepackage[affil-it]{authblk}
\usepackage[a4paper,left=2.5cm,right=2.5cm,top=2.8cm,bottom=3.7cm]{geometry}
	
\usepackage{braket}	
	
	\usepackage{amsmath}
\numberwithin{equation}{section}
\usepackage{amssymb}
\usepackage{amsfonts}
\usepackage{mathrsfs}
\usepackage{mathtools}

\usepackage{amsthm}
	\newtheorem{thm}{Theorem}[section]
	\newtheorem{lem}[thm]{Lemma}
	\newtheorem{prop}[thm]{Proposition}
	\newtheorem{cor}[thm]{Corollary}

	\theoremstyle{definition}
	\newtheorem{defin}[thm]{Definition}
    \newtheorem{stm}[thm]{Statement}

\DeclareMathOperator{\R}{\mathbb{R}}
\DeclareMathOperator{\C}{\mathbb{C}}
\DeclareMathOperator{\N}{\mathbb{N}}
\DeclareMathOperator{\Z}{\mathbb{Z}}
\DeclareMathOperator{\cs}{\mathbb{S}}

\newcommand{\mz}{\mathcal{Z}}
\newcommand{\mf}{\mathcal{F}}
\newcommand{\mN}{\mathcal{N}}
\newcommand{\ms}{\mathcal{S}}

\newcommand{\bsym}{\beta_{\text{\rm\tiny SYM}}}
\newcommand{\bgau}{\beta_{\text{\rm\tiny g}}}
\newcommand{\dd}{\mathrm{d}}

\newcommand{\tr}{\mathrm{Tr}}
\newcommand{\DR}{\mathsf{G}}
\newcommand{\ai}{\mathsf{a}}
\newcommand{\si}{\mathsf{s}}

\newcommand{\lllang}{\langle\!\!\langle}
\newcommand{\rrrang}{\rangle\!\!\rangle}

\DeclareMathOperator{\atan}{\mathrm{arctan}}
\DeclareMathOperator{\aco}{\mathrm{arccos}}
\DeclareMathOperator{\supp}{\mathrm{supp}}

\usepackage[utf8]{inputenc}
\usepackage[english]{babel}

\usepackage[pdftex]{graphicx}
\usepackage[dvipsnames]{xcolor}
\definecolor{darkred}{rgb}{0.8,0.1,0.1}

\usepackage[linktocpage=true,colorlinks=true,citecolor=Blue,linkcolor=darkred,urlcolor=Blue]{hyperref}

\usepackage[toc,page]{appendix}
\usepackage[nottoc]{tocbibind}
\usepackage{caption}
\usepackage{enumerate}
\usepackage{ytableau}
\newcommand{\symbox}{\begin{ytableau} \ & \ \end{ytableau}}

\usepackage{tikz}
		\usetikzlibrary{er,positioning,arrows.meta,calc}

\usepackage[numbers,compress,square,comma]{natbib}	
\begin{document}
\bibliographystyle{myJHEP}
\captionsetup[figure]{labelfont={bf,small},labelformat={default},labelsep=period,font=small}

\title{\textbf{\huge Explicit large \texorpdfstring{$N$}{N} von Neumann algebras from matrix models}}

\author[$\spadesuit$]{Elliott Gesteau\footnote{egesteau@caltech.edu}}
\affil[$\spadesuit$]{\small Division of Physics, Mathematics, and Astronomy \protect\\ California Institute of Technology, Pasadena, CA 91125, U.S.A. \vspace{0.2cm}}

\author[$\clubsuit$]{Leonardo Santilli\footnote{santilli@tsinghua.edu.cn}\vspace{0.35cm}}
\affil[$\clubsuit$]{\small Yau Mathematical Sciences Center\protect\\ Tsinghua University, Haidian district, Beijing, 100084, China. \vspace{0.2cm}}

\date{\hspace{8pt}}

{
	\pagenumbering{arabic}
	\setcounter{page}{1}
	\renewcommand*{\thefootnote}{\fnsymbol{footnote}}
	\maketitle
	\thispagestyle{empty}
			
	\begin{abstract}
        We construct a large family of quantum mechanical systems that give rise to an emergent type III$_1$ von Neumann algebra in the large $N$ limit. Their partition functions are matrix integrals that appear in the study of various gauge theories. We calculate the real-time, finite temperature correlation functions in these systems and show that they are described by an emergent type III$_1$ von Neumann algebra at large $N$. The spectral density underlying this algebra is computed in closed form in terms of the eigenvalue density of a discrete matrix model. Furthermore, we explain how to systematically promote these theories to systems with a Hagedorn transition, and show that a type III$_1$ algebra only emerges above the Hagedorn temperature. Finally, we empirically observe in examples a correspondence between the space of states of the quantum mechanics and Calabi--Yau manifolds.
	\end{abstract}

	\clearpage
{
\tableofcontents 
\par\noindent\rule{\linewidth}{0.4pt}
}}
\renewcommand*{\thefootnote}{\arabic{footnote}}
\setcounter{footnote}{0}
\vspace{0.8cm}

\section{Introduction}

In the past two decades, our understanding of the emergence of spacetime in quantum gravity has immensely improved, in particular in the controlled setting of the AdS/CFT correspondence. This progress is in large part due to the study of the interplay between the emergent geometry of spacetime in the bulk and the entanglement structure of the boundary theory. Important developments include the discovery of the Ryu--Takayanagi formula \cite{Ryu:2006bv} and its covariant generalization \cite{Hubeny:2007xt}, of the quantum extremal surface formula \cite{Lewkowycz:2013nqa,Engelhardt:2014gca}, of the quantum error-correcting properties of holography \cite{Almheiri:2014lwa,Harlow:2016vwg}, of the ``island" prescription \cite{Almheiri:2019hni}, and of the replica wormhole configurations \cite{Penington:2019kki,Almheiri:2019qdq}. Perhaps most famously, these developments culminated in a derivation of the Page curve for black hole evaporation \cite{Penington:2019npb,Almheiri:2019psf}.\par

This plethora of results clarifies the emergence of \textit{space} in the bulk, as they relate the entanglement entropy of subregions of the boundary theory to the area of surfaces in the bulk. In contrast, the emergence of \textit{time} in the bulk is much less explored, especially inside the black hole interior. While time outside the black hole can be directly mapped to the boundary time \cite{Jafferis:2015del}, interior time is much more mysterious, and faces all sorts of paradoxes, including puzzles related to diffeomorphism invariance \cite{PhysRev.160.1113} and an apparent incompatibility with the axioms of conformal field theory \cite{Marolf:2012xe}. Importantly, signatures of the black hole singularity remain largely elusive.\par
\medskip
A new approach to the emergence of time in the AdS/CFT correspondence was recently put forward by Leutheusser and Liu \cite{Leutheusser:2021qhd,Leutheusser:2021frk}, where they argued that the interior time could only be sharply understood in the large $N$ limit of the boundary theory. After defining a von Neumann algebra of single trace operators at $N=\infty$ in the thermofield double state $\lvert \Psi_{\beta} \rangle$, they argued that above the Hawking--Page temperature, this algebra becomes type III$_1$, and that this highly nontrivial entanglement property between the two sides of the thermofield double is related to the emergence of a new interior time.

It was conjectured in \cite{Leutheusser:2021qhd,Leutheusser:2021frk} that the sudden transition to type III$_1$ can be diagnosed by the real-time two-point function of single trace operators in the boundary theory. More precisely, given a single trace operator $\phi$:
\begin{equation}
\label{Gtbeta}
G_+(t)=\langle \Psi_{\beta} \lvert \phi (t)\phi (0) \lvert \Psi_{\beta} \rangle .
\end{equation}
The Fourier transform of \eqref{Gtbeta} satisfies 
\begin{align}
\widetilde{G}_+(\omega)=\frac{\rho(\omega)}{1-e^{-\beta\omega}},
\end{align}
where $\rho(\omega)$ is the finite temperature  K\"all\'en--Lehmann spectral function at inverse temperature $\beta$. The conjecture of Leutheusser--Liu is then that the type of the von Neumann algebra of each single trace operator at finite temperature is related to the structure of the spectral density $\rho (\omega)$: if it consists in delta-functions, then the von Neumann algebra is type I, whereas if it consists in a continuum,\footnote{Originally the support was required to be the whole real axis, but we will show that this is not necessary.} then the von Neumann algebra is type III$_1$ (and so is the von Neumann algebra of the full theory). This conjecture therefore proposes to use each single trace operator as a ``probe", which, through its interactions with the whole system, will help diagnose the emergence of spacetime. We will take a similar viewpoint in this paper.  \par
The presence of a phase transition with this sort of structural change in the spectral function is expected in holographic theories, and it is associated to the emergence of a black hole horizon in the bulk. In particular, the mixing properties associated to the presence of a sharp horizon require that the two-point functions of single trace operators must vanish at late times \cite{Festuccia:2005pi,Festuccia:2006sa}. This is related to apparent information loss at large $N$ (see \cite{Festuccia:2005pi,Festuccia:2006sa}), and is only possible if the spectral function is not discrete, and hence if there is an emergent type III$_1$ algebra in the large $N$ limit. The relation between the chaotic properties of the two-point function and the type of the underlying von Neumann algebra was recently explored in \cite{Furuya:2023fei}, which in particular used these methods to obtain the type III$_1$ property from the continuity (or measurability) of the spectral density of generalized free fields. Note that type III$_1$ is only a necessary condition for this horizon structure, but that it is not obvious it is sufficient, in particular, there exists a whole hierarchy of chaotic properties that are expected for black hole horizons and that are strictly stronger than type III$_1$ \cite{Gesteau:2023rrx,Ouseph:2023juq}. \par
In the phase in which the large $N$ gauge theory is characterized by a type III$_1$ algebra, it is argued in \cite{Leutheusser:2021qhd,Leutheusser:2021frk} that there can exist a new, emergent notion of time defined from the extension of a \textit{half-sided modular inclusion}, which is the algebraic structure encoding the presence of a horizon in the bulk at large $N$. This extension allows to take operators outside the horizon to operators in the black hole interior, and allows to potentially ask sharp questions about the emergence of the interior and of the singularity. It also gives yet another clue that interior time should be intrinsically related to the large $N$ limit of the boundary theory.\par
\medskip
The proposal in \cite{Leutheusser:2021qhd,Leutheusser:2021frk} for the emergence of type III$_1$ factors at large $N$ is holographic in essence, and relies on subregion dualities as well as on the assumption of smoothness of the dual bulk spacetime in the large $N$ limit. In particular, the type III$_1$ nature of the von Neumann algebras is argued for in a rather indirect way, by resorting to results from algebraic quantum field theory. In this way an algebraic structure of \textit{half-sided modular inclusion} emerges, and it can be shown that this structure only exists for type III$_1$ von Neumann algebras. However, it is less clear how the emergence of nontrivial algebras arises directly at the level of the string/gauge theory of interest, and how the type of the algebra can be directly inferred from a gauge theory calculation.\par
\medskip
The purpose of this paper is threefold: 
\begin{itemize}
    \item First, to explain how to explicitly construct large $N$ algebras from explicit quantum mechanical theories and rigorously derive their type from considerations related to the spectral function. In particular, a statement similar to, but slightly different from, the conjectures of \cite{Leutheusser:2021qhd,Leutheusser:2021frk}, and that agrees with the relevant results in \cite{Furuya:2023fei} in the type III$_1$ case, will be rigorously shown and applied to our examples.
    \item Second, to propose a broad class of toy models inspired from gauge/string duality, of which \cite{Iizuka:2008eb} constitutes a particularly simple instance, and apply our formalism to them to show they give rise to an emergent type III$_1$ von Neumann algebra. This part, which is the most prominent both in terms of length and of technical developments, involves mathematics related to random matrix theory and combinatorial representation theory.
    \item In addition, we further elaborate on these examples and construct a family of models that undergo a first order phase transition in the large $N$ limit. This is a general result based on the study of the partition functions of the systems, and the prescription we provide is independent of the considerations on von Neumann algebras. Combining this with the analysis of the spectral density, we show that these modified models carry an associated type III$_1$ algebra only above the critical temperature.
\end{itemize}\par
Our strategy for the first goal will be to apply standard techniques from algebraic quantum field theory to associate a special kind of von Neumann algebra to any quantum dynamical system of harmonic oscillators, under the assumption that correlation functions factorize in the large $N$ limit. The von Neumann algebras describe a generalized free field theory, can be systematically studied, and their type turns out to be classified by the spectrum of an operator closely related to the spectral density (these techniques should be compared to \cite{Furuya:2023fei}).\par
The second goal will be met thanks to a general construction of quantum mechanical systems whose partition function can be computed from matrix integrals of a kind that often appears in the string theory literature. In this simple setup, we will be able to calculate thermal correlation functions entirely explicitly. In particular, the large $N$ value of these correlation functions can be expressed in terms of the saddle point eigenvalue density of a discrete matrix ensemble. From this, we will be able to rigorously derive the type of the von Neumann algebra describing these systems at large $N$. Figure \ref{fig:basics} sketches the basic ideas.\par
Finally, we will meet the third goal by enlarging our quantum systems into a direct sum of sectors that transform under different flavor symmetry groups. This will have the effect of promoting the third order phase transition present in our initial examples to a first order one, which will be sharp enough to induce a change at the level of the large $N$ algebra depending on the temperature. This procedure interestingly parallels and revisits early calculations in gauge theory \cite{Liu:2004vy}.\par
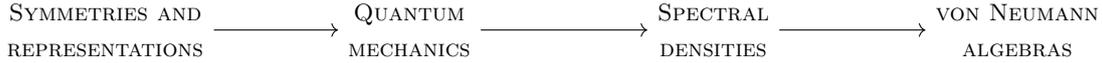
\begin{figure}[tb]
    \centering
    \begin{tikzpicture}
        \node[align=center] (b) at (-4,0) {\textsc{\footnotesize Symmetries and}\\ \textsc{\footnotesize representations}};
        \node[align=center] (c) at (0,0) {\textsc{\footnotesize Quantum}\\ \textsc{\footnotesize mechanics}};
        \node[align=center] (d) at (4,0) {\textsc{\footnotesize Spectral}\\ \textsc{\footnotesize densities}};
        \node[align=center] (e) at (8,0) {\textsc{\footnotesize von Neumann}\\ \textsc{\footnotesize algebras}};
        \draw[->] (b) -- (c);
        \draw[->] (c) -- (d);
        \draw[->] (d) -- (e);
    \end{tikzpicture}
    \caption{In this work we explore the implications of writing the finite temperature partition function of a theory in terms of representations of the global symmetries. We construct a quantum mechanics from these representations and determine the large $N$ von Neumann algebra of operators.}
    \label{fig:basics}
\end{figure}\par
\medskip

Effective descriptions of gauge/string dualities in terms of matrix models have a long history, and the incarnation we consider was pioneered in \cite{Sundborg:1999ue,Aharony:2003sx}. A particularly appealing feature of such descriptions is that matrix models often admit character expansions, which can be useful to carry calculations explicitly in a setup more akin to bulk variables. The most precise statements along these lines have been made in $\mN=4$ super-Yang--Mills theory, in the case of half-BPS states \cite{Lin:2004nb}, as well as in the more recent giant graviton expansion \cite{Gaiotto:2021xce,Lee:2022vig}. In our case, the states of the systems we will study will be packaged into such character expansions.\par

\medskip
\begin{figure}[thb]
    \centering
    \begin{tikzpicture}
        \draw[fill=gray,opacity=0.15] (-6,-2) -- (0,-2) -- (0,2) -- (-6,2) -- (-6,-2);
        \draw[fill=gray,opacity=0.5] (0,-2) -- (6,-2) -- (6,2) -- (0,2) -- (0,-2);
        \draw[black,thick,->] (-6,0) -- (6,0);
        \draw[black,thick,dashed] (0,-2) -- (0,2);
        \node[anchor=north east] at (0,0) {$T_H$};

        \node at (3,1) {\large Type III$_1$};
        \node at (-3,1) {\large Type I};
        \node at (-3,-1) {\large $\mz \sim O(1)$};
        \node at (3,-1) {\large $\mz \sim O(\exp N^2)$};

        \node[anchor=south west] at (-6,0) {\footnotesize low $T$};
        \node[anchor=south east] at (6,0) {\footnotesize high $T$};
    \end{tikzpicture}
\caption{Below the Hagedorn temperature ($T<T_H$) the partition function is finite and the algebra of single-trace operators is a type I von Neumann algebra. Above the Hagedorn temperature ($T>T_H$), the partition function diverges and the von Neumann algebra becomes type III$_1$.}
\label{fig:arrow}
\end{figure}
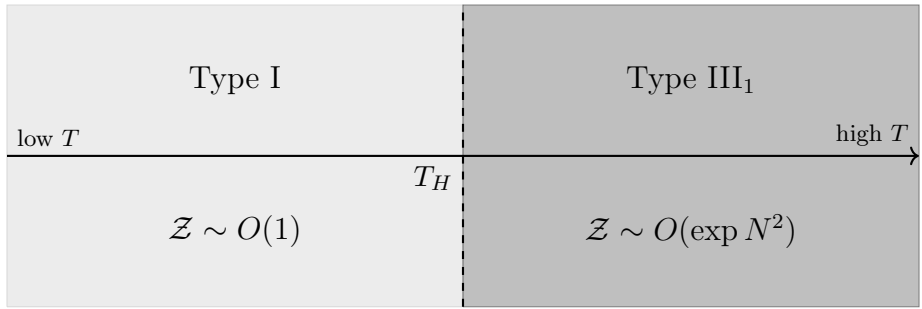\par
The main results of this work can be summarized as follows:
\begin{itemize}
   \item Building on the foundational work on operator algebras \cite{doi:10.1063/1.1704002,ArakiWyss} summarized in \cite{Derezinski} and its relation to algebraic quantum field theory (see for example \cite{Fewster:2019ixc}), we give a completely general and rigorous procedure to construct large $N$ von Neumann algebras from quantum systems satisfying large $N$ factorization, and to determine their types. In particular, we show that a more precise and entirely rigorous version of the conjecture of \cite{Leutheusser:2021qhd,Leutheusser:2021frk} relating the type of the von Neumann algebras to the support of the large $N$ spectral function directly follows from these results (see also \cite{Furuya:2023fei}).
    \item We put forward a general procedure to construct quantum systems with a third order phase transition in the large $N$ limit, and an extension of these systems to systems with a first order Hagedorn-like transition. An alternative description of these systems is given in terms of unitary matrix models, which often have a more direct gauge theory interpretation. This is ultimately a matrix model result, potentially of independent interest.
    \item For our class of models, we introduce a general notion of probe operator, analogous to a detector interacting with the gauge theory. We derive a \emph{universal relation} between the K\"all\'en--Lehmann spectral density associated to the real-time two point functions of such probes, and the eigenvalue density of an underlying matrix model. 
    \item From this K\"all\'en--Lehmann density, we apply our general method to construct an associated large $N$ von Neumann algebra. This algebra can be constructed from the correlation functions of the probe only. We exploit the relation between spectral density and eigenvalue density to understand the types of our algebras (Figure \ref{fig:arrow}).
    \item We illustrate our results by computing the partition functions and K\"all\'en--Lehmann densities entirely explicitly in selected examples inspired by various physical theories.
\end{itemize}

\subsection{Organization}

The contents of the paper are organized in two main parts plus appendices. Part \ref{part1}, which is more formal in nature, contains all general constructions, both of large $N$ algebras associated to quantum systems with large $N$ factorization, and of the quantum systems of interest in this paper. We introduce all the main theoretical results in this part. Part \ref{part2} is dedicated to explicit examples, where the machinery of Part \ref{part1} finds explicit realization. We tried to keep the two parts as independent as possible, so that the busy reader interested only in concrete calculations can directly jump to Part \ref{part2}, while the more mathematically minded reader can content themselves with the general and formal constructions of Part \ref{part1}. Figure \ref{fig:concepts} contains a sketch of the main concepts and ideas discussed in the text.\par
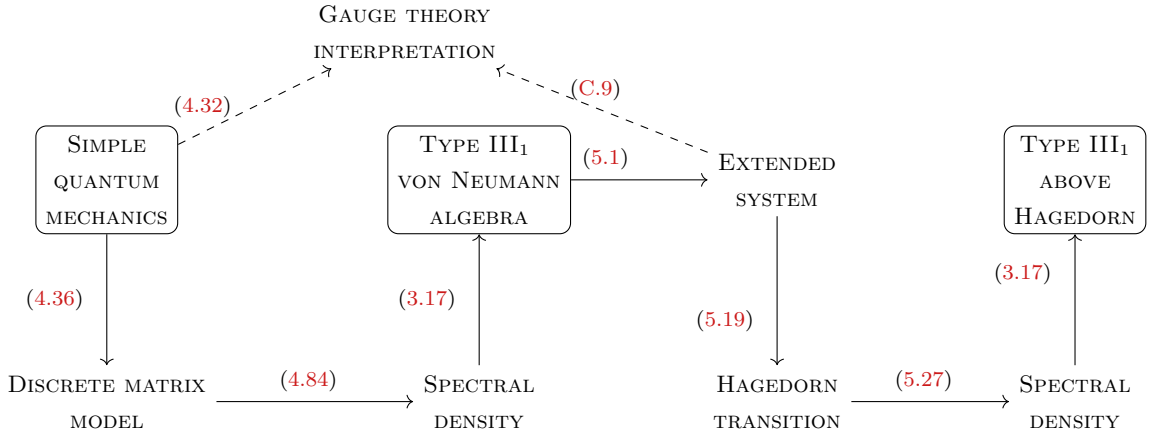
\begin{figure}[t]
\centering
\begin{tikzpicture}[scale=0.98]
\node[align=center,rectangle,rounded corners,draw] (a) at (-1,0) {\textsc{\footnotesize Simple}\\ \textsc{\footnotesize quantum}\\ \textsc{\footnotesize mechanics}};
\node[align=center] (gt) at (3,2) {\textsc{\footnotesize Gauge theory}\\ \textsc{\footnotesize interpretation}};
\node[align=center] (ex) at (8,0) {\textsc{\footnotesize Extended}\\ \textsc{\footnotesize system}};
\node[align=center] (h) at (8,-3) {\textsc{\footnotesize Hagedorn}\\ \textsc{\footnotesize transition}};
\node[align=center,rectangle,rounded corners,draw] (t) at (4,0) {\textsc{\footnotesize Type III$_1$ } \\ \textsc{\footnotesize von Neumann}\\ \textsc{\footnotesize algebra}};
\node[align=center] (s2) at (12,-3) {\textsc{\footnotesize Spectral}\\ \textsc{\footnotesize density}};
\path[->,dashed] (ex) edge node[anchor=south] {\scriptsize \eqref{eq:defHagMM} \footnotesize\hspace{2pt} } (gt);
\path[->] (t) edge node[anchor=south,pos=0.3] {\scriptsize \eqref{eq:ZFissumLUMM} \footnotesize\hspace{2pt} } (ex);
\path[->,dashed] (a) edge node[anchor=east] {\scriptsize \eqref{dictionaryrepQM} \footnotesize\hspace{2pt} } (gt);
\path[->] (ex) edge node[anchor=east,pos=0.7] {\scriptsize \eqref{eq:FermMMHagedorn} \footnotesize\hspace{2pt} } (h);
\path[->] (h) edge node[anchor=south] {\scriptsize \eqref{eq:rhofromHagedorn} \footnotesize\hspace{2pt} } (s2);
\node[align=center] (rho) at (-1,-3) {\textsc{\footnotesize Discrete matrix}\\ \textsc{\footnotesize model}};
\node[align=center] (vn) at (4,-3) {\textsc{\footnotesize Spectral} \\ \textsc{\footnotesize density}};
\path[->] (a) edge node[anchor=east] {\scriptsize \eqref{eq:genericdiscreteMM} \footnotesize\hspace{2pt} } (rho);
\path[->] (rho) edge node[anchor=south] {\scriptsize \eqref{eq:WigthmanVeneziano} \footnotesize\hspace{2pt} } (vn);
\path[->] (vn) edge node[anchor=east] {\scriptsize \eqref{eq:typefromrho} \footnotesize\hspace{2pt} } (t);
\node[align=center,rectangle,rounded corners,draw] (f) at (12,0) {\textsc{\footnotesize Type III$_1$}\\ \textsc{\footnotesize above}\\ \textsc{\footnotesize Hagedorn}};
\path[->] (s2) edge node[anchor=east,pos=0.7] {\scriptsize \eqref{eq:typefromrho} \footnotesize\hspace{2pt} } (f);
\end{tikzpicture}
\caption{The steps through which this work associates von Neumann algebras to quantum mechanical systems with large $N$ factorization.}
\label{fig:concepts}
\end{figure}

\medskip
In Section \ref{sec:reviewholo}, we begin with a brief review of the salient features of von Neumann algebras and their relevance for the emergence of interior time in holography. Our construction and results on von Neumann algebras are in Section \ref{sec:vNtot}. We explain how to construct general von Neumann algebras associated to large $N$ factorizing quantum systems, with a key role played by the K\"all\'en--Lehmann spectral density. In Section \ref{sec:QM} we change gears and introduce the class of systems we aim to study. The operators required as input for the machinery developed in Section \ref{sec:vNtot} are explicitly constructed in Subsection \ref{sec:probe}. In Subsection \ref{sec:spectral} we calculate the real-time two-point function of these operators using matrix models, and obtain the spectral density $\rho (\omega)$. This allows us to deduce the type of the associated von Neumann algebra.\par
Then, in Section \ref{sec:Fermi}, we extend the quantum systems by including a weighted sum over the flavor rank. In Subsection \ref{sec:FermionicMM} we investigate the asymptotic growth of the partition function of these extended models and argue for the presence of a large $N$, first order phase transition with Hagedorn behavior. Subsection \ref{sec:specAvg} replicates the analysis of the spectral density, and shows its different behavior on the two sides of the Hagedorn transition. Combining the three steps therein leads to a jump in the type of von Neumann algebra across the Hagedorn phase transition, which becomes type III$_1$, as indicated in Figure \ref{fig:arrow}.\par
\medskip
We then proceed to apply our construction to different examples of systems inspired by gauge/string duality, and perform explicit calculations. This is the content of Part \ref{part2}, consisting of Sections \ref{sec:variationsIOP} to \ref{sec:YM2Global}. In Section \ref{sec:variationsIOP} we consider a generalization of a model from \cite{Iizuka:2008eb} (see also \cite{Iizuka:2008hg,Michel:2016kwn}) and unveil its connection with supersymmetric QCD in 4d (SQCD$_4$) as well as with Calabi--Yau varieties. We consider a newly introduced ensemble based on a toy model of lattice QCD in 2d (QCD$_2$) \cite{Hallin:1998km} in Section \ref{sec:ExQCD2}. The third example, in Section \ref{sec:conifold}, is the generating function of Donaldson--Thomas invariants of the resolved conifold, and we comment on the geometric meaning of our construction along the way. The final example is based on the partition function of pure Yang--Mills theory in two dimensions, which also recently appeared in the description of generic holographic systems with global symmetry \cite{Kapec:2019ecr,Kang:2022orq}, as reported in Section \ref{sec:YM2Global}.\par
Section \ref{sec:discussion} contains a summary of our results and presents future research directions. The text is complemented by extensive appendices, of which \ref{sec:ccrcar} and \ref{sec:PTMM} may be of independent interest. Appendix \ref{sec:ccrcar} formally generalizes the construction of Section \ref{sec:vNtot} to any bosonic or fermionic system made of free oscillators.  We discuss an effective matrix model for four-dimensional $\mN=4$ super-Yang--Mills theory \cite{Liu:2004vy,Dutta:2007ws} in Appendix \ref{sec:effN4}, which is amenable to be analyzed by the same means as the examples in Part \ref{part2}, even though we lack a quantum mechanical interpretation. The scope of Appendix \ref{sec:effN4} is to showcase our techniques in a familiar matrix model, without attempting any claim concerning $\mathcal{N}=4$ super-Yang--Mills theory. Appendix \ref{sec:PTMM} discusses unitary matrix models at large $N$ and puts forward a general procedure to promote the characteristic third order phase transitions to first order ones. Certain long proofs are gathered in Appendix \ref{app:longproof}.\par
\medskip

\subsubsection*{Disclaimer on notation}
In order to facilitate the reading and the retrieval of the main ingredients, the fundamental steps in the derivation are framed as Theorems and Propositions. We will also informally summarize the main achievements as Statements, referring to the accompanying theorem for a precise version. Previously known results are accompanied by citation of the original works or other standard references, while propositions without citations are new. Besides, while aiming at a broad picture in Part \ref{part1}, we highlight with a symbol $\diamond$ caveats that may obstruct the direct application of certain steps to more involved physical models.\par
In order to make this work self-contained and easily reproducible for the black hole inclined audience, we sketch the proofs of a handful of statements that are likely known to the matrix model practitioner.\par

\subsection{The speed-reader's guide to the contents}
Given the length of the paper, we now point at some important parts that can be looked at independently. We also emphasize that the two parts of this work can be read largely independently, and that our examples can be understood without the general theory.

\begin{itemize}
    \item The reader interested in von Neumann algebras for systems with large $N$ factorization can directly look into Subsection \ref{sec:vNpivot}. The classification result is given in Subsection \ref{sec:vNalgebra}.
    \item For a hands-on definition of our quantum systems, we suggest looking at \eqref{eq:ZLNCharExp}, whose reinterpretation as a Hilbert space trace is explained in Subsection \ref{sec:RepsToQM}. Our definition of a probe and its interaction with the systems is in Subsection \ref{sec:probe} --- or one may jump to Corollary \ref{c:rholargeN} and Statement \ref{thm:istype3} for the properties of the large $N$ spectral density.
    \item The main technique to promote a third order transition to a first order one is explained in Subsections \ref{sec:QMLsum} and \ref{sec:SummaryFF}. The results about the corresponding new systems are summarized in Subsection \ref{sec:summary3steps}.
    \item The reader interested in concrete examples may directly look at Part \ref{part2}. Example 1 in Section \ref{sec:variationsIOP} displays all the features discussed abstractly in Part \ref{part1}.
\end{itemize}

\clearpage

\part{General theory}
\label{part1}

The goal of this first part is to introduce all our main results in the most general language possible. We will begin with a brief review of the increasingly important role von Neumann algebras have been playing in holography in recent years. We will then provide a general construction of von Neumann algebras capturing the finite temperature correlations of systems with large $N$ factorization, as well as an entirely explicit criterion to determine the type of these algebras based on the support of the spectral function. Then, we will introduce our family of toy models, their relationship to matrix integrals akin to those appearing in string theory, and a general prescription to determine the support of the spectral function of our models in terms of the asymptotic behavior of the underlying matrix integral. In particular, this formula will be enough to determine the types of the large $N$ algebras of our systems.

\section{Von Neumann algebras in holography}
\label{sec:reviewholo}

The role of von Neumann algebras in holography is ubiquitous, and has become increasingly important in the recent years (see also the earlier works \cite{Papadodimas:2012aq,Papadodimas:2013jku}). They are especially relevant to make sense of statements related to quantum error correction in AdS/CFT \cite{Kang:2018xqy,Kang:2019dfi,Gesteau:2020rtg,Faulkner:2020hzi,Gesteau:2021jzp,Chandrasekaran:2022eqq,Faulkner:2022ada,Gesteau:2023hbq}, which are the backbone of the quantum extremal surface formula \cite{Harlow:2016vwg}. They also appear in the study of the emergence of the bulk radial direction, through the notion of modular chaos \cite{Faulkner:2018faa,DeBoer:2019kdj}.\par
In this preliminary section, we will focus on one particular aspect of the appearance of von Neumann algebras in holography, that is related to the emergence of time in the black hole interior above the Hawking--Page temperature. We hence first review the particular aspects of the theory of von Neumann algebras that are relevant to this topic. Then, in Subsection \ref{sec:LLsummary}, we summarize the construction of emergent times due to Leutheusser--Liu \cite{Leutheusser:2021qhd,Leutheusser:2021frk}.

\subsection{Classification of factors}

Let us first recall the definition of a von Neumann algebra and of a factor on a Hilbert space.

\begin{defin}
Let $\mathscr{H}$ be a Hilbert space. The commutant $M^\prime$ of an algebra of bounded operators $M$ on $\mathscr{H}$ is the algebra of bounded operators on $\mathscr{H}$ that commute with all elements of $M$. The bicommutant $M^{\prime\prime}:=(M^{\prime})^{\prime}$ of $M$ is the commutant of the commutant of $M$.
\end{defin}

\begin{defin}
Let $\mathscr{H}$ be a Hilbert space. An algebra of bounded operators $M$ on $\mathscr{H}$ is said to be a von Neumann algebra if it is closed under Hermitian conjugation and is equal to its bicommutant. If the center of $M$ is reduced to multiples of the identity, one says that $M$ is a factor.
\end{defin}\par
It is an important problem to classify all von Neumann algebras up to isomorphism. As any von Neumann algebra can be written as a direct integral of factors, this problem reduces to the easier problem of classifying all \textit{factors}.\par
In general, this problem is very hard, and still wide open. However, it has been solved in a particular case by Connes \cite{Connes:1973} and Haagerup \cite{Haagerup:1987} in the 1970s, namely, the case of \textit{hyperfinite factors}, which can be approximated by finite-dimensional algebras.

\begin{defin}
A von Neumann algebra $M$ is hyperfinite if there exists an increasing sequence of finite-dimensional algebras $(M_n)_{n\in\mathbb{N}}$, such that $M$ is the weak operator closure of $\underset{n\in\mathbb{N}}{\bigcup}M_n$.
\end{defin}
It turns out that hyperfinite von Neumann algebras are most often the ones that appear in physics. In particular, under reasonable assumptions, one can show that in quantum field theory, local algebras are hyperfinite \cite{Haag_1992}.\par
The classification of von Neumann algebras relies on the structure of the set of projections inside the algebras \cite{Takesaki_1979}. This structure does not necessarily give a lot of physical insight for our purposes, and rather than studying it here, we will simply enumerate the possible types of factors and their main properties.

\begin{itemize}
    \item[I)] Type I factors are isomorphic to the algebra of operators $\mathcal{B}(\mathscr{H})$, where $\mathscr{H}$ is a separable Hilbert space. There is of course exactly one type I factor for each dimension $\dim (\mathscr{H}) = n\in\mathbb{N} \cup \left\{ \infty \right\}$. This unique factor is said to have type I$_n$.
    \item[II)] Type II factors can no longer be written in the form $\mathcal{B}(\mathscr{H})$, but they still possess a \textit{tracial weight}, namely a (potentially infinite) functional $\tau$ that satisfies \begin{equation}\tau(xy)=\tau(yx).\end{equation} This trace is \textit{not} the same as the trace in $\mathcal{B}(\mathscr{H})$, and can be thought of as a rescaled version of the trace, where an ``infinite amount of entanglement'' between the algebra and the rest of the system is subtracted \cite{Witten:2021unn}. Type II factors come in two types: type II$_1$ factors for which the trace is finite on all the algebra, and type II$_\infty$ factors, isomorphic to the tensor product of a type II$_1$ and a type I$_\infty$ factor, for which the trace is infinite on the identity. Both types are unique up to isomorphism.
    \item[III)] Type III factors are all those that remain.
\end{itemize}

The breakthrough of the work of Connes and Haagerup \cite{Connes:1973,Haagerup:1987} is the classification of type III factors in the hyperfinite case, which seemed completely out of reach beforehand. This classification was made possible by the modular theory of Tomita and Takesaki \cite{takesaki_2010}, whose main result we now briefly recall.\par

\subsubsection{Type III factors}

Let $M$ be a type III factor. To any faithful, normal state $\varphi$ on $M$ ({weak-${}^\ast$} continuous linear functional of norm one which is nonzero on nonzero positive elements), Tomita--Takesaki theory \cite{takesaki_2010} associates an unbounded self-adjoint operator\footnote{In the physics literature, $\Delta_\varphi$ corresponds to the exponential of the (full) modular Hamiltonian of the system.} $\Delta_\varphi$ such that for all $t\in\mathbb{R}$, \begin{align}\Delta_\varphi^{-it}M\Delta_\varphi^{it}=M.\end{align} This means that conjugation by imaginary powers of $\Delta_\varphi$ gives rise to a one-parameter automorphism group of $M$, called the \textit{modular automorphism group}. The state $\varphi$ satisfies the KMS condition with respect to that group, which is an infinite-dimensional version of a thermal equilibrium condition.\par
One fundamental difference between type III factors and the other factors is that in the former case, there is no way of implementing modular automorphisms of the algebra by conjugating by a unitary that pertains to the algebra. This means that modular flow is an \textit{outer} automorphism for type III factors, whereas it is inner in the other cases.\par
The way type III factors are classified then is the following: introduce the \textit{Connes invariant} \begin{align} S(M):=\underset{\varphi}{\bigcap}\;\mathrm{Sp}\;\Delta_\varphi,\end{align}where the intersection is taken over all faithful normal states $\varphi$ on $M$, and $\mathrm{Sp}$ denotes the spectrum of an unbounded operator.\par
It turns out that the Connes invariant can take just a few values for type III factors. Let $M$ be a type III factor,
\begin{itemize}
\item If $S(M)=\{0,1\}$, then $M$ is said to have type III$_0$.
\item If for $0<\lambda<1$, $S(M)=\{0\}\cup\{\lambda^n,n\in\mathbb{Z}\}$, then $M$ is said to have type III$_\lambda$.
\item If $S(M)=[0,\infty)$, then $M$ is said to have type III$_1$.\end{itemize}

From this, it is clear that the modular operator $\Delta_\varphi$ of the system has continuous spectrum for all faithful normal states on $M$ if and only if $M$ has type III$_1$. This is a property that is expected from the modular Hamiltonian of black hole states on one side of the horizon. It turns out that in the hyperfinite case, type is enough to completely classify type III factors. More precisely, there is a unique hyperfinite type III$_\lambda$ factor for $0<\lambda\leq 1$. The type III$_0$ case is more subtle.

In quantum field theory, one of the main applications of Tomita--Takesaki theory is the Bisognano--Wichmann theorem \cite{Haag_1992}, which asserts that the modular automorphisms of local algebras associated to Rindler wedges implement Lorentz boosts along the Rindler horizon. This can be seen as a statement of the Unruh effect, which asserts that the QFT vacuum is thermal with respect to Lorentz boosts. 

In fact, due to the horizon structure, the generator of null translations along a Rindler horizon can also be reconstructed from modular theory. Null translations of magnitude $a$ around the horizon are implemented by a semigroup of unitaries $U(a)$. If we denote by $\mathcal{M}$ the algebra of the wedge, and by $\mathcal{N}$ the translated algebra $\mathcal{N}:=U(1)\mathcal{M}U^\dagger(1)$, the following relations are satisfied in the vacuum state, for all $t \ge 0$:
\begin{equation}
\begin{aligned}
\Delta_\mathcal{M}^{it}\mathcal{N}\Delta_\mathcal{M}^{-it}& \subset \mathcal{N},\\
\Delta_\mathcal{M}^{it}U(a)\Delta_\mathcal{M}^{-it}&=\Delta_\mathcal{N}^{it}U(a)\Delta_\mathcal{N}^{-it}=U(e^{-2\pi t}a),\\
J_\mathcal{N}J_\mathcal{M}&=U(2),\\
\Delta_\mathcal{N}^{it}&=U(1)\Delta_\mathcal{M}^{it}U^\dagger(1).
\end{aligned}
\end{equation}\par
This structure is known in algebraic quantum field theory under the name of \textit{half-sided modular inclusion} \cite{cmp/1104253848}, and it is necessary to get a well-defined causal structure in the emergent spacetime. It turns out that under reasonable conditions on the vacuum state, the only von Neumann algebras that are consistent with the half-sided modular inclusion structure are hyperfinite factors of type III$_1$. This implies that in quantum field theory, local von Neumann algebras must have type III$_1$. This leads to the following result: 

\begin{thm}[\cite{Haag_1992}]
In a quantum field theory, the von Neumann algebras associated to Rindler wedges are type III$_1$ factors. 
\end{thm}

\subsection{The Leutheusser--Liu half-sided modular inclusion}
\label{sec:LLsummary}

The recent, key insight of Leutheusser and Liu \cite{Leutheusser:2021qhd,Leutheusser:2021frk} is that a von Neumann algebra associated to a spacetime region outside of a \textit{general} horizon should also have type III$_1$, even in the case of a black hole. The reason is that there is also a half-sided modular inclusion along the black hole horizon.\par
This fact has strong consequences in holography. Consider two copies of strongly interacting $\mN=4$ super-Yang--Mills theory at large $N$, in the thermofield double state \cite{Maldacena:2001kr}. The AdS/CFT dictionary tells us that the bulk effective description of the theory is a theory of fields on an asymptotically AdS spacetime. The boundary theory admits a deconfinement phase transition at large $N$ \cite{Aharony:2003sx}, which is dual in the bulk to the Hawking--Page phase transition \cite{Witten:1998zw}, with the emergence of a wormhole and a horizon on both sides of the thermofield double \cite{Maldacena:2001kr}. For extensive discussion on Hagedorn, deconfinement and Hawking--Page transitions, see \cite{Aharony:2003sx,Alvarez-Gaume:2005dvb}.\par
The operator algebraic interpretation put forward in \cite{Leutheusser:2021qhd,Leutheusser:2021frk} is the following: the algebra of generalized free fields at large $N$ in the bulk is dual to the algebra generated by (appropriately normalized) single-trace operators at $N=\infty$. The emergence of a black hole horizon then implies that the type of this algebra goes from type I to type III$_1$ at the Hawking--Page temperature. \par
Actually, even more can be said. What is shown in \cite{Leutheusser:2021qhd,Leutheusser:2021frk} is that in the case where the large $N$ algebra corresponds to a generalized free field, this half-sided modular inclusion can be extended to positive values of the coordinate on the horizon. It is this extension that allows to cross the horizon and construct a new, emergent time inside the black hole. Therefore, the emergence of time in the black hole interior is very closely related to the type III$_1$ nature of the algebra of the exterior. This emergent time should furthermore break down at a finite value of the parameter, which would be a signature of the black hole singularity.\par
Let us note that these observations do not, in principle, rely on the strength of the 't Hooft coupling on the boundary. While bulk physics is no longer approximated by a supergravity theory in the weakly coupled regime, it has been argued \cite{Festuccia:2005pi,Festuccia:2006sa} that the K\"all\'en--Lehmann spectral density is still continuous at weak (but nonzero) coupling, which implies that the boundary algebra should still become type III$_1$ above the Hawking--Page temperature. If one could still make sense of the notion of half-sided modular inclusion in this regime, this would imply that there is a precise notion of ``stringy horizon" for holographic duals of weakly coupled gauge theories.\par

\section{Von Neumann algebras for systems with large \texorpdfstring{$N$}{N} factorization}
\label{sec:vNtot}
In this rather formal section, we describe a general procedure that allows to construct a von Neumann algebra describing the correlation functions of any system satisfying large $N$ factorization. In particular, this method would apply to $\mathcal{N}=4$ super-Yang--Mills if we were able to compute the large $N$ two-point functions of single trace operators rigorously. In order not to clutter the main body of the paper too much, and also because this will be the case in all our examples, we restrict the discussion here to the case of bosonic systems, however, the fermionic counterpart of the statements presented here can be derived in a completely analogous way thanks to the fermionic algebras introduced in Appendix \ref{sec:ccrcar}. More generally, Appendix \ref{sec:ccrcar} systematically introduces all the mathematical objects involved here in the language of \cite{Derezinski}. On top of this reference, this section also inspires itself from standard techniques of algebraic quantum field theory, see for example \cite{cmp/1103899044,Fewster:2019ixc}, as well as \cite{Longo:2021rag} and \cite[App.C]{Faulkner:2022ada} for recent discussions and comments.

\subsection{The construction}
\label{sec:vNpivot}

Consider a sequence of quantum systems indexed by an integer $N$, at finite inverse temperature $\beta$. Let $A$ be an operator defined at each $N$. We start by defining what we mean by large $N$ factorization.
 
\begin{defin}\label{def:largeNfactor}
A self-adjoint operator $\phi(t)$ is said to satisfy (bosonic) large $N$ factorization if all its large $N$, finite temperature correlation functions have a limit, and these limits satisfy:
\begin{equation}
\begin{aligned}
\braket{\phi (t_1)\dots \phi(t_{2m-1})} & =0, \\ 
\braket{\phi(t_1)\dots \phi(t_{2m})} & =\sum_{\text{Wick pairing }\varpi }\prod_{j=1}^m\braket{\phi (t_{\varpi (2j-1)}) \phi(t_{\varpi (2j)})},
\end{aligned}
\label{eq:Wick}
\end{equation}
where the sum runs over pairings of operators that respect the ordering, with the corresponding permutations of $2m$ objects denoted $\varpi$.
\end{defin}
A quantum field theory at finite temperature is completely determined by its correlation functions. In what follows, we shall be interested in quantum systems containing some observables that satisfy large $N$ factorization. Equation \eqref{eq:Wick} then teaches us that for these observables, all-point functions can be recovered from the datum of two-point functions. The question is now: how does one construct a net of von Neumann algebras and an expectation value functional (i.e. a state on this algebra) whose values reproduce the correlation functions of such a form \eqref{eq:Wick}, and determine its type? It turns out that this problem can be answered by a standard construction of operator algebra theory, namely, a representation of the \textit{canonical commutation relations} (CCR). Note that the algebra we construct here has no direct definition in terms of the finite $N$ observables of the quantum mechanical theory. The idea is to directly engineer it from the data of the correlation functions at large $N$, in such a way that it recovers them.\par
\medskip

\subsubsection{CCR and Araki--Woods theory}
More formally, we wish to construct an algebra of quantum fields in 0+1 dimensions, and a state on this algebra, such that this state captures the correlation functions of observables satisfying large $N$ factorization. There is a generic operator algebraic construction that allows to describe the correlations of bosonic creation and annihilation operators in states that factorize (often dubbed \textit{quasi-free states}). This construction is widely used in algebraic quantum field theory, and is precise enough to know about the type of von Neumann algebra when the correlation functions factorize, see for example \cite{Fewster:2019ixc} for a review. Here we will adapt the method of algebraic quantum field theory to the format of \cite{Derezinski}, which is well-adapted to compute the types of the algebras under consideration. \par

The algebra we will consider is given by a representation of the canonical commutation relations (CCR algebra) over a one-particle Hilbert space. This Hilbert space is a space of functions of spacetime: here we are in $0+1$ dimensions, so it simply is a space of functions of $\mathbb{R}$, interpreted as the time direction. In this manner, the one-particle Hilbert space knows about the creation of individual particles at different times. Equivalently, it will be more convenient to think of the Hilbert space as a space of functions indexed by the frequency $\omega$. More formally, recall the definition of the Wightman function
\begin{align}
G_+(t)=\braket{\phi(t)\phi(0)}_\beta,
\end{align}
and of the K\"all\'en--Lehmann spectral density
\begin{align}
\label{eq:RhoandGrel}
 \rho(\omega):=(1-e^{-\beta\omega})\widetilde{G}_+(\omega).
\end{align}
Throughout the paper, to make contact with the sign conventions in \cite{Festuccia:2006sa}, the definition used for the Fourier transform is \begin{align}\widetilde{f}(\omega):=\frac{1}{2\pi}\int_{\mathbb{R}}f(t)e^{+i\omega t}dt.\end{align}
One can now introduce the following definition.\footnote{We will assume $\rho(\omega)d\omega$ is a measure on $\mathbb{R}$ that induces a tempered distribution. If $\rho$ is not supported on the whole real axis, then one needs to quotient by the null space of the inner product: functions whose support is included in the space where $\rho$ is zero.}
\begin{defin}
The real Hilbert space $L^2(\mathbb{R},\rho)$ is the completion of the space of rapidly decaying real-valued functions $\mathcal{S}(\mathbb{R})$ on $\mathbb{R}$ whose Fourier transform is square-integrable against $\rho$, endowed with the inner product defined by
\begin{align}
    (f_1,f_2):=\int_\mathbb{R}d\omega \tilde{f}_1^\ast(\omega)\tilde{f}_2(\omega)\mathrm{sgn}(\omega)\rho(\omega).
\end{align}
\end{defin}
This is a real inner product because it can be checked that $\rho(\omega)$ is an odd function of $\omega$, positive for positive $\omega$.
Out of this Hilbert space and our inner product, we now construct a representation of the CCR (see Appendix \ref{sec:ccrcar} for more details) through a standard procedure known as the \textit{Araki--Woods construction}. The first step is to introduce the operation
\begin{align}
\widetilde{\ \mathsf{j}f}(\omega):=-i\,\mathrm{sgn}(\omega)\tilde{f}^\ast(-\omega).
\end{align}
This operation endows $L^2(\mathbb{R},\rho)$ with the structure of a real symplectic space, with symplectic structure given by 
\begin{align}
\sigma(f_1,f_2):=\int_\mathbb{R}d\omega \tilde{f}_1^\ast(\omega)\widetilde{\ \mathsf{j}f}_2(\omega)\mathrm{sgn}(\omega)\rho(\omega).
\end{align}
Moreover, note that $\mathsf{j}$ squares to $-1$ (minus the identity operator) and $\mathsf{j}^\ast=-\mathsf{j}$. This implies that one can promote the structure of $L^2(\mathbb{R},\rho)$ to that of a complex Hilbert space $Z$ \cite{Fewster:2019ixc}, on which $-\mathsf{j}$ implements the imaginary unit, with a sesquilinear inner product given by
\begin{align}
\braket{f_1,f_2}:=\int_\mathbb{R}d\omega \tilde{f}_1^\ast(\omega)\widetilde{f}_2(\omega)\mathrm{sgn}(\omega)\rho(\omega)+\int_\mathbb{R}d\omega \tilde{f}_1^\ast(\omega)\tilde{f}_2^\ast(-\omega)\rho(\omega)=(f_1,f_2)+i(f_1,\mathsf{j}f_2).
\end{align}\par
The construction is then given by the following procedure. Let $Z$ be defined as before, and let $\Gamma$ be the bosonic Fock space over $Z\oplus\bar{Z}$. For $f_1$, $f_2$ in the image of $\mathcal{S}(\mathbb{R})$ inside $L^2(\mathbb{R},\rho)$, we introduce the operators \begin{align}W(f_1,\bar{f}_2):=e^{i\phi(f_1,\bar{f}_2)},\end{align} where \begin{align}\phi(f_1,\bar{f}_2)=\frac{1}{\sqrt{2}}(a^\dagger(f_1,\bar{f}_2)+a(f_1,\bar{f}_2)),\end{align}$a$ and $a^\dagger$ being the usual raising and lowering operators on the Fock space $\Gamma$.\par
We then introduce the operator $\rho_{\beta}$ (not to be confused with the spectral density $\rho(\omega)$ --- we chose this notation for consistency with \cite{Derezinski}), defined by 
\begin{align}
\label{eq:rhoandgamab}
\rho_\beta:=\gamma_\beta(1-\gamma_\beta)^{-1},
\end{align} where 
\begin{align}
    \widetilde{(\gamma_\beta f)}(\omega):=e^{-\beta\vert\omega\vert}\tilde{f}(\omega).
\end{align}
Note that $\gamma_\beta$ is a self-adjoint operator satisfying $0\leq\gamma_\beta\leq 1$. Then, for $f\in\mathrm{Dom}(\rho_\beta^{\frac{1}{2}})$, we can define two unitary operators $W_{\gamma,l}$ and $W_{\gamma,r}$ on $\Gamma$ by
\begin{align}
    W_{\beta,l}(f) &:=W((\rho_\beta+1)^{\frac{1}{2}}f,\bar{\rho}_\beta^{\frac{1}{2}}\bar{f}), \\
    W_{\beta,r}(\bar{f})& :=W(\rho_\beta^{\frac{1}{2}}f,(\bar{\rho}_\beta+1)^{\frac{1}{2}}\bar{f}).
\end{align}
It can then be shown that
\begin{prop}[\cite{Derezinski}]
The operators $W_{\beta,l}(f)$ satisfy the canonical commutation relations
\begin{align}
    W_{\beta,l}(f_1)W_{\beta,l}(f_2)=e^{-\frac{i}{2}\mathrm{Im}\braket{f_1,f_2}}W_{\beta,l}(f_1+f_2).
\end{align}
\end{prop}\par
Importantly, the algebra of operators generated by the $W_{\beta,l}(f)$ is not yet a von Neumann algebra: it is only a formal algebra generated by unitary operators, with no extra topological structure. The missing ingredient in order to obtain a von Neumann algebra is to take its \textit{bicommutant} on the Fock space $\Gamma$.
\begin{defin}
The large $N$ von Neumann algebra $M$ associated to the large $N$ spectral density $\rho$ is the bicommutant of the algebra generated by the operators $W_{\beta,l}$ on the Fock space $\Gamma$.
\end{defin}

Although this last step of taking the bicommutant might seem like a mere technicality, it is a crucial one: without it, the $W_{\beta,l}(f)$ do not know anything about the underlying entanglement structure of the vacuum state. It is this last step that ``teaches" the algebra of the $W_{\beta,l}(f)$ about the entanglement structure of the underlying Hilbert space. In particular, the different types that the von Neumann algebra $M$ can have arise due to this bicommutant operation.

The interest of the construction above is that the correlation functions in the vacuum state $\ket{\Omega_\beta}$ of $\Gamma$ exactly reproduce those of the system of interest. In particular, while Dirac deltas are not necessarily in $L^2(\mathbb{R},\rho)$, one can still formally write down the correlation functions of the operators associated to  $\delta(\cdot-t)$ and $\delta(\cdot)$. Then, using \cite[Proposition 40, 1st point of item 4]{Derezinski}, one finds:
\begin{align}
\label{eq:rho2ptTFD}
\bra{\Omega_\beta}\phi(t)\phi(0)\ket{\Omega_\beta}&=\mathrm{Re}\braket{(\rho_\beta)\delta(\cdot-t),\delta(\cdot)}+\frac{1}{2}\braket{\delta(\cdot-t),\delta(\cdot)} \notag \\
    &=\int_\mathbb{R}d\omega\; e^{-i\omega t}\mathrm{sgn}(\omega)\rho(\omega)\left(\frac{e^{-\beta\vert\omega\vert}}{1-e^{-\beta\vert\omega\vert}}+\frac{1}{2}\right)+\frac{1}{2}\int_\mathbb{R}d\omega\; e^{-i\omega t}\rho(\omega) \notag \\
    &=\int_\mathbb{R}d\omega\;e^{-i\omega t}\rho(\omega)\left(\frac{1}{2}\mathrm{coth}\left(\frac{\beta\omega}{2}\right)+\frac{1}{2}\right) \notag \\
    &=\int_\mathbb{R}d\omega e^{-i\omega t}\frac{\rho(\omega)}{1-e^{-\beta\omega}},
\end{align}
which matches with the Wightman function $G_+(t)$ by Fourier inversion. Similar results can be found for correlations of two creation and/or annihilation operators. Moreover, the correlation functions factorize (this is a consequence of the Araki--Woods construction that directly follows from \cite{Derezinski}), which shows that the algebra $M$ satisfies the requirement of describing all correlation functions at finite temperature.

\subsection{Type of large \texorpdfstring{$N$}{N} algebras}
\label{sec:vNalgebra}
With this construction at hand, we can study the type of the von Neumann algebra generated by the $W_{\gamma_\beta,l}(f)$ thanks to the following result summarized by Derezi\'{n}ski \cite{Derezinski}:
\begin{thm}[\cite{Derezinski}]\label{thm:Derezinski}
Let $M$ be the bicommutant of the Araki--Woods representation of the CCR defined above. If $\gamma_\beta$ is trace-class, then $M$ has type I. If $\gamma_\beta$ has some continuous spectrum, then $M$ has type III$_1$. 
\end{thm}

A more detailed justification of the type III$_1$ part of this result based on the triviality of the centralizer can for example be found in \cite{Golodets_Neshveyev_1998,Furuya:2023fei}. Using \eqref{eq:rhoandgamab}, this theorem implies in the context considered by Leutheusser--Liu (which will be enough for our cases as well):
\begin{equation}\boxed{ \hspace{0.5cm}
\begin{tabular}{l c l}
\hspace{8pt} & \hspace{8pt}\vspace{-8pt} \\
 Type I & \hspace{0.7cm} & $\text{if } \supp \rho = \bigsqcup \left\{ \text{ isolated pts } \right\}$ and $\sum_{\omega\in\mathrm{supp}\,\rho}e^{-\beta\vert\omega\vert}<\infty$ .\\
Type III$_1$ & \hspace{0.7cm} & $\text{if } \supp \rho  \subseteq \R \text{ continuous }$ . \vspace{-8pt} \\
     \hspace{8pt} & \hspace{8pt} 
\end{tabular} 
\hspace{0.5cm}  }
\label{eq:typefromrho}
\end{equation}\par
This follows because if $\rho$ is a continuous function (non-constant equal to zero), then $\gamma_\beta$ is just the multiplication operator that multiplies admissible functions by $e^{-\beta|\omega|}$, and has continuous spectrum by continuity of $\rho$. For the type III$_1$ case we are assuming $\rho$ is a nonzero continuous function, so the preimage of every interval around a point in the image of $e^{-\beta\vert\omega\vert}$, which contains an interval, has nonzero measure. 

The theory lurking behind the scenes here is that of \textit{quasi-free states} on CCR/CAR algebras. It is a very well-developed theory and an important tool in the proof of many results of algebraic quantum field theory. Appendix \ref{sec:ccrcar} of this paper summarizes the main properties of quasi-free states, following \cite{Derezinski}.

\subsubsection{Intermezzo: Comparison with Leutheusser--Liu}
At this stage, it is useful to make a remark on the difference between our findings and the conjecture of Leutheusser--Liu \cite{Leutheusser:2021frk,Leutheusser:2021qhd}. This discrepancy was also recently identified in \cite{Furuya:2023fei}. The initial conjecture was that a type III$_1$ algebra emerges in the large $N$ limit if and only if the spectral density not only gets a continuous support, but also gets a continuous support \textit{on the full real line}. This is in slight contradiction with our findings: as long as we are looking at a thermal state and the spectral function of the theory has \textit{some} continuous support, we find that the type III$_1$ property will be satisfied. The subtlety that explains this discrepancy is that the argument given in \cite{Leutheusser:2021frk,Leutheusser:2021qhd} for their conjecture was based on the fact that the modular Hamiltonian of the thermal state needs to have support on the whole real line in order for the algebra to have type III$_1$. This is certainly true. However, the spectral density considered here does not encode the spectrum of the (full) modular Hamiltonian itself, but rather, the spectrum of its counterpart on the one-particle Hilbert space. This means that there is no contradiction between the modular spectrum having to be supported on the whole real line and our result as long as the additive group generated by the one-particle spectrum is dense in $\mathbb{R}$. This is clearly true if the one-particle spectrum has any continuous component.\footnote{Similar arguments about the additive group generated by the one-particle spectrum can be found in the algebraic QFT literature, see for example \cite{10.2307/24714460}.}

\section{Quantum mechanical systems from matrix models}
\label{sec:QM}
Now that a general construction of large $N$ algebras has been presented, we are ready to introduce and investigate the basic properties of a large class of quantum mechanical models, which will be the object of study in the rest of the work. These models have the very convenient property of having a partition function which is a unitary matrix integral, which matches the form of effective descriptions of gauge theories on $\mathbb{S}^{d-1} \times \mathbb{S}^1 _{\bgau}$, where $\bgau$ is the inverse temperature of the gauge theory. It is the character expansion of the matrix integral that recasts it as a manifest finite temperature partition function of a quantum system.\par
Alongside with the definition of our models, we introduce the main ingredients that will play a part in our analysis. Our novel construction of a quantum mechanical system from the representations of the flavor symmetry of the theory under consideration, is laid out in Subsection \ref{sec:QMflavor}. We then introduce a formal notion of probe and its associated Hilbert space in Subsection \ref{sec:probe}, and use it to compute the spectral densities in the quantum mechanics of interest in Subsection \ref{sec:spectral}.\par
These ideas are the technical core of this work, and will be exploited at length in the rest of the paper. They will also be concretely realized in the examples we will discuss in Part \ref{part2}.

\subsection{Gauge theories, flavor symmetries, representations}

Before moving on to the construction of our models, we briefly review some important features of matrix integrals appearing in gauge theory, which will be instrumental in the rest of our analysis.

\subsubsection{Gauge theories as unitary matrix models}
\label{sec:GaugeUMMsetup}
Our starting point is the family of unitary one-matrix models: 
\begin{equation}
\label{eq:genericUMMf}
    \mz_{\text{UMM}} (Y,\tilde{Y}) := \oint_{SU(N+1)} [\dd U ]  ~ f(U,U^{\dagger} ; Y, \tilde{Y}) .
\end{equation}
In this expression, $[\dd U]$ is the normalized Haar measure on the compact Lie group $SU(N+1)$, and $Y,\tilde{Y} \in SU(L+1)$ or $Y,\tilde{Y}\in U(L)$ are matrices of parameters, whose eigenvalues will play the role of fugacities for the flavor symmetry. Besides, $f$ is some class function, whose details depend on the specifics of the model of interest. Finally, the integration $\oint_{SU(N+1)}$ projects the integrand onto its gauge-invariant part, assuming the interpretation of $SU(N+1)$ as a gauge group holds. It is sometimes referred to as Molien--Weyl projector.\par
An important aspect of \eqref{eq:genericUMMf} is that the above integral depends on two integer parameters:
\begin{itemize}
\item $N$, which is the rank of the gauge group of the theory, and 
\item $L$, which is (or scales with) the rank of the flavor symmetry group of the theory.
\end{itemize}
In this section, we will be interested in the limit of this model when $L$ and $N$ both become large with fixed ratio, cf. Definition \ref{def:VenezianoLim} below.\par
\begin{itemize}
\item[$\diamond$] For practical convenience, we will often assume that 
\begin{equation}
\label{eq:exchangeUY}
	\text{\eqref{eq:genericUMMf} is invariant under} \quad (U,Y) \  \longleftrightarrow \ (U^{\dagger},\tilde{Y}) .
\end{equation}
In a motivating example in Section \ref{sec:IOP}, this assumption will have the physical meaning of requiring the theory to be free of chiral anomalies.
\end{itemize}\par
Models of the generic form \eqref{eq:genericUMMf} have a long history in gauge theory, and include the prototypical Gross--Witten--Wadia model \cite{Gross:1980he,Wadia:1980cp,Wadia:2012fr}. In some instances, \eqref{eq:genericUMMf} is derived directly from a quantum gauge theory placed on the $d$-dimensional cylinder geometry $\mathbb{S}^{d-1} \times \mathbb{S}^1$ \cite{Sundborg:1999ue,Aharony:2003sx}. Denoting $r_{d-1}$ and $ \bgau$ the radii of $\mathbb{S}^{d-1} $ and $ \mathbb{S}^1$, respectively, in the limit $r_{d-1}/\bgau \ll 1$ all the fields are heavy and can be integrated out, except for the holonomy of the gauge field around the thermal $\mathbb{S}^1$. The model reduces to an effective (0+1)-dimensional QFT, i.e. a Euclidean quantum mechanics theory on $\mathbb{S}^1$. Denoting the gauge connection $A$ and 
\begin{equation}
    \exp \oint_{\mathbb{S}^1} A = U \in SU(N+1) ,
\end{equation}
one is left with a (fairly complicated) effective action depending on $U, U^{\dagger}$ and the external parameters $Y, \tilde{Y}$ \cite{Aharony:2003sx}. Taking a further simplifying limit of very weak gauge coupling, in which the contributions from massive modes attain a tractable form, one arrives at the expression \eqref{eq:genericUMMf}. It is in this class of the examples that \eqref{eq:exchangeUY} will be related to a condition on the fields participating in the gauge theory.\par
\medskip
One well known feature of gauge theories is that their states come arranged into irreducible representations of the global symmetry. This information can be naturally extracted from the unitary matrix models \eqref{eq:genericUMMf}, through their character expansion. We now proceed to explain this procedure. This picture will become very explicit in the concrete examples we will consider in Part \ref{part2}.\par
We begin by stating useful generalities about $SU(N+1)$ representations (see e.g. \cite{Macdonaldbook}) that we will use, and then move on to stating the result. The expert reader can skip this brief review and move on to Subsection \ref{sec:charexprep}.\par

\subsubsection{Intermezzo: Properties of irreducible representations}

Irreducible representations of $SU(N+1)$ are in one-to-one correspondence with Young diagrams of length at most $N$. For example, the fundamental and adjoint representations of $SU(5)$ correspond to the Young diagrams:
\begin{equation}
	\text{fund} \ : \ \Box , \qquad \qquad \text{adj} \ : \ \ytableausetup{centertableaux,smalltableaux}\begin{ytableau} \ & \ \\ \ \\ \ \\ \ \end{ytableau} \ .
\end{equation}
The diagrams, in turn, are in one-to-one correspondence with partitions of length at most $N$:
\begin{equation}
	R= (R_1, R_2, \dots, R_N) , \qquad R_1 \ge R_2 \ge \cdots \ge R_N \ge 0 .
\end{equation}
The non-negative integers $R_i$ in the partition specify the rows of the diagram. 
\begin{defin}
    Let $R$ be an irreducible representation of $SU(N+1)$ and identify it with a Young diagram $R$. The length of $R$ is the positive integer $\ell (R)$ such that 
\begin{equation}
	R_{\ell (R)} > 0 , \quad R_{\ell (R)+1} =0 .
\end{equation}
If $R= \emptyset $ is trivial, $\ell (\emptyset):=0$.
\end{defin}
We will also use the notation 
\begin{equation}
	\lvert R \rvert := \sum_{i=1}^{\infty} R_i .
\end{equation}\par
The Lie group $SU(N+1)$ is endowed with an involution 
\begin{equation}
\label{eq:chargeconjU}
	\mathsf{C} \ : \ U \ \mapsto \ U^{\dagger} 
\end{equation}
which has the physical meaning of charge conjugation. It acts on the irreducible representations as $\mathsf{C} (R) = \overline{R}$, with $\overline{R}$ the complex conjugate representation. The Young diagram for $\overline{R}$ is the complement to $R$ inside a rectangle of edge lengths $(R_1 , N+1)$, rotated by $180^{\circ}$. For example, the anti-fundamental of $SU(5)$ is 
\begin{equation}
\overline{\text{fund}}\ : \ \overline{\Box} = \ytableausetup{centertableaux,smalltableaux}\begin{ytableau} \ \\ \ \\ \ \\ \ \end{ytableau} \ .
\end{equation}
Another example is, in $SU(4)$
\begin{equation}
 \ytableausetup{centertableaux,smalltableaux}
 R=(5,3,1) \quad \begin{ytableau} \ & \ & \ & \ & \   \\ \ & \ & \ & *(gray) \ & *(gray) \ \\ \ & *(gray) \ & *(gray) \ & *(gray) \ & *(gray) \ \\ *(gray) \ & *(gray) \ & *(gray) \ & *(gray) \ & *(gray) \ \end{ytableau}  \qquad \Longrightarrow \ \begin{color}{gray} \overline{R} =(5,4,2) \end{color}
\end{equation}
whereas the same diagram $R=(5,3,1)$ seen as an $SU(5)$ representation of length $3$ gives 
\begin{equation}
 \ytableausetup{centertableaux,smalltableaux}
 R=(5,3,1) \quad \begin{ytableau} \ & \ & \ & \ & \   \\ \ & \ & \ & *(gray) \ & *(gray) \ \\ \ & *(gray) \ & *(gray) \ & *(gray) \ & *(gray) \ \\ *(gray) \ & *(gray) \ & *(gray) \ & *(gray) \ & *(gray) \ \\ *(gray) \ & *(gray) \ & *(gray) \ & *(gray) \ & *(gray) \ \end{ytableau}  \qquad \Longrightarrow \ \begin{color}{gray} \overline{R} =(5,5,4,2) \end{color} .
\end{equation}
The Young diagram consisting of a column with $N+1$ boxes corresponds to the determinant representation, which is isomorphic to the trivial representation of $SU(N+1)$.
\begin{prop}[\cite{Brocker1985}]
The set of equivalence classes of $SU(N+1)$ representations up to isomorphism, endowed with the direct sum $\oplus$ and tensor product $\otimes$, is a commutative ring, called the \emph{representation ring} of $SU(N+1)$. It is generated by the set $\mathfrak{R}^{SU(N+1)}$ of isomorphism classes of irreducible representations.
\end{prop}

\subsubsection{Unitary matrix models as ensembles of representations: Preliminaries}
\label{sec:charexprep}
Character expansions are a widely used tool.\footnote{Implications of characters expansions in 2d gravity have been explored in \cite{Kazakov:1995ae}, and more recently in \cite{Kimura:2021hph,Betzios:2022pji}.} We prove a general character expansion formula here in an abstract context, which encompasses the case-by-case studies present in the literature. While the outcome will certainly be familiar to the practitioners, to our knowledge a formulation of the character expansion at such a level of generality has not appeared previously.\par
We list our setup and working assumptions and then derive Lemma \ref{lemma:ChExp}. Let the notation be as in \eqref{eq:genericUMMf}, with $f$ a class function for both $U$ and $U^{\dagger}$ separately. The reason why this should be the case in a matrix model derived from a gauge theory was explained in Subsection \ref{sec:GaugeUMMsetup}. The same reasoning applied to connections in the flavor bundle tells us that $f$ should also be a class function with respect to $Y, \tilde{Y}$. Physically, this is to require that the partition function of the theory does not transform under flavor symmetry transformations (possibly up to an anomalous phase, which plays no role in our discussion thus we neglect it).\par
Being $f$ a class function by hypothesis it admits a character expansion, that is, an expansion in the basis of functions on $SU(N+1)$ invariant under the adjoint action of the group. With the hypothesis just stated we write $f$ in the character basis as
\begin{equation}
\label{eq:fcharexp}
	f(U,U^{\dagger} ; Y, \tilde{Y}) = \sum_{R \in \mathfrak{R}^{SU(N+1)}} c_R \chi_R (U) \chi_{\varphi(R)} (Y) \sum_{\tilde{R} \in \mathfrak{R}^{SU(N+1)}} \tilde{c}_{\tilde{R}} \chi_{\tilde{R}} (U^{\dagger}) \chi_{\tilde{\varphi}(\tilde{R})} (\tilde{Y}) 
\end{equation}
(we allow the coefficients $c, \tilde{c}$ to be different in general). Let us unpack the notation. The sums run over isomorphism classes of irreducible representations of $SU(N+1)$. The functions $\chi_R$ are the characters of the Lie group, in the representation $R$. $\varphi(R)$ is a $U(L)$ or $SU(L+1)$ representation, uniquely fixed by $R$ via an map $\varphi : \ \mathfrak{R}^{SU(N+1)} \to  \mathfrak{R}^{U(L)}$ (or $ \to  \mathfrak{R}^{SU(L+1)}$). The concrete form of $\varphi$ depends on $f$, hence on the specifics of the gauge theory under consideration. In many cases, $\varphi$ is just the pullback of a map $U(L) \rightarrow SU(N+1)$, seen as the natural embedding if $L \le N$ and as a projection if $L>N$. Likewise, $\tilde{\varphi} (\tilde{R})$ is a representation uniquely determined by $\tilde{R}$ via a map $\tilde{\varphi}$. These maps will have very explicit realizations in the examples.\par
The examples in Part \ref{part2} will manifestly satisfy these hypothesis. One may think of \eqref{eq:fcharexp} either as 
\begin{itemize}
    \item our working assumption on the matrix models \eqref{eq:genericUMMf}; or
    \item as a physics-motivated requirement for \eqref{eq:genericUMMf} to be a valid approximation of the gauge theory on $\mathbb{S}^{d-1} \times \mathbb{S}^1$.
\end{itemize}

\subsubsection{Unitary matrix models as ensembles of representations: Character expansion}
\begin{lem}\label{lemma:ChExp}
Consider the setup of Subsection \ref{sec:charexprep}. There exist a subset $\mathfrak{R}_{L} ^{(N)} \subseteq \mathfrak{R}^{U(L)}$ (respectively $\subseteq \mathfrak{R}^{SU(L+1)}$) of irreducible $U(L)$ (resp. $SU(L+1)$) representations, a bijection $\phi : \mathfrak{R}_{L} ^{(N)} \to \mathfrak{R}_{L} ^{(N)}$ and numbers $\mathfrak{d}_R, \tilde{\mathfrak{d}}_R \in \mathbb{R}$ labelled by elements $R \in \mathfrak{R}_{L} ^{(N)}$, such that 
\begin{equation}
\label{eq:characterexp}
	\mz_{\text{\rm UMM}} (Y,\tilde{Y})  = \sum_{R \in \mathfrak{R}_{L} ^{(N)}} \mathfrak{d}_R \tilde{\mathfrak{d}}_{R} \chi_R (Y) \chi_{\phi (R)} (\tilde{Y}) .
\end{equation}
\end{lem}
\begin{proof}
The function $f$ in \eqref{eq:genericUMMf} is a class function for both $U$ and $U^{\dagger}$ separately, by hypothesis, and we take its character expansion \eqref{eq:fcharexp}.\par
The character $\chi_{\varphi(R)} (Y)$ in \eqref{eq:fcharexp} vanishes if $\ell (\varphi(R)) > L$. This imposes a constraint on the representations $R$ that contribute non-trivially to \eqref{eq:genericUMMf}, which intertwines the $N$- and $L$-dependence. Moreover, it may happen that the function $f$ is such that some of the coefficients $c_R$ or $\tilde{c}_R$ vanish for certain representations $R$. We will denote by $\mathfrak{R}_L ^{(N)} \subseteq  \mathfrak{R}^{U(L)}$ (or $ \subseteq  \mathfrak{R}^{SU(L+1)}$) the subset of irreducible representations that survive these selection rules.\par
Plugging \eqref{eq:fcharexp} back into \eqref{eq:genericUMMf} and using the orthogonality of characters:
\begin{equation}
    \oint_{SU(N+1)} [\dd U ] \chi_R (U) \chi_{\tilde{R}} (U^{\dagger}) = \delta_{R \tilde{R}} , 
\end{equation}
we arrive at 
\begin{equation}
	\mz_{\text{\rm UMM}} (Y,\tilde{Y})  = \sum_{R \in \varphi^{-1} \left( \mathfrak{R}_{L} ^{(N)} \right) } c_R \tilde{c}_{R} \chi_{\varphi(R)} (Y) \chi_{\tilde{\varphi} (R)} (\tilde{Y}) .
\end{equation}
The sum has been restricted to those $R$ that yield a non-trivial contribution. Changing variables $\hat{R} = \varphi(R)$, this is precisely \eqref{eq:characterexp}, after the identification 
\begin{equation}
    \mathfrak{d}_{\hat{R}} := c_{\varphi^{-1} (\hat{R})} , \qquad \tilde{\mathfrak{d}}_{\hat{R}} := \tilde{c}_{\varphi^{-1} (\hat{R})}, \qquad \phi := \tilde{\varphi} \circ \varphi^{-1}
\end{equation}
(and eventually renaming $\hat{R} \mapsto R$).
\end{proof}
Below we list side remarks concerning formula \eqref{eq:characterexp}:
\begin{itemize}
    \item In practice, we will discuss models such that $f$ is invariant under the involutions \eqref{eq:exchangeUY} and \eqref{eq:chargeconjU}. This implies $\tilde{c}_R= c_R$, and hence $\tilde{\mathfrak{d}}_{\hat{R}} = \mathfrak{d}_{\hat{R}}$.
    \item Moreover, applying \eqref{eq:exchangeUY} followed by \eqref{eq:chargeconjU}, we have that the model remains unchanged under the exchange $Y \leftrightarrow \tilde{Y}$, which also implies that 
\begin{equation}
    \tilde{\varphi} = \varphi  \qquad \text{ or } \qquad \tilde{\varphi} = \mathsf{C} \circ \varphi .
\end{equation}
The partition function is insensitive to which of the two options is actually realized.
    \item Note that, assuming invariance under the charge conjugation involution, 
\begin{equation}
	\phi (R) =R \quad \forall R \qquad \text{or} \qquad \phi (R) =\overline{R} \quad \forall R .
\end{equation}
    \item By construction, $\phi$ is the restriction to $\mathfrak{R}_L ^{(N)} $ of an isomorphism of the Grothendieck ring $\mathfrak{R}^{U(L)}$ (or $\mathfrak{R}^{SU(L+1)}$).
    \item The character expansion extends to other choices of classical gauge group. For integration over $SO(N+1)$ in \eqref{eq:genericUMMf}, only representations of $Sp(L+1)$ will enter in \eqref{eq:characterexp}, and conversely for integration over $Sp(N+1)$ one gets irreducible $SO(L+1)$ representations. The corresponding Young diagrammatic techniques were pioneered in \cite{Koike}.
\end{itemize}\par
\medskip
Before moving on, we introduce the unrefined version of the partition function \eqref{eq:genericUMMf}, 
\begin{equation}
\label{eq:ZLNUMM}
    \mz_L ^{(N)} (y) := \oint_{SU(N+1)} [\dd U ] f ( U,U^{\dagger} ; \mathrm{diag} ( \underbrace{\sqrt{y},\dots ,\sqrt{y}}_{L}) , \mathrm{diag} ( \underbrace{\sqrt{y},\dots ,\sqrt{y}}_{L}) ) .
\end{equation}
We now apply the character expansion to \eqref{eq:ZLNUMM}.
\begin{cor}\label{corol:CharExp}
Let $\mz_L ^{(N)}$ be as in \eqref{eq:ZLNUMM}. With the notation of Lemma \ref{lemma:ChExp}, it holds that 
\begin{equation}
\label{eq:ZLNCharExp}
	 \mz_L ^{(N)} (y)  =  \sum_{R \in \mathfrak{R}_{L} ^{(N)}} y^{\lvert R \rvert } ~  \mathfrak{d}_R \tilde{\mathfrak{d}}_{R} ~ (\dim R) (\dim \phi (R)) .
\end{equation}
If moreover the initial unitary matrix model is invariant under the involution \eqref{eq:chargeconjU}, then 
\begin{equation}
\label{eq:ZFCharExp}
	 \mz_L ^{(N)} (y)  = \sum_{R \in \mathfrak{R}_{L} ^{(N)}} y^{\lvert R \rvert } ~  \mathfrak{d}_R^2 (\dim R)^2 .
\end{equation}
\end{cor}
The gauge-invariant content of a gauge theory must be assembled into representations of the global symmetries, and one can restrict to irreducible ones without loss of generality. The character expansion does the job, and provides us with an ensemble directly in terms of representations, in which the requirement of gauge invariance only appears through the constraint $R \in \mathfrak{R}_{L} ^{(N)}$.\par
Expression \eqref{eq:ZLNCharExp} is then interpreted as an ensemble of gauge-invariant operators. These are grouped into superselection sectors labelled by the pair $(R, \phi (R))$ of irreducible representations of the flavor symmetry. The elementary gauge-invariant operators of the theory are the generators of the pair $(R, \phi (R))$. We now proceed to give an explicit construction of quantum mechanical systems based on this character expansion.\par

\subsection{Quantum mechanics of flavor symmetry}
\label{sec:QMflavor}

Inspired from the general results above, we now construct quantum mechanical systems whose partition functions are given by the character expansions of the matrix models described above. We will very explicitly construct the Hilbert space of these systems as well as their Hamiltonian.

\subsubsection{Ensembles of representations as quantum systems}
\label{sec:RepsToQM}
Let us now revisit and push forward the previous analysis. We set $y=e^{-\beta}$ in \eqref{eq:ZLNCharExp}, so that $y^{\lvert R\rvert} = e^{- \beta \lvert R\rvert }$ is formally written as a Boltzmann factor. We give this hint credit.\par 
In this way, we are led to interpret \eqref{eq:ZLNCharExp} as the partition function of a quantum mechanical system involving the elementary gauge-invariant states, organized into irreducible representations, with Hamiltonian $H$ diagonalized in the representation basis, with eigenvalues $\lvert R \rvert $. The flavor symmetry imposes that all the states belonging to the same pair $(R,\phi(R))$ carry the same energy $\lvert R \rvert$. The degeneracy of all these elementary states of equal energy is already resummed and is accounted for by $(\dim R)(\dim \phi (R))$ in  \eqref{eq:ZLNCharExp}.\par
\medskip
It is important to note here that $\beta$ will be the inverse temperature of the effective quantum mechanics describing the gauge-invariant states, which may be distinct from the inverse temperature $\bgau$ of the gauge theory from which the matrix integral is derived. The two will be related, but possibly in non-trivial ways.
\begin{itemize}
\item[$\diamond$] We henceforth assume $\mathfrak{d}_R, \tilde{\mathfrak{d}}_R\in \mathbb{N}$. This allows us to think of them as additional degeneracy, due to quantum numbers that we are not taking into account. This assumption holds in a vast list of examples of character expansions (even well beyond the current scope, see e.g. \cite{Santilli:2021iks,Mironov:2022yhd,Mironov:2022fsr}), but we do not know of a deeper mathematical justification.
\end{itemize}
\par

Following this intuition, we now write down a quantum mechanical system whose Hilbert space $\mathscr{H}_L$ is graded by (isomorphism classes of) irreducible representations $R$ of the flavor symmetry, 
\begin{equation}
\label{eq:splitHLintoR}
	\mathscr{H}_L = \bigoplus_{R \in \mathfrak{R}^{U(L)} } \mathscr{H}_L (R) \otimes  \mathscr{H}_L (\phi(R)) .
\end{equation}
The vector space for given $R$ is spanned by 
\begin{equation}
\label{eq:vbasisgenR}
\begin{aligned}
	& \lvert R ,  v_{\ai} , \lambda_{\si} \rangle \otimes \lvert \phi(R) , \tilde{v}_{\dot{\ai}} ,   \tilde{\lambda}_{\dot{\si}} \rangle , \\
 & \text{for } \ \ai=1, \dots, \dim R , \ \dot{\ai}=1, \dots, \dim \phi(R), \ \si = 1, \dots, \mathfrak{d}_R , \ \dot{\si} = 1, \dots, \tilde{\mathfrak{d}}_R .
\end{aligned}
\end{equation}
Throughout, we use undotted and dotted indices, respectively, to refer to the $R$ and $\phi (R)$ part of \eqref{eq:splitHLintoR}. 
\begin{itemize}
\item The vectors $\lvert R, v_{\ai} \rangle $ are in one-to-one correspondence with the generators of $R$, and likewise for $\lvert \phi(R), \tilde{v}_{\dot{\ai}} \rangle $ with the generators of $\phi (R)$. 
\item The additional quantum numbers $\lambda_{\si}$ and $\tilde{\lambda}_{\dot{\si}}$ account for the additional degeneracy. We interpret them as associated to additional global symmetries that have remained spectators in the character expansion, but which impose additional selection rules.
\end{itemize}
The states \eqref{eq:vbasisgenR} are in the canonical normalization
\begin{equation}
\begin{aligned}
	&\langle R^{\prime} ,  v_{\mathsf{b}} , \lambda_{\mathsf{r}} \rvert \otimes \langle \phi(R^{\prime}) , \tilde{v}_{\dot{_{\mathsf{b}}}} , \lambda_{\dot{_{\mathsf{b}}}} \vert R ,  v_{\ai} , \lambda_{\si} \rangle \otimes \lvert \phi(R) , \tilde{v}_{\dot{\ai}}  \lambda_{\dot{\si}} \rangle  \\
	= & \langle\phi(R^{\prime}) , \tilde{v}_{\dot{_{\mathsf{b}}}} , \lambda_{\dot{_{\mathsf{r}}}}  \vert \phi(R) , \tilde{v}_{\dot{\ai}} , \lambda_{\dot{\si}} \rangle  ~ \langle R^{\prime} ,  v_{\mathsf{b}} , \lambda_{\mathsf{r}} \vert R ,  v_{\ai} , \lambda_{\si} \rangle = \left( \delta_{R R^{\prime}} \delta_{\ai \mathsf{b}}  \delta_{\lambda_{\mathsf{s}} \lambda_{\mathsf{r}}} \right) \left( \delta_{\phi(R) \phi(R^{\prime})} \delta_{\dot{\ai}\dot{\mathsf{b}}} \delta_{\lambda_{\dot{\si}} \lambda_{\dot{\mathsf{r}}} }  \right).
\end{aligned}
\end{equation}
To reduce clutter, we introduce the notation
\begin{equation}
	\lvert R , \ai ,\si ;\ \phi (R) , \dot{\ai},\dot{\si} \rangle := \lvert R ,  v_{\ai} , \lambda_{\si} \rangle \otimes \lvert \phi(R) , \tilde{v}_{\dot{\ai}} , \tilde{\lambda}_{\dot{\si}} \rangle
\end{equation}
and $\sum_{\ai,\si,\dot{\ai},\dot{\si}}$ to indicate the sum over the elements of the basis \eqref{eq:vbasisgenR} at given $R$.\par
The Hamiltonian acts as 
\begin{equation}
	H \lvert R , \ai, \si;\ \phi (R) , \dot{\ai},\dot{\si} \rangle  = \left\lvert R  \right\rvert ~\lvert R , \ai, \si ;\ \phi (R) , \dot{\ai},\dot{\si} \rangle   
\end{equation}
where we recall that $\lvert R \rvert = \sum_{i \ge 1} R_i$. The fact that the eigenvalues are independent of the generators of the representation $R$, producing degeneracy in the spectrum, is imposed by the assumption of the existence of global symmetries in the gauge theory.\par
\begin{defin}
    Let $\mathfrak{R}_L^{(N)} $ be as in Subsection \ref{sec:charexprep} and $\mathscr{H}_L$ as in \eqref{eq:splitHLintoR}. We define the \emph{gauge-invariant Hilbert space} to be 
\begin{equation}
\label{eq:HLNintoRLN}
	\mathscr{H}_L ^{(N)} := \bigoplus_{R \in \mathfrak{R}_L^{(N)} } \mathscr{H}_L (R) \otimes  \mathscr{H}_L (\phi(R)) .
\end{equation}
\end{defin}
The definition is motivated as follows.
\begin{itemize}
    \item Note that $\mathscr{H}_L ^{(N)}$ is not simply a vector space, but inherits a Hilbert space structure from $\mathscr{H}_L $, where the completeness is a consequence of the discreteness of the representation spectrum.
    \item The restriction to $R \in \mathfrak{R}_L^{(N)}$ stems from having a finite gauge rank $N$ in the gauge theory with start with. We loosely refer to this constraint as stemming from gauge invariance, with a slight abuse of notation.
\end{itemize}
Putting the pieces together, we have the following.
\begin{thm}
 Under the notation above, it holds that 
 \begin{equation}
\label{eq:ZisUMM}
    \mz_L ^{(N)} (e^{-\beta}) = \mathrm{Tr}_{\mathscr{H}_L ^{(N)}} \left( e^{- \beta H} \right) .
\end{equation}
\end{thm}
\begin{proof}
At any fixed $L$, we use \eqref{eq:HLNintoRLN} and take the trace in the energy eigenbasis \eqref{eq:vbasisgenR}. This gives 
\begin{equation}
\begin{aligned}
	 \mathrm{Tr}_{\mathscr{H}_L^{(N)}} \left( e^{- \beta H} \right) & = \sum_{R \in \mathfrak{R}_{L} ^{(N)}} \sum_{\ai,\si,\dot{\ai},\dot{\si}}  \langle R, \ai,\si;\ \phi (R) ,\dot{\ai},\dot{\si}  \lvert e^{- \beta H} \rvert R, \ai,\si;\ \phi (R) ,\dot{\ai},\dot{\si}  \rangle \\
	 	& = \sum_{R \in \mathfrak{R}_{L} ^{(N)}} \sum_{\ai=1} ^{\dim R} \sum_{\si=1} ^{\mathfrak{d}_R} \sum_{\dot{\ai}=1} ^{\dim \phi(R)} \sum_{\dot{\si}=1} ^{\tilde{\mathfrak{d}}_R}  e^{- \beta \lvert R \rvert } \\
   & = \sum_{R \in \mathfrak{R}_{L} ^{(N)}} \mathfrak{d}_R \dim (R) \tilde{\mathfrak{d}}_R \dim (\phi(R))  ~ e^{- \beta \lvert R \rvert }  .
\end{aligned}
\end{equation}
\end{proof}
We emphasize that, despite the apparent simplicity of the Hamiltonian $H$, these systems behave very differently from free systems. One major distinction is the degeneracy factor $\dim R$, which grows fast with the energy $\lvert R\rvert$.
\begin{defin}
Let\footnote{We adopt the notation E (east) and W (west), instead of the more customary left and right, to avoid confusion with the symbols $L,R$.} $ \mathscr{H}_L^{\mathrm{E}} ,\mathscr{H}_L^{\mathrm{W}} \cong \mathscr{H}_L$ be two copies of $\mathscr{H}_L$, and consider the tensor product space $\mathscr{H}^{\mathrm{E}}_L \otimes \mathscr{H}^{\mathrm{W}}_L$. We endow it with the canonical tensor product basis and promote the Hamiltonian $H$ on $\mathscr{H}_L$ to $H \otimes 1$ on $\mathscr{H}^{\mathrm{E}}_L \otimes \mathscr{H}^{\mathrm{W}}_L$. The \emph{thermofield double state} of the system, at inverse temperature $\beta$, is 
\begin{equation}
\label{eq:TFDL}
	\lvert \Psi_{\beta} \rangle_L := \frac{1}{\sqrt{ \mz_L ^{(N)} } }  \sum_{R \in \mathfrak{R}_{L} ^{(N)}} \sum_{\ai ,\si ,\dot{\ai},\dot{\si}} e^{- \frac{\beta}{2} \lvert R \rvert } ~ \lvert R, \ai , \si ;\ \phi (R) ,\dot{\ai},\dot{\si}  \rangle^{\mathrm{E}} \otimes \lvert R, \ai,\si;\ \phi (R) ,\dot{\ai},\dot{\si}  \rangle^{\mathrm{W}} .
\end{equation}
\end{defin}
Of course, this definition only makes sense if $\mz_L ^{(N)}$ is finite, which is expected at finite $N$ but not at large $N$. The usefulness of the thermofield double state is to compute the thermal expectation values of any operator $\hat{\mathcal{O}}$, through the identity 
\begin{equation}
	{}_L \langle \Psi_{\beta} \lvert \hat{\mathcal{O}} \rvert \Psi_{\beta} \rangle_L = \frac{1}{\mz_L^{(N)} }  \mathrm{Tr}_{\mathscr{H}_L^{(N)}} \left( e^{- \beta H}~ \hat{\mathcal{O}} \right) .
\end{equation}\par
\medskip

The construction so far is summarized as:
\begin{equation}
\begin{tabular}{r c l}
\textsc{Gauge + flavor symmetry} & \hspace{8pt} & \textsc{Quantum mechanics}\\
\hline
flavor rank $L$ & $\quad \longleftrightarrow \quad $ & choice of Hilbert space \\
flavor symmetry generators & $\quad \longleftrightarrow \quad $ & microscopic states \\
gauge rank $N$ & $\quad \longleftrightarrow \quad $ & constraint on allowed states \\
objects in representation ring & $\quad \longleftrightarrow \quad $ & superselection sectors \\
\hline
\end{tabular} 
\label{dictionaryrepQM}
\end{equation}\par

\subsubsection{Global symmetry representations and meson-like formulation}
\label{sec:QMFreeGas}
In preparation for Subsection \ref{sec:probe}, we now proceed to reformulate the quantum system in terms of more standard physical operators. This is an insightful rewriting, although not strictly necessary for the construction.\par
We start with the trivial representation, which provides us with the vacuum $\lvert \emptyset; \emptyset \rangle$. We then identify the number $\lvert R \rvert $ with the number of particles created on top of the vacuum. On the other hand, $\lvert R \rvert$ is total the number of boxes in the Young diagram $R$, thus 
\begin{equation}
	\text{ \# of particles } = \lvert R \rvert = \text{ \# of boxes}.
\end{equation}
All the one-particle states are obtained tensoring the vacuum with the representation $(\Box, \phi (\Box))$. 
In this way we identify the state 
\begin{equation}
    \lvert \Box, \ai ; \phi (\Box) , \dot{\ai} \rangle , \qquad \ai =1, \dots, \dim \Box , \ \dot{\ai} =1, \dots, \dim \phi(\Box )
\end{equation}
with a \emph{meson-like} state.
We emphasize that, when $\phi (R)= \overline{R}$, $\lvert \Box, \ai ; \overline{\Box} , \dot{\ai} \rangle$ resembles the standard meson, with one index $\ai$ in the fundamental and the other index $\dot{\ai}$ in the anti-fundamental, combined in a gauge-invariant fashion.\par
Iterating this procedure and tensoring with arbitrary symmetric powers of $(\Box, \phi (\Box))$ one ends up constructing all the necessary states, which is a consequence of the isomorphism 
\begin{equation}
    \bigoplus_{n=0} ^{\infty} \mathrm{Sym}^{n} (\Box, \phi (\Box)) \cong \bigoplus_{R \in \mathfrak{R}^{SU(L+1)}} (R, \phi (R)) ,
\end{equation}
which holds as well replacing $R \in \mathfrak{R}^{SU(L+1)}$ with $R \in \mathfrak{R}^{U(L)}$ on the right-hand side and interpreting $\Box$ on the left-hand side as the fundamental representation of $U(L)$.

\subsubsection{The partition function as a discrete matrix model}
\label{sec:RepToFermi}

In Subsection \ref{sec:charexprep}, the character expansion has been used to repackage the gauge-invariant operators into an ensemble of irreducible representations, or of Young diagrams. We now proceed to rewrite this system as a matrix model in which the eigenvalues live on the lattice $\mathbb{N}^L$, called \emph{discrete matrix models}. This presentation is more suitable for computations.\par

\begin{prop}[\cite{Rusakov:1992uf,Douglas:1993iia}]
 For every $L \in \mathbb{N}$, there exists an injective map $\iota: \ \mathfrak{R}_L ^{(N)} \to \mathbb{N}^{L+1}$ such that the equality holds:
\begin{equation}
\label{eq:genericdiscreteMM}
	\mz_L ^{(N)} (y) =  y^{-\frac{L^2}{2}}  \sum_{\vec{h} \in \mathfrak{H}_L ^{(N)} } \frac{y^{\sum_{i=1} ^{L} \left( h_i + \frac{1}{2} \right)} }{L! G(L+2)^2} ~ \mathfrak{d}_{\vec{h}} ^2 \prod_{1 \le i<j \le L+1} (h_i - h_j)^2 ,
\end{equation}
where $G(\cdot)$ is Barnes's $G$-function \cite{Barnes}, $\mathfrak{H}_L ^{(N)} := \iota \left( \mathfrak{R}_L ^{(N)} \right) \subseteq \mathbb{N}^{L}$, $\vec{h} =(h_1, \dots, h_L, -1) \in \mathbb{N}^{L+1}$, and $\mathfrak{d}_{\vec{h}} $ denotes the image of $ \mathfrak{d}_{R}$ under $\iota$.
\end{prop}
We have restricted to the case $\phi (R) =R$ or $\phi (R) =\overline{R}$, but the argument can be easily adapted to a more general situation. In that case, the discrete matrix model is non-standard, because the squared Vandermonde factor gets modified.\par
Besides, \eqref{eq:genericdiscreteMM} is stated for $SU(L+1)$, but the adaptation to $U(L)$ is straightforward, the only difference residing in the value of $\dim R$.\par
\begin{proof}
The equivalence passes through the change of variables 
\begin{equation}
\label{eq:changeRtoH}
    h_i = R_i -i + L ,
\end{equation}
with the new variables $h_i$ satisfying 
\begin{equation}
\label{eq:orderedh}
	h_1 > h_2 > \cdots > h_L \ge 0 .
\end{equation}
The map $\iota$ is thus specified by \eqref{eq:changeRtoH}. Direct computation gives $\frac{L^2}{2} + \lvert R \rvert = \frac{L}{2} + \sum_{i=1}^{L} h_i $. It remains to express $\dim R$ as \cite{Macdonaldbook}
\begin{equation}
\label{eq:dimRformula}
    \dim R = \prod_{1 \le i < j \le L+1} \left( \frac{ R_i -R_j - i + j}{j-i} \right) .
\end{equation}
The standard substitution \eqref{eq:changeRtoH} casts $\dim R$ in terms of the Vandermonde determinant:
\begin{equation}
	\dim R = \frac{1}{G(L+2)} \prod_{1 \le i<j \le L+1} (h_i - h_j) ,
\end{equation}
where the Barnes's $G$-function at integer argument satisfies $G(L+2)= \prod_{n=1}^{L} (n!)$, and with the understanding $h_{j>L} =0$. For $U(L)$ representations we have 
\begin{equation}
\label{eq:dimRVdm}
	\dim R = \frac{1}{G(L+1)} \prod_{1 \le i<j \le L} (h_i - h_j) .
\end{equation}\par
The last step of the proof is to notice that the summand is totally symmetric in the variables $h_i$, and vanishes whenever two are equal, due to the Vandermonde determinant. Therefore, one can lift the restriction \eqref{eq:orderedh} and divide by the order of the symmetric group of $L$ elements, i.e. $L!$, to remove the over-counting.\par
The proof we have given here is a generalization of the one in \cite{Douglas:1993iia}. The result was certainly known in the matrix model community, and our only contribution it to cast it in a very general form.
\end{proof}

The discreteness of the matrix ensemble \eqref{eq:genericdiscreteMM} and the presence of additional constraints will generically induce large $N$ third order phase transitions. This is a recurrent theme, initiated in \cite{Douglas:1993iia}. This property is nothing but the mirror of the ubiquitous third order phase transitions in the unitary one-matrix models we have started with. We will see that, in several cases, passing to a larger ensemble will promote the phase transition from third to first order, with Hagedorn-like behavior. We elaborate on this comment in Subsection \ref{sec:1PTargument}, and in Appendix \ref{sec:PTMM} from the point of view of the unitary matrix model.\par
The two-step dictionary that emerges, in going from the gauge theory to the discrete matrix model passing through the ensemble of representations, is:
\begin{equation}
\begin{tabular}{r c l}
\textsc{Gauge + flavor symmetry} & \hspace{8pt} & \textsc{Discrete matrix model} \\
\hline
gauge rank $N$ & $\quad \longleftrightarrow \quad $ & constraint on allowed configurations \\
flavor rank $L$ & $\quad \longleftrightarrow \quad $ & \# of eigenvalues $L$ \\
gauge-invariant state & $\quad \longleftrightarrow \quad $ & eigenvalue configuration \\
\hline
\end{tabular} 
\label{dictionaryfermiongauge}
\end{equation}\par

\subsection{Coupling to a probe}
\label{sec:probe}

Now that the Hilbert space of our system and a Hamiltonian have been specified, we recall our main objective: constructing a large $N$ von Neumann algebra of type III$_1$. As we have proven earlier in Section \ref{sec:vNtot}, one can do so by constructing observables satisfying large $N$ factorization such that the spectral density has continuous support. This is the goal of this subsection. The strategy will be to couple a probe operator to the system and to explicitly calculate its real-time correlation functions. The correlation functions of the probe will then be described by a large $N$ von Neumann algebra of type III$_1$. Importantly, constructing this large $N$ von Neumann algebra will only require the correlation functions of the probe. The idea is that just like in $\mathcal{N}=4$ super-Yang--Mills, the correlation functions of one single trace operator allow to probe the emergence of spacetime, here the correlation functions of our probe will give information on the emergent properties of the full quantum system in the large $N$ limit. \par

\subsubsection{Probe operators}

Very schematically, the idea is to define an auxiliary object that creates a probe excitation at time $0$ and let it propagate until it is annihilated at a later time $t>0$. The probe particle is a bosonic particle in the representation $(\Box,\emptyset)$. Of course, the alternative choice $(\emptyset,\overline{\Box} )$ is equally valid and can be dealt with in exactly the same manner. The probe is coupled to the system by tensoring the corresponding Hilbert spaces.
\begin{defin}
    Let 
    \begin{equation}
    \Gamma_{\mathrm{probe}} := \bigoplus_{n=0}^{\infty} \mathrm{Sym}^{n} (\Box, \emptyset)
    \end{equation}
    denote the bosonic Fock space of a probe in the fundamental representation. The $n=0$ component is the trivial vector space, consisting only of the vacuum. The probe is said to be coupled to the system if $\mathscr{H}_L ^{(N)}$ is replaced with the total Hilbert space
    \begin{equation}
    \label{eq:coupledprobeH}
    \begin{aligned}
        \mathscr{H}_L ^{\mathrm{tot}} & := \mathscr{H}_L^{(N)} \otimes \Gamma_{\mathrm{probe}} \\
        & \cong \left[ \mathscr{H}_L^{(N)} \otimes ( \emptyset , \emptyset )_{\rm probe} \right] \oplus  \left[ \mathscr{H}_L^{(N)}  \otimes (\Box, \emptyset)_{\rm probe} \right] \oplus \cdots 
     \end{aligned}
     \end{equation}
\end{defin}
\begin{defin}\label{def:ProbeApprox1}
    The \emph{probe approximation} consists in neglecting all terms in \eqref{eq:coupledprobeH} other than the one giving the first non-trivial contribution to the correlation functions.
\end{defin}
In practice, the probe approximation can be enforced by modifying the Hamiltonian adding a large mass term on $\Gamma_{\mathrm{probe}}$. We will provide a more formal definition below, in Definition \ref{def:probeapprox2}.\par
The $n=0$ component of $\Gamma_{\mathrm{probe}} $ consists only of the vacuum 
\begin{equation}
    \lvert \emptyset; \emptyset \rangle_{\rm probe} \cong \lvert \underbrace{0, \dots , 0}_{L} ; \underbrace{0, \dots , 0}_{L} \rangle_{\rm probe} .
\end{equation}
Here and in what follows, the semicolon denotes the tensor product between the left and the right for the probe, while the commas separate the occupation numbers of the different modes in the one-particle Hilbert space.
We define the creation and annihilation operators $ c_p ^{\dagger} , c_p$ acting on the probe sector $\Gamma_{\mathrm{probe}}$. The lowest non-trivial probe sector is the $n=1$ term in $\Gamma_{\mathrm{probe}}$, consisting of 
\begin{equation}
    \left\{ \left( c_p ^{\dagger} \otimes 1 \right) \lvert 0, \dots , 0  ;  0, \dots , 0  \rangle_{\rm probe} , \ \forall p =1, \dots, L \right\} .
\end{equation}
The full $\Gamma_{\mathrm{probe}}$ is built acting with $c_p ^{\dagger} \otimes 1$ on the Fock vacuum of the probe. We emphasize that $c^{\dagger}_p, c_p$ act on the probe sector rather than $\mathscr{H}_L ^{(N)}$. We therefore indicate them with a different letter. Nevertheless, the two sets of operators behave in the same way when expressed in terms of $SU(L+1)$ representations.\par
In the composite system, in which the probe is coupled to the rest, the push-forward of the creation operator is 
\begin{equation}
    (1 \otimes 1)  \otimes \left( c_p ^{\dagger} \otimes 1 \right)_{\mathrm{probe}}
\end{equation}
where the first parenthesis is the identity operator on $\mathscr{H}_L$ (or its restriction to the subspace $\mathscr{H}_L^{(N)}$), that decomposes into $1 \otimes 1$ on the $R$-graded terms $\mathscr{H}_L (R) \otimes \mathscr{H}_L (\phi (R))$ in \eqref{eq:splitHLintoR}.

\subsubsection{Hamiltonian and interactions}
\label{sec:Hint}
The next step is to promote the Hamiltonian $H$ on $\mathscr{H}_L$ to a Hamiltonian $H^{\prime}$ on the combined Hilbert space $\mathscr{H}_L \otimes \Gamma_{\mathrm{probe}}$ of the system coupled to the probe. This Hamiltonian must necessarily contain a term $H \otimes 1_{\mathrm{probe}}$. We would also like to add the tensor product of the identity on $\mathscr{H}_L$ and a large mass term acting only on $\Gamma_{\mathrm{probe}}$, and add an interaction term between the probe and the rest of the system.
In order to be able to discuss possible interaction terms more precisely, we will make the following assumption:

\begin{itemize}
\item[$\diamond$] In what follows, we assume \eqref{eq:exchangeUY}, and moreover take $\phi(R)=R$. The situation with $\phi (R)= \overline{R}$ can easily be retrieved from the ensuing discussion, with minor variations. The generic situation for arbitrary $\phi (R)$ can also be worked out along the same lines, but is not immediate, as it requires a more careful analysis of allowed interaction terms.
\end{itemize}
Recall the Hilbert space of the composite system of the probe and our quantum theory in \eqref{eq:coupledprobeH}, given by 
\begin{align}
\label{eq:Hilbtot}
\mathscr{H}^{\mathrm{tot}}_L =\mathscr{H}_L^{(N)} \otimes\Gamma_{\mathrm{probe}},
\end{align}
where $\Gamma_{\rm probe}$ is the Fock space of the probe. Define the free Hamiltonian operator on $\Gamma_{\rm probe}$:
\begin{align}
H_{0, \mathrm{probe}}:= \mu \sum_{p=1}^{L} \left( c_p^\dagger c_p \otimes 1\right)_{\mathrm{probe}}  ,
\end{align}
where $\mu >0$ has the meaning of a mass for the probe, and we will later assume $\mu \gg 1$ (and also larger than all other relevant scales in the problem), which enforces the probe approximation.\par
Then, without coupling the probe to the gauge theory, one can define the Hamiltonian associated to the two systems as an operator on $\mathscr{H}^{\mathrm{tot}}_L$:
\begin{align}
H_{\mathrm{decoupled}}:= 1\otimes H_{0, \mathrm{probe}}+H \otimes 1,
\end{align}
where $H$ is the Hamiltonian of our quantum mechanical theory from Subsection \ref{sec:QMflavor}, which satisfies $H(R,\phi(R))=|R|$, where $H(R,\phi(R))$ denotes the restriction of $H$ to $\mathscr{H}_L (R) \otimes \mathscr{H}_L (\phi (R))$.\par 
Clearly, letting the system in $\mathscr{H}^{\mathrm{tot}}_L$ evolve with $H_{\mathrm{decoupled}}$ is a trivial operation, as the probe and the system do not talk to each other. In order to witness a non-trivial behavior and the ensuing appearance of a large $N$ type III$_1$ algebra, the probe must interact with our quantum system. Hence the full Hamiltonian reads:
\begin{equation}
\label{eq:defHprime}
H^{\prime} := 1 \otimes H_{0, \mathrm{probe}}+H \otimes 1 +H_{\mathrm{int}},
\end{equation}
where $H_{\mathrm{int}}$ is an interaction Hamiltonian that we now define directly inside the tensor product Hilbert space $\mathscr{H}^{\mathrm{tot}}_L$. If our models were derived top down in string theory, $H_{\mathrm{int}}$ would be determined by the action on the probe brane. Instead, our construction of the quantum mechanics is bottom up, and we will introduce the interaction by hand, selecting a tractable sample of all the possible interactions $H_{\mathrm{int}}$.\par
\medskip
The idea is that in the original quantum Hilbert space (without adding the probe), all pairs of representations appearing are of the form $(R,\phi(R))$. However, once one tensors these with the symmetric powers of the fundamental representation of $SU(L+1)$, which corresponds to the one-particle Hilbert space of the probe, more generic pairs of representations start appearing in the Hilbert space. This motivates the definition of $H_{\mathrm{int}}$ as an operator that is zero on pairs of the form $(R,\phi(R))$, and nonzero on other pairs $(R_1,R_2)$ (upon identifying $\mathscr{H}^{\mathrm{tot}}_L$ with a Hilbert space spanned by representations of the form $(R_1,R_2)$, where $R_2$ is some $\phi(R)$, and $R_1$ is an irreducible representation that appears in the decomposition of the tensor product of some symmetric power of the fundamental representation with $R$). It should also respect $SU(L+1)$ symmetry.\par

With this in mind, one possible interaction term is 
\begin{equation}
\label{eq:Hint}
	H_{\text{int}} (R_1, R_2) = \frac{g}{2} \left[ Q (R_1) - Q (R_2) \right] , \qquad Q(R) := C_2 (R) + (L+1) \lvert R \rvert ,
\end{equation}
where $C_2$ is the quadratic Casimir. We will work with \eqref{eq:Hint} in the rest of this part. Of course, other similar choices can be made, for example by switching $R$ with $\phi(R)$, or by switching signs. For example, to ease the comparison with \cite{Iizuka:2008eb}, we will make a slightly different choice of interaction Hamiltonian in Section \ref{sec:IOP}, see \eqref{eq:HintIOPEx}.\par
The coupling $g$ sets the strength of the interaction. Coming from a matrix model, it will be natural to consider a 't Hooft limit with $g \propto 1/L$ in the large $N$ and large $L$ limit. The choice of the factor $L\lvert R \rvert $ is merely to avoid some cumbersome shifts in the ensuing expressions. Dropping it, or changing its coefficient, would not alter the conclusions.

\begin{itemize}
\item[$\diamond$] Higher Casimir invariants may be included in $H_{\text{int}} $, as well as more sophisticated interactions. We restrict ourselves to the simplest non-trivial interaction. The choice is motivated by the sake of tractability and an analogy with \cite{Iizuka:2008eb}.
\end{itemize}\par
Now that we have introduced the Hamiltonian, we can write down a more precise version of Definition \ref{def:ProbeApprox1}.
\begin{defin}\label{def:probeapprox2}
    Consider the $\C$-valued map 
    \begin{equation}
        \hat{\mathcal{O}} \ \mapsto \ \tr_{\mathscr{H}_L ^{\mathrm{tot}} } \left( e^{- \beta H^{\prime}} ~ \hat{\mathcal{O}} \right) , \qquad \qquad \hat{\mathcal{O}} \in \mathcal{L} \left( \mathscr{H}^{\mathrm{tot}}_L \right) 
    \end{equation}
    and write the right-hand side in the form 
    \begin{equation}
    \label{eq:correlsumprobe}
        \tr_{\mathscr{H}_L ^{\mathrm{tot}} } \left( e^{- \beta H^{\prime}} ~ \hat{\mathcal{O}} \right)  =  \sum_{\vec{n} \in \N^{L+1} } e^{- \beta \mu \lvert \vec{n} \rvert } ~ \tr_{\mathscr{H}_L ^{(N)} } \left( {}_{\text{\rm probe}} \langle \vec{n} ; \emptyset \rvert e^{- \beta \left( H + H_{\text{\rm int}} \right) }  \hat{\mathcal{O}} \lvert \vec{n} ; \emptyset \rangle_{\text{\rm probe}} \right) 
    \end{equation}
    where we have used $\dim \Box =L+1$ for $SU(L+1)$ and introduced the shorthand notations 
    \begin{equation}
        \vec{n}  = (n_1, \dots, n_{L+1} ) , \qquad \lvert \vec{n} ; \emptyset \rangle_{\text{\rm probe}} := \lvert n_1, \dots , n_{L+1} ; 0, \dots , 0 \rangle_{\text{\rm probe}} , \qquad \lvert \vec{n} \rvert := \sum_{i=1} ^{L+1} n_i
    \end{equation}
    and the Hamiltonian \eqref{eq:defHprime}. The outer sum runs over all states in the probe sector. The partition functions \eqref{eq:correlsumprobe} and the ensuing correlation functions are treated as series expansions in the parameter $e^{- \beta \mu} $. The \emph{probe approximation} consists in discarding all contributions except the lowest order one.
\end{defin}
This definition based on the series expansion in the parameter $e^{- \beta \mu}$ formalizes the physical intuition that a probe is an object with a large mass $\mu \gg 1$, so that the excited states are not accessible. Note that the order of limits will be important later. We first take $\beta \mu \gg 1$ with every other parameter fixed. Then, when considering the large $N$ limit, we will take $N,L\gg 1$ but assuming $\mu$ is still large enough with respect to $L$. In the continuation we make two remarks on this definition.
\begin{itemize}
    \item Notice that the probe approximation is a different concept than a cutoff at energy scales $ O(\mu)$, and when the probe is decoupled from the system, the latter is not affected.
    \item Retaining only the lowest order contribution to the correlation functions, as per Definition \ref{def:probeapprox2}, \emph{is not} the same as setting $\vec{n}=(0, \dots, 0)$ in \eqref{eq:correlsumprobe}. While this latter prescription will often work, it may happen that one needs to compute correlation functions of the form 
    \begin{equation}
        e^{\beta \mu} ~ \tr_{\mathscr{H}_L ^{\mathrm{tot}} } \left( e^{- \beta H^{\prime}} ~ \hat{\mathcal{O}} \right) 
    \end{equation}
    for some operator $\hat{\mathcal{O}}$ that annihilates the probe vacuum. In such a scenario, the contribution from $\vec{n}=(0, \dots, 0)$ trivializes, whereas the first excited states with $\lvert \vec{n} \rvert =1$ will yield the first non-trivial contribution,
    \begin{equation}
    \begin{aligned}
        & e^{\beta \mu} ~ \tr_{\mathscr{H}_L ^{\mathrm{tot}} } \left( e^{- \beta H^{\prime}} ~ \hat{\mathcal{O}} \right) \\
        =&  e^{\beta \mu} \cdot  \left\{ 0 \ + \ e^{- \beta \mu} ~ \sum_{\substack{\vec{n} \in \N^{L+1} \\ \lvert \vec{n} \rvert =1} }  \tr_{\mathscr{H}_L ^{(N)} } \left( {}_{\text{\rm probe}} \langle   \vec{n} ; \emptyset \rvert  e^{- \beta \left( H + H_{\text{\rm int}} \right) } ~\hat{\mathcal{O}} \lvert \vec{n} ; \emptyset \rangle_{\text{\rm probe}} \right)  \ + \ O(e^{- 2\beta \mu})  \right\} \\
        = & \sum_{\substack{\vec{n} \in \N^{L+1} \\ \lvert \vec{n} \rvert =1} }  \tr_{\mathscr{H}_L ^{(N)} } \left( {}_{\text{\rm probe}} \langle   \vec{n} ; \emptyset \rvert e^{- \beta \left( H + H_{\text{\rm int}} \right) } ~ \hat{\mathcal{O}} \lvert \vec{n} ; \emptyset \rangle_{\text{\rm probe}} \right)  \ + \ O(e^{- \beta \mu}) .
    \end{aligned}
    \end{equation}
\end{itemize}\par

\subsubsection{Correlation functions}
\label{sec:probecorrel}

We wish to calculate real-time two-point correlation functions of the probe. We now introduce new operators whose correlations can be computed easily. They are related to the $c_p$'s by a $U(L)$ rotation, and consequently, by flavor symmetry, the correlations of these operators will be the same as the ones of any creation or annihilation operator of the probe in a given basis. However it will be useful for computations to write operators down in a form that treats all basis vectors democratically. With the ingredients defined insofar, we hence introduce the operator
\begin{equation}
\label{eq:defOfromcp}
    \mathcal{O}_L := \frac{1}{\sqrt{L+1}} \sum_{p=1} ^{L+1} (1 \otimes 1)  \otimes \left( c_p ^{\dagger} \otimes 1 \right)_{\mathrm{probe}} 
\end{equation}
on $\mathscr{H}^{\mathrm{tot}}_L$. Likewise, the annihilation operator is $ \mathcal{O}^{\dagger}_L = \frac{1}{\sqrt{L+1}} \sum_{p=1}^{L+1} (1 \otimes 1) \otimes \left( c_p  \otimes 1 \right)_{\mathrm{probe}}$.\par
As the Hamiltonian of the systems respects the flavor symmetry, it is clear that these operators have correlation functions equal to the ones of any $c_p^\dagger$. They also satisfy the commutation relations:
\begin{equation}
    \left[ c_p, c_q ^{\dagger} \right] = \delta_{p,q}  \quad \Longrightarrow \quad \left[ \mathcal{O}_L ^{\dagger}, \mathcal{O}_L \right] = 1 .
\end{equation}
The normalization by $1/\sqrt{L+1}$ in \eqref{eq:defOfromcp} ensures that the correlation functions are properly normalized. The reason why we introduced the $\mathcal{O}^{\dagger}_L$ instead of calculating correlation functions directly at the level of the $c_p^\dagger$ is because it will be useful to have explicit sums over flavor indices in our calculations.\par
In order to lighten the notation, we will omit the subscript on the probe part, and identify $c_p ^{\dagger}$ with $(c_p ^{\dagger} \otimes 1)_{\rm probe} $, so to write 
\begin{equation}
    \mathcal{O}_L = \frac{1}{\sqrt{L+1}} \sum_{p=1} ^{L+1} 1 \otimes  c_p ^{\dagger} .
\end{equation}
The left operator in the tensor product acts on $\mathscr{H}_L ^{(N)}$ and the right acts on $\Gamma_{\rm probe}$.\par
\medskip
Equipped with the Hamiltonian $H^{\prime}$ from \eqref{eq:defHprime}, we generate the time evolution in the Heisenberg picture of the operators in 
\begin{equation}
     \mathrm{End} \left( \mathscr{H}^{\mathrm{tot}}_L \right) \cong \mathrm{End} \left( \mathscr{H}_L ^{(N)} \otimes \Gamma_{\mathrm{probe}} \right) .
\end{equation}
What we will require is the time-evolved annihilation operator at time $t$:
\begin{equation}
    \mathcal{O}^{\dagger}_L (t) = e^{i t H^{\prime}} \left( \frac{1}{\sqrt{L+1}} \sum_{p=1} ^{L+1} 1 \otimes c_p \right) e^{-it H^{\prime}} .
\end{equation}
The correlation functions of time-ordered products of operators $\mathcal{O}_L, \mathcal{O}^{\dagger}_L$ inserted at different times provide well-posed observables of the quantum system. Observe that 
\begin{itemize}
\item Correlation functions will be non-vanishing only if an equal number of $\mathcal{O}_L$ and $\mathcal{O}^{\dagger}_L$ is taken. This is of course a consequence of the simplicity of the probe operators, and differs from the more sophisticated models of e.g. \cite{Witten:2021unn,Chandrasekaran:2022eqq}.
\item In the two-point functions, i.e. in the expectation values of operators 
\begin{equation}
    \mathcal{O}^{\dagger}_L (t)  \mathcal{O}_L (0) = \frac{1}{L+1} \sum_{p,q=1}^{L+1} e^{i t H^{\prime}} \left(  1 \otimes c_q \right) e^{-it H^{\prime}} (  1 \otimes c_p ^{\dagger} ) 
\end{equation}
only terms with $q=p$ will yield a non-vanishing contribution, by flavor symmetry. 
\item We also notice the normalization by $\dim \Box =L+1$, that keeps the correlation functions properly normalized. To work with $U(L)$ instead of $SU(L+1)$, one lets $p$ run from $1$ to $L$ and normalizes $\mathcal{O}_L$ by $1/\sqrt{L}$.
\end{itemize}
We arrive at the following statement. 
\begin{lem}\label{prop:sendprobe}
    The two-point function of a probe is computed by the finite temperature expectation value of $\mathcal{O}^{\dagger}_L (t)  \mathcal{O}_L (0) \in \mathcal{L} \left( \mathscr{H}^{\mathrm{tot}}_L \right) $, evolved with Hamiltonian \eqref{eq:defHprime}. Explicitly:
    \begin{equation}
    \label{eq:sumpapadagp}
        \mathcal{O}^{\dagger}_L (t)  \mathcal{O}_L (0) \qquad \text{ and } \qquad \frac{1}{L+1} \sum_{p=1}^{L+1} e^{i t H^{\prime}} \left(  1 \otimes c_p \right) e^{-it H^{\prime}} \left(  1 \otimes c_p ^{\dagger} \right)
    \end{equation}
    have the same correlation functions.
\end{lem}\par
\begin{itemize}
\item[$\diamond$] In the present discussion, we are neglecting the additional quantum numbers $\lambda_s, \lambda_{\dot{s}}$ in the probe sector. They can be reinstated replacing $c_p ^{\dagger} $ and $c_p$ with the appropriate creation and annihilation operators that create a state decorated with additional quantum numbers, but we do not sum over them.
\end{itemize}\par
Finally, to make contact with the conventions used in the study of von Neumann algebras, we introduce the self-adjoint operator
\begin{equation}
\label{eq:defphifromO}
    \phi_L := \frac{1}{\sqrt{2}} \left( \mathcal{O}^{\dagger}_L + \mathcal{O}_L \right) .
\end{equation}
In terms of the probe creation and annihilation operators, $\phi_L$ reads 
\begin{equation}
    \phi_L = \sum_{p=1} ^{L+1} \frac{c_p + c^{\dagger}_p}{\sqrt{2 (L+1)}} .
\end{equation}

\subsubsection{Intermezzo: On creation and annihilation operators}

Let us clarify a subtlety about the creation and annihilation operators from the get go. In the concrete models we have constructed, there are creation and annihilation operators associated to a probe particle coupled to our quantum systems. It is important to note that these operators, although closely related, are formally \textit{not} the same as the ones appearing in Section \ref{sec:vNtot}.\par
Indeed, in the case of Section \ref{sec:vNtot}, the one-particle Hilbert space is $L^2(\mathbb{R},\rho)$ and corresponds to functions of the time variable (or equivalently, their Fourier transforms), whereas in our quantum systems the one-particle Hilbert space will be indexed by the number of degrees of freedom of the theory. As this number goes to infinity, we will see momentarily that the real-time correlation functions of the system will factorize according to Wick's theorem. What Section \ref{sec:vNtot} tells us is that then, the same correlation functions can be recovered by a quasi-free state on the CCR algebra over $L^2(\mathbb{R},\rho)$. It is in that sense that the two sets of creation and annihilation operators are related, and that one will be able to identify the von Neumann algebras constructed in Section \ref{sec:vNtot} and Appendix \ref{app:CCR} with the large $N$ von Neumann algebras of our quantum systems.

\subsection{Spectral densities}
\label{sec:spectral}

We have shown in Subsection \ref{sec:vNalgebra} that, in order to determine the type of the von Neumann algebras, one needs to calculate the real-time Wightman function at finite temperature, as well as its associated K\"all\'en--Lehmann spectral density $\rho (\omega)$. Here, we introduce appropriate Wightman functions for the class of models of Subsection \ref{sec:QMflavor}, and derive general properties of the associated spectral density. These properties will allow us to determine the type of the large $N$ von Neumann algebras constructed in Section \ref{sec:vNtot}.\par

\subsubsection{Definitions}
We begin introducing the notation, and then proceed with the main statements. Recall the quantum system of Subsection \ref{sec:QMflavor} and the correlation function of creation and annihilation operators from Subsection \ref{sec:probecorrel}. With these ingredients at hand, we introduce the Wightman functions of the systems. Our conventions follow \cite[App.B]{Festuccia:2006sa}.\par
\begin{defin}
The real-time, finite temperature Wightman function is 
\begin{equation}
\label{eq:defGplus}
    G_{L,+} (t) :=  \frac{1}{\mz_L ^{(N)} } \tr_{\mathscr{H}_L ^{\mathrm{tot}}} \left( e^{- \beta H^{\prime}} ~ \phi_L (t)  \phi_L (0)  \right) .
\end{equation}
The dependence on the parameter $N$ and $\beta$ is left implicit in the notation $G_{L,+} (t)$. 
\end{defin}
The importance of this quantity lies in its appearance in the right-hand side of \eqref{eq:rho2ptTFD}. Let $\widetilde{G}_{L,+} (\omega)$ be the Fourier transform of $ G_{L,+} (t)$. From \eqref{eq:RhoandGrel}, $\widetilde{G}_{L,+} (\omega)$ is simply related to the K\"all\'en--Lehmann spectral density $\rho (\omega) $, or \emph{spectral density} for short, through 
\begin{equation}
\label{eq:rhofromG}
    \widetilde{G}_{L,+} (\omega) = \frac{\rho (\omega) }{1-e^{- \beta \omega}} ,
\end{equation}\par
There are related types of Wightman functions that are useful in practice for certain computations \cite[App.B]{Festuccia:2006sa}. 
\begin{defin}
    With the definitions and prescription of Subsection \ref{sec:probe}, the real-time, finite temperature, retarded and advanced Wightman functions are, respectively, 
\begin{align}
    G_{L,\mathrm{R}} (t) & := \hspace{8pt} i \theta (t) \hspace{8pt} \frac{1}{\mz_L ^{(N)}}  \mathrm{Tr}_{\mathscr{H}_L ^{\mathrm{tot}}}\left( e^{- \beta H^{\prime}}  \left[ \phi_L (t) , \phi_L (0) \right] \right) , \label{eq:defGRfunc} \\
    G_{L,\mathrm{A}} (t) & := - i \theta (- t) \frac{1}{\mz_L ^{(N)}} \mathrm{Tr}_{\mathscr{H}_L ^{\mathrm{tot}}} \left( e^{- \beta H^{\prime}}  \left[ \phi_L (t) , \phi_L (0) \right] \right) , \label{eq:defGAfunc}
\end{align}
with $\theta (t)$ the Heaviside step function.
\end{defin}

\begin{defin}\label{def:notationthmrho}
    With the notation as in Subsection \ref{sec:QMflavor}, let $H_{\text{\rm int}}$ be as in \eqref{eq:Hint} and $\lvert \Psi_{\beta} \rangle_L$ be the thermofield double state \eqref{eq:TFDL}. Besides, for every $R \in \mathfrak{R}_L ^{(N)}$ let
    \begin{equation}
    \label{eq:defsetJ}
        \mathscr{J}_R := \left\{ J=1 \right\} \cup \left\{ J \in \left\{2, \dots , L +1 \right\}  \ : \ R_{J-1} > R_{J } \right\} 
    \end{equation}
    and, $\forall J \in \mathscr{J}_R$, denote by $R \sqcup \Box_J$ the Young diagram obtained adding a box at the end of the $J^{\text{th}}$ row of $R$. Let also
    \begin{align}
         E^{\mathrm{int}} _J & := H_{\text{\rm int}} (R \sqcup \Box_J, \phi (R) ) , \\
         E_J (\mu) & := E^{\mathrm{int}} _J + \mu .
    \end{align}
    For the specific $H_{\text{\rm int}}$ in \eqref{eq:Hint} one gets 
    \begin{equation}
        E^{\mathrm{int}} _J := g(R_J -J+L +1) .
    \end{equation}\par
    For every $\omega \in \R$, introduce the shifted variable $\omega_{\mathrm{r}}$ defined as 
    \begin{equation}
    \label{eq:defomegaren}
        \omega_{\mathrm{r}} := \lvert \omega \rvert  - \mu .
    \end{equation}
    Moreover, for $\phi (R) =R$ or $\phi (R) =\overline{R}$, define $\DR (\omega)$ through
    \begin{equation}
    \label{eq:defoprhsrho}
        \langle R, \ai ,\si ; \phi(R) , \dot{\ai}, \dot{\si} \lvert  \DR (\omega) \lvert  R, \ai ,\si ; \phi(R) , \dot{\ai}, \dot{\si}  \rangle := \frac{1}{L} \sum_{J \in \mathscr{J}_R} \delta \left( \omega - E_J (\mu) \right)  ~ \left( \frac{ \dim (R \sqcup \Box_J) }{ \dim R} \right) 
    \end{equation}
    and $\Omega (\omega) $ through 
    \begin{equation}
    \label{eq:defOmega}
        \langle  R, \ai ,\si ; \phi(R) , \dot{\ai}, \dot{\si}  \lvert \Omega (\omega) \rvert  R, \ai ,\si ; \phi(R) , \dot{\ai}, \dot{\si}  \rangle = \frac{1}{L} \sum_{J \in \mathscr{J}_R } \frac{1}{\omega- E_J (\mu) } ~ \frac{\dim (R \sqcup \Box_J) }{  \dim R} .
    \end{equation}
\end{defin}

For later reference, we introduce here the notion of Veneziano limit. It will have a prominent role in the study of the planar limit of the unitary matrix models \eqref{eq:ZLNUMM}.
\begin{defin}\label{def:VenezianoLim}
    Consider a gauge theory with gauge rank $N$ and flavor rank $L$. The \emph{Veneziano parameter} is the ratio 
    \begin{equation}
    \gamma= \frac{L}{N} .
\end{equation}
The \emph{Veneziano limit} is a large $N$ planar limit in which $\gamma$ is kept finite.
\end{defin}\par

\subsubsection{Main results on the spectral density}

The following results are the linchpin of the subsequent derivation. For a cleaner presentation, we have collected the proofs in Appendix \ref{app:longproof}. We encourage the reader to consult that appendix to gain familiarity with the very explicit computations in these models.\par

\begin{thm}\label{thmrhoisrho}
    With the notation as in Definition \ref{def:notationthmrho} and $\theta (\omega)$ the Heaviside step function, the identity 
    \begin{equation}\label{eq:rhoisrho}
      \rho (\omega) = \frac{1}{2} ~ {}_L\langle \Psi_{\beta} \lvert \theta (\omega) \DR (\omega) - \theta (- \omega) \DR (- \omega)  \rvert \Psi_{\beta} \rangle_L 
    \end{equation}
    holds in the probe approximation.
\end{thm}
\begin{proof}
    The proof is instructive but lengthy, thus it is spelled out in detail in Appendix \ref{app:proofthmrho}. It is based on writing $\rho (\omega) = (1- e^{- \beta \omega})  \widetilde{G}_{L,+} (\omega) $, explicitly computing the right-hand side and retaining the lowest order in $e^{- \beta \mu}$. 
\end{proof}
A few remarks on the theorem are in order. 
\begin{itemize}
    \item First, $\rho (\omega)$ is manifestly odd, as it should be.
    \item Second, notice that $\rho (\omega)$ really depends on $\omega_{\mathrm{r}}$ in \eqref{eq:defomegaren}. Recall that the variable $\omega$ measure differences in the energy levels. Due to the probe mass term, the correlation functions are centered around $\mu$; that is, there is a uniform shift by $\mu$ for every $R$. The physically meaningful variable, that probes the interaction of the non-trivial quantum mechanics with the probe, is \eqref{eq:defomegaren}.
    \item Third, we are not including the additional quantum numbers $\lambda_s$ in our probe for simplicity. They can be reinstated by replacing $c_{p} ^{\dagger}$ appropriately in the definition of $\mathcal{O}_L$ in Subsection \ref{sec:probecorrel}. With an appropriate normalization, this modification inserts a ratio $\frac{\mathfrak{d}_{R \sqcup \Box_J }}{\mathfrak{d}_R}$ in the right-hand side of \eqref{eq:defoprhsrho}.
    \item Fourth, an interesting aspect of the proof of \eqref{eq:rhoisrho}, given below, is that the degeneracy factors remain spectators. As a consequence, several extensions are automatically built-in in our formula, including: the replacement $(\dim R)^2 \mapsto (\dim R)^{2-2\mathtt{g}}$ for an arbitrary integer $\mathtt{g} \in \mathbb{N}$, and refining $(\dim R)^2$ into a $q$-deformed or Macdonald measure \cite{Borodin:2011}.
\end{itemize}\par
Next, we give an equivalent characterization of \eqref{eq:rhoisrho}. In practice we will use Theorem \ref{thm:rhoandOmega} for the computations in Part \ref{part2}.
\begin{thm}\label{thm:rhoandOmega}
    With the notation as in Definition \ref{def:notationthmrho}, it holds that
    \begin{align}
    \label{eq:GtildeOmega}
        \widetilde{G}_{L,\mathrm{R}} (\omega + i \varepsilon ) &= - \frac{1}{2} ~  {}_L\langle \Psi_{\beta} \lvert \Omega (\omega + i \varepsilon)  + \Omega (-\omega - i \varepsilon) \rvert \Psi_{\beta} \rangle_L , \\
        \widetilde{G}_{L,\mathrm{A}} (\omega - i \varepsilon ) &= - \frac{1}{2} ~  {}_L\langle \Psi_{\beta} \lvert \Omega (\omega - i \varepsilon)  + \Omega (-\omega + i \varepsilon) \rvert \Psi_{\beta} \rangle_L .
    \end{align}
\end{thm}
\begin{proof} The proof is done in Appendix \ref{app:proofthmrho} by direct calculation.
\end{proof}
The crucial aspect of Theorem \ref{thm:rhoandOmega} is that, thanks to \eqref{eq:rho2ptTFD}, we can equivalently compute $\rho (\omega)$ via the discontinuity equation \cite[Eq.(B.7)]{Festuccia:2006sa}
\begin{equation}
\label{eq:rhodiscG}
    \rho (\omega) = - i \lim_{\varepsilon \to 0^{+}} \left[ \widetilde{G}_{L,\mathrm{R}} (\omega + i \varepsilon ) - \widetilde{G}_{L,\mathrm{A}} (\omega - i \varepsilon ) \right] .
\end{equation}
This expression relates the support of $\rho (\omega)$ to the branch cuts of $ \widetilde{G}_{L,\mathrm{R}}$ and $ \widetilde{G}_{L,\mathrm{A}}$. 
\begin{cor}\label{cor:jumprhoOmega}
    \begin{equation}
    \begin{aligned}
        \rho (\omega) =\frac{i}{2} \lim_{\varepsilon \to 0^{+}} &  \left\{ \theta (\omega)   \left[  {}_L\langle \Psi_{\beta} \lvert \Omega (\omega + i \varepsilon) - \Omega (\omega - i \varepsilon) \rvert \Psi_{\beta} \rangle_L  \right] \right. \\
        & \left. - \theta (- \omega) \left[  {}_L\langle \Psi_{\beta} \lvert \Omega (-\omega + i \varepsilon) - \Omega (-\omega - i \varepsilon) \rvert \Psi_{\beta} \rangle_L  \right] \right\} . 
    \end{aligned}
    \end{equation}
\end{cor}
\begin{proof}
    It follows immediately from Theorem \ref{thm:rhoandOmega} and some rewriting. Noting that the singularities of $\Omega (\omega)$ are located on $\R_{>0}$, only the terms with $+\omega$ will contribute to the discontinuity if $\omega >0$ and only the terms with $-\omega$ will contribute if $\omega <0$.
\end{proof}
\begin{itemize}
    \item We will crucially resort to Corollary \ref{cor:jumprhoOmega} for the explicit calculation of the eigenvalue density in all the examples.
    \item The relation between $\rho (\omega)$ and ${}_L \langle \Psi_{\beta} \vert \Omega (\omega) \vert \Psi_{\beta} \rangle_L$ provided by Corollary \ref{cor:jumprhoOmega} for each fixed $\mathrm{sign}(\omega)$ is (up to a normalization by $\pi$) exactly the relation between any density $\rho (\omega)$ and its Stieltjes transform. That is, $\pi \Omega (\cdot ) $ is the resolvent for the spectral density $\rho (\omega)$.
\end{itemize}
We are now interested in the large $N$ limit. With the change of variables \eqref{eq:changeRtoH}, $\widetilde{G}_{L,\mathrm{R}} (\omega) $, and similarly $\widetilde{G}_{L,\mathrm{A}} (\omega) $, are in turn related via \eqref{eq:GtildeOmega} to the planar resolvent of the matrix model \eqref{eq:genericdiscreteMM} in the Veneziano limit (Definition \ref{def:VenezianoLim}). It is a standard fact that the branch cuts of the latter quantity determine the eigenvalue density of the matrix ensemble. This chain of identities intertwines the two key quantities, the spectral density $\rho (\omega)$ and the eigenvalue density, and implies the following result.\par
\begin{center}
\noindent\fbox{%
\parbox{0.98\linewidth}{%
\begin{cor}\label{c:rholargeN}
Consider the planar Veneziano limit $N\to \infty$ with $\gamma =L/N$ and $\lambda = g L$ fixed. If the discrete matrix model $\mz_L ^{(N)}$ in \eqref{eq:genericdiscreteMM} has a continuous eigenvalue density, $\lim_{N \to \infty} \supp \rho \subseteq \mathbb{R}$ is continuous. 
\end{cor}
}}\end{center}\par
From the perspective of the quantum mechanics of Subsection \ref{sec:QMflavor}, and especially of the matrix model \eqref{eq:genericdiscreteMM}, it is natural to consider a planar limit in which $g$ scales as $1/L$, rather than $1/N$. Insisting on the gauge theory origin of these models, one may prefer to take $N$ as the reference scale. For the current section, our choice will result in a more convenient normalization, but the other choice would have worked equally well, simply differing by $\lambda \mapsto \gamma \lambda $.\par
\medskip
The discontinuity equation \eqref{eq:rhodiscG} establishes a direct correspondence between the support of $\rho (\omega)$ to the branch cuts of $ \widetilde{G}_{L,\mathrm{R}}$ and $ \widetilde{G}_{L,\mathrm{A}}$. The next Lemma expresses these quantities using the eigenvalue densities of the matrix models \eqref{eq:genericdiscreteMM}. Combining these two facts will lead to a proof of Corollary \ref{c:rholargeN}.
\begin{lem}\label{lemma:OmegalargeN}
    Consider the discrete matrix model \eqref{eq:genericdiscreteMM} in the Veneziano large $N$ limit. Assume it admits a non-trivial saddle point configuration and denote $\varrho_{\ast} (x)$ the corresponding saddle point eigenvalue density. Then, at leading order in the planar limit, 
    \begin{equation}
    \label{eq:WigthmanVeneziano}
          {}_L\langle \Psi_{\beta} \lvert  \Omega (\omega)  \rvert \Psi_{\beta} \rangle_L = \frac{1}{\lambda } \left[ - 1 +  \exp \left( \int \dd x \frac{ \varrho_{\ast} (x)}{\frac{\omega - \mu }{ \lambda} - x } \right) \right] .
    \end{equation}
\end{lem}
A different definition of planar limit in which $\lambda= gN$ is kept fixed, instead of our choice $\lambda=g L$, would simply result is a redefinition $\lambda \mapsto \gamma \lambda $ in the formula.
\begin{proof}
    The idea behind the proof of this lemma is to exploit Theorem \ref{thm:rhoandOmega} and then use Cauchy's theorem to write the Wightman function as a contour integral. In the Veneziano limit, the integral can be evaluated by a saddle point approximation, yielding \eqref{eq:WigthmanVeneziano}. The reader can consult all the details in Appendix \ref{app:proofLemmaOmega}.
\end{proof}
\begin{proof}[Proof of Corollary \ref{c:rholargeN}]
Let us start assuming that the matrix model $\mz_L ^{(N)}$ in \eqref{eq:genericdiscreteMM} admits a non-trivial saddle point, which means that $\ln \mz_L ^{(N)}$ shows the usual $O(N^2)$ growth. Lemma \ref{lemma:OmegalargeN} relates the Wightman functions $\widetilde{G}_{L,\mathrm R} (\omega) , \widetilde{G}_{L,\mathrm A} (\omega)$ to the saddle point eigenvalue density. The formula has logarithmic branch cuts if $\frac{\omega_{\mathrm{r}}}{\lambda} \in \supp \varrho_{\ast}$. Expressing $\varrho (\omega)$ through the discontinuity equation \eqref{eq:rhodiscG}, these branch cuts will contribute to $\supp \rho$ with a continuous interval in $\R$, proving the claim.\par
\end{proof}\par
If the matrix model does not possess a large $N$ scaling for the chosen values of the parameters, the large $N$ argument that led to Lemma \ref{lemma:OmegalargeN} does not apply. The lack of a large $N$ scaling prevents the coalescence of the eigenvalues of the matrix model, which do not form a continuum. It is certainly possible to take the inductive limit over $N$ of the sequence of eigenvalue densities as $N \to \infty$, but in this case the limiting eigenvalue density $\varrho_{\ast} (x)$ will be a distribution, i.e. an infinite sum of Dirac deltas. Taking the inductive $N \to \infty$ limit of Theorem \ref{thmrhoisrho} still \emph{formally} relates the eigenvalue density $\varrho_{\ast}$ to $\rho$, whenever the limit of the latter exists. This would formally imply the discreteness of $\supp \rho$ in the phase(s) in which the eigenvalues do not coalesce in the planar limit, with the caveat that we are not proving existence of the limit in this case.\par
Stated differently, we have shown that 
\begin{equation}
    \supp \varrho_{\ast} \ \text{continuous} \quad \Longrightarrow \quad \supp \rho \ \text{continuous} ,
\end{equation}
and heuristically we expect the same implication for discrete support, but there may be potential obstructions in taking the limit of Theorem \ref{thmrhoisrho} in the phase(s) that do not satisfy the assumptions of Lemma \ref{lemma:OmegalargeN}.\par
\medskip
In any case, discrete random matrix ensembles such as the ones considered throughout this section typically have a continuous density of eigenvalues in the large $N$ limit. This is an extremely widespread property of these models, so that Corollary \ref{c:rholargeN} generically applies. In Section \ref{sec:Fermi} we will have to work harder and add extra ingredients to get the situation without continuous eigenvalue density, at low temperature. In particular, the third order transitions that typically appear in the discrete ensembles, separate two phases both with continuous eigenvalue density.

\subsubsection{Intermezzo: Comparison with IOP}
We ought to stress that \eqref{eq:GtildeOmega} is closely related to \cite[Eq.(5.13)]{Iizuka:2008eb}, which in fact was a source of inspiration for the present analysis, especially in the choice of $H_{\text{int}}$. We have also chosen conventions for the probe vacuum energy that agree with \cite{Iizuka:2008eb}. Nevertheless, \eqref{eq:GtildeOmega} remains valid for a wide class of models, and also incorporates the $N$-dependent constraint, which enriches the dynamics of the toy models. A more detailed comparison is in Section \ref{sec:IOP}.\par
Another difference with the treatment in \cite{Iizuka:2008eb} is that, while the probe approximation works essentially in the same way, the shift $\omega \mapsto \omega - \mu$ due to the probe mass is removed in \cite{Iizuka:2008eb} by adding a counterterm by hand. For consistency with the probe approximation, we keep the shift by $\mu$. This has turned out to be important for $\rho (\omega)$ and the derivation of the von Neumann algebra, while this subtlety would not be appreciated in the correlation functions of $\mathcal{O}^{\dagger}_L (t) \mathcal{O}_L (0)$, which generalize the computation of \cite{Iizuka:2008eb}.

\subsection{Large \texorpdfstring{$N$}{N} algebras}
\label{sec:MMalgebra}

In order to make statements about large $N$ algebras, we need our systems to satisfy large $N$ factorization, as defined in Subsection \ref{sec:vNpivot}. In this short subsection, we argue that it is the case and deduce the types of the large $N$ algebras. 

\begin{lem}\label{lem:factor}
    Consider the discrete matrix model \eqref{eq:genericdiscreteMM} in the large $N$ Veneziano limit, and assume it admits a non-trivial saddle point configuration. At leading order, the $n$-point functions of $\phi_L (t)$ satisfy the large $N$ factorization property of Definition \ref{def:largeNfactor}.
\end{lem}

\begin{proof}
    The proof is done explicitly for the four-point function and then by induction on $n$. We relegate the lengthy details to Appendix \ref{app:LargeNproof}. The basic idea goes as follows.
    \begin{enumerate}[(1)]
        \item We first expand $\phi_L$ in terms of $\mathcal{O}_L, \mathcal{O}_L ^{\dagger}$, so that we reduce to study $n$-point functions of these operators. Only combinations with an equal number of $\mathcal{O}_L$ and $ \mathcal{O}_L ^{\dagger}$ give a non-vanishing contribution.
        \item We study the various non-trivial combinations by direct computation. 
        \item We approximate the thermal ensemble of representations by its saddle point approximation.
        \item We match the resulting terms with the ones predicted by the factorization property (cf. Proposition \ref{lemmaAppFactorization}), thus showing the lemma.
    \end{enumerate}
    A few remarks on this proof:
    \begin{itemize}
        \item A technical assumption in this factorization lemma is that the number $n$ of operator insertions is given from the onset and kept fixed in the large $N$ limit. This is a standard requirement, see for instance \cite{Papadodimas:2012aq} for a neat discussion of this and related matters.
        \item Theorem \ref{thmrhoisrho} shows that the two-point function ${}_L \langle \Psi_{\beta} \lvert \mathcal{O}_L^{\dagger} (t_2) \mathcal{O}_L (t_1) \rvert \Psi_{\beta} \rangle_L $ behaves as the expectation value of a single-trace operator in the ensemble of representations \eqref{eq:ZLNCharExp}. The lemma morally follows from this fact and the factorization properties of matrix models at large $N$. However, there are complications due to the time dependence, which make the proof more technical and are dealt with in Appendix \ref{app:LargeNproof}.
        \item We warn the reader of some subtleties to pass from (2) to (3), which are dealt with in Appendix \ref{app:completeAppD}. The hypotheses of the planar limit and that the saddle point corresponds to a representation of large size are both crucial.
        \item Let us mention that it should be possible to give a Feynman diagram derivation of the factorization property. The models we consider generalize \cite{Iizuka:2008eb}, and one should be able to perform a computation of the correlation functions through more standard QFT techniques, very much along the lines of \cite{Michel:2016kwn}, and show that they factorize. Such a formulation is worth to be studied in detail and is left as an open problem.
    \end{itemize}
\end{proof}\par
\medskip
The results so far are succinctly summarized in the statement:
\begin{center}\setlength{\fboxrule}{1.9pt}
\noindent\fbox{%
\parbox{0.89\linewidth}{%
\begin{stm}[Corollary \ref{cor:jumprhoOmega} and Theorem \ref{thm:Derezinski}]\label{thm:istype3}
The correlation functions of the probe operators in our quantum systems are reproduced by large $N$ von Neumann algebras of Type III$_1$.
\end{stm}\setlength{\fboxrule}{.4pt}
}}\end{center}\par
The main technical achievements of this section are \eqref{eq:rhodiscG}-\eqref{eq:WigthmanVeneziano}. Combining these two, we immediately get Corollary \ref{cor:jumprhoOmega}, whose assumptions are generically satisfied by the ensembles of representations \eqref{eq:ZLNCharExp}. Combining this fact with Theorem \ref{thm:Derezinski}, we have that, generically, the quantum systems constructed and studied in this section have an associated Type III$_1$ von Neumann algebra.\par
We thus state the more precise version of the qualitative Statement \ref{thm:istype3}.
\begin{thm}\label{thm:istype3Rigor}
    Consider the quantum systems of Subsection \ref{sec:QMflavor}, coupled to a probe as detailed in Subsection \ref{sec:probe}. The Hilbert space is $\mathscr{H}^{\mathrm{tot}} _L$ \eqref{eq:Hilbtot} and the Hamiltonian is $H^{\prime}$ \eqref{eq:defHprime}. Assume that the parameters entering the definition \eqref{eq:ZLNCharExp} are such that the eigenvalue density of the discrete matrix model \eqref{eq:genericdiscreteMM} has a saddle point approximation $\varrho_\ast(x)$ which coalesces to a continuum. Then, their large $N$ algebras of operators are von Neumann algebras of Type III$_1$.\par
    Additionally, the models based on the ensembles \eqref{eq:cIOPtableaux}, \eqref{eq:fermionMMRepsNosum}, \eqref{eq:ZEx3def}, \eqref{eq:YM2MM} satisfy the above hypothesis $\forall ~0<\beta<\infty$.
\end{thm}
\begin{proof}
    Corollary \ref{cor:jumprhoOmega} holds by assumption. Moreover, the technical result of Lemma \ref{lem:factor} shows that the hypotheses of Theorem \ref{thm:Derezinski} are satisfied. Combining the two yields the first part of the result.\par
    For the ensembles listed in the theorem, it is known in the literature and we will show explicitly in Part \ref{part2} that they satisfy the hypothesis on the non-trivial saddle point in the Veneziano limit. 
\end{proof}

\section{Systems with Hagedorn transitions}
\label{sec:Fermi}

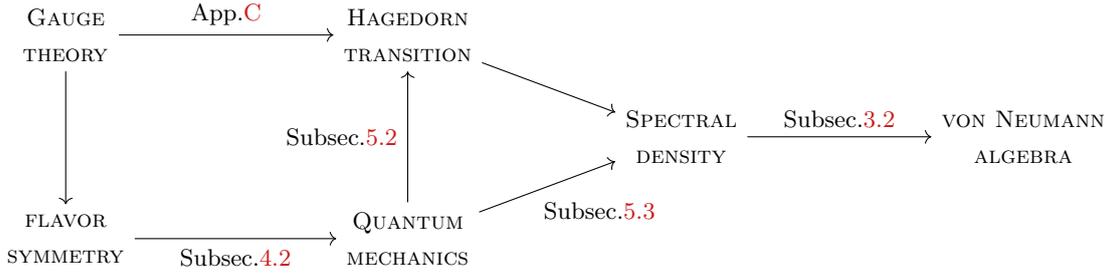
\begin{figure}[th]
\centering
\begin{tikzpicture}[scale=0.9]
\node[align=center] (gt) at (-9,3) {\textsc{\footnotesize Gauge}\\ \textsc{\footnotesize theory}};
\node[align=center] (fs) at (-9,0) {\textsc{\footnotesize flavor}\\ \textsc{\footnotesize symmetry}};
\node[align=center] (ht) at (-4,3) {\textsc{\footnotesize Hagedorn}\\ \textsc{\footnotesize transition}};
\node[align=center] (qm) at (-4,0) {\textsc{\footnotesize Quantum}\\ \textsc{\footnotesize mechanics}};
\node[align=center] (sd) at (0,1.5) {\textsc{\footnotesize Spectral}\\ \textsc{\footnotesize density}};
\node[align=center] (vn) at (5,1.5) {\textsc{\footnotesize von Neumann}\\ \textsc{\footnotesize algebra}};

\path[->] (gt) edge node[anchor=south] {\footnotesize App.\ref{sec:PTMM}} (ht);
\draw[->] (gt) -- (fs);
\path[->] (fs) edge node[anchor=north] {\footnotesize Subsec.\ref{sec:QMflavor}} (qm);
\path[->] (qm) edge node[anchor=east] {\footnotesize Subsec.\ref{sec:FermionicMM}} (ht);
\path[->] (qm) edge node[anchor=north west,pos=0.4] {\footnotesize Subsec.\ref{sec:specAvg}} (sd);
\draw[->] (ht) -- (sd);
\path[->] (sd) edge node[anchor=south] {\footnotesize Subsec.\ref{sec:vNalgebra}} (vn);

\end{tikzpicture}
\caption{Chart of the main concepts explored in this section.}
\label{fig:flowchart}
\end{figure}

The initial inputs to derive the quantum systems of the previous section were unitary matrix integrals. Such systems generically do not exhibit a sharp Hagedorn-like behavior; rather, they tend to display a third order phase transition \cite{Gross:1980he,Wadia:1980cp,Wadia:2012fr}. Mapping the models to a Hilbert space of states organized in representations of the flavor symmetry, these third order phase transitions crop up in the discrete matrix models and are due to the gauge constraint indexed by $N$. In this section, we will describe a general procedure to promote these systems to systems with a first order Hagedorn-like transition. The basic idea, presented in Subsection \ref{sec:QMLsum}, is to introduce an extended Hilbert space with sectors indexed by the rank $L$ of the flavor symmetry group.\par
\medskip
The central result of this work is to assign a quantum system to a matrix model and determine the type of its von Neumann algebra of operators. It hinges upon three main steps:
\begin{itemize}
    \item[1.] Show how the type of von Neumann algebra is encoded in the K\"all\'en--Lehmann spectral density;
    \item[2.] Obtain the planar limit of the quantum mechanical systems;
    \item[3.] Compute the K\"all\'en--Lehmann spectral densities in the quantum systems.
\end{itemize}
In Section \ref{sec:vNtot} we have already explained how the type of von Neumann algebra of operators is determined by looking at the K\"all\'en--Lehmann spectral density. The main outcome is \eqref{eq:typefromrho}. This section lies down the remaining two of the three main steps.
\begin{itemize}
    \item[2.] In Subsection \ref{sec:FermionicMM} we discuss the phase structure of our new class of models. The outcome is summarized in \eqref{eq:FermMMHagedorn}.
    \item[3.] In Subsection \ref{sec:specAvg} we compute the spectral densities for our class of quantum mechanical models, obtaining \eqref{eq:rhofromHagedorn}. 
\end{itemize}
The three steps are put together in Subsection \ref{sec:summary3steps}. The web of relations between the various concepts is sketched in Figure \ref{fig:flowchart}. The take-home message of this central section is that a type III$_1$ algebra can only emerge above the Hagedorn temperature.

\subsection{Extended quantum mechanical system}
\label{sec:QMLsum}
In this section we go one step beyond the unitary one-matrix models \eqref{eq:ZLNUMM}, and construct models whose partition function takes the form:
\begin{equation}
\label{eq:ZFissumLUMM}
    \mz (\mathfrak{q},y) := \sum_{L \ge 0 } \mathfrak{q}^{L^2} \mz_L ^{(N)} (y)  .
\end{equation}
We have introduced a weighted sum over the flavor symmetry rank $L$, with weight $\mathfrak{q}^{L^2}$ controlled by the free parameter $\mathfrak{q}$, which we will equivalently write 
\begin{equation}
\label{eq:defqexpa}
    \mathfrak{q}= \exp \left( - \frac{1}{2a} \right) , \qquad a>0 .
\end{equation}
The idea will be to interpret \eqref{eq:ZFissumLUMM} as the partition function of an extended quantum system with sectors indexed by the flavor symmetry rank, and an extra ``fugacity'' $\mathfrak{q}$ associated to these sectors.\footnote{Viewing $ \mz_L ^{(N)}$ as a statistical ensemble of representations, the importance of considering generating functions like \eqref{eq:ZFissumLUMM} has been advocated for in the mathematical literature \cite{Betea}.} We stress that the sum is over the rank of the flavor symmetry of the systems, while the integer $N$, related to the rank of the gauge group, is not summed. Thus $N$ remains as a genuine parameter of the extended systems.\par
To be precise, the convergence of \eqref{eq:ZFissumLUMM} is not guaranteed in general, and a more rigorous definition requires to truncate it to a large value $L_{\max}$. Our regularization is as follows: we choose a very large number $\gamma_{\max}\gg 1$ and let $L_{\max}(N):=\lfloor{\gamma_{\max}N}\rfloor+1$, and define 
\begin{equation}
    \mz (\mathfrak{q},y) := \sum_{L=0} ^{L_{\max}(N)} \mathfrak{q}^{L^2} \mz_L ^{(N)} (y)  .
\end{equation}
This implies that the sum has linearly many terms in $N$. We may let $\gamma_{\max}$ depend on the inverse temperature $\beta$, in a convenient way depending on the specific model of interest in each case. In practice, the idea is to wisely choose a regularization scheme $\gamma_{\max}$ so that $\mz$ will not miss the interesting physical phenomena for a vast range of temperatures.\par
\medskip
It is shown in Appendix \ref{sec:PTMM} that a common feature of the ensemble \eqref{eq:ZFissumLUMM} is to promote a unitary matrix model \eqref{eq:ZLNUMM} with third order phase transition to a model with first order, Hagedorn-like phase transition. This statement can equivalently be argued for from the perspective of a discrete matrix model, using the rewriting of $\mz_L ^{(N)}$ from Subsection \ref{sec:RepToFermi}. This approach is presented in full generality in Subsection \ref{sec:FermionicMM}, and exemplified in Part \ref{part2}.
\begin{defin}
Let $\mz (\mathfrak{q},y)$ be as in \eqref{eq:ZFissumLUMM}. The parameter space of the model is 
\begin{equation}
    \left\{ 0< \lvert \mathfrak{q}\rvert  <1  , \ 0 \le \arg (\mathfrak{q}) < 2 \pi \right\} \times \left\{ y >0 \right\}  .
\end{equation}
We will set $\arg (\mathfrak{q})=0$ throughout. The \textit{constant-$ \mathfrak{q}$ slice} of the parameter space is the region $\left\{ y>0 \right\} $ at a fixed real value $0 < \mathfrak{q} < 1$.\par
The \textit{Schur slice} of the parameter space is the region $\mathfrak{q} = \sqrt{y}$ and we call the partition function on the Schur slice the limit 
\begin{equation}
\label{eq:schurZFUMM}
	\mz (y) := \mz (\sqrt{y},y) .
\end{equation}\par
For $a$ as in \eqref{eq:defqexpa} and $y=e^{-\beta}$, we restrict the parameter space to be the positive quadrant $(a, \beta) \in \R_{>0} \times \R_{>0}$. The first kind of slice is a slice of constant $a>0$, and the Schur slice is $a=\beta^{-1}$.
\end{defin}
\begin{itemize}
    \item Note that the physical interpretation of the Schur slice is particularly natural, as it amounts to adding a \emph{vacuum energy} to each sector and considering a thermal state with no other fugacity involved.
    \item The nomenclature \emph{Schur slice} is chosen in (vague) analogy with the degeneration of the Macdonald polynomials into Schur polynomials when the two fugacities are equal.
    \item In the constant-$ \mathfrak{q}$ slice we assume that $0<\mathfrak{q}<1$ is a given number and describe the structure of the theory as a function of $y$ alone.
    \item Finally, note that $\mathrm{Arg} (\sqrt{y})$ is not a free parameter, because it can be reabsorbed by a gauge transformation, equivalently in a change of variables in \eqref{eq:ZLNUMM}.
\end{itemize}\par
Now, the quantity $\mz (\mathfrak{q},y)$ can be identified with the partition function of a quantum system living on the extended Hilbert space
\begin{align}
\label{eq:defextHN}
\mathscr{H}^{(N)} :=\bigoplus_{0\leq L\leq L_{\max}(N)}\mathscr{H}_L ^{(N)} = \bigoplus_{0\leq L\leq L_{\max}(N)} \bigoplus_{R \in \mathfrak{R}_L^{(N)} } \mathscr{H}_L (R) \otimes  \mathscr{H}_L (\phi(R)) , 
\end{align}
where the second equality recalls the definition of the gauge-invariant Hilbert space $\mathscr{H}^{(N)} _L$ from \eqref{eq:HLNintoRLN}.
More precisely, we have:
\begin{prop}
With the notation as just explained, it holds that 
\begin{equation}
\label{eq:ZthermalQM}
	\mz (\mathfrak{q}, e^{-\beta}) =\sum_{L=0}^{L_{\max}(N)} \mathfrak{q}^{L^2} ~ \mathrm{Tr}_{\mathscr{H}_L ^{(N)}} \left( e^{- \beta H} \right) .
\end{equation}
\end{prop}

Once again, note that in the case of the Schur slice, the new fugacity $\mathfrak{q}$ itself depends on the temperature, in such a way that \eqref{eq:ZthermalQM} can be interpreted as the partition function of a quantum system at inverse temperature $\beta$ with a vacuum energy.
\subsection{Hagedorn transition and partition function}
\label{sec:FermionicMM}

In this subsection we establish the phase structure of the extended models \eqref{eq:ZFissumLUMM} at large $N$.\par
Remember that the scope of the whole construction is to derive a tractable quantum mechanical system from a gauge theory on $\mathbb{S}^{d-1} \times \mathbb{S}^1$, whenever the partition function of the latter can be expressed in the form \eqref{eq:genericUMMf}. We now proceed to show that, passing to the extended systems introduced in Subsection \ref{sec:QMLsum}, we generically obtain models with a first order phase transition, in which the partition function stays finite at large $N$ for $\beta^{-1}<T_H$ and diverges like $O(e^{N^2})$ when $\beta^{-1}>T_H$, where $T_H$ is a critical temperature to be determined in each case.\par

\subsubsection{Physical significance}

Before diving in the analysis of the phase structure, a disclaimer is in order. In a realistic holographic model of a black hole, we expect the existence of two distinct temperatures: $T_{\mathrm{HP}}$ and $T_{\mathrm{breakdown}}$. At $T= T_{\mathrm{HP}}$, the thermal AdS solution and the black hole solution exchange their dominance, in a process known as Hawking--Page transition. At $T = T_{\mathrm{breakdown}} \ge T_{\mathrm{HP}}$, the perturbative expansion around the thermal AdS solution breaks down. Since the gauge theory interpretation of our matrix models holds only at very weak gauge coupling \cite{Sundborg:1999ue,Aharony:2003sx}, it is expected that $T_{\mathrm{breakdown}} \approx T_{\mathrm{HP}}$ \cite{Aharony:2003sx,Alvarez-Gaume:2005dvb}. In the light of this discussion we simply refer to the transitions we observe as Hagedorn transitions \cite{Aharony:2003sx,Alvarez-Gaume:2005dvb}.\par

\subsubsection{Hagedorn transitions from matrix models}
\label{sec:1PTargument}

The idea is schematically as follows. Consider the fixed-$L$ unitary matrix model \eqref{eq:ZLNUMM} in the planar large $N$ limit, as a function of the Veneziano parameter $\gamma$ (as introduced in Definition \ref{def:VenezianoLim}). Assume it has a third order phase transition at a critical curve $\gamma = \gamma_c (\beta)$ in the $(\beta,\gamma)$-plane. This phase transition is mapped to a third order phase transition in the discrete matrix model \eqref{eq:genericdiscreteMM}.\par
Passing to \eqref{eq:ZFissumLUMM} with $\mathfrak{q}= e^{-1/(2a)}$, we write 
\begin{equation}
    \mz (e^{-\frac{1}{2a}}, e^{-\beta}) = \sum_{L=0} ^{L_{\max}(N)} \exp \left[ - N^2 \left(\frac{\gamma^2 }{2a} - \mf (\gamma, \beta ) + \cdots \right) \right]
\end{equation}
where the $N^2$ scaling is the typical growth of matrix models, so that $\mf (\gamma, \beta)=O(1) $, or more precisely 
\begin{equation}
    \lim_{N \to \infty } \left. \frac{1}{N^2} \ln \mz_{L} ^{(N)} (e^{-\beta}) \right\rvert_{L= \gamma N } = \mf (\gamma, \beta ) ,
\end{equation}
is finite and independent of $N$, and the ellipses indicate sub-leading contributions. Moreover, by assumption, we have 
\begin{equation}
    \mf (\gamma,\beta) = \begin{cases} \mf_{-} (\gamma , \beta) & \gamma \le \gamma_c  \\ \mf_{+} (\gamma , \beta) & \gamma > \gamma_c , \end{cases}
\end{equation}
with $\mf_{-} $ and $\mf_{+}$ coinciding up to the second derivative at $\gamma_c$, but distinct away from the critical curve $\gamma_c (\beta)$.\footnote{Abstractly, $\mf_{\pm}$ are local functions on a two-dimensional manifold with local coordinates $(\gamma, \beta)$, whose 2-jets are equal at every point of a codimension-one submanifold parametrized by the embedding $\beta \mapsto \gamma_c (\beta)$.}\par
At large $N$, the sum over $L$ effectively enforces an average over $\gamma$. The leading planar contribution to $\mz$ comes from the saddle points $\gamma_{\ast}$ of the quantity $\gamma^2 /(2a) - \mf (\gamma, \beta )$. We can start by looking for the saddle point in the phase $\gamma>\gamma_c$, for which $\gamma_{\ast}$ is determined by 
\begin{equation}
\label{eq:genericSPEgamma}
    \left. \frac{\partial \ }{\partial \gamma} \mf_{+} (\gamma , \beta) \right\rvert_{\gamma_{\ast}} = \frac{\gamma_{\ast} }{a} .
\end{equation}
The solution is $\gamma_{\ast} = \gamma_{\ast} (\beta, a)$. In each constant-$a$ slice of the parameter space, $\gamma_{\ast}$ is moved as we dial the inverse temperature, and in the Schur slice $a = \beta^{-1} $ the saddle point $\gamma_{\ast}$ is uniquely determined as a function of $\beta$.\par
\begin{figure}[tb]
\centering
\includegraphics[width=0.5\textwidth]{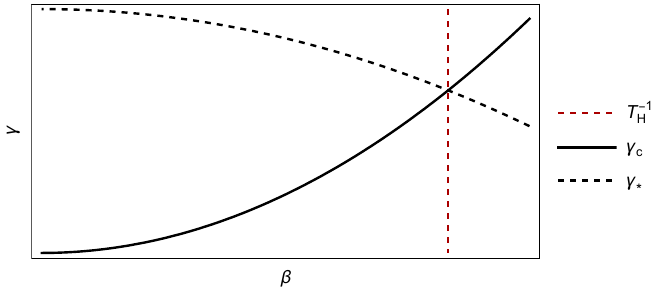}
\caption{Illustration of a constant-$a$ slice of the parameter space. The intersection of the curves $\gamma_c$ and $\gamma_{\ast}$ produces a phase transition. The solution $\gamma_{\ast}$ to \eqref{eq:genericSPEgamma} is valid only for values of $\beta$ such that $\gamma_{\ast} > \gamma_c$, on the left of the dashed vertical line.}
\label{fig:gammavsbeta}
\end{figure}\par
\medskip
However, as we dial $\beta$, $\gamma_{\ast} $ may cross $\gamma_c$. At that precise value of $\beta$, $\mf_{+}$ in \eqref{eq:genericSPEgamma} is replaced by $\mf_{-}$. The explicit form of the saddle point equation for $\gamma_{\ast}$ changes. Therefore, the value of $\gamma_{\ast}$ jumps at that critical value of $\beta$, which we denote $1/T_H$.\par
More precisely, in any fixed slice $a=a(\beta)$ (possibly constant) of the $(a,\beta)$-parameter space, there are two curves 
\begin{equation}
\gamma_c (\beta) \ \text{ and } \ \gamma_{\ast} (\beta) \ \subset \ (\beta, \gamma)\text{-plane} .
\end{equation}
The curve $\gamma_c (\beta)$ is independent of $a$ and remains the same in every slice. The curve $\gamma_{\ast} (\beta)$ is fibered in the parameter space over the $a$-direction $\R_{>0}$. Assume that, in a chosen slice of the parameter space, the curve $\gamma_{\ast} (\beta)$ obtained solving \eqref{eq:genericSPEgamma} intersects $\gamma_c (\beta)$. Then, the solution is consistent only for values of $\beta$ such that $\gamma_{\ast} > \gamma_c$, see Figure \ref{fig:gammavsbeta}. Crossed that critical inverse temperature, namely $1/T_H$, the new branch of the curve $\gamma_{\ast} (\beta)$ solves a different equation.\par

When $\beta^{-1}$ crosses $T_H$ a new value $\gamma_{\ast}$ becomes the saddle point. This discontinuity will produce a phase transition in the ensemble \eqref{eq:ZFissumLUMM} at a critical temperature $T_H$. The transition can in principle be first or second order. As we will show in examples and further support with a general prescription in Appendix \ref{sec:PTMM}, a widespread situation in this type of ensembles is that 
\begin{equation}
\label{eq:gammaastvanish}
    \gamma_{\ast} (\beta) = 0 \qquad \text{ if } \beta^{-1} < T_H .
\end{equation}
In fact, under mild assumptions, 
\begin{equation} 
 \mf_{-} (\gamma, \beta)= \gamma^2 f_{-} (\beta)
\end{equation}
for some smooth function $f_{-}$ (cf. Lemma \ref{lem:Szego0}),\footnote{Both the presence of a third order transition and the form of $\mf_\pm$ are consistent with and supported by \cite{Schnitzer:2004qt}, which studies QCD-like theories on $\cs^{d-1}\times \cs^1$ in the Veneziano limit.} whence the trivial saddle \eqref{eq:gammaastvanish} in the low temperature phase, valid for $\frac{1}{2a} > f_{-} (\beta)$.\par
The trivial saddle point kills the $O(N^2)$ growth of $\ln \mz$, leading to a first order\footnote{Strictly speaking, ``weakly'' first order, because there is no coexistence of saddles.} phase transition at $T_H$. When this is the case, Gaussian integration in the phase $1/\beta<T_H$ leads, near the critical point, to 
\begin{equation}
\mz \sim \frac{1}{\sqrt{T_H-\beta^{-1}}} , 
\end{equation}
which in turn is associated with Hagedorn behavior as the critical temperature $T_H$ is approached from below. Strictly at $N=\infty$ we have $\mz (\beta^{-1}<T_H)= O(1) $ and $\mz (\beta^{-1}>T_H)= \infty$, with $\ln \mz$ diverging proportionally to $N^2$.\par
\medskip
More precisely, we have the following result:
\begin{prop}
    Assume that if $\beta^{-1}<T_H$, there exist $c_0,K>0$ independent of $L$ such that for all $N$, 
    \begin{align}
        q^{L^2}\mz_{L} ^{(N)} (e^{-\beta})\leq Ke^{-c_0 L^2}.
    \end{align} 
    Also assume that if $\beta^{-1}>T_H$, there exist $A,B,c_1,c_2$ with $0<c_1<c_2$, $A>0,B>0$, such that for all $N$, the largest contribution to the sum is comprised in $[Ae^{c_1 N^2},Be^{c_2 N^2}]$. Then, 
    \begin{equation}
        \ln \mz (e^{-\beta})  = \begin{cases} O(1) & \beta^{-1} < T_H \\ O(N^2) & \beta^{-1}>T_H . \end{cases}
    \end{equation}
\end{prop}
\begin{proof}
In the low temperature phase, the result is a simple consequence of domination by a convergent series. In the high temperature phase it suffices to bound the sum by its largest term times the number of terms, and to notice that this number of terms is very small compared to the largest contribution. Taking the logarithm on both sides yields the desired result.
\end{proof}

To conclude, let us mention that it would be interesting to explore examples in which $\frac{\gamma^2}{2a} - \mf_{+}$ develops new minima at some high temperature $1/\beta > T_H$.\par

\subsubsection{Probing the Hagedorn transition with a Polyakov loop}
Phase transitions are famously diagnosed by suitable order parameters. In the present case, the transition can be probed at the level of the matrix model by looking at a Polyakov loop. Inserting a Polyakov loop, it is not hard to show that is acquires a non-trivial expectation value at $\beta^{-1}>T_H$, indicating deconfinement. We state the result here and defer the details to Appendix \ref{app:Polyakov}.\par
At the level of the matrix model, it holds that
\begin{equation}
\begin{tabular}{l c l}
$\langle $Polyakov loop$\rangle =0$ & \hspace{0.5cm} & $\text{if }T<T_H $\\
$\langle $Polyakov loop$\rangle \ne 0$ & \hspace{0.5cm} & $\text{if } T>T_H$ .
\end{tabular} 
\label{eq:PolyHagedorn}
\end{equation}\par

\subsubsection{Summary: Hagedorn transition}
\label{sec:SummaryFF}
The bottom line of our analysis is that a wide class of extended models \eqref{eq:ZFissumLUMM} undergoes a first order phase transition at a critical temperature $T_H$, with 
\begin{equation}\boxed{ \hspace{1cm}
\begin{tabular}{l c l}
\hspace{8pt} & \hspace{8pt}\vspace{-8pt} \\
$\ln \mz = O(1)$ & \hspace{0.5cm} & $\text{if }T<T_H $\\
$\ln \mz =  O(N^2) $ & \hspace{0.5cm} & $\text{if } T>T_H$ . \vspace{-8pt} \\
     \hspace{8pt} & \hspace{8pt} 
\end{tabular} 
\hspace{1cm}  }
\label{eq:FermMMHagedorn}
\end{equation}
This feature of the quantum mechanical models is demonstrated explicitly in three examples in Part \ref{part2}. The transition \eqref{eq:FermMMHagedorn} is obtained summing over the parameter $L$, which, from the point of view of the discrete ensemble, is the maximum length of the Young diagrams $R$. The index $N$, labelling the constraint, is kept fixed. The phase transition is along a critical curve $T_H (a)$ in the $(a,T)$-plane, which projects onto a unique point $T_H$ on the slice $a=T$.\par
It has been proposed that the algebras of single-trace operators are type I von Neumann algebras in the phase $T<T_H$, and become type III$_1$ factors at $T>T_H$. The main result of this work is to prove this expectation rigorously in the class of models whose microscopic description has been given in Subsection \ref{sec:QMflavor}. This is done next, in Subsection \ref{sec:specAvg}.\par

\subsection{Spectral densities}
\label{sec:specAvg}

\subsubsection{Coupling the extended quantum mechanics to a probe}
\label{sec:sumprobe}

Akin to Subsection \ref{sec:probe}, we now introduce the operators of interest in the extended quantum mechanical systems of Subsection \ref{sec:QMLsum}. Our take on \eqref{eq:ZFissumLUMM} and the associated Hilbert space \eqref{eq:defextHN} is that each fixed-$L$ sector describes a full-fledged quantum system, and the sum over $L$ renders the whole ensemble closer to a holographic interpretation. Therefore, the prescription for the correlation function in the systems of Subsection \ref{sec:QMLsum} is to 
\begin{enumerate}[(i)]
    \item couple a probe to each fixed-$L$ system as prescribed in Subsection \ref{sec:probe}, 
    \item compute the correlation function, and 
    \item take the weighted sum over $L$ of the result.
\end{enumerate}
A practical way of doing so is to promote $\mathscr{H}_L ^{(N)}$ to $\mathscr{H}_L ^{\mathrm{tot}}$ as explained in Subsection \ref{sec:probe}, \emph{before} summing over $L$ and defining \eqref{eq:defextHN}. That is, we work with:
\begin{equation}
\label{eq:Hilbsumtot}
    \mathscr{H}^{\mathrm{tot}} := \bigoplus_{L \ge 0}  \mathscr{H}_L ^{\mathrm{tot}} .
\end{equation}
Alternatively, one may first define $\mathscr{H}^{(N)}$ according to \eqref{eq:defextHN}, and then define a probe on that space. To specify such a probe, we need in addition the choice of sector $L$ to which it couples, and the choice must be fine-tuned to explore the regime of interest in each case. We refrain from this definition and stick to the former approach.\par
At this point, we promote the operators $\mathcal{O}_L, \mathcal{O}_L ^{\dagger}$ of Subsection \ref{sec:probecorrel} to block-diagonal operators on $ \mathscr{H}^{\mathrm{tot}}$, 
\begin{equation}
\label{eq:2ptopDirectSum}
    \mathcal{O},  \mathcal{O}^{\dagger}  \in  \mathcal{L} \left(  \mathscr{H}^{\mathrm{tot}} \right) \cong \mathcal{L} \left( \bigoplus_{L\ge 0} \mathscr{H}_L ^{(N)} \otimes \Gamma_{\rm probe} \right) ,
\end{equation}
that act sector-wise as $\mathcal{O}_L, \mathcal{O}_L ^{\dagger}$ for every $L$. Note that for each summand, the one-particle Hilbert space inside $\Gamma_{\rm probe}$ has a different dimension, equal to the flavor rank $L$.

\subsubsection{Spectral densities of the extended quantum mechanics}
We now sketch the adaptation of the procedure of Subsection \ref{sec:spectral} to the new case with the sum over $L$. The ingredients we need, as explained in Section \ref{sec:vNtot}, are the finite temperature Wightman functions and the associated K\"all\'en--Lehmann spectral density $\rho (\omega)$. The novelty in this section is the presence of a Hagedorn transition, highlighted in Subsection \ref{sec:FermionicMM}. We will pay special attention to how the properties of $\rho (\omega)$ change across the transition.\par
The obvious version of the Wightman function \eqref{eq:defGplus} for the ensemble \eqref{eq:ZLNUMM} is:
\begin{defin}
The real-time, finite temperature Wightman function is 
\begin{equation}
\label{eq:GplusAvg}
    G_{+} (t) :=  \frac{1}{\mz} \sum_{L=0} ^{L_{\max}(N)} \mathfrak{q}^{L^2} \tr_{\mathscr{H}_L ^{\mathrm{tot}}} \left( e^{- \beta H^{\prime}} ~ \mathcal{O}^{\dagger} (t)  \mathcal{O} (0)  \right) .
\end{equation}
The dependence on the parameters $N,a, \beta$ is left implicit in the notation.
\end{defin}
We need to prescribe the probe parameters in the same way as we did before. We choose the mass $\mu$ to be the same in all sectors, and large enough that the probe approximation can always be applied (so it is always large enough compared to $L_{\max}(N)$). The coupling of the interaction in a given sector is scaled in a 't Hooft way.

Once again, in the case of the Schur slice $a=\beta^{-1}$, \eqref{eq:GplusAvg} is interpreted simply as a thermal correlation function in an extended quantum system at inverse temperature $\beta$. The only difference between the right-hand side of \eqref{eq:GplusAvg} and its fixed-$L$ counterpart \eqref{eq:defGplus} is the weighted sum over $L$ and the overall normalization by $\mz$ instead of $\mz_L ^{(N)}$. Let $\widetilde{G}_{+} (\omega)$ be the Fourier transform of $G_+ (t)$. Once again, $\widetilde{G}_{+} (\omega)$ is simply related to the spectral density $\rho (\omega) $ through \eqref{eq:rhofromG}, 
\begin{equation}
    \widetilde{G}_{+} (\omega) = \frac{\rho (\omega) }{1-e^{- \beta \omega}} .
\end{equation}
We elaborate further on the properties of $\rho (\omega)$ and the probe approximation in the case of summing over sectors in Subsection \ref{sec:commentsSumL}.\par
\medskip
Theorems \ref{thmrhoisrho} and \ref{thm:rhoandOmega} go through, with the only modification of the weighted sum over $L$ and the normalization by $\mz$ outside the sum, instead of normalizing by $\mz_L ^{(N)}$. This fact can be easily checked directly, and it does not involve any large $N$ limit. To get an equivalent of Corollary \ref{c:rholargeN}, we will follow the extremization procedure over the Veneziano parameter $\gamma$ explained in Subsection \ref{sec:FermionicMM}. The Corollary holds using the eigenvalue density $\varrho_{\ast}$ evaluated at the saddle point $\gamma_{\ast}$.\par
In slightly more detail, for every $N$ we can schematically write 
\begin{equation}
    \mz = \sum_{L=0} ^{L_{\max}(N)} \exp \left[ - N^2 \ms_N (\gamma) \right] ,
\end{equation}
where the dependence on $\beta$ and $\mathfrak{q}$ is left implicit to reduce clutter, and the dependence on $\gamma$ means that we replace $L \mapsto N \gamma$ everywhere. The sequence $\left\{ \ms_N \right\}_{N \in \N}$ admits a finite and well-defined point-wise limit $\lim_{N\to \infty} \ms_N (\gamma) = \ms (\gamma)$. Consider the family of functions of the form $F_N (\vec{\omega}) /\mz $, where
\begin{equation}
    F_N (\vec{\omega})  := \sum_{L=0} ^{L_{\max}(N)} \exp \left[ - N^2 \ms_N (\gamma) \right] ~ f_N (\vec{\omega}, \gamma) ,
\end{equation}
and the sequence $f_N$ is subject to the constraint 
\begin{equation}
\label{eq:condseqfN}
   \lim_{N \to \infty} \frac{ \ln f_N (\vec{\omega}, \gamma) }{N^2} = 0.
\end{equation}
Here $\vec{\omega}$ generically indicates variables in the domain of $f_N$ which do not appear in the definition of $\mz$. In the case of interest to us presently, there is just one variable $\omega \in \C$ and we are looking at a correlation function given by Theorem \ref{thmrhoisrho} or Theorem \ref{thm:rhoandOmega}.\par
In the large $N$ limit, the saddle points of $F_N$ are entirely determined by $\ms_N$, and therefore are the same as for $\mz$ (since the correlation functions converge to a finite value). In computing the ratios $F_N (\vec{\omega})/\mz$ at leading order in the large $N$ limit, the first non-trivial contribution is given precisely by $f_N (\vec{\omega}, \gamma) \vert_{\text{saddle}}$, i.e. the defining function evaluated at the saddle point.\par
We therefore approximate $ \widetilde{G}_{+} (\omega) $ by its value at the saddle point $\gamma_\ast$ in the phase in which the saddle point $\gamma_{\ast}$ is non-trivial, and we deduce Corollary \ref{c:rholargeN} with the modification that we must evaluate the eigenvalue density, and in particular the endpoints of its support, at $\gamma_{\ast}$.\footnote{It is usual to approximate the large $N$ expectation value of observables with a saddle point analysis, although proving a completely rigorous result here would require a study of the rate of convergence to the saddle which is beyond the scope of this paper.}\par
\medskip
We now explicitly relate the support of the spectral density $\rho (\omega)$ to the two phases of the matrix model, as we did in Subsection \ref{sec:spectral}. We observe that the factorization Lemma \ref{lem:factor} works in the high temperature phase, in which $\ln \mz = O(N^2)$. The proof is identical, with the large $N$ limit being evaluated at the saddle point $\gamma_{\ast}$. Moreover, we showed that the spectral density there is continuous. In the low temperature phase, the assumption of a non-trivial saddle point in Lemma \ref{lem:factor} fails. The partition function remains finite in this phase. If we scale the interaction $\lambda$ between the probe and the system like $1/L$ (in a sector-dependent way), then there is no large $N$ factorization and it doesn't really make sense to talk about the spectral density. However, here, unlike in Section \ref{sec:QM}, $N$ and $L$ have fundamentally inequivalent roles: only $N$ is a free parameter, while now $L$ is summed over. Then, if we scale the interaction like $1/N$ (in a sector-independent way), given that the partition function is finite the correlation function simply converges to its limit as $N\rightarrow\infty$ with $L$ fixed, which is the limit in which the probe decouples from the rest of the system. In this limit, in the probe approximation, the spectral function of the probe reduces to delta functions at $\pm\mu$. Therefore, the large $N$ Hilbert space is that of a free oscillator, and in particular the algebra of observables for the probe has type I. Hence, for the choices of scaling of $\lambda$ described in this paragraph:
\begin{equation}\boxed{ \hspace{1cm}
\begin{tabular}{l c l}
\hspace{8pt} & \hspace{8pt}\vspace{-8pt} \\
$\supp \rho = \bigsqcup \left\{ \text{ isolated pts } \right\}$ & \hspace{0.7cm} & $\text{if } \ln \mz = O(1) $\\
$\supp \rho  \subseteq \R \text{ continuous } $ & \hspace{0.7cm} & $\text{if } \ln \mz = O(N^2)$  \vspace{-8pt} \\
     \hspace{8pt} & \hspace{8pt} 
\end{tabular} 
\hspace{1cm}  }
\label{eq:rhofromHagedorn}
\end{equation}
and in the former case $e^{- \beta\lvert \omega\rvert}$ is integrable on the support of $\rho$.\par
\medskip

\subsubsection{Comments on the probe approximation in the extended quantum mechanics}
\label{sec:commentsSumL}
 The probe approximation in the extended quantum mechanical systems can be defined in the analogous fashion to Definition \ref{def:probeapprox2}. Again we expand partition functions and correlation functions as power series in the parameter $e^{- \beta \mu}$. This time the expansion is less trivial, because the probe sectors depend on $L$ and therefore cannot be brought out of the sum over $L$. This does not change the strategy, and we simply retain the term independent of $e^{- \beta \mu}$ in the series expansion.\par
The comments about the mass of the probe in Subsection \ref{sec:spectral} apply identically to this case. In particular, one may express everything in terms of the shifted $\omega_{\mathrm{r}}$ as in \eqref{eq:defomegaren}, which corresponds to neglecting the uniform shift of the energy levels given by the probe mass. The latter is a contribution due to the probe alone, not directly relevant for the quantum mechanical systems.

\subsection{Summary}
\label{sec:summary3steps}
The three main results in the planar limit so far are:
\begin{equation*}
    \begin{tabular}{r l|c|c}
     \hspace{4pt} & \textsc{Result} & \textsc{Section}& \textsc{Crystallized in} \\
     \hline 
         1. & von Neumann algebra type from $\supp \rho$\hspace{8pt} & Subsec. \ref{sec:vNalgebra} & Eq. \eqref{eq:typefromrho} \\
         2. & growth of $\ln \mz$ from Temperature & Subsec. \ref{sec:FermionicMM} & Eq. \eqref{eq:FermMMHagedorn} \\
         3. & $\supp \rho$ from growth of $\ln \mz$ & Subsec. \ref{sec:specAvg} & Eq. \eqref{eq:rhofromHagedorn} 
    \end{tabular}
\end{equation*}
Putting them together we arrive at the central result of our work:
\begin{center}\setlength{\fboxrule}{1.9pt}
\noindent\fbox{%
\parbox{0.89\linewidth}{%
\begin{stm}[\eqref{eq:typefromrho}, \eqref{eq:FermMMHagedorn} and \eqref{eq:rhofromHagedorn}]\label{thm:ItoIII}
In the saddle point approximation, the large $N$ von Neumann algebras associated to our probe operators jump to type III$_1$ above the Hagedorn temperature, in accordance with the holographic picture of a Hawking--Page transition.
\end{stm}\setlength{\fboxrule}{.4pt}
}}\end{center}

\begin{figure}[ht]
    \centering
\begin{equation*}\boxed{ \hspace{0.3cm}
\begin{tikzpicture}
    \draw[black,thick,->] (0,-2) -- (0,2);
    \draw[blue,thick,dashed] (-3.95,0) -- (5.45,0);
    \node[anchor=south] at (0,2) {\textsc{Temperature}};
    \node[anchor=east] at (-4,0) {$\begin{color}{blue}T_H\end{color}$};
    \node[] at (0,0) {$\begin{color}{blue}\bullet\end{color}$};
    \node at (-1.75,1.25) {$\ln \mz \sim O(N^2)$};
    \node at (-1.75,-1.25) {$\ln \mz \sim O(1)$};
    \node at (1.75,1.25) {Type III$_1$};
    \node at (1.75,-1.25) {Type I};
    \path[->] (2.,-0.4) edge[bend right] node[align=left,anchor=west] {\footnotesize jump in\\ \footnotesize von Neumann algebra} (2,0.4);
\end{tikzpicture}
\hspace{0.3cm}  }
\end{equation*}
\caption{Illustration of Theorem \ref{thm:ItoIII}. At the Hagedorn temperature $T_H$, the partition function starts diverging, and the von Neumann algebra experiences a sharp transition from type I to type III$_1$.}
\label{fig:typefromT}
\end{figure}
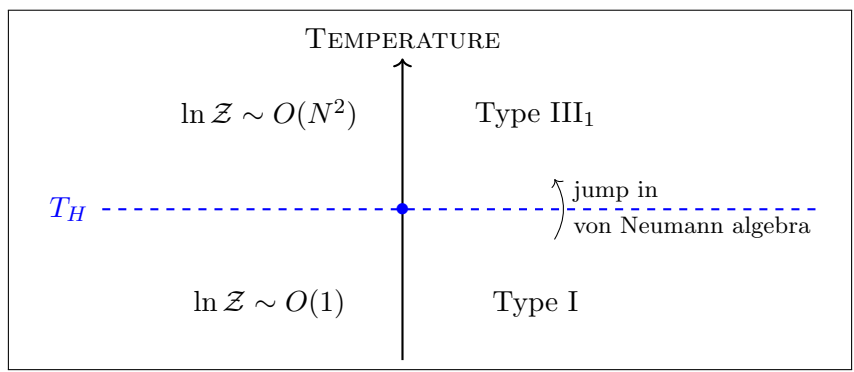\par
The content of this statement is represented pictorially in Figure \ref{fig:typefromT}. We now give a more detailed and rigorous formulation of the qualitative claim of Statement \ref{thm:ItoIII}.
\begin{thm}\label{thm:ItoIIIRigor}
    Consider a quantum mechanical system as defined in Subsection \ref{sec:QMLsum}, coupled to a probe as detailed in Subsection \ref{sec:sumprobe}, and such that the system at fixed $L$ satisfies the assumptions of Theorem \ref{thm:istype3Rigor}. The Hilbert space is $\mathscr{H}^{\mathrm{tot}}$ \eqref{eq:Hilbsumtot}. Assume that \eqref{eq:ZLNCharExp} undergoes a third order phase transition, and that there exists $T_H (a)>0$ such that \eqref{eq:ZFissumLUMM} has the Hagedorn-like behaviour \eqref{eq:FermMMHagedorn} in the planar limit. Then, the large $N$ algebra of probe operators is a Type III$_1$ von Neumann algebra if and only if $T>T_H$.\par
    Additionally, the models \eqref{eq:avgBQCD2} or \eqref{eq:defZAMM} satisfy the above hypotheses $\forall ~0<a<\infty$. 
\end{thm}
\begin{proof}
    The first part follows from combining the results of this section. As emphasized in Subsection \ref{sec:FermionicMM}, the construction is tailored so that, if \eqref{eq:ZLNCharExp} has a third order phase transition, it is a generic feature that  \eqref{eq:ZFissumLUMM} satisfies \eqref{eq:FermMMHagedorn} in the planar limit.\par
    The second part of the theorem follows from the explicit derivations in Part \ref{part2}. 
\end{proof}\par

\clearpage
\part{Examples}
\label{part2}

In this second part of the work we move on to apply the general principles of Part \ref{part1} to selected examples. We present in detail four case studies, where all the quantities discussed in Part \ref{part1} are shown to display the desired properties.

\section{Example 1: Variations on the IOP model}
\label{sec:variationsIOP}

We begin in Subsection \ref{sec:IOP} by introducing a family of toy models that includes and generalizes a model of holographic relevance from \cite{Iizuka:2008eb}. In detail, we will give four equivalent characterizations of the system.
\begin{itemize}
    \item[$i$)] It is a constrained version of a matrix model from \cite[Sec.5]{Iizuka:2008eb}. We will refer to this matrix model (and not to the underlying quantum mechanics) as \emph{IOP model}, and it will be the simplest instance of our broader family that we dub \emph{constrained IOP model} (cIOP for short), introduced in Subsection \ref{sec:introIOP}.
    \item[$ii$)] It is the low energy theory of the chiral operators in four-dimensional supersymmetric QCD, or SQCD$_4$ for short. This perspective allows us to establish a direct relation with Calabi--Yau varieties in Subsection \ref{sec:CYIOP}. In this framework, the Calabi--Yau attached to the original IOP model is just $\C^{L^2}$, whereas the cIOP model bears a connection with non-trivial Calabi--Yau varieties.
    \item[$iii$)] It is a toy model of two-dimensional lattice QCD with bosonic quarks (Subsection \ref{sec:BQCD2IOP}).
    \item[$iv)$] It is one-dimensional QCD with bosonic quarks.
\end{itemize}
To fit in the paradigm of Section \ref{sec:Fermi}, in Subsection \ref{sec:IOPgrandcan} we consider the sum over an integer $L \in \mathbb{N}$, which has the meaning of number of flavors in all the presentations. 
In Subsections \ref{sec:IOPlargeN}-\ref{sec:cIOPlargeN} and \ref{sec:IOPFAplanar} we study the Veneziano limit of the theory without and with sum over $L$, respectively. The spectral density $\rho (\omega)$ is computed in Subsection \ref{sec:IOPrho}.\par

\subsection{Partition function of variations on IOP}
\label{sec:IOP}

\subsubsection{Variations on IOP: Definition}
\label{sec:introIOP}
\begin{defin}
Let $\beta >0$, $y= e^{-\beta}$ and $L \in \mathbb{N}$. The IOP matrix model is 
\begin{equation}
      \mz_{\text{IOP}} (L, y) = \sum_R y^{\lvert R \rvert } (\dim R )^2 , \label{eq:IOPtableaux}     
\end{equation}
with the sum running over Young diagrams of length $\ell (R) \le L$.
\end{defin}
The sum over Young diagrams \eqref{eq:IOPtableaux} was derived from a quantum mechanics in \cite[Sec.5]{Iizuka:2008eb}, whence the nomenclature IOP. It is easy to prove that \cite{Macdonaldbook}
\begin{equation}
\label{eq:IOPexact}
    \mz_{\text{IOP}}  = (1-y)^{-L^2} .
\end{equation}\par
The IOP matrix model does not show a phase transition and $\ln \mz = O(L^2)$ at all temperatures, however, it serves as inspiration for the subsequent discussion. For any $N \in \mathbb{N}$, $N \ge L$, one can embed $SU(L+1) \hookrightarrow SU(N+1)$. This allows for a generalization of \eqref{eq:IOPtableaux}.
\begin{defin}
Let $0<y<1$ and $L,N \in \mathbb{N}$. The \emph{constrained} IOP matrix model (cIOP) is 
\begin{equation}
      \mz_{\text{\rm cIOP}}^{(N)} (L, y) = \sum_R y^{\lvert R \rvert } (\dim R )^2 , \label{eq:cIOPtableaux}     
\end{equation}
with the sum running over Young diagrams of bounded length 
\begin{equation}
\label{eq:lengthminLN}
    \ell (R) \le \min \left\{L,N \right\} .
\end{equation}
\end{defin}
Of course, for every $N \ge L$, the constraint \eqref{eq:lengthminLN} is immaterial and one recovers \eqref{eq:IOPtableaux}. However, the presence of the constraint plays a role in the Veneziano limit 
\begin{equation}
    L \to \infty, \quad N \to \infty , \qquad \qquad \text{with } \quad \gamma = \frac{L}{N} \text{ fixed} .
\end{equation}
Indeed, the cIOP model undergoes a third order phase transition as a function of $\gamma$ \cite{Baik:2000,Santilli:2020ueh}.\par
It is convenient to adopt the change of variables $h_i = R_i - i + L$ introduced in \eqref{eq:changeRtoH} and then use formula \eqref{eq:dimRformula}. With these expressions at hand, we write 
\begin{equation}
\label{eq:ZIOP}
    \mz_{\text{IOP}} = \frac{ y^{- L^2/2} }{G(L+2)^2} \sum_{h_1 > h_2 > \cdots > h_{L} \ge 0} y^{\sum_{j=1}^{L} \left( h_j + \frac{1}{2} \right) }  \prod_{1 \le i < j \le L+1} (h_i - h_j)^2 ,
\end{equation}
with $G(\cdot) $ being Barnes's $G$-function \cite{Barnes} and understanding $h_{L+1} \equiv -1$.

\subsubsection{From cIOP to one-plaquette bosonic \texorpdfstring{QCD$_2$}{QCD2}}
\label{sec:BQCD2IOP}

Let us now introduce a unitary matrix model equal to the cIOP partition function. 
\begin{lem}[\cite{Gessel}]
\label{lemma:cIOPBQCD2}
For every $y\in \C$, $\lvert y \rvert \ne 1$, and $L,N \in \mathbb{N}$ it holds that 
\begin{equation}
\label{eq:ZBQCD2equalsZcIOP}
     \mz_{\text{\rm cIOP}}^{(N)} (L,y) = \oint_{SU(N+1)} [\dd U ] \left[ \det \left(1-\sqrt{y} U \right)  \det \left(1-\sqrt{y} U^{-1} \right) \right]^{- L } . 
\end{equation}
\end{lem}
With $y=e^{- \beta}$, the right-hand side is a toy model for one-plaquette lattice QCD$_2$ with bosonic quarks, studied in \cite{Minahan:1991pv,Santilli:2021eon}. The interpretation of $L$ as the flavor rank and $N$ as the gauge rank of the gauge theory, put forward in Section \ref{sec:QM}, finds a neat realization here.
Comparing \eqref{eq:ZBQCD2equalsZcIOP} and \cite[Eq.(3.8)]{Betzios:2017yms}, we also observe that $\mz_{\text{\rm cIOP}}^{(N)}$ is the partition function of a gauged matrix quantum mechanics describing 1d QCD with $L$ bosonic quarks, where again $y=e^{- \beta}$ (note that \cite{Betzios:2017yms} includes an additional adjoint field).
\begin{proof}[Proof of Lemma \ref{lemma:cIOPBQCD2}]
This lemma is well-known both in the mathematics and physics literature, see for instance \cite{BumpDiaconis,GGT1} for overviews and generalizations. For completeness, we provide one proof here. The so-called Cauchy identity states that \cite{Macdonaldbook}
\begin{equation}
\label{eq:Cauchyid}
    \prod_{i=1}^{L+1} \prod_{k=1} ^{N+1} \frac{1}{1-y_i z_k} = \sum_{R} \chi_R (Y) \chi_R (U) ,
\end{equation}
where on the right-hand side the sum runs over Young diagrams $R$ with bounded length 
\begin{equation}
    \ell (R) \le \min \left\{L,N \right\} ,
\end{equation}
exactly as in \eqref{eq:lengthminLN}. $Y$ and $U$ are special unitary matrices with eigenvalues $\left\{ y_i \right\} $ and $\left\{ z_k \right\}$, respectively, and $\chi_R$ is the character of the group.\par
Setting $y_i = \sqrt{y}$ and using the property \cite[Sec.3]{Macdonaldbook}\footnote{The first equality in \eqref{eq:schurbringout} can be shown using the definition of $\chi_R $ as a ratio of determinants, $\chi_R ( \mathrm{diag} (y_1, \dots , y_{L+1})) = \det_{1 \le j<k \le L+1} \left[y_k ^{R_j +L -j+1} \right] / \det_{1 \le j<k \le L+1} \left[y_k ^{L -j+1} \right] $. Particularizing to $y_k = \sqrt{y}$, one pulls $\sqrt{y}$ out of the determinants and simplifies between numerator and denominator to get \eqref{eq:schurbringout}.}
\begin{equation}
\label{eq:schurbringout}
    \chi_R ( \mathrm{diag} (\sqrt{y} , \dots \sqrt{y} )) = \sqrt{y}^{\sum_i R_i } \chi_R (\mathrm{diag} (1, \dots,1)) = \sqrt{y}^{\lvert R \rvert} \dim R ,
\end{equation}
\eqref{eq:Cauchyid} implies 
\begin{equation}
    \sum_R \sqrt{y}^{\lvert R \rvert} \dim R  \chi_R (U) = \det \left(1-\sqrt{y} U \right)^{-L} .
\end{equation}
Inserting this expression and its complex conjugate in the right-hand side of \eqref{eq:ZBQCD2equalsZcIOP}, one gets 
\begin{equation}
\begin{aligned}
   &\oint_{SU(N+1)} \hspace{-12pt}  [\dd U ] \left[ \det \left(1-\sqrt{y} U \right)  \det \left(1-\sqrt{y} U^{-1} \right) \right]^{- L }  \\
   & \hspace{3.5cm} = \sum_{R, R^{\prime}} \dim(R) \dim (R^{\prime}) \sqrt{y}^{\lvert R \rvert + \lvert R^{\prime} \rvert} \oint_{SU(N+1)}  \hspace{-12pt} [\dd U ]  \chi_R (U)  \chi_{R^{\prime}} (U^{-1}) .
\end{aligned}
\end{equation}
Taking into account the orthogonality of characters:
\begin{equation}
    \oint_{SU(N+1)} [\dd U ] \chi_R (U) \chi_{R^{\prime}} (U^{-1}) = \delta_{R R^{\prime}} ,
\end{equation}
we obtain the identity \eqref{eq:ZBQCD2equalsZcIOP}.
\end{proof}

\subsubsection{IOP: Veneziano limit}
\label{sec:IOPlargeN}
For completeness, we now rederive the large $L$ limit of the IOP and cIOP models, in the presentation \eqref{eq:ZIOP}, so to directly compute the relevant quantities in our formalism and make the paper self-contained. However, we omit the more technical details and refer to the pertinent literature. We start with IOP, and explain how the picture gets modified in the constrained models in Subsection \ref{sec:cIOPlargeN}.\par
\medskip
To begin with, the restriction $h_1 >  \cdots > h_L$ in \eqref{eq:ZIOP} can be lifted using the invariance of the summand under action of the Weyl group. This step introduces an overall factor $1/L!$.\par
Equation \eqref{eq:ZIOP} describes a discrete ensemble of $L$ variables with a hard wall at $h_i=0$. It is an easy task to solve its large $L$ limit, following for instance \cite{Douglas:1993iia}. The summand in \eqref{eq:ZIOP} is rewritten in the form $e^{-S (h_1, \dots, h_L)}$ with 
\begin{equation}
\label{eq:IOPSeff}
    S (h_1, \dots, h_L) = \beta \sum_{i=1} ^{L} h_i - \sum_{i \ne j} \ln \lvert h_i - h_j \rvert .
\end{equation}
At large $L$, we may assume a scaling of the eigenvalues 
\begin{equation}
\label{eq:scalinghtox}
    h_i = L^{\eta} x_i ,
\end{equation}
for $\left\{ x_i \right\} \sim O(1)$ and a power $\eta >0$ which we now determine. With the replacement \eqref{eq:scalinghtox}, the first term in \eqref{eq:IOPSeff} has a growth $\propto L^{1+\eta}$, while the second term in \eqref{eq:IOPSeff} yields a piece $\eta L (L-1) \ln L$, independent of the eigenvalues and thus irrelevant for the sake of the saddle point analysis, and a piece $\propto L^2$. We only need to focus on the two terms carrying a dependence on the eigenvalue. Demanding that they compete, so to find a nontrivial equilibrium, imposes $1+\eta =2$, thus fixing $\eta=1$.\par 
It is customary to introduce the density of eigenvalues 
\begin{equation}
    \varrho (x) = \frac{1}{L} \sum_{i=1} ^{L} \delta (x-x_i) , \qquad \qquad x  >0 ,
\end{equation}
which is normalized to 1 by definition. Furthermore, the discreteness of the ensemble \eqref{eq:ZIOP} imposes a minimal distance among any two eigenvalues, which translates into the condition \cite{Douglas:1993iia}
\begin{equation}\label{eq:condrhole1}
     \varrho (x) \le 1  \qquad \qquad \forall x  >0 .
\end{equation}
With these definitions, $S(h_1, \dots , h_L)$ is a functional of $\varrho$ and \eqref{eq:IOPSeff} becomes (up to $x$-independent terms)
\begin{equation}
\label{eq:SeffIOPrho}
    S [\varrho] = L^2 \int \dd x \varrho (x) \left[  \beta x - \mathrm{P}\!\!\!\int \dd x^{\prime} \varrho (x^{\prime}) \ln \lvert x-x^{\prime} \rvert \right] ,
\end{equation}
where $ \mathrm{P}\!\!\!\int $ stands for the principal value integral. We have thus expressed the summand in \eqref{eq:ZIOP} as $e^{-L^2 (\cdots )}$, with the dots indicating a positive $O(1)$ term. At large $L$, the leading contributions come from the saddle points of \eqref{eq:SeffIOPrho}. The saddle point equation, also known as equilibrium equation, is 
\begin{equation}
\label{eq:IOPSPE}
    2 \mathrm{P}\!\!\!\int \dd x^{\prime} \frac{ \varrho_{\ast} (x^{\prime}) }{x-x^{\prime}} = \beta .
\end{equation}
Here $\varrho_{\ast}$ is the eigenvalue density that extremizes the functional \eqref{eq:SeffIOPrho}. It must be looked for in the functional space subject to the constraints:
\begin{equation}
\label{eq:rhoIOPspacesol}
    \int_0 ^{\infty} \dd x \varrho (x) =1 , \qquad 0 \le  \varrho (x) \le 1 , \qquad \supp \varrho \subseteq [0, \infty) .
\end{equation}
While the derivation was standard so far, the IOP matrix model has two peculiar features, encapsulated in \eqref{eq:rhoIOPspacesol}: 
\begin{itemize}
    \item the discreteness of the ensemble, and
    \item the hard wall at $x=0$.
\end{itemize}\par
Matrix models with a hard wall typically have an inverse square root behavior near the edge \cite{Forrester:1993,Cunden:2018}, which in our case would produce 
\begin{equation}
    \varrho (x) \sim \frac{1}{\sqrt{x}} , \qquad \qquad x \to 0 .
\end{equation}
Clearly, this is incompatible with the restriction \eqref{eq:condrhole1}. Following the seminal work \cite{Douglas:1993iia}, we thus look for a ``capped'' eigenvalue density of the form
\begin{equation}
\label{eq:DKansatz}
    \varrho (x) = \begin{cases} 1 & \hspace{8pt} 0 \le x < x_{-} \\ \hat{\varrho} (x) \  & x_{-} \le x \le x_{+} \\ 0 &  \hspace{8pt} x > x_{+} \end{cases}
\end{equation}
with $\hat{\varrho}$ a nontrivial function subject to the continuity conditions 
\begin{equation}
    \hat{\varrho} (x_{-}) =1 , \qquad \hat{\varrho} (x_{+}) =0 .
\end{equation} 
The values of $x_{\pm}$ will be determined by normalization. See Figure \ref{fig:cartoonhw1} for a sketch.\par
\begin{figure}[tb]
    \centering
    \includegraphics[width=0.4\textwidth]{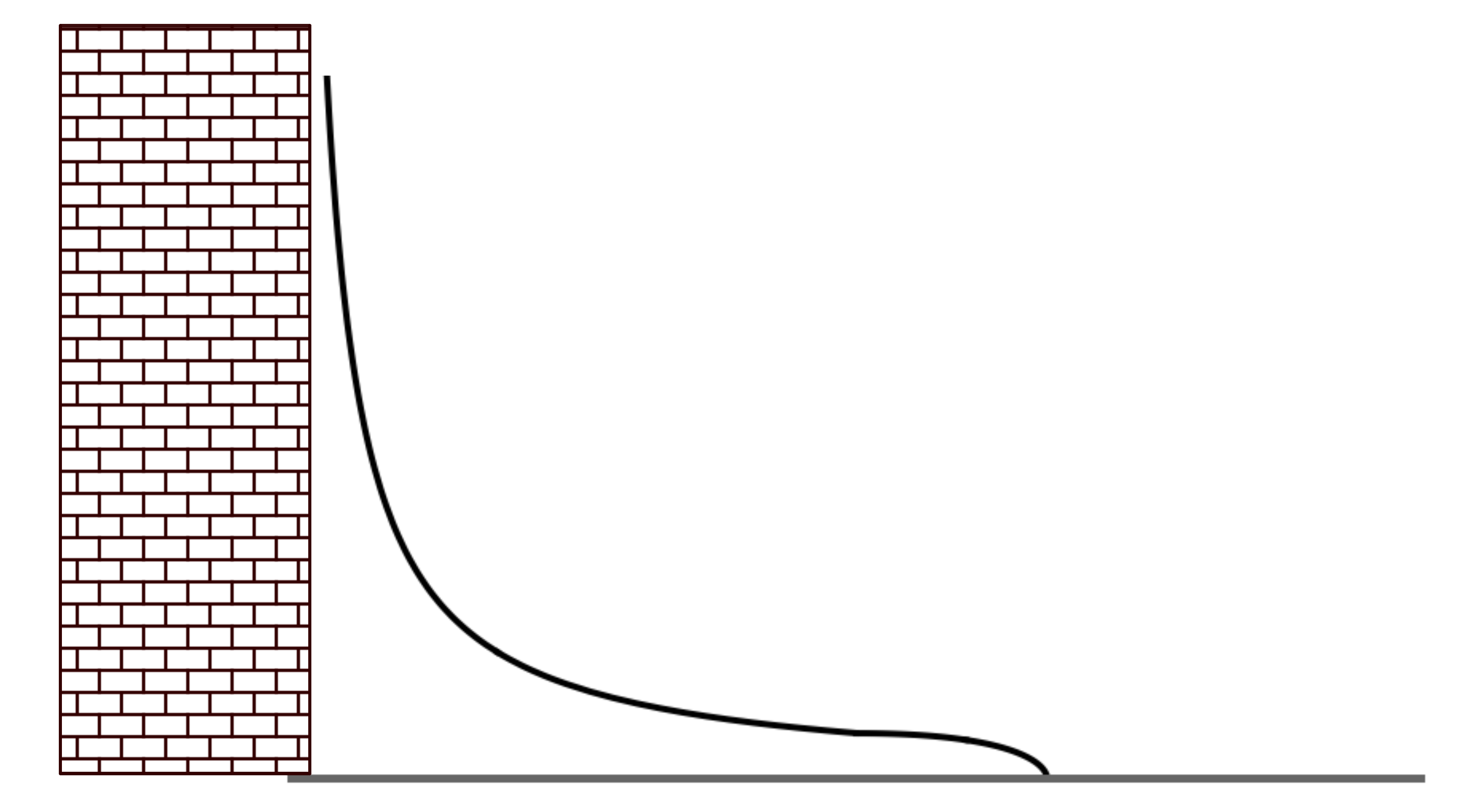}\hspace{0.05\textwidth}
    \includegraphics[width=0.4\textwidth]{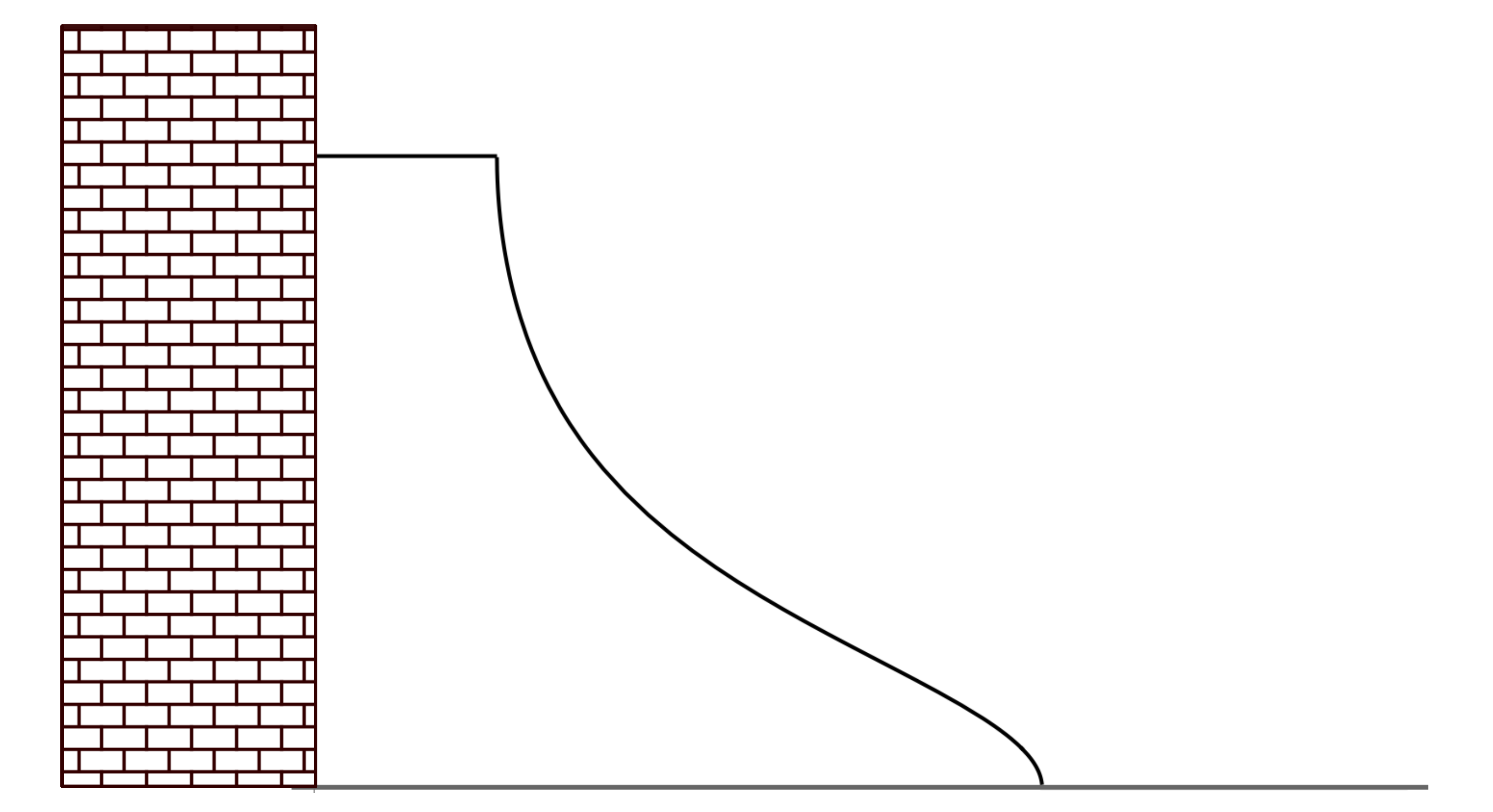}
    \caption{Schematic representation of the eigenvalue density in presence of a hard wall for the eigenvalues. Left: Continuous matrix models typically have a square root singularity near the hard wall. Right: Discrete matrix models have eigenvalue densities bounded above.}
    \label{fig:cartoonhw1}
\end{figure}\par

To solve the saddle point equation \eqref{eq:IOPSPE} with the ansatz \eqref{eq:DKansatz} is a standard procedure, see e.g. \cite{Douglas:1993iia}. The non-trivial part to be determined lies on the domain $x_{-} \le x \le x_{+}$. In this region, \eqref{eq:DKansatz} leads to 
\begin{equation}
\label{eq:rhohatSPEIOP}
    \mathrm{P}\!\!\!\int_{x_-}^{x_+} \dd x^{\prime} \frac{ \hat{\varrho}_{\ast} (x^{\prime}) }{x-x^{\prime}} = \frac{\beta}{2} + \ln \left( \frac{x}{x-x_-} \right) .
\end{equation}
This is the standard equilibrium equation, now for the equilibrium measure $\hat{\varrho}_{\ast} (x) \dd x$ supported on $[x_-, x_+]$ and with a modified effective potential
\begin{equation}
    \beta x \mapsto \beta x + 2 x \ln (x) - (x-x_-)\ln (x-x_-) ,
\end{equation}
to account for the term $ \ln \left( \frac{x}{x-x_-} \right)$ on the right-hand side of \eqref{eq:rhohatSPEIOP}. Therefore, we have effectively reduced the problem of finding an eigenvalue density with a hard wall which is bounded above, to the problem of finding the eigenvalue density $\hat{\varrho}_{\ast}$, at the price of trading the linear potential for a more complicated one. The latter problem admits a solution via a standard procedure. We skip the more technical details, and refer in particular to \cite[Sec.5]{CP:2013}, where a very similar calculation is performed.\par 
In a nutshell, the textbook prescription consists in introducing a complex function, called resolvent, supported on $\C \setminus [x_-, x_+]$. This function uniquely determines $\hat{\varrho}_{\ast}$ and $x_{\pm}$. Then, \eqref{eq:rhohatSPEIOP} is interpreted as a discontinuity equation for the resolvent, which can be solved by complex analytic methods. Once the resolvent is found, one extracts the eigenvalue density $\hat{\varrho}_{\ast}$ and its support.\par
Omitting the more technical part and jumping to the last step, we find (see also \cite{CP:2013}) 
\begin{equation}
    \hat{\varrho}_{\ast} (x) = \frac{2}{\pi} \atan \left( \sqrt{\frac{x_{-}}{x_{+}}}  \sqrt{\frac{x_{+} - x}{x - x_{-}}}  \right) 
\end{equation}
with $\supp \hat{\varrho}_{\ast} = [x_{-} , x_{+}]$ given by 
\begin{equation}
    x_{-} = \tanh \left( \frac{\beta}{4} \right) , \qquad  x_{+} =  \coth \left( \frac{\beta}{4} \right) .
\end{equation}
The full density of eigenvalues $\varrho_{\ast} (x)$ is shown in Figure \ref{eq:varrhoIOP}.\par
\begin{figure}[tb]
    \centering
    \includegraphics[width=0.4\textwidth]{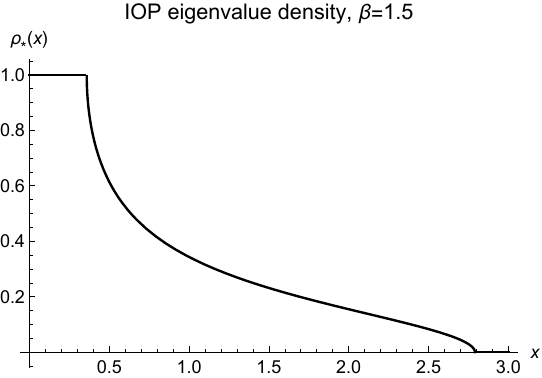}
    \caption{Eigenvalue density of the IOP matrix model, shown at $\beta = 1.5$.}
    \label{eq:varrhoIOP}
\end{figure}\par
\medskip
For completeness, we compare with the original derivation of \cite[Sec.5]{Iizuka:2008eb}. There, the large $L$ limit of $\mz_{\text{IOP}}$ was addressed directly in the sum over Young diagrams form \eqref{eq:IOPtableaux}. The result was encoded in a function $f (x)$ capturing the shape of the Young diagram $R$ that dominates in the limit. Notice that our conventions are such that the flavor rank, denoted $L$ here, corresponds to $N$ of \cite{Iizuka:2008eb}.\par
Comparing the definitions of $\varrho (x)$ with that of $f(x)$ using \eqref{eq:changeRtoH} to compare $h_j$ and $R_j$, we get the relation
\begin{equation}
\label{eq:varrhotofIOP}
    \hat{\varrho}_{\ast} (x) = \frac{1- f_{\ast} ^{\prime} (x) }{2}  \qquad \forall x \in \R 
\end{equation}
and our formulas agree with the existing literature using this substitution.\par
\medskip
The remaining step is to evaluate the effective action \eqref{eq:SeffIOPrho} onto the solution $\varrho_{\ast}$ that extremizes it. A direct computation gives 
\begin{equation}
    \ln \mz_{\text{\rm IOP}} = - L^2 \log \left( 1-e^{- \beta /2} \right) ,
\end{equation}
reproducing the closed form evaluation \eqref{eq:IOPexact}.

\subsubsection{Variations on IOP: Veneziano limit}
\label{sec:cIOPlargeN}
The IOP model is recovered from the constrained IOP in the region $N \ge L$, i.e. $\gamma \le 1$ in the planar limit. In turn, the cIOP model has a third order phase transition at $\gamma=\gamma_c$, with 
\begin{equation}
\label{eq:gammaccIOP}
    \gamma_c= \frac{1+ \sqrt{y}}{2 \sqrt{y}} .
\end{equation}
This can be proven along the lines of \cite{CP:2013}. We therefore only sketch the main ideas here.\par
In a nutshell, the difference between IOP and cIOP is not in the summand in the discrete matrix models, but in the domain. This means that the effective matrix model action for cIOP is again \eqref{eq:IOPSeff}, thus leading to the exact same saddle point equation \eqref{eq:IOPSPE} as in IOP. We can thus use the same solution for $\varrho_{\ast}$ as long as it belongs to the constrained functional space. In cIOP, the additional constraint, yielding the additional parameter $\gamma$, implies that the IOP solution is only valid for certain values of $\gamma$ (which in particular include $\gamma \le 1$), but fails to satisfy the constraint when $\gamma$ is large, leading to a phase transition.\par
\begin{lem}
    In the cIOP model, the saddle point eigenvalue density $\varrho_{\ast}$ is the same found in Subsection \ref{sec:IOPlargeN} for $\gamma \le \left( e^{\beta/2}+ 1\right)/2$.
\end{lem}
\begin{proof}
    The constraint in the cIOP model is formulated in the large $N$ limit as follows.\par
    Let $0 \le s \le 1$ be the limiting value of the index $j/L$ and denote $x(s)$ the limit of $h_j /L$. Then, when $N \le L$, we have that $R_j = 0$ if $j >N$, which, by the change of variables, implies $h_j =L-j$ if $j>N$. Hence, dividing both sides of the latter equation by $L$, we find the constraint 
\begin{equation}
    x(s) = 1-s \quad \text{ if } \quad s > \gamma^{-1} .
\end{equation}
Inverting the relation using the inverse function theorem, we have $\varrho (x) = - \frac{\dd s (x) }{ \dd x} $. Therefore, the eigenvalue density of the cIOP model must be consistent with
\begin{equation}
\label{eq:rhocIOPgammaxminus}
   \varrho (x)=1 \quad \text{ if } \quad 0 \le x <1- \gamma^{-1} .
\end{equation}
Recalling \eqref{eq:DKansatz}, we see that the solution derived in the unconstrained case has $\varrho (x)=1 $ if and only if $0 \le x \le x_-$. Thus, it is inconsistent with the requirement \eqref{eq:rhocIOPgammaxminus} if $x_{-} < 1-\gamma^{-1}$.\par
We conclude that, if $x_{-} < 1-\gamma^{-1}$, we cannot simultaneously satisfy the defining constraint on cIOP and \eqref{eq:rhoIOPspacesol}. The solution breaks down at the value of $\gamma=\gamma_c$ for which $x_- = 1-\gamma_c^{-1}$. Substituting $x_- = \tanh (\beta/4)$ we obtain \eqref{eq:gammaccIOP}.\par
\end{proof}\par
Beyond the critical point, we ought to determine a new density $\varrho_{\ast}$ that belongs to the restricted functional space \eqref{eq:rhoIOPspacesol} and moreover satisfies the additional constraint \eqref{eq:rhocIOPgammaxminus} in the region $\gamma > \gamma_c$. The computation is technical, but can be achieved adapting from \cite[Sec.5]{CP:2013}. However, we omit it, because the only ingredients we need to know are (i) whether $\supp \varrho_{\ast}$ is continuous, and (ii) the large $N$ expression of the partition function.\par 
The answer to (i) is affirmative, since the solution is a variation of the ansatz \eqref{eq:DKansatz} where we additionally impose \eqref{eq:rhocIOPgammaxminus}. For later reference, we note that imposing \eqref{eq:rhocIOPgammaxminus} together with the normalization of $\varrho (x)$, implies that the non-constant part of $\varrho_{\ast} (x)$ bounds an area of $\gamma^{-1}$. Thus it always have finite support, which shrinks as $\gamma$ is increased at fixed temperature.\par
The answer to (ii) was computed in \cite{Santilli:2020ueh} by other methods, so it suffices to quote the result therein. Defining 
\begin{equation}
    \mf _{\text{\rm cIOP}} := \lim_{N \to \infty}\frac{1}{N^2} \ln \mz_{\text{\rm cIOP}}^{(N)} ,
\end{equation}
where the limit is understood to be the Veneziano limit, we have:
\begin{equation}
\label{eq:FcIOP}
    \mf _{\text{\rm cIOP}} = \begin{cases} - \gamma^2 \ln (1-y) & \gamma < \frac{1+\sqrt{y}}{2 \sqrt{y}} \\ -(2\gamma -1) \ln (1-\sqrt{y}) - \frac{1}{4} \ln y + C (\gamma) & \gamma > \frac{1+\sqrt{y}}{2 \sqrt{y}}  \end{cases}
\end{equation}
where, in the second phase, the piece $C(\gamma)$ is independent of the temperature and is explicitly given by 
\begin{equation}
    C(\gamma)= -\gamma^2  \ln \left(\frac{4 \gamma (\gamma -1) }{(2\gamma -1)^2}\right)  +(2\gamma -1) \ln \left(\frac{2(\gamma-1) }{2\gamma-1}\right)-\frac{1}{2} \ln \left(2\gamma -1\right) .
\end{equation}
This expression presents a discontinuity in the third derivative of $\mf_{\text{\rm cIOP}}$ at $\gamma_c$. Obviously, for $\gamma < \gamma_c$, the constraint is not active and \eqref{eq:FcIOP} agrees with the partition function of the IOP model.\par

\subsection{Variations on IOP: Sum over flavor symmetries}
\label{sec:IOPgrandcan}

The finite temperature partition function of the IOP matrix model can be evaluated exactly. We are interested in summing over the integer $L$, which has the meaning of the rank of the flavor symmetry. We let $0<\mathfrak{q}= e^{-1/(2a)} <1$ and consider 
\begin{equation}
    \sum_{L=0} ^{\infty} \mathfrak{q}^{L^2} \mz_{\text{\rm IOP}} = \sum_{L=0}^{\infty} \exp \left[ - L^2 \left( \frac{1}{2a} + \ln (1-e^{- \beta}) \right) \right] ,
\end{equation}
where the right-hand side follows from \eqref{eq:IOPexact}. This expression is only well-defined and convergent for low temperatures, that is, for $\beta$ such that 
\begin{equation}
\label{eq:IOPbetaineq}
     \frac{1}{2a} + \ln (1-e^{- \beta}) >0 .
\end{equation}
As it was discussed in Part \ref{part1}, Section \ref{sec:Fermi}, to render the sum defined at every temperature, we must therefore truncate it at a large value $L_{\max}$ which, in the case of cIOP, may depend on $N$.\par 
While this procedure is not especially enlightening for the IOP model, it already highlights that the behavior of the partition function summed over the rank of the flavor symmetry drastically changes at the critical temperature at which the inequality \eqref{eq:IOPbetaineq} is saturated. We now move on and consider the richer cIOP model. We will show how the interplay between the third order phase transition present in cIOP and the sum over $L$ leads to an interesting Hagedorn-like behaviour.

\subsubsection{Variations on IOP: Flavor sum and planar limit}
\label{sec:IOPFAplanar}

The cIOP model has a third order phase transition with $\ln \mz_{\text{\rm cIOP}}^{(N)} =O(N^2)$ on both sides. We now proceed to promote it to a first order transition with our ``sum over flavors'' prescription.
\begin{thm}
\label{thm:avgCIOP}
With the notation as in \eqref{eq:cIOPtableaux} and $y=e^{- \beta}, \mathfrak{q}= e^{-1/(2a)} $, let 
\begin{equation}
    \mz_{\text{\rm Ex1}} (\mathfrak{q}, y) = \sum_{L=0} ^{L_{\max}(N)} \mathfrak{q}^{L^2} \mz_{\text{\rm cIOP}}^{(N)} (L,y) . \label{eq:avgBQCD2}
\end{equation}
$\forall \ a >0$ there exists $T_H>0$ such that, in the large $N$ limit, \eqref{eq:avgBQCD2} has a first order phase transition at $1/\beta= T_H$. Moreover, 
\begin{equation}
    \ln \mz_{\text{\rm Ex1}} = \begin{cases} O(1) & \frac{1}{\beta} < T_H \\ O(N^2) & \frac{1}{\beta} > T_H .\end{cases}
\end{equation}
\end{thm}
The free energy $ \ln \mz_{\text{\rm Ex1}} $ as a function of the temperature $T=1/\beta$ at fixed $N$ is shown in Figure \ref{fig:ZcIOPplotN4}. Likewise, plotting a numerical evaluation of $ \ln \mz_{\text{\rm Ex1}} $ at various $N$, one may check that the trend is consistent with the analytic prediction that it remains small and approximately constant for $T<T_H$ and it shows a polynomial growth in $N$ if $T>T_H$.

\begin{figure}[th]
    \centering
    \begin{tikzpicture}
    \node (p) at (0,0) {\includegraphics[width=0.5\textwidth]{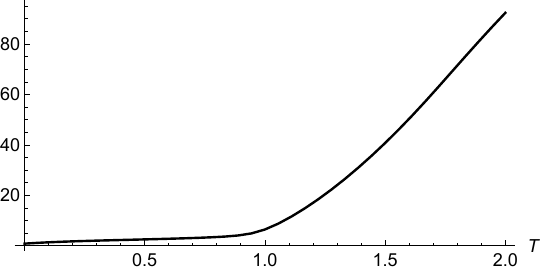}};
    \node[anchor=east] at ($(p.north west)+(0,-0.5)$) {$\scriptstyle \ln \mz_{\text{\rm Ex1}}$};
    \node at ($(p.north)+(-1,-1)$) {$\scriptstyle N=4$};
    \end{tikzpicture}
    \caption{Plot of $ \ln \mz_{\text{\rm Ex1}} (y)$ as a function of the temperature $T=-1/\ln (y)$ at $N=4$. The sum over $L$ in \eqref{eq:avgBQCD2} has been truncated at $L \le 20$. A change in behavior around $T \approx 1$ is already visible at such low value of $N$, although the sharp phase transition is smoothed by finite $N$ effects. The plot is taken in the Schur slice $a=T$, but the qualitative behavior is the same in other slices.}
    \label{fig:ZcIOPplotN4}
\end{figure}\par

\begin{proof}
The details of the proof of Theorem \ref{thm:avgCIOP} are essentially identical to the calculations in Subsection \ref{sec:AvgQCD2FF} below, thus we omit them. The basic idea is that, introducing the weighted sum over $L$, the large $N$ limit is taken in two steps:
\begin{itemize}
    \item[$(i)$] Extremize $\mz_{\text{\rm cIOP}}^{(N)}$ in the Veneziano limit as a function of $\gamma$ and $y$;
    \item[$(ii)$] Evaluate the resulting effective action at the saddle point $\gamma = \gamma_{\ast}$.
\end{itemize}
Step $(i)$ has been sketched in Subsection \ref{sec:cIOPlargeN}, and the result we need is expression \eqref{eq:FcIOP}.\par
To set up step $(ii)$, we define 
\begin{equation}
    \mathcal{S}_{\beta,a} ^{\text{\rm cIOP}} (\gamma) := \frac{1}{N^2} \ln \mz_{\text{\rm cIOP}}^{(N)} (L,e^{-\beta}) - \frac{\gamma^2}{2a}
\end{equation}
and write 
\begin{equation}
    \mz_{\text{\rm Ex1}} (\mathfrak{q}, e^{-\beta}) = \sum_{L=0} ^{L_{\max}(N)}\exp \left\{ N^2 \mathcal{S}_{\beta,a} ^{\text{\rm cIOP}} \left( \gamma= \frac{L}{N} \right) \right\} .
\end{equation}
At large $N$, $\mathcal{S}_{\beta,a} ^{\text{\rm cIOP}} (\gamma)$ approaches a function of continuous variable $\gamma \ge 0$, and the leading contribution comes from its maximum, located at $\gamma_{\ast}$. The function to extremize is 
\begin{equation}
\label{eq:ScIOPform}
    \mathcal{S}_{\beta,a} ^{\text{\rm cIOP}} (\gamma) =  \mf _{\text{\rm cIOP}}  - \frac{\gamma^2}{2a} ,
\end{equation}
with the first summand written explicitly in \eqref{eq:FcIOP}. We see that, if $\gamma < \gamma_c:= \frac{1+e^{- \beta/2}}{2 e^{- \beta/2}}$, then $ \mathcal{S}_{\beta,a} ^{\text{\rm cIOP}} (\gamma)$ is quadratic, and moreover if
\begin{equation}
\label{eq:defTHIOP}
    \frac{1}{2a} \ge - \ln (1-e^{- \beta}) ,
\end{equation}
the partition function $ \mz_{\text{\rm Ex1}}$ is Gaussian and damped at $\gamma > 0$. We denote $1/T_H$ the value of $\beta$ for which the inequality \eqref{eq:defTHIOP} is saturated and $ \mathcal{S}_{\beta,a} ^{\text{\rm cIOP}} (\gamma)$ flips sign. For example, if $a=1/\beta$, $T_H \approx 1.039$.\par
Step $(ii)$ is akin to Subsection \ref{sec:AvgQCD2FF} --- especially Figure \ref{fig:plotSbetaZoom}, and the discussion around \eqref{eq:calStwophases}. Here we discuss explicitly only the case $a=1/\beta$, but the case of constant $a$ independent of $\beta$ can be dealt with in exactly the same manner, similarly to Subsection \ref{sec:AvgQCD2FF}.
Studying the maximum of \eqref{eq:ScIOPform} as a function of $\gamma$, for different ranges of values of $\beta$, we find the following:
\begin{itemize}
    \item $\mathcal{S}_{\beta,a} ^{\text{\rm cIOP}}$ is non-negative definite with a global maximum at $\gamma_{\ast} =0$ if $\beta>1/T_H$;
    \item $\mathcal{S}_{\beta>T_H^{-1},a} ^{\text{\rm cIOP}} (0) =0$.
    \item At $\beta < 1/T_H$ we ought to look for a different saddle point;
    \item By direct inspection of $\mathcal{S}_{\beta,a} ^{\text{\rm cIOP}}$, we observe that it has a global maximum at a certain $\gamma_{\ast} >0$ if $\beta<1/T_H$; examples at fixed $\beta$ are shown in Figure \ref{fig:plotSbetacIOP}.
    \item At this value, $\mathcal{S}_{\beta<T_H^{-1},a} ^{\text{\rm cIOP}} (\gamma_{\ast}) >0$, with strict inequality.
    \item The two saddle points do not coexist and exchange dominance exactly at $\beta=1/T_H$.
\end{itemize}
\begin{figure}[tb]
	\centering
	\includegraphics[width=0.32\textwidth]{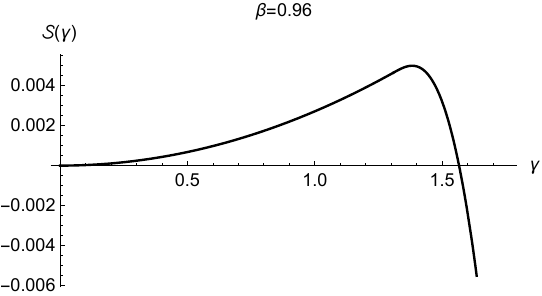}%
    \hspace{0.01\textwidth}%
    \includegraphics[width=0.32\textwidth]{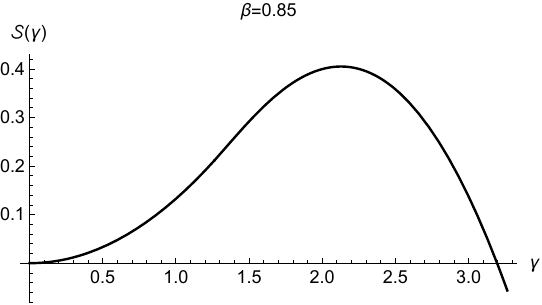}
    \hspace{0.01\textwidth}%
    \includegraphics[width=0.32\textwidth]{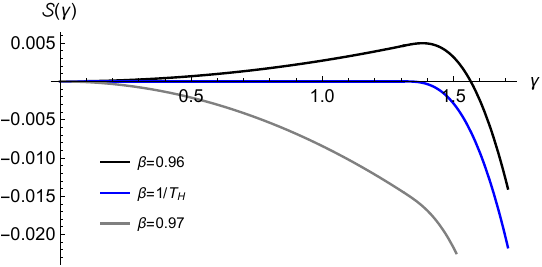}
    	\caption{Plot of $\mathcal{S}_{\beta,\beta^{-1}} ^{\text{\rm cIOP}} (\gamma)$ as a function of $\gamma$ for different values of $\beta$. Left: $\beta=0.96$. Center: $\beta=0.85$. Right: a comparison for $\beta$ slightly above (gray), at the critical point (blue), and slightly below (black). The absolute maximum is $\gamma_{\ast}=0$ if $\beta >1/T_H$, while $\mathcal{S}_{\beta,\beta^{-1}} ^{\text{\rm cIOP}} (\gamma)$ attains a positive absolute maximum at $\gamma_{\ast}>0$ if $\beta <1/T_H$.}
	\label{fig:plotSbetacIOP}
\end{figure}\par
Next, we have to evaluate \eqref{eq:ScIOPform} on the saddle point, to obtain the leading order of $\ln \mz_{\text{\rm Ex1}}$. In the phase $\beta^{-1} <T_H$, we have seen that the leading term trivializes, leaving behind a $O(1)$ correction. In the high temperature phase $\beta^{-1} >T_H$, on the other hand, we have a strictly positive leading order, accompanied by the overall factor $N^2$. Accordingly, $\ln \mz_{\text{\rm Ex1}}$ has a jump from $O(1)$ to $O(N^2)$ at $T_H$.\par
While we have performed the computation explicitly for $a=\beta^{-1}$, one can plot \eqref{eq:ScIOPform} for other choices, in particular on constant-$a$ slices, and check that the same behaviour holds. We stress that 
\begin{itemize}
\item The precise value of $T_H$ depends on $a$, and in the Schur slice $a= 1/\beta$ we find $T_H \approx 1$;
\item Nevertheless, the \emph{existence} of a jump in the saddle point value $\gamma_{\ast}$ holds for all $a>0$. 
\end{itemize}\par
Moreover, the Gaussian integration over fluctuations around the saddle point, in the low temperature phase, produces a universal factor 
\begin{equation}
    \mz_{\text{\rm Ex1}} \propto \left( y - e^{-1/T_H} \right)^{-1/2} ,
\end{equation}
which gives rise to Hagedorn behavior of the Laplace transform of $ \mz_{\text{\rm Ex1}}$ as $\beta^{-1} \to T_H$ from below. The behavior of this model is analogous to the one observed in the matrix model for $\mathcal{N}=4$ super-Yang--Mills \cite{Liu:2004vy,Dutta:2007ws}.
\end{proof}

\subsection{Spectral density of variations on IOP}
\label{sec:IOPrho}
We now analyze the spectral density $\rho(\omega)$ in the variations of the IOP model. As above, we construct from IOP, and eventually pass to the model summed over the number of flavors, which fits in the overarching formalism of Section \ref{sec:Fermi}. The definition of Wightman functions and their relation with $\rho (\omega)$ is taken from Subsection \ref{sec:spectral}. We do not discuss explicitly the microscopic operators $\mathcal{O}$, but take the general expressions for the Wightman functions in Subsection \ref{sec:spectral} as our starting point.

\subsubsection{Intermezzo: Comparison with IOP}
For the sake of comparison, let us notice that the space of states generated by the adjoint fields $A_{ij}$ in the notation of \cite{Iizuka:2008eb} agrees with our Hilbert space. The fundamental field in \cite{Iizuka:2008eb} is the analogue of our probe, and in the SQCD$_4$ language (cf. Subsection \ref{sec:SQCD4andCY}) it is an auxiliary, massive chiral field without anti-chiral counterpart. Both our cIOP generalization and the map to SQCD$_4$ yield an embedding of the IOP model in broader setups.\par 
To compare the Wightman functions, one should be aware of the following differences: our conventions for the Fourier transform are such that $\omega \mapsto - \omega$ with respect to \cite{Iizuka:2008eb}; moreover, the interaction Hamiltonian differs by an overall sign, and our representations are allowed to be $(R, R)$ or $(R, \overline{R})$, which differ by a renaming from the basis $(\overline{R},R)$ considered in \cite{Iizuka:2008eb}. To more easily relate our expressions to the existing literature, throughout this example we fix $\phi (R)= \overline{R}$ and deform the definition of the Hamiltonian $H^{\prime}$ according to 
\begin{align}
    H (R_1, \overline{R_2}) & = \lvert\overline{R_2} \rvert ; \label{eq:HnewbarIOP} \\
    H_{\text{\rm int}} (R_1, \overline{R_2}) &= \frac{g}{2} \left[ Q ( \overline{R_2} ) -  Q (\overline{R_1} ) \right] ; \label{eq:HintIOPEx} \\
     Q(R) &= C_2 (R) + (L+1) \lvert R \rvert . 
\end{align}
\begin{itemize}
    \item The choice of letting \eqref{eq:HnewbarIOP}-\eqref{eq:HintIOPEx} depend on $\overline{R}$ instead of $R$ is adapted so that \eqref{eq:OmegachangedIOP} below matches the conventions of \cite[Sec.5]{Iizuka:2008eb}. 
    \item The redefinition of $H$ does not affect the partition function, which is invariant under charge conjugation $R \mapsto \overline{R}$.
    \item In the definition of $H_{\text{\rm int}}$, 
        \begin{equation}
            C_2(R) = \sum_{i=1}^{L} R_i (R_i -2i +L)
        \end{equation}
        is the quadratic Casimir of $R$, and in $Q(R)$ we use $(L+1) \lvert R \rvert$ because $R$ are representations of $SU(L+1)$.
\end{itemize}
In this way, our explicit computations of the Wightman functions differ by \cite[Sec.5]{Iizuka:2008eb} only by the shift by $\mu$, which was removed ``by hand'' with a counterterm in \cite[Eq.(2.2)]{Iizuka:2008eb} (be aware that $L$ here was denoted $N$ in \cite{Iizuka:2008eb}).

\subsubsection{The IOP spectral density}
Recall the general formula 
 \begin{equation}
    \begin{aligned}
        \rho (\omega) =\lim_{\varepsilon \to 0^{+}} &  \left\{ \frac{i}{2}  \left[  {}_L\langle \Psi_{\beta} \lvert \Omega (\omega + i \varepsilon) - \Omega (\omega - i \varepsilon) \rvert \Psi_{\beta} \rangle_L  \right] \right. \\
        & \left. - \frac{i}{2}  \left[  {}_L\langle \Psi_{\beta} \lvert \Omega (-\omega + i \varepsilon) - \Omega (-\omega - i \varepsilon) \rvert \Psi_{\beta} \rangle_L  \right] \right\} . 
    \end{aligned}
    \end{equation}
derived in Corollary \ref{cor:jumprhoOmega}, where 
\begin{equation}
\label{eq:OmegachangedIOP}
        \langle  R, \ai ; \phi(R) , \dot{\ai} \lvert \Omega (\omega) \rvert  R, \ai ; \phi(R) , \dot{\ai}  \rangle = \frac{1}{(L+1)} \sum_{J \in \mathscr{J}_R } \frac{1}{\omega- \mu - E_{J}^{\text{\rm int}} } ~ \frac{\dim (R \sqcup \Box_J) }{  \dim R} .
\end{equation}
Here 
\begin{equation}
    E_{J}^{\text{\rm int}} = \frac{\lambda}{2(L+1)} \left[ C_2 (\overline{R}) - C_2 ( \overline{R \sqcup \Box_J} ) + L +1 \right]
\end{equation}
with $\lambda= g (L+1)$ the 't Hooft coupling for the interaction Hamiltonian \eqref{eq:HintIOPEx}. We henceforth set $\lambda=1$ without loss of generality, and it can be reinserted at any point by multiplying the Wightman functions by $1/\lambda$ and replacing $\omega - \mu \mapsto (\omega - \mu)/\lambda$.\par
\medskip
\begin{figure}[th]
    \centering
    \begin{equation*}
    \begin{tikzpicture}
         \node[anchor=north] (R0) at (-6,0) {\begin{ytableau} \ & \ & \ & \  \\ \ & \ & \  \\ \ \\ \ \end{ytableau}};
         \node[anchor=north] (R1) at (-2,0) {\begin{ytableau} \ & \ & \ & \  \\ \ & \ & \ & *(gray) \  \\ \ & *(gray) \ & *(gray) \ & *(gray) \  \\ \ & *(gray) \ & *(gray) \ & *(gray) \  \\ *(gray) \ & *(gray) \ & *(gray) \ & *(gray) \ \\ *(gray) \ & *(gray) \ & *(gray) \ & *(gray) \  \end{ytableau}};
        \node[anchor=north] (R2) at (2,0) {\begin{ytableau} \ & \ & \ & \  \\ \ & \ & \ & *(gray) \  \\ \ & *(Cerulean) \ & *(gray) \ & *(gray) \  \\ \ & *(gray) \ & *(gray) \ & *(gray) \  \\ *(gray) \ & *(gray) \ & *(gray) \ & *(gray) \ \\ *(gray) \ & *(gray) \ & *(gray) \ & *(gray) \  \end{ytableau}};
        \node[anchor=north] (R3) at (6,0) { \begin{ytableau}*(gray) \ & *(gray) \ & *(gray) \ & *(gray) \ \\ *(gray) \ & *(gray) \ & *(gray) \ & *(gray) \ \\ *(gray) \ & *(gray) \ & *(gray) \  \\ *(gray) \ & *(gray) \ & *(orange) \ \\ *(gray) \ \end{ytableau}};

        \node[anchor=south] at (-6,0) {$\scriptstyle \overline{R}=(4,3,1,1,0,0)$};
        \node[anchor=south] at (-2,0) {$\scriptstyle R=(4,4,3,3,1,0)$};
        \node[anchor=south] at (2,0) {$\scriptstyle \overline{R}\sqcup \Box_3=(4,3,2,1,0,0)$};
        \node[anchor=south] at (6,0) {$\scriptstyle R\setminus \Box_{4}=(4,4,3,2,1,0)$};

         \node[anchor=north] (P0) at (-6,-5) {\begin{ytableau} \ & \ & \ & \  \\ \ & \ & \  \\ \ \\ \ \end{ytableau}};
         \node[anchor=north] (P1) at (-2,-5) {\begin{ytableau} \ & \ & \ & \  \\ \ & \ & \ & *(gray) \  \\ \ & *(gray) \ & *(gray) \ & *(gray) \  \\ \ & *(gray) \ & *(gray) \ & *(gray) \  \\ *(gray) \ & *(gray) \ & *(gray) \ & *(gray) \ \\ *(gray) \ & *(gray) \ & *(gray) \ & *(gray) \  \end{ytableau}};
        \node[anchor=north] (P2) at (2,-5) {\begin{ytableau} \ & \ & \ & \ & *(Cerulean) \ \\ \ & \ & \ & *(gray) \ & *(CadetBlue) \ \\ \ & *(gray) \ & *(gray) \ & *(gray) \ & *(CadetBlue) \ \\ \ & *(gray) \ & *(gray) \ & *(gray) \ & *(CadetBlue) \ \\ *(gray) \ & *(gray) \ & *(gray) \ & *(gray) \ & *(CadetBlue) \ \\ *(gray) \ & *(gray) \ & *(gray) \ & *(gray) \ & *(CadetBlue) \ \end{ytableau}};
        \node[anchor=north] (P3) at (6,-5) { \begin{ytableau} *(CadetBlue) \ & *(gray) \ & *(gray) \ & *(gray) \ & *(gray) \ \\ *(CadetBlue) \ & *(gray) \ & *(gray) \ & *(gray) \ & *(gray) \ \\ *(CadetBlue) \ & *(gray) \ & *(gray) \ & *(gray) \  \\ *(CadetBlue) \ & *(gray) \ & *(gray) \ & *(gray) \\ *(CadetBlue) \ & *(gray) \ \\ *(orange) \ \end{ytableau}};

        \node[anchor=south] at (-6,-5) {$\scriptstyle \overline{R}=(4,3,1,1,0,0)$};
        \node[anchor=south] at (-2,-5) {$\scriptstyle R=(4,4,3,3,1,0)$};
        \node[anchor=south] at (2,-5) {$\scriptstyle \overline{R}\sqcup \Box_1=(5,3,1,1,0,0)$};
        \node[anchor=south] at (6,-5) {$\scriptstyle R\setminus \Box_{6}=(5,5,4,4,2,0)$};
    \end{tikzpicture}
    \end{equation*}
    \caption{Appending a box at row $J$ and taking the conjugate, exemplified for $SU(6)$ representations, $L=5$. The Young diagram $\overline{R}=(4,3,1,1,0,0)$ is shown in white, $R=(4,4,3,3,1,0)$ in gray, the added box $\Box_J$ in cyan and the removed box $\Box_{L-J+2}$ in orange. Above: $J=3$. Below: $J=1$. In this case, $R$ gains a whole new column at the beginning, shown is darker blue in the two rightmost diagrams.}
    \label{fig:RcupJconj}
\end{figure}
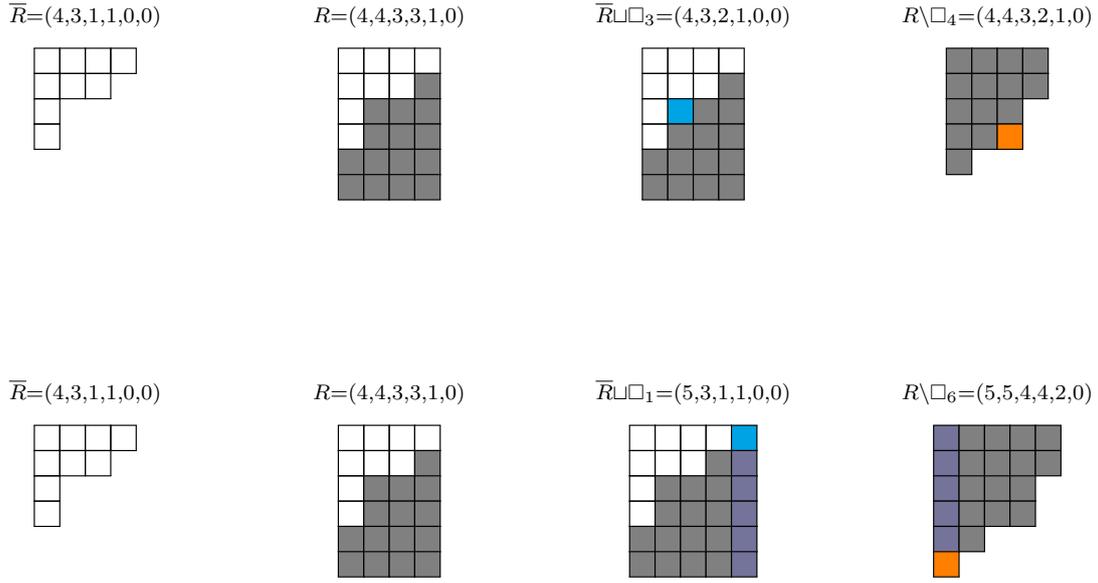\par

We first change variables $R^{\prime} = \overline{R}$ and use the property $\dim \overline{R} = \dim R$, and also 
\begin{equation}
\label{eq:tableauxsubtract}
    \overline{\left( \overline{R^{\prime}} \sqcup \Box_J \right) } = R^{\prime} \setminus \Box_{L-J+2} , 
\end{equation}
where the right-hand side is the conjugate to the Young diagram with the last box of the $(L-J+2)^{\text{th}}$ row removed. Expression \eqref{eq:tableauxsubtract} is valid for $J=2, \dots, L+1 $. To address the case $J=1$ we observe that 
\begin{equation}
\label{Rprimeshifted}
    R^{\prime} = (R^{\prime}_1  , R^{\prime}_2 , \dots, R^{\prime}_L , 0) \cong (R^{\prime}_1 + 1 , R^{\prime}_2 + 1, \dots, R^{\prime}_L + 1, 1) ,
\end{equation}
where the two expressions yield isomorphic $SU(L+1)$ representation with same dimensions, although their $\lvert R^{\prime} \rvert$ is shifted by $L$. The representation $\overline{R^{\prime}} \sqcup \Box_1$ is obtained deleting the last box from the right-hand side of \eqref{Rprimeshifted}. We can therefore let $R^{\prime}$ run over non-empty Young diagrams with $R^{\prime}_{L+1} \in \left\{ 0,1 \right\}$, which is equivalent to declare that the conjugate of $R=\emptyset$ is the determinant representation of $SU(L+1)$, which is trivially true representation theoretically and is most suited for dealing with Young diagrams. For instance:
\begin{equation}
    \emptyset \sqcup \Box_1 = \Box = \overline{ \left( \overline{ \Box } \right) } = \overline{ (\underbrace{1,\dots, 1}_{L},0)} = \overline{ (\underbrace{1,\dots, 1}_{L+1}) \setminus \Box_{L+1} } .
\end{equation}
These observations imply that we are allowed to simply use \eqref{eq:tableauxsubtract} with the understanding that $R^{\prime}$ cannot be empty but can have at most one box in its $(L+1)^{\text{th}}$ row.\footnote{Strictly speaking, this would only be true if we compensate the shift in $ \lvert R^{\prime} \rvert$ with a factor $e^{\beta L \delta_{J,1}} $ in the right-hand side of \eqref{eq:OmegachangedIOP}. This factor does not change in any way the ensuing discussion thus we omit it.}\par
Finally, we rename the summation index $I=L+2-J$ and use the fact that 
\begin{equation}
     \dim (R^{\prime} \setminus \Box_I ) =0 \qquad \text{ if } - I + L + 2 \notin \mathscr{J}_{\overline{R^{\prime}}}
\end{equation}
to let the sum run over $I \in \left\{ 1, \dots, L+1 \right\}$.\par
We stress that there is nothing wrong with $J=1$, and this digression is just to explain what it means by deleting a box from the row $I=L+1$: it means to first shift the Young diagram by one column using \eqref{Rprimeshifted}, and then delete the unique box in the $I^{\text{th}}=(L+1)^{\text{th}}$ row so to get a valid $SU(L+1)$ representation. This is manifest in Figure \ref{fig:RcupJconj}.\par
In this way we get
\begin{equation}
\label{eq:OmegaIOPfullformula}
     {}_L\langle \Psi_{\beta} \lvert \Omega (\omega )  \rvert \Psi_{\beta} \rangle_L  = \frac{1}{\mz_{\text{\rm cIOP}}} \sum_{R^{\prime}} e^{- \beta \lvert R^{\prime}\rvert } (\dim R^{\prime})^2 \frac{1}{(L+1)} \sum_{I=1}^{L+1} \frac{1}{\omega- \mu - E_{I}^{\prime} } ~ \frac{\dim (R^{\prime} \setminus \Box_I) }{  \dim R ^{\prime} }
\end{equation}
where now 
\begin{equation}
    E_{I}^{\prime}  = \frac{\lambda}{(L+1)} (R_I^{\prime} - I + L)    
\end{equation}
descends from $E_J ^{\text{\rm int}}$ after the various rearrangements. With the change of variables \eqref{eq:changeRtoH}, 
\begin{equation}
\label{eq:EprimeIishI}
    E_I^{\prime} = \frac{\lambda}{(L+1)} h_I .
\end{equation}
We drop the $^{\prime}$ from now on. We have thus recovered the formula of \cite[Eq.(5.13)]{Iizuka:2008eb} as a particular case of the framework developed in Part \ref{part1}.\par
The proof of the spectral density of the IOP model relies on the following technical result, which is a particular case of Lemma \ref{lemma:OmegalargeN}. 
\begin{lem}[\cite{Iizuka:2008eb}]\label{lem:GRIOPfromint}
The asymptotic behavior of ${}_L\langle \Psi_{\beta} \lvert \Omega (\omega \pm  i \varepsilon) \rvert \Psi_{\beta} \rangle_L$ of the IOP model in the planar limit is
\begin{equation}
\label{eq:OmegaIOPlargeN}
    {}_L\langle \Psi_{\beta} \lvert  \Omega (\omega \pm i \varepsilon) \rvert \Psi_{\beta} \rangle_L \approx 1- \exp \left(   \int_{0} ^{x_{+}} \dd x  \frac{ \varrho_{\ast} (x) }{x+\mu -\omega \mp i \varepsilon } \right) .
\end{equation}
\end{lem}
\begin{proof}
    The proof is essentially identical to Appendix \ref{app:proofLemmaOmega} up to certain signs, and we simply sketch the main ideas here. By direct computation, the ratio of dimensions in \eqref{eq:OmegaIOPfullformula} yields 
\begin{equation}
\frac{\dim (R \setminus \Box_I) }{  \dim R} = \prod_{j \ne I}  \left(  1 - \frac{1}{h_I-h_j} \right) .
\end{equation}
For every $(L+1)$-tuple $\vec{h}=(h_1, \dots, h_L,-1)$ and $\xi \in \mathbb{C}$, we define the function (it slightly differs from Appendix \ref{app:proofLemmaOmega} to make contact with the sign conventions of \cite{{Iizuka:2008eb}} in this example)
\begin{equation}
    \Phi_{\vec{h}} (\xi) := \prod_{j=1}^{L+1} \left( 1 - \frac{1}{\xi -h_j} \right) ,
\end{equation}
and also denote $\Omega_{\vec{h}}$ the right-hand side of \eqref{eq:OmegachangedIOP} after substituting \eqref{eq:EprimeIishI}. Thanks to the residue theorem we equate 
\begin{equation}
     \Omega_{\vec{h}} (\omega \pm i \varepsilon ) = - \frac{1}{(L+1) g} \oint_{\mathcal{C}} \frac{ \dd \xi}{2 \pi i }  \frac{\Phi_{\vec{h}} (\xi) }{\frac{\omega \pm i \varepsilon - \mu}{g} - \xi } ,
\end{equation}
with the integration contour $\mathcal{C} = \mathcal{C} (\vec{h})$ encircling the points $h_{I}$ and leaving outside the point $(\omega \pm i \varepsilon- \mu)/g$ (compared to Appendix \ref{app:proofLemmaOmega}, here the residues pick a minus sign due to the slight redefinition of $ \Phi_{\vec{h}} (\xi)$, and we will see that our present computation will reproduce \cite[Eq.(5.38)]{Iizuka:2008eb}). From now on we set $Lg = \lambda=1$ without loss of generality. By contour deformation we get 
\begin{equation}
    \Omega_{\vec{h}} (\omega \pm i \varepsilon ) =  1 - \Phi_{\vec{h}} \left( (L+1) (\omega \pm i \varepsilon - \mu)\right) ,
\end{equation}
see Appendix \ref{app:proofLemmaOmega} for the details. Writing 
\begin{equation}
    \Phi_{\vec{h}} \left( (L+1) (\omega \pm i \varepsilon - \mu)\right) = \exp \left[ \sum_{j=1}^{L+1} \ln  \left( 1 - \frac{1}{(L+1) \left( \omega \pm i \varepsilon - \mu -\frac{h_j}{L+1} \right) } \right) \right]
\end{equation}
and approximating at leading order in the planar limit $\ln \left( 1- \frac{c}{L+1} \right) \approx - \frac{c}{L+1}$, we get 
\begin{equation}
     \left. \Omega_{\vec{h}} (\omega \pm i \varepsilon ) \right\rvert_{\text{saddle point}} =  1 -  \exp \left( \int \dd x \frac{ \varrho_{\ast} (x)}{x + \mu - (\omega \pm i \varepsilon) } \right) .
\end{equation}
The claim follows from the usual asymptotics
\begin{equation}
{}_L\langle \Psi_{\beta} \lvert  \Omega (\omega \pm i \varepsilon) \rvert \Psi_{\beta} \rangle_L \approx \left. \Omega_{\vec{h}} (\omega \pm i \varepsilon ) \right\rvert_{\text{saddle point}} .
\end{equation}
\end{proof}

\begin{center}
\noindent\fbox{%
\parbox{0.98\linewidth}{%
\begin{thm}\label{thmIOPspec}
The spectral density of the IOP model is
\begin{equation}
\label{eq:specrhoIOP}
    \rho (\omega) =  \frac{(1-y)}{2\omega_{\mathrm{r}}} \left[ \theta (\omega-\mu) \sqrt{ ( \omega_{\mathrm{r}} - x_{-}) (x_{+} -  \omega_{\mathrm{r}} ) } - \theta (-\omega-\mu) \sqrt{ ( \omega_{\mathrm{r}} - x_{-}) (x_{+} -  \omega_{\mathrm{r}} ) }  \right] 
\end{equation}
with compact support 
\begin{equation}
    \supp \rho = \left[ -x_{+} -\mu , -x_{-} -\mu \right]  \cup \left[  x_{-}+\mu , x_{+} +\mu \right]  \subset \mathbb{R}, \qquad \qquad x_{\pm} = \left[ \coth (\beta/4) \right]^{\pm 1} .
\end{equation}
\end{thm}
\begin{cor}
    The large $N$ correlation functions of the IOP model are described by a von Neumann algebra of type III$_1$.
\end{cor}
}}\end{center}\par
The positive branch of the square roots is taken in \eqref{eq:specrhoIOP}, and the signs in front are adjusted accordingly, and we have used the definition $\omega_{\mathrm{r}}=\lvert \omega \rvert - \mu$ from \eqref{eq:defomegaren} to shorten the expressions.
\begin{proof}[Proof of Theorem \ref{thmIOPspec}]
To reduce clutter, we define the shorthand 
\begin{equation}
    \Delta \Omega (\pm \omega, i \varepsilon) := {}_L\langle \Psi_{\beta} \lvert  \Omega (\pm \omega + i \varepsilon) - \Omega (\pm \omega - i \varepsilon)\rvert \Psi_{\beta} \rangle_L 
\end{equation}
and hence our ultimate goal is to compute 
\begin{equation}
\label{eq:rhoIOPdiscG}
    \rho (\omega) =\lim_{\varepsilon \to 0^{+}}  \left[ \frac{i}{2} \Delta \Omega (\omega, i \varepsilon) - \frac{i}{2} \Delta \Omega (-\omega, i \varepsilon) \right] .
\end{equation}
We use Lemma \ref{lem:GRIOPfromint} to write 
\begin{align}
    \Delta \Omega (\omega, i \varepsilon) \approx & \exp \left(   \int_{0} ^{x_{+}} \dd x  \frac{ \varrho_{\ast} (x) }{x+\mu -\omega + i \varepsilon } \right)  - \exp \left(   \int_{0} ^{x_{+}} \dd x  \frac{ \varrho_{\ast} (x) }{x+\mu -\omega - i \varepsilon } \right) , 
\end{align}
and then plug the eigenvalue density \eqref{eq:DKansatz} and integrate on the interval $[0, x_{-}]$, where $\varrho_{\ast}$ is constant. We get 
\begin{equation}
\begin{aligned}
     \Delta \Omega (\omega, i \varepsilon) \approx  &  \frac{ \left( \sqrt{ \omega - \mu - x_{-} - i \varepsilon } \right)^2 }{ \left( \sqrt{ \omega-\mu -i \varepsilon  } \right)^2 }   \exp \left[ \int_{x_{-}} ^{x_{+}} \dd x  \frac{ \hat{\varrho}_{\ast} (x) }{x-\omega+\mu + i \varepsilon } \right]  \  - \ ( i \varepsilon \mapsto - i \varepsilon ) .
\end{aligned}
\end{equation}
We use the substitution \eqref{eq:varrhotofIOP}, which we recall is $\left. \varrho_{\ast} (x)\right\rvert_{[x_-,x_+]} = \frac{1}{2} - \frac{f_{\ast} ^{\prime} (x) }{2}$, and integrate the constant $\frac{1}{2}$ on $[x_{-},x_{+}]$: 
\begin{align}
\label{eq:DeltaOmegaIOPfprime}
    \Delta \Omega (\omega, i \varepsilon) \approx  &  \frac{\sqrt{ (\omega-\mu - x_{-} - i \varepsilon)} \sqrt{ (\omega-\mu - x_{+} - i \varepsilon) }}{\left( \sqrt{ \omega-\mu -i \varepsilon  } \right)^2 }  \exp \left[ \frac{1}{2} \int_{x_{-}} ^{x_{+}} \dd x  \frac{ f^{\prime}_{\ast} (x) }{\omega-\mu -x -i \varepsilon} \right] \\ & \  - \ ( i \varepsilon \mapsto - i \varepsilon ) ,  \notag 
\end{align}
which is analogous to \cite[Eq.(5.38)]{Iizuka:2008eb}. Reading the integral off from there, and denoting momentarily $\tilde{x}_{\pm} := x_{\pm} + \mu$ to lighten the expressions, one finds 
\begin{equation}
\label{eq:GtildeIOPexplicit}
     \Delta \Omega (\omega, i \varepsilon) \approx  \frac{(1-y) }{2(\omega - \mu -i \varepsilon)}  \left[  -1+ \sqrt{ \omega-i \varepsilon -\tilde{x}_{-}} \sqrt{\omega-i \varepsilon -\tilde{x}_{+}} \right]  \ - \ ( i \varepsilon \mapsto - i \varepsilon ) .
\end{equation}
The expression for $\Delta \Omega (-\omega, i \varepsilon)$ is obtained with the obvious replacement $\omega \mapsto - \omega$.\par
To evaluate the spectral density of the IOP model, we use the discontinuity equation \eqref{eq:rhoIOPdiscG}. The pieces $ \frac{(1-y) }{2( \omega - \mu  \pm i \varepsilon)}$ are smooth in the limit $\varepsilon \to 0^{+}$, therefore they cancel in \eqref{eq:rhoIOPdiscG} and do not contribute to $\rho (\omega)$. The pieces 
\begin{equation}
    \sqrt{ \omega \pm i \varepsilon - \tilde{x}_{-} } \sqrt{\omega \pm i \varepsilon - \tilde{x}_{+} }
\end{equation}
in $\Delta \Omega (\omega, i \varepsilon)$ have branch cuts on $[\tilde{x}_-, \tilde{x}_{+} ]$ in the limit $\varepsilon \to 0 ^{+}$, thus contribute non-trivially to \eqref{eq:rhoIOPdiscG} only if $\omega \in [\tilde{x}_-, \tilde{x}_{+} ]$. Likewise, the corresponding pieces in $\Delta \Omega (-\omega, i \varepsilon)$ have branch cuts along $[-\tilde{x}_+, - \tilde{x}_{-} ]$.\par
Formula \eqref{eq:rhoIOPdiscG} instructs us to assemble these building blocks into $\frac{i}{2} \Delta \Omega (\omega, i \varepsilon) - \frac{i}{2} \Delta \Omega (-\omega, i \varepsilon)$. Let us omit the smooth prefactor for a moment. 
\begin{enumerate}[(i)]
    \item Note that the signs combine to give, schematically, 
\begin{align}
     & \left[ - \frac{i}{2} \left( \sqrt{ - (\omega + i \varepsilon - \tilde{x}_{-})( \tilde{x}_{+} - \omega  - i \varepsilon)} - \sqrt{ - (\omega - i \varepsilon - \tilde{x}_{-})( \tilde{x}_{+} - \omega  + i \varepsilon)} \right) \right. \\
     & \ \left. + \frac{i}{2} \left( \sqrt{ - (\omega + i \varepsilon - (-\tilde{x}_{+}) )( - \tilde{x}_{-} - (\omega  + i \varepsilon))} - \sqrt{ - (\omega - i \varepsilon - (-\tilde{x}_{+}) )( - \tilde{x}_{-} - (\omega  - i \varepsilon))} \right) \right] , \notag 
\end{align}
where the first (respectively second) line account for the contribution from $\Delta \Omega (\omega, i \varepsilon)$ (respectively $\Delta \Omega (-\omega, i \varepsilon)$).
\item Note also that the first (respectively second) square root in each parenthesis approaches the branch cut from above (respectively below) as $\varepsilon \to 0^+$, see Figure \ref{fig:branchIOP}.
\item Hence the first (respectively second) square root in each parenthesis produces a factor $e^{i \pi/2}$ (respectively $e^{-i\pi/2}$).
\end{enumerate}\par
\begin{figure}[th]
    \centering
    \includegraphics[width=0.35\textwidth]{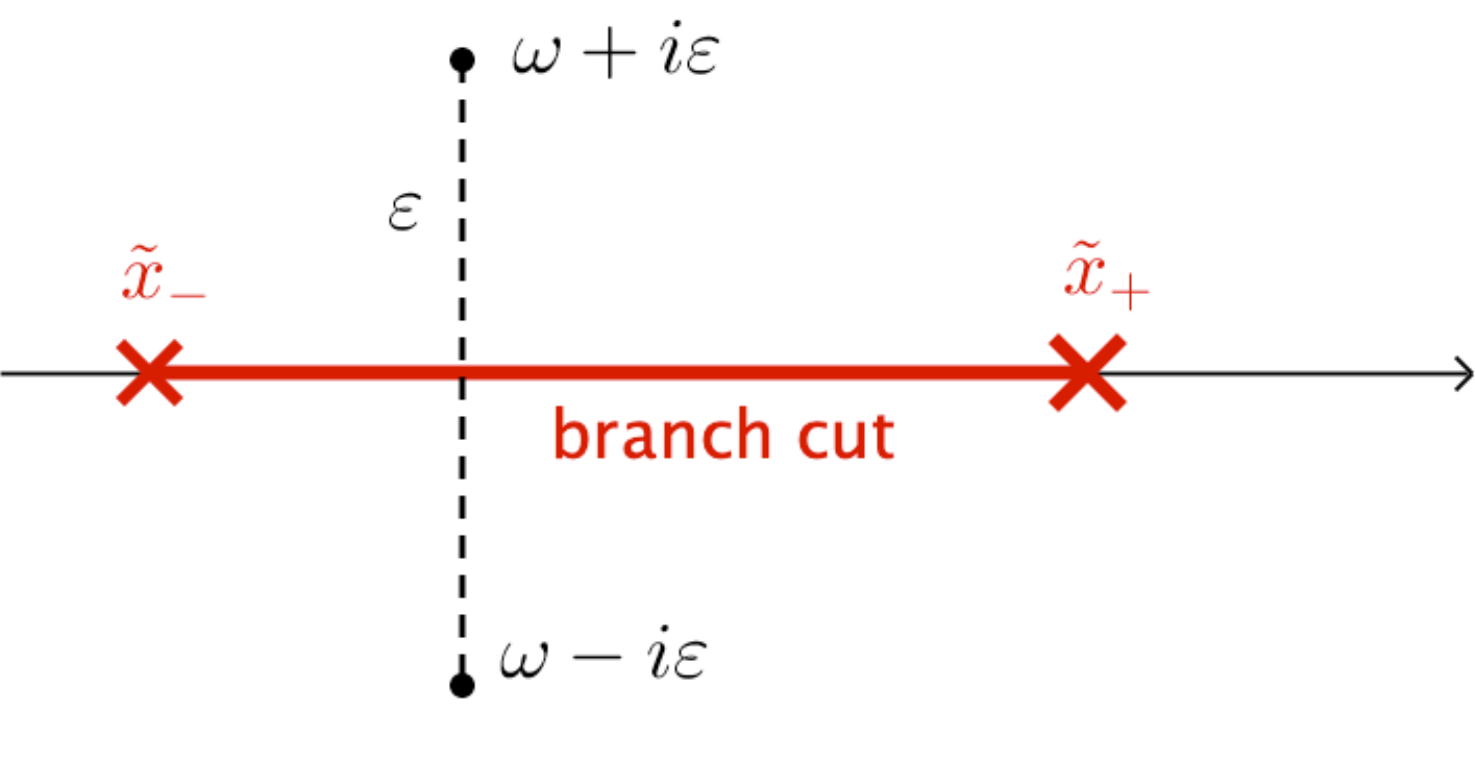}
    \caption{The branch cut that contributes to the spectral density of IOP. There is a second branch with $\omega$ replaced by $-\omega$.}
    \label{fig:branchIOP}
\end{figure}\par
Combining these observations we get 
\begin{equation}
   \begin{cases}  - \frac{i}{2} \left( e^{i \pi/2} - e^{-i \pi/2} \right)  \sqrt{(\omega  - \tilde{x}_{-})( \tilde{x}_{+} - \omega )} & \text {if } \hspace{13pt} \omega \in [\tilde{x}_-, \tilde{x}_{+} ] \\ + \frac{i}{2} \left( e^{i \pi/2} - e^{-i \pi/2} \right)  \sqrt{( (- \omega) -\tilde{x}_{-}  )(\tilde{x}_{+} - (- \omega ))}  & \text {if } - \omega \in [\tilde{x}_-, \tilde{x}_{+} ] \end{cases} 
\end{equation}
which, after simplifications, leaves a plus sign in front when $\omega \in [\tilde{x}_-, \tilde{x}_{+} ]$ and a minus sign in front when $- \omega \in [\tilde{x}_-, \tilde{x}_{+} ]$. We have observed in Subsection \ref{sec:spectral} that $\Omega$ is a resolvent for the density $\rho (\omega)$, and this analysis of the overall sign reproduces the standard sign conventions, with $\lim_{\varepsilon \to 0^{+}} i \Delta (\omega, i\varepsilon)$ giving a positive $\rho (\omega)$ if $\omega \in [\tilde{x}_-, \tilde{x}_{+} ]$ and vanishing otherwise.\par
Altogether we obtain \eqref{eq:specrhoIOP}. In particular, $ \rho (\omega)$ has compact support 
\begin{equation}
   \omega_{\mathrm{r}} \in \left[ x_{-} , x_{+} \right] ,
\end{equation}
where recall from Subsection \ref{sec:IOPlargeN} that $x_{\pm} = \left[ \coth (\beta/4) \right]^{\pm 1} $. The support thus extends to the whole real axis in the infinite temperature limit $\beta \to 0$.
\end{proof}

\subsubsection{Spectral density and von Neumann algebra of cIOP}
\begin{center}
\noindent\fbox{%
\parbox{0.98\linewidth}{%
\begin{thm}
    The large $N$ von Neumann algebra of the quantum mechanical system associated to the cIOP model is of type III$_1$.
\end{thm}}}\end{center}
\begin{proof}
    According to the prescription from Subsection \ref{sec:vNalgebra}, we need to show that $\supp \rho$ is continuous. On one side of the third order phase transition, namely $\gamma < \gamma_c$, the $\gamma$-dependent constraint is inactive and the proof is identical to that of Theorem \ref{thmIOPspec}. We remark that inverting the relation $\gamma_c (\beta)$ tells us that, for fixed $\gamma$, at low enough temperature $\beta>2 \ln (2\gamma -1)$ we recover the IOP solution, for which we know the spectral density in closed form from Theorem \ref{thmIOPspec}.\par
    At $\gamma > \gamma_c$, the eigenvalue density $\varrho_{\ast}$ is determined as highlighted in Subsection \ref{sec:cIOPlargeN}. Even without having $\supp \varrho_{\ast}$ in closed form, we know from Subsection \ref{sec:cIOPlargeN} that it is compact and continuous on $\R_{\ge 0}$. To prove this, it suffices to recall that the constraint imposes $\varrho_{\ast} (x)=1$ for all $0 \le x \le 1-\gamma^{-1}$. In particular, that part of the support extends as $\gamma$ is increased for given $\beta$, while $\supp \varrho_{\ast}$ also stretches along $\R_{\ge 0}$ as the temperature is increased, for fixed $\gamma$.\par
    Lemma \ref{lem:GRIOPfromint} still holds but now we need to use the form of $\varrho_{\ast}$ in the new phase. We arrive at expression \eqref{eq:DeltaOmegaIOPfprime}, except that now $f^{\prime}_{\ast}$ will be a different function, namely the shape of the (different) representation that dominates the planar limit in the phase $\gamma > \gamma_c$. The non-trivial contribution stems from the non-constant part of $\varrho_{\ast}$, which, as discussed in Subsection \ref{sec:cIOPlargeN}, always exists and has compact and continuous support. The integrals are more involved, but the upshot is that they have branch cuts along $\pm \supp \varrho_{\ast}$ shifted by $\mu$. By \eqref{eq:rhoIOPdiscG}, this is sufficient to conclude that $\rho (\omega)$ has compact support on $\R$.
\end{proof}

\subsubsection{Spectral density of cIOP with sum over flavor symmetries}
\label{sec:speccIOP}

\begin{thm}
    The spectral density $\rho(\omega)$ associated to the matrix model $\mz_{\text{\rm Ex1}}$ has compact, continuous support at $1/\beta >T_H$.
\end{thm}
\begin{proof}
    In the high temperature phase, $1/\beta >T_H$, $\ln \mz_{\text{\rm Ex1}}$ has a large $N$ growth and the saddle point argument applies. Lemma \ref{lem:GRIOPfromint} goes through, except that now $\widetilde{G}_{\mathrm{R}} (\omega) $ is still given by \eqref{eq:OmegaIOPlargeN}, but with $\varrho_{\ast}$ the eigenvalue density of the cIOP model in the planar limit, and further evaluated at the saddle point $\gamma = \gamma_{\ast}$.\par
    Due to the lack of knowledge of $\gamma_{\ast}$ as an explicit function of the inverse temperature $\beta$ and the control parameter $a$, we did not manage to write down $\supp \varrho_{\ast}$ in closed form in the phase $\beta^{-1} > T_H$. Nevertheless, for our purposes it suffices to observe that $\supp \varrho_{\ast} \subset \R_{\ge 0}$ is compact and non-trivial. This claim stems from the combination of the facts proven in previous subsections:
    \begin{itemize}
        \item The analysis of Subsection \ref{sec:cIOPlargeN} shows that the eigenvalue density of cIOP in the phase $\gamma > \gamma_c > 1$ has compact and continuous support; moreover, there exists a finite compact interval on which $\varrho_{\ast}$ is non-constant and subtends a finite area $\gamma^{-1}$;
        \item In the high temperature phase of $\mz_{\text{\rm Ex1}}$, the saddle point is an absolute maximum located at $\gamma_{\ast}>1$, as shown in the proof of Theorem \ref{thm:avgCIOP}.
    \end{itemize}
    Therefore, $\supp \varrho_{\ast}$ at $\gamma_{\ast}$ is a finite interval.\par
    From the branch cut of $\widetilde{G}_{\mathrm{R}} (\omega)$ and $\widetilde{G}_{\mathrm{A}} (\omega)$ we still obtain 
\begin{equation}
    \supp \rho \subset \R_{>0} \text{ compact and continuous} .
\end{equation}
This result follows from the knowledge of the branch cuts, although without a general expression in closed form.\par
In the low temperature phase, however, Lemma \ref{lem:GRIOPfromint} fails. The $\delta$ functions in the definition of the spectral density do not coalesce in this phase, and $\rho (\omega)$ becomes trivial as the probe decouples from the rest of the system.
\end{proof}
\begin{equation*}
\begin{tabular}{c|c|c}
\hspace{8pt} \textsc{Temperature} \hspace{8pt} & \hspace{8pt} $ \ln \mz_{\text{\rm Ex1}} $ \hspace{8pt} & \hspace{8pt} \textsc{Algebra type} \hspace{8pt} \\
\hline 
$\beta^{-1}<T_H$ & $O(1)$ & I\ \\
$\beta^{-1}>T_H$ & $O(N^2)$ & III$_1$ \\
\hline
\end{tabular}
\end{equation*}

\subsection{cIOP, \texorpdfstring{SQCD$_4$}{QCD4}, and Calabi--Yau}
\label{sec:SQCD4andCY}

Before moving on to the next example, we make a few comments in this independent subsection on the link between the different variants of the IOP model and the Hilbert series of some Calabi--Yau varieties appearing in the analysis of SQCD$_4$.

\subsubsection{From bosonic \texorpdfstring{QCD$_2$}{QCD2} to \texorpdfstring{SQCD$_4$}{QCD4}}
Unitary matrix integrals such as \eqref{eq:ZBQCD2equalsZcIOP} appear in the computation of Hilbert series of algebraic varieties \cite{Stanley}, and play a prominent role in determining the moduli spaces of vacua of supersymmetric gauge theories \cite{Benvenuti:2006qr,Gray:2008yu}.\par
The basic idea is that supersymmetric gauge theories admit moduli spaces of vacua. These are complex algebraic varieties carved out by the vacuum equations modulo gauge transformations, and typically carry additional structure imposed by supersymmetry. As usual, the algebraic varieties can be specified in terms of their generators, which are certain gauge-invariant operators, and relations among them. Physical information is encoded in these data of the moduli space of vacua, and in particular in their Hilbert series, defined below. 
\begin{lem}[\cite{Gray:2008yu}]
    The Hilbert series of the moduli space of vacua of four-dimensional supersymmetric QCD (SQCD$_4$) with 
    \begin{equation}
        N_{\mathrm{colors}} = N+1 , \qquad N_{\mathrm{flavors}} =L
    \end{equation}
    equals \eqref{eq:ZBQCD2equalsZcIOP}.
\end{lem}
In Section \ref{sec:QM} we have built quantum mechanical systems out of gauge-invariant operators. Here we are seeing this feature emerging cleanly in SQCD$_4$. Furthermore:
\begin{itemize}
    \item The working assumption \eqref{eq:exchangeUY} here corresponds to demand that SQCD$_4$ is free from chiral anomalies.
    \item The Hilbert space $\mathscr{H}_L ^{(N)}$ is constructed by the meson operators in SQCD$_4$, with the $R$-part and $\overline{R}$-part arising, respectively, from chiral and anti-chiral indices.
    \item From the SQCD$_4$ perspective, the partition function enumerates generators of the ring of gauge-invariant operators (this statement is a rephrasing of Lemma \ref{lem:CYHS} below). To compare with IOP, it suffices to note that, when $L<N$, the spectrum of this ring is simply $\C^{L^2}$. Besides, the collection of gauge-invariant operators is identified with the span of the operators denoted $A_{ij}$ in the quantum mechanical model of IOP \cite{Iizuka:2008hg,Iizuka:2008eb}. 
\end{itemize}

\subsubsection{Calabi--Yau variations on IOP}
\label{sec:CYIOP}

\begin{defin}
Let $\mathfrak{R} = \bigoplus_{n \ge 0} \mathfrak{R}_n$ be a Noetherian graded commutative ring over $\mathbb{C}$, and denote $X= \mathrm{Spec}(\mathfrak{R})$. The \emph{Hilbert series} of $X$ in the indeterminate $y$ is 
\begin{equation}
    \mathrm{HS}_y (X) = \sum_{n=0} ^{\infty} \dim (\mathfrak{R}_n) y^n \ \in \Z [\![ y ]\!] .
\end{equation}
\end{defin}
The Noetherian assumption guarantees that $\mathfrak{R}$ is finitely generated and \cite{Stanley}
\begin{equation}
    \dim (\mathfrak{R}_n) < \infty \qquad \qquad \forall n \ge 0 .
\end{equation} 
\begin{lem}[\cite{Gray:2008yu}]
\label{lem:CYHS}
There exists a toric Calabi--Yau manifold $X_{L}^{(N)}$, with 
\begin{equation}
    \dim  \left( X_{L}^{(N)} \right)  = \begin{cases} L^2 & N \ge L \\ 2L(N+1) -N(N+2) & N < L , \end{cases}
\end{equation}
such that the identity 
\begin{equation}
    \mz_{\text{\rm cIOP}}^{(N)} (L,y) = \mathrm{HS}_{y} \left( X_{L}^{(N)} \right) 
\end{equation}
holds. Moreover, $X_{L\le N} ^{(N)} = \C^{L^2}$ and $X_{N+1,N}$ is a hypersurface in $\C^{N^2+2N +3}$.
\end{lem}
\begin{proof}The statement follows from comparison with \cite[Sec.4.3]{Gray:2008yu}, with $(N_c, N_f)$ therein replaced by $(N+1,L)$. The second part follows essentially by the definition of the Hilbert series for a freely generated complex algebraic variety, cf. also \cite[Eq.(3.6)]{Gray:2008yu}.
\end{proof}\par
This lemma explains geometrically that the constrained IOP model carries a richer structure than the original IOP model.
\begin{center}
\noindent\fbox{%
\parbox{0.98\linewidth}{%
\begin{thm}\label{thm:CYtoIOP}
The states of the IOP model are in one-to-one correspondence with elements in the ring of functions on $\C^{L^2}$. The states in the cIOP model are in one-to-one correspondence with elements in the ring of functions on a non-trivial Calabi--Yau variety $X_L ^{(N)}$.
\end{thm}
}}\end{center}\par
\par
\medskip
The Hilbert series approach to the cIOP model allows to introduce yet another characterization, corresponding to a simple model of bosons.\par
Consider a square lattice of $L^2$ sites, labelled by $(i,j) \in \mathbb{Z}_L \times \mathbb{Z}_L$, with bosonic particles at its lattice sites. Generate a random configuration of energy levels $\left\{ e_{ij} + \frac{1}{2}\right\}_{i,j} \subset(\frac{1}{2} +\mathbb{N })^{L^2}$. By taking bosons, we allow two or more particles to occupy the same energy level. See Figure \ref{fig:latticeboson} for an illustration.\par
\begin{figure}[th]
\centering
\begin{tikzpicture}

\draw[] (-3,0) -- (0,0);
\draw[] (1,0.3) -- (-2,0.3);
\draw[] (2,0.6) -- (-1,0.6);

\draw[] (0,0) -- (2,0.6);
\draw[] (-1.5,0) -- (0.5,0.6);
\draw[] (-3,0) -- (-1,0.6);
\draw[->] (-3,0)--(-3,2);
\node[anchor=east] (e) at (-3,1.5) {$e_{ij}$};
\node[anchor=north east] (i) at (-0.5,0) {$\scriptstyle i$};
\draw[->] (i) -- ($(i)+(-0.75,0)$);
\node[anchor=west] (j) at (0.5,0) {$\scriptstyle j$};
\draw[->] (j) -- ($(j)+(0.75,0.225)$);

\node (b1) at (-0.75,0.95) {$\bullet$};
\draw[dashed,thin] (b1) -- (-0.75,0.45);
\node (b2) at (0.75,1.95) {$\bullet$};
\draw[dashed,thin] (b2) -- (0.75,0.45);
\node (b3) at (-1.75,1.15) {$\bullet$};
\draw[dashed,thin] (b3) -- (-1.75,0.15);
\node (b4) at (-0.25,0.15) {$\bullet$};
\end{tikzpicture}
\caption{Free bosons on a square lattice. $L=2$ in this example.}
\label{fig:latticeboson}
\end{figure}
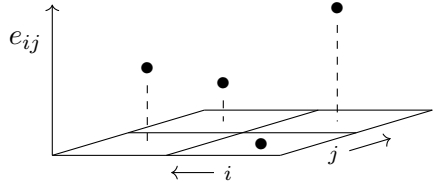\par

The partition function of the free system simply corresponds to summing over all such configurations, weighted by the Boltzmann factor:
\begin{equation}
    \mz_B ^{\text{(free)}} = \prod_{(i,j) \in \mathbb{Z}_L \times \mathbb{Z}_L} \sum_{e_{ij} =0}^{\infty} e^{-\beta \left( e_{ij}+ \frac{1}{2} \right) } =  \left( 2 \sinh (\beta/2) \right)^{- L^2 } ,
\end{equation}\par
It is convenient to introduce $\alpha=(i,j)$ and define a lexicographic order on the set of pairs, so that $\alpha \in \mathbb{Z}_{L^2}$. To enrich the model, we impose a constraint on this system, forbidding certain energy levels. We thus define
\begin{equation}
\label{eq:defBosonicGCPF}
    \mz_B^{(N)} =  \prod_{\alpha=0}^{L^2-1} \sum_{ e_{\alpha}=0 }^{\infty} d^{(N)} _{\left\{ e_{\alpha}\right\} } ~ e^{-\beta \sum_{\alpha} \left( e_{\alpha} + \frac{1}{2} \right)} , \qquad \qquad d^{(N)} _{\left\{ e_{\alpha}\right\} } \in \left\{ 0 ,1 \right\} .
\end{equation}
Whenever a coefficient $d^{(N)} _{\left\{ e_{\alpha}\right\} }$ is taken to vanish, the corresponding configuration $\left\{ e_{\alpha}\right\}$ of $L^2$ particles is forbidden. As in Section \ref{sec:QM}, we are imposing a constraint indexed by $N \in \mathbb{N}$. 

\begin{prop}
There exists a choice of collection 
\begin{equation}
    \left\{ d^{(N)} _{\left\{ e_{\alpha}\right\} } , \ e_{\alpha} \in \mathbb{N} , \ \alpha=0,1,\dots, L^2-1 \right\}
\end{equation}
such that 
\begin{equation}
    \mz_B^{(N)} =  e^{- \frac{\beta}{2} L^2} \mz_{\text{\rm cIOP}}^{(N)} .
\end{equation}
\end{prop}
\begin{proof}
We may introduce a refined Hilbert series as follows. Let $\vec{y}=(y_1, \dots, y_L)$ and  define $y_{\alpha} := \sqrt{y_i y_j}$, with $\alpha =(i,j)$, $\alpha \in \mathbb{Z}_{L^2}$. We have 
\begin{equation}
    \mathrm{HS}_y (X) = \sum_{\left\{e_{\alpha} \ge 0 \right\} } d_{\left\{e_{\alpha}\right\}} \left( X_{L}^{(N)} \right)  \prod_{\alpha=0} ^{L^2-1} y_{\alpha} ^{e_{\alpha}} ,
\end{equation}
where $d_{\left\{e_{\alpha}\right\}}$ are non-negative integers such that 
\begin{equation}
     \dim (\mathfrak{R}_n) = \sum_{\left\{e_{\alpha} \ge 0 \ \vert \ \sum_{\alpha} e_{\alpha} = n \right\} } d_{\left\{e_{\alpha}\right\}} \left( X_{L}^{(N)} \right)  .
\end{equation}
Lemma \ref{lem:CYHS} still holds, upon replacement
\begin{equation}
    \left[ \det \left(1-\sqrt{y}U \right) \left(1-\sqrt{y}U^{-1} \right)\right]^{-L} \quad \longrightarrow \quad \prod_{i=1} ^{L}  \det \left(1-\sqrt{y_i}U \right) \left(1-\sqrt{y_i}U^{-1} \right)
\end{equation} 
in \eqref{eq:ZBQCD2equalsZcIOP}. Unrefining the series, $y_i \to y \ \forall i=1, \dots, L$, we get the combinatorial identity 
\begin{equation}
\label{eq:cIOPCYHSfermions}
    \mz_{\text{\rm cIOP}}^{(N)} (L,y) =  \prod_{\alpha=0}^{L^2-1} \sum_{e_{\alpha}=0 }^{\infty}  d_{\left\{e_{\alpha}\right\}} \left( X_{L}^{(N)} \right) ~ y^{\sum_{\alpha} e_{\alpha}} .
\end{equation}
We now use the fact that $d_{\left\{e_{\alpha}\right\}} \left( X_{L}^{(N)} \right) \in \left\{ 0, 1 \right\}$ for $X_{L}^{(N)}$ a toric Calabi--Yau variety \cite{Benvenuti:2006qr}. In particular, $d_{\left\{e_{\alpha}\right\}}  \left( \C^{L^2} \right) = 1$ for all sets $\left\{ e_{\alpha} \right\}_{\alpha=0, \dots , L^2-1}$, while some of the coefficients $d_{\left\{e_{\alpha}\right\}} \left(X_{L}^{(N)} \right)$ will vanish if $L>N$, corresponding to imposing constraints, called syzygys, on the generators of the ring $\mathfrak{R}$. With this fact, multiplying \eqref{eq:cIOPCYHSfermions} by $e^{- \beta \sum_{\alpha=0}^{L^2-1}\frac{1}{2}} = e^{- \beta L^2/2} $ proves the claim.
\end{proof}

\subsubsection{Flavor sum: Calabi--Yau ensembles}
We finish this investigation of the relationship between the IOP matrix model and Hilbert series by looking at the sum over $L$ through the lens of our Calabi--Yau varieties.
\begin{prop}
Let $N \in \mathbb{N}$, $0<y<1$ and $\mz_{\text{\rm Ex1}}$ as in \eqref{eq:avgBQCD2}. It holds that 
\begin{equation}
\label{eq:avgCIOPHS}
    \mz_{\text{\rm Ex1}} (\mathfrak{q}=1,y) = 1 + \mathrm{HS}_{y} \left( \bigsqcup_{L \ge 1} X_{L}^{(N)} \right) .
\end{equation}
\end{prop}
\begin{proof}
Let $\mathcal{O}_{X_{L}^{(N)}}$ be the sheaf of functions over $X_{L}^{(N)}$. For any $L,L^{\prime} \in \mathbb{N}$ we have the short exact sequence 
\begin{equation}
    0 \longrightarrow \mathcal{O}_{X_{L^{\prime}}^{(N)}} \longrightarrow \mathcal{O}_{X_{L}^{(N)} \sqcup X_{L^{\prime}}^{(N)} } \longrightarrow \mathcal{O}_{X_{L}^{(N)}} \longrightarrow 0 .
\end{equation}
Additivity of the Hilbert series states that 
\begin{equation}
    \mathrm{HS}_{y} \left( X_{L}^{(N)} \right) + \mathrm{HS}_{y} \left( X_{L^{\prime}}^{(N)}\right) =  \mathrm{HS}_{y} \left( X_{L}^{(N)} \sqcup X_{L^{\prime}}^{(N)} \right) .
\end{equation}
Starting with $L=1, L^{\prime}=2$ and iterating this identity in combination with Lemma \ref{lem:CYHS} gives formula \eqref{eq:avgCIOPHS}. The latter is understood as an inductive limit on $L_{\text{max}}$, with the sum over $L$ in $\mz_{\text{\rm Ex1}}$ truncated at $L \le L_{\text{max}}$.
\end{proof}
In summary, the cIOP model with Gaussian sum over flavors becomes a generating function of Hilbert series of Calabi--Yau varieties. It points toward a categorification, in which the structure rings of Calabi--Yau of different dimensions are assembled in a direct sum weighted by $\mathfrak{q}^{L^2}$.\par
As a side remark, notice that summing over moduli spaces of different dimensions is not unfamiliar in string theory. In fact, in the path integral of 2d gravity, one integrates over the moduli space of genus $\mathtt{g}$ surfaces, and eventually sums over $\mathtt{g} \in \mathbb{N}$.\par

\subsubsection{Summary: Calabi--Yau variations on IOP}
In Theorem \ref{thm:CYtoIOP} we have pointed out a correspondence between the space of states in the cIOP model, which inherits the ring structure from the ring of representations of the flavor symmetry, and the ring of functions on a certain Calabi--Yau variety $X_L ^{(N)}$. The partial dictionary so far is:
\begin{equation*}
\begin{tabular}{r c l}
\textsc{Quantum mechanics} & \hspace{8pt} & \textsc{Calabi--Yau}\\
\hline
IOP & $\quad \longleftrightarrow \quad $ & trivial, $\C^{L^2}$ \\
cIOP & $\quad \longleftrightarrow \quad $ & non-trivial, $X_L ^{(N)}$ \\
constraint $R\in \mathfrak{R}_L ^{(N)}$ & $\quad \longleftrightarrow \quad $ & syzygys \\
(unweighted) flavor sum & $\quad \longleftrightarrow \quad $ & disjoint union $\bigsqcup_{L } X_{L}^{(N)}$ \\
\hline
\end{tabular} 
\end{equation*}\par
The role of Calabi--Yau varieties, especially when summing over $L$, remains to be elucidated. We hope to report on this topic in the future.

\section{Example 2: Matrix model of \texorpdfstring{QCD$_2$}{QCD2}
}
\label{sec:ExQCD2}

Motivated by the properties evidenced in the previous example, we introduce another similar model that illustrates the paradigm of Part \ref{part1}. The corresponding matrix integral is an extremely streamlined low energy toy model of QCD$_2$. In Subsection \ref{sec:AvgQCD2FF} we extend the model according to Section \ref{sec:Fermi}, by including a sum over the number of flavors with Gaussian weight. Once again, the sum over flavors will promote a third order phase transition to a first order one. The corresponding spectral density is analyzed in Subsection \ref{sec:rhoQCD2}.

\subsection{Partition function of \texorpdfstring{QCD$_2$}{QCD2} matrix model}
\label{sec:noFAQCD2FF}

We introduce a matrix model that captures the behaviour of thermal QCD$_2$ with $N+1$ colors and $L$ flavors \cite{Hallin:1998km} (further explored and generalized in \cite{Santilli:2020ueh,Santilli:2021eon}). Its partition function is 
\begin{align}
    \mz_{\text{\rm QCD$_2$}}^{(N)} (L, y) &= \oint_{SU(N+1)} [\dd U ] \left[ \det \left(1+\sqrt{y} U \right)  \det \left(1+\sqrt{y} U^{-1} \right) \right]^{L} \label{eq:QCD2MM} 
\end{align}
with parameter $0 < y <1$. In the QCD interpretation, $L \in \mathbb{N}$ has the meaning of the number of quarks.\footnote{For simplicity, the matrix model corresponds to the low energy limit (equivalently, it takes the approximation of large quark mass $m$) of the matrix model in \cite{Hallin:1998km}, since the more general expression would not add anything new to our discussion.} We have the identification $ y = e^{- \beta m} $, with $m$ the quark mass. To avoid clutter, we simply redefine $m \beta \mapsto \beta$.\par
The matrix model \eqref{eq:QCD2MM} admits other interpretations:
\begin{enumerate}[(i)]
    \item Comparing with \cite{Minahan:1991pv}, it is a simplified model of one-plaquette lattice QCD$_2$ \cite{Santilli:2021eon}.
    \item A unitary matrix model that generalizes \eqref{eq:QCD2MM} was examined in \cite{Betzios:2017yms}, in the context of a black hole phase in AdS$_2$/CFT$_1$. The authors of \cite{Betzios:2017yms} also considered correlation functions of operators similar to but distinct from our $\mathcal{O}_L (t)$.
\end{enumerate}\par
\begin{lem}[\cite{Gessel}]
 For every $L, N \in \mathbb{N}$, let $\mz_{\text{\rm QCD$_2$}}^{(N)}$ be as in \eqref{eq:QCD2MM} with $y=e^{-\beta}$. Let also $\mathfrak{h}_L ^{(N)} \subset \N^{L}$ denote the set 
 \begin{equation}
    \mathfrak{h}_L ^{(N)} = \left\{  (h_1, \dots, h_{L}) \in \N^{L} \ : \ N+L-1 > h_1 > h_2 > \cdots > h_{L} \ge 0 \right\} .
 \end{equation}
 It holds that 
\begin{equation}
\label{eq:fermMMequalsQCD2MM}
    \mz_{\text{\rm QCD$_2$}}^{(N)} (L, e^{-\beta}) = \frac{e^{\beta L^2/2}}{G(L+1)^2} \sum_{(h_1, \dots, h_{L}) \in \mathfrak{h}_L ^{(N)} }  \prod_{1 \le i < j \le L} (h_i - h_j)^2 ~ e^{- \beta\sum_{j=1}^{L} \left( h_j + \frac{1}{2} \right) } .
\end{equation}
\end{lem}
\begin{proof}The proof is very similar to the one of Lemma \ref{lemma:cIOPBQCD2}. The change of variables \eqref{eq:changeRtoH} rewrites the right-hand side of \eqref{eq:fermMMequalsQCD2MM} as a sum over Young diagrams:
\begin{equation}
\label{eq:fermionMMRepsNosum}
     \sum_{R \subseteq (N)^L} y^{\lvert R \rvert} ~ (\dim R)^2.
\end{equation}
Here we have used formula \eqref{eq:dimRformula} for $\dim (R)$. The sum is restricted to Young diagrams that fit inside a rectangular tableau of $L$ rows and $N$ columns. For instance:
\begin{equation*}\ytableausetup{centertableaux}
			\begin{matrix} L=4 \text{  and  } N=4, \quad \\ \begin{color}{orange}R= (3,2)\end{color} \subset (4,4,4,4) \end{matrix} \qquad \qquad  
			\begin{ytableau}
		 	*(orange) \ & *(orange) \ & *(orange) \ & \  \\
		 	*(orange) \ & *(orange) \ & \ & \ \\
		  	\ & \ & \ & \ \\
		   	\ & \ & \ & \ \\
			\end{ytableau}
	\end{equation*}
On the other hand, we apply the so-called dual Cauchy identity \cite{Macdonaldbook} 
\begin{equation}
\label{eq:Cauchydual}
    \prod_{i=1}^{L} \prod_{j=1} ^{N+1} (1+y_i z_j) = \sum_{R \subseteq (N)^L} \chi_R (Y) \chi_{R^{\top}} (U)
\end{equation}
and its complex conjugate to \eqref{eq:QCD2MM}, where $R^{\top}$ is the transpose partition to $R$. We set $Y$ to have $L$ eigenvalues all equal to $\sqrt{y}$. Using the orthogonality of characters, as done to get \eqref{eq:IOPtableaux}, one finds \eqref{eq:fermionMMRepsNosum}. This proves the relation \eqref{eq:fermMMequalsQCD2MM}.
\end{proof}\par

\subsubsection{Comparison with cIOP}
The matrix model \eqref{eq:QCD2MM} is closely related to the cIOP matrix model \eqref{eq:ZBQCD2equalsZcIOP}, the difference residing in the plus signs in front of the fugacity $\sqrt{y}$ and in front of the power $L$. Because the sign in front of $L$ appears in the exponent, it changes the geometric \emph{series} $(1-\sqrt{y} z_a)^{-L}$ into \emph{polynomials} in $z_a$ of degree $L$, $\left\{ z_a \right\}_{a=1, \dots , N+1}$ being the eigenvalues of the random matrix $U$. This is reflected in the different constraints appearing in the two character expansions, where in particular $\mz_{\text{\rm QCD$_2$}}^{(N)}$ includes a sum over only a finite number of representations.\par
The distinction between the two models becomes starker if one compares their massless limits, which send $\sqrt{y} \to 1$. The limit is ill-posed in the cIOP model, due to the integrand developing a singularity. On the other hand, \eqref{eq:QCD2MM} is finite and well-defined in the limit $\sqrt{y} \to 1$. As observed in \cite{Santilli:2021eon}, these behaviors reflect the expected singularity and lack thereof in massless bosonic versus fermionic QCD$_2$. Furthermore, the latter model admits a closed form solution in the limit \cite{BS:1985}, expressed entirely with Barnes's $G$-functions:
\begin{equation}
    \lim_{y \to 1}  \mz_{\text{\rm QCD$_2$}}^{(N)} (L, y) = \frac{G(L+1)^2 G(N+1) G(2L+N+1)}{G(2L+1) G(L+N+1)^2} .
\end{equation}

\subsubsection{Schubert cells and quantum mechanics}
From \eqref{eq:fermionMMRepsNosum}, the states of the quantum mechanical model read off from $\mz_{\text{\rm QCD$_2$}}^{(N)}$ are associated with Young diagrams that fit into the $N \times L$ rectangle $(N)^L$. In turn, these Young diagrams are in one-to-one correspondence with the Schubert cells of the decomposition of the Grassmannian $\mathrm{Gr}(L,N+L)$. Thus, while we have seen Calabi--Yau varieties emerging from the cIOP model in Subsection \ref{sec:SQCD4andCY}, the relevant algebraic geometry in the current example is that of Schubert varieties.\par
The appearance of the Grassmannian and its Schubert cells is of course consistent with a QCD-type interpretation of the model.

\subsection{Partition function of \texorpdfstring{QCD$_2$}{QCD2} with sum over flavor symmetries}
\label{sec:AvgQCD2FF}

According to our paradigm, we want to enrich the model by introducing a Gaussian weight and summing over the flavor rank in QCD$_2$.
\begin{defin}
Let $0<y,\mathfrak{q}<1$, $L, N \in \mathbb{N}$ and $\mz_{\text{\rm QCD$_2$}}^{(N)}$ as in \eqref{eq:QCD2MM}. The QCD$_2$ matrix model with sum over flavors is 
\begin{equation}
\label{eq:defZAMM}
    \mz_{\text{\rm Ex2}} (\mathfrak{q},y) = \sum_{L=0} ^{L_{\max}(N)} \mathfrak{q}^{L^2} \mz_{\text{\rm QCD$_2$}}^{(N)} (L, y)  .
\end{equation}
\end{defin}
Note that in the Schur slice $\mathfrak{q}=\sqrt{y}$, the factor $\mathfrak{q}^{L^2}$ in \eqref{eq:defZAMM} cancels the $y^{-L^2/2}$ from \eqref{eq:fermMMequalsQCD2MM}.

\subsubsection{Planar limit of \texorpdfstring{QCD$_2$}{QCD2} with sum over flavor symmetries}
\begin{thm}\label{thm:Amodel}
Let $\mz_{\text{\rm Ex2}}$ be as in \eqref{eq:defZAMM}, with $y=e^{- \beta}$ and $\mathfrak{q}=e^{-1/(2a)}$. $\forall a >0$, there exists $T_H >0$ such that $\mz_{\text{\rm Ex2}}$ undergoes a first order phase transition in the planar limit at $\frac{1}{\beta} = T_H$, with 
\begin{equation}
    \ln \mz_{\text{\rm Ex2}} = \begin{cases} O(1) & \frac{1}{\beta} < T_H \\ O(N^2) & \frac{1}{\beta} > T_H . \end{cases}
\end{equation}
Moreover, in the Schur slice $a=\beta^{-1}$, $T_H \approx 1.039 $.
\end{thm}

\begin{figure}[th]
    \centering
    \begin{tikzpicture}
    \node (p) at (0,0) {\includegraphics[width=0.5\textwidth]{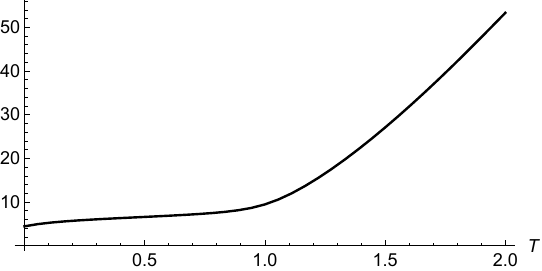}};
    \node[anchor=east] at ($(p.north west)+(0,-0.5)$) {$\scriptstyle \ln \mz_{\text{\rm Ex2}}$};
    \node at ($(p.north)+(-1,-1)$) {$\scriptstyle N=10$};
    \end{tikzpicture}
    \caption{Plot of $\ln \mz_{\text{\rm Ex2}} (\mathfrak{q},y)$ as a function of the temperature $T=-1/\ln (y)$ at $N=10$. The sum over $L$ in \eqref{eq:defZAMM} has been truncated at $L \le 20$. A change in behavior around $T \approx 1$ is already visible at small $N$, although the sharp phase transition is smoothed by finite $N$ effects. The plot is taken in the Schur slice $a=T$, but the qualitative behavior is the same in all $a$-slices of the parameter space.}
    \label{fig:ZAplotN10}
\end{figure}\par

Theorem \ref{thm:Amodel} argues that, passing to the ensemble \eqref{eq:defZAMM}, the third order phase transition experienced by $\mz_{\mathrm{QCD}_2}$ in the Veneziano limit is promoted to a first order one, at a critical temperature $T_H$. The function $\ln \mz_{\text{\rm Ex2}}$ at fixed $N$ is plotted in Figure \ref{fig:ZAplotN10}. It is also possible to evaluate numerically $\ln \mz_{\text{\rm Ex2}}$ at fixed $T$ for various $N$, and check that the plot is consistent with an approximately constant behavior if $T<T_H$ and a polynomial shape if $T>T_H$. In Appendix \ref{sec:AvgQCD2UMM} we re-derive the first order phase transition from the unitary matrix model expression \eqref{eq:QCD2MM} for $\mz_{\text{\rm QCD$_2$}}^{(N)}$, and argue that it is accompanied by Hagedorn-like behavior.\par
In conclusion, Theorem \ref{thm:Amodel} derives a first order transition with the desired holographic properties from the ingredients: 
\begin{enumerate}[$i)$]
    \item A simple quantum system built out of representations of the flavor symmetry;
    \item A constraint on the states descending from gauge invariance;
    \item A Gaussian sum over the number of quarks;
\end{enumerate}
in accordance with the general prescription of Part \ref{part1}.

\begin{proof}[Proof of Theorem \ref{thm:Amodel}]
The proof is done in two steps. 
\begin{itemize}
    \item[(1)] First we solve the planar Veneziano limit of $\mz_{\text{\rm QCD$_2$}}^{(N)}$, with 
        \begin{equation}
            L \to \infty , \quad N \to \infty , \qquad \gamma = \frac{L}{N} \text{ fixed}.
        \end{equation}
    \item[(2)] Second, we plug the result in $\mz_{\text{\rm Ex2}}$ and extremize over $\gamma$.
\end{itemize}\par
\underline{Step (1).} Most of the procedure for the large $L$ limit of \eqref{eq:fermMMequalsQCD2MM} goes through exactly as in the cIOP matrix model, see Subsection \ref{sec:cIOPlargeN}. In particular, the discrete matrix models have same summand. Writing it in the form $e^{-L^2 S (\cdots)}$ and extremizing the effective action, we we get the same saddle point equation \eqref{eq:IOPSPE}. However, the main difference here is that the discrete matrix model we deal with has two hard walls: at $h_j=0$ and $h_j = L+N-1$. In the large $N$ limit, the second hard wall for the scaled eigenvalues $x$ is placed at 
\begin{equation}
    \frac{L+N-1}{L} \approx 1+\gamma^{-1} . 
\end{equation}
In the planar limit, the eigenvalue density $\varrho (x)$ is subject to the constraints 
\begin{equation}
    \int_0 ^{\infty} \dd x \varrho (x) =1 , \qquad  \varrho (x) \le 1 , \qquad \supp \varrho \subseteq [0, 1 + \gamma^{-1} ] .
\end{equation}
That is, compared to the cIOP model, there is an additional hard wall at $x= 1+\gamma^{-1}$. This will play a role later. This discrete matrix model was first solved in \cite{CP:2013}. For the explicit comparison with \cite{CP:2013}, the dictionary is: 
\begin{equation*}\begin{tabular}{r|c|c|c|c|c}
    \cite{CP:2013} & $r$ & $s$ & $R$ & $\alpha$ & $I_{r,s}$ \\
    \hline
    here & $N+L$ & $L$ & $1+\gamma^{-1}$ & $e^{-\beta}$ & $e^{-\beta L^2/2} \mz_{\text{\rm QCD$_2$}}^{(N)}$ 
\end{tabular}
\end{equation*}\par
We thus find again a solution \eqref{eq:DKansatz}, with the same $\varrho (x)$ as in the IOP model, as long as $x_{+} < 1 + \gamma^{-1}$. When $x_{+}$ hits the hard wall, we ought to find a new solution. This takes place at 
\begin{equation}
    x_{+} (\gamma_c) = 1 + \frac{1}{\gamma_c} \quad \Longrightarrow \quad \gamma_c = \frac{1-\sqrt{y}}{2 \sqrt{y}} .
\end{equation}
As a function of the inverse temperature, $2 \gamma_c = e^{\beta/2} -1$.\par
\begin{figure}[tb]
    \centering
    \includegraphics[width=0.4\textwidth]{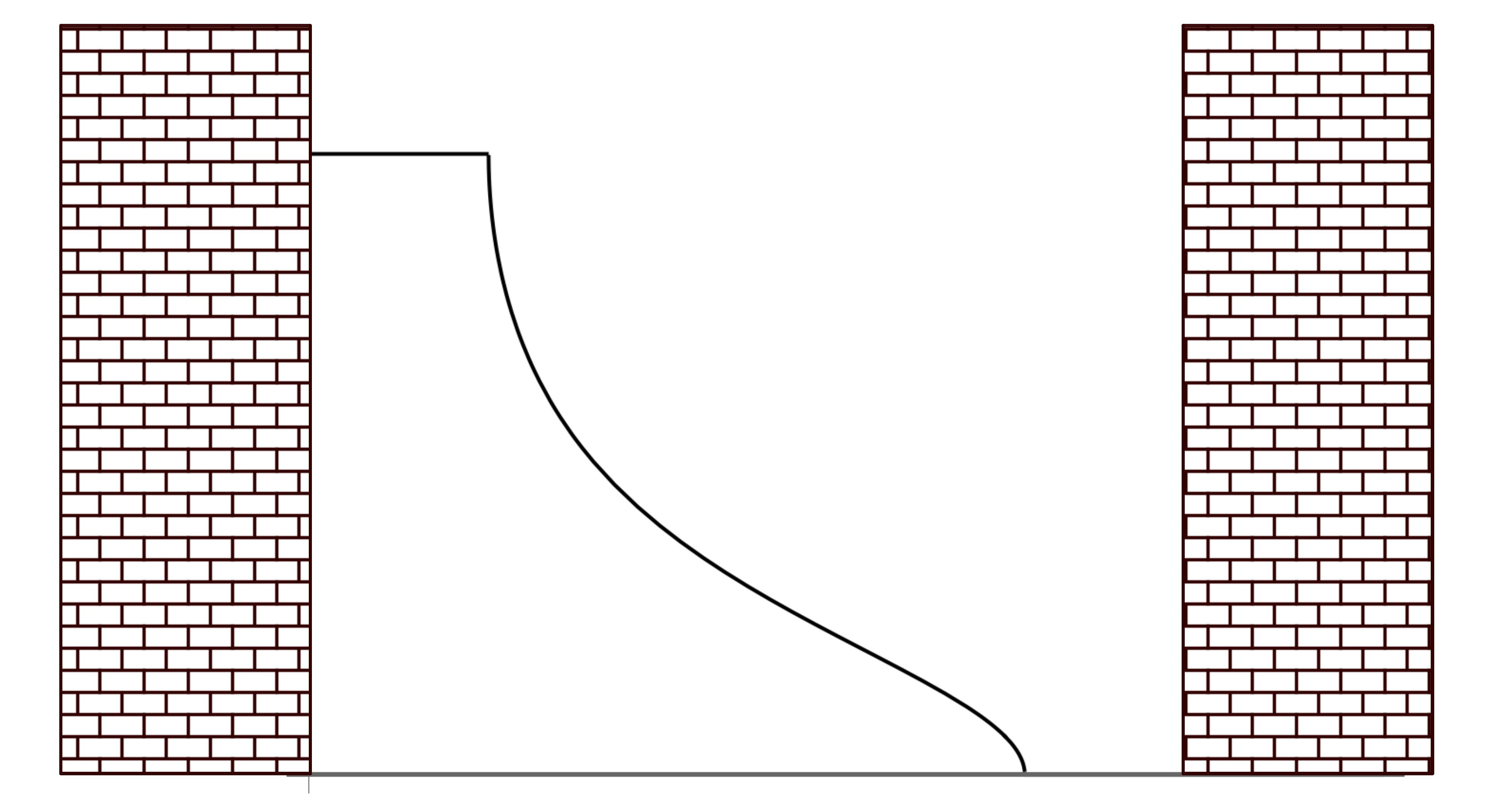}\hspace{0.05\textwidth}
    \includegraphics[width=0.4\textwidth]{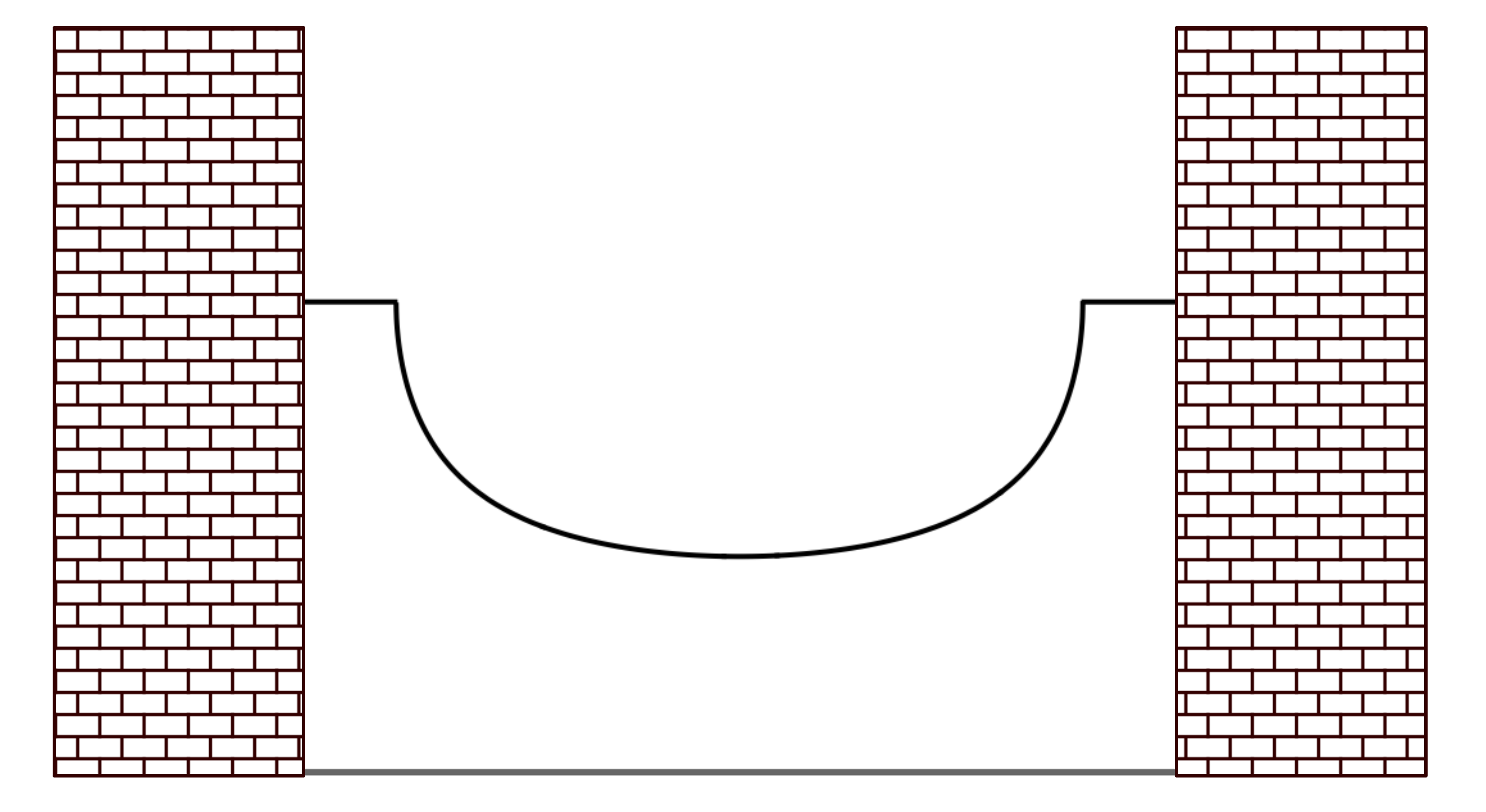}
    \caption{Schematic representation of the eigenvalue density in presence of two hard walls for the eigenvalues. Left: The hard wall on the right is not active, the eigenvalue density is capped only at the left edge. Right: Both hard walls are active, the eigenvalue density is capped at left and right edges.}
    \label{fig:cartoonhw2}
\end{figure}\par

Because of the square root singularity at the second hard wall, we replace the ansatz \eqref{eq:DKansatz} with an eigenvalue density ``capped'' both on the left and on the right edges:
\begin{equation}
\label{eq:CPansatz}
    \varrho (x) = \begin{cases} 1 & \hspace{8pt} 0 \le x < x_{-} \\ \hat{\varrho} (x) & x_{-} \le x \le x_{+} \\ 1 &  x_{+} < x \le 1+\gamma^{-1} . \end{cases}
\end{equation}
See Figure \ref{fig:cartoonhw2} for a schematic view. In this new scenario, the nontrivial part of the eigenvalue density is \cite{CP:2013}
\begin{equation}
\label{eq:CPrhohat}
    \hat{\varrho}_{\ast} (x) = 1 + \frac{2}{\pi} \left[ \atan \left( \sqrt{\frac{x_{-}}{x_{+}}}  \sqrt{\frac{x_{+} - x}{x - x_{-}}}  \right)  - \atan \left( \sqrt{\frac{1+\gamma^{-1} - x_{-}}{1+\gamma^{-1} - x_{+}}}  \sqrt{\frac{x_{+} - x}{x - x_{-}}}  \right) \right]
\end{equation}
with 
\begin{equation}
\label{eq:xpmCP}
    x_{\pm} = \frac{1}{2(1+\sqrt{y})} \left[ \left( 2 + \gamma^{-1} \right)^{1/2} \pm y^{1/4} \gamma^{-1/2} \right]^2  .
\end{equation}\par

Computing 
\begin{equation}
    \mf_{\text{\rm QCD$_2$}} = \left. \frac{1}{N^2} \ln \mz_{\text{\rm QCD$_2$}}^{(N)} \right\rvert_{\text{saddle point}} 
\end{equation}
with this eigenvalue density we get 
\begin{equation}
	 \mf_{\text{\rm QCD$_2$}}  = \begin{cases} - \gamma^2 \ln (1-y) & \gamma < \frac{1-\sqrt{y}}{2\sqrt{y}} \\ (1+2 \gamma) \ln \left(1 + \sqrt{y}\right) -  \frac{1}{4} \ln y + C(\gamma)  & \gamma > \frac{1-\sqrt{y}}{2\sqrt{y}}  \end{cases}
\end{equation}
where, in the second phase, $C(\gamma)$ is the $y$-independent term
\begin{equation}
    C(\gamma)= -\gamma^2  \ln \left(\frac{4 \gamma }{\gamma +1}\right)  +2\gamma  (\gamma +1) \ln \left(\frac{1+2\gamma }{\gamma+1}\right)-\frac{1}{2} \ln \left(\frac{(\gamma +1)^2}{1+2 \gamma}\right) - \gamma \ln 4 - \ln 2.
\end{equation}\par
\medskip
\underline{Step (2).} We express the summands in $\mz_{\text{\rm Ex2}}$ as functions of the Veneziano parameter $\gamma=L/N$ and approximate them for $N \gg 1$:
\begin{align}
    \mz_{\text{\rm Ex2}} (e^{- \beta}) & \approx \sum_{L=0}^{L_{\max}(N)} \exp \left\{  N^2 \mathcal{S}_{\beta,a} (\gamma) \right\} , \label{eq:ZEx2gammaEff} \\
   \mathcal{S}_{\beta,a} (\gamma) & := \mf_{\text{\rm QCD$_2$}}  - \frac{\gamma^2}{2a}  . \label{eq:expZEx2phase2}
\end{align}
In the phase $\gamma < \gamma_c$, the partition function takes the Gaussian form
\begin{equation}
    \mz_{\text{\rm Ex2}} (e^{- \beta}) \approx \sum_{L=0}^{L_{\max}(N)} \exp \left\{  - N^2 \gamma^2 \left[\frac{1}{2a} + \ln (1-e^{- \beta}) \right]\right\} .
\end{equation}
The saddle point is $\gamma_{\ast} =0$, valid in the low temperature region (see Figure \ref{fig:plot1})
\begin{equation}
    \frac{1}{2a} + \ln (1-e^{- \beta}) >0 .
\end{equation}
\begin{itemize}
    \item[a)] In every constant-$a$ slice of the parameter space $\left\{ (a,\beta) \in \R_{> 0} \times \R_{> 0} \right\}$, the trivial saddle $\gamma_{\ast}=0$ holds if 
    \begin{equation}
        \beta > \beta_c = - \ln \left( 1- e^{-1/(2a)} \right) ,
    \end{equation}
    equivalently $0<y<1-\mathfrak{q}$. For instance:
    \begin{equation}
        \beta_c \big\rvert_{a= 1/2} \approx 0.459, \qquad \qquad \beta_c \big\rvert_{a= 2} \approx 1.509.
    \end{equation}
    \item[b)] In the Schur slice $a=\beta^{-1}$, i.e. $\mathfrak{q}=\sqrt{y}$, the trivial saddle $\gamma_{\ast}=0$ holds if (Figure \ref{fig:plot1})
    \begin{equation}
        \frac{\beta}{2} + \ln (1-e^{- \beta}) >0 \quad \Longrightarrow \quad \beta > \beta_c \approx 0.962 .
    \end{equation}
\end{itemize}
Notice that $0<\beta_c<\infty$ exists and is finite $\forall a >0$, so that the features we describe are valid for every choice of the parameter $a$.
\begin{figure}[t]
	\centering
	\includegraphics[width=0.35\textwidth]{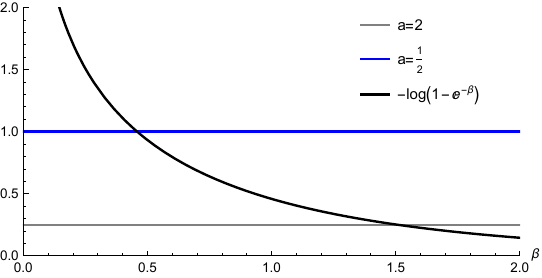}%
    \hspace{0.05\textwidth}%
    \includegraphics[width=0.35\textwidth]{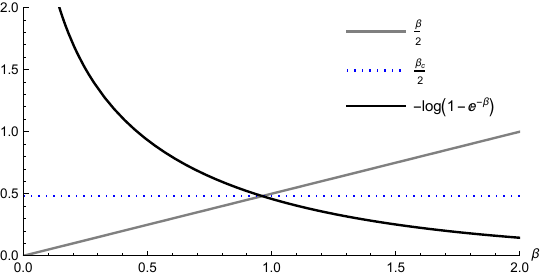}
	\caption{The saddle point $\gamma_{\ast}=0$ is valid at low temperature, $\beta > \beta_{c}$. Left: slice of constant $a$. Right: Schur slice $a=\beta^{-1}$.}
	\label{fig:plot1}
\end{figure}\par
\medskip
\begin{figure}[tb]
	\centering
	\includegraphics[width=0.32\textwidth]{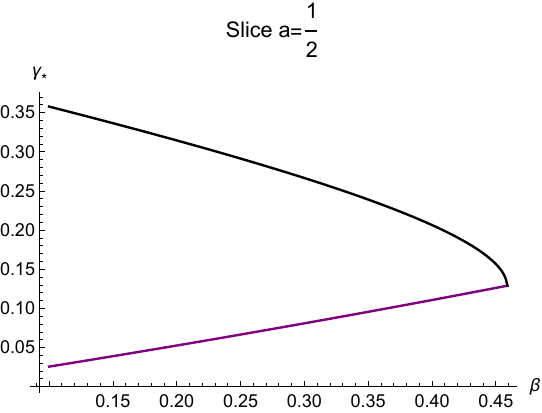}%
    \hspace{0.01\textwidth}%
    \includegraphics[width=0.32\textwidth]{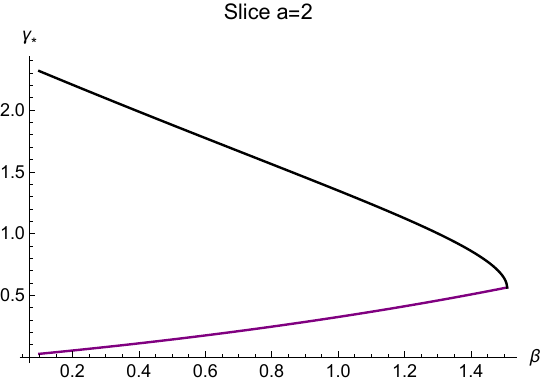}
    \hspace{0.01\textwidth}%
    \includegraphics[width=0.32\textwidth]{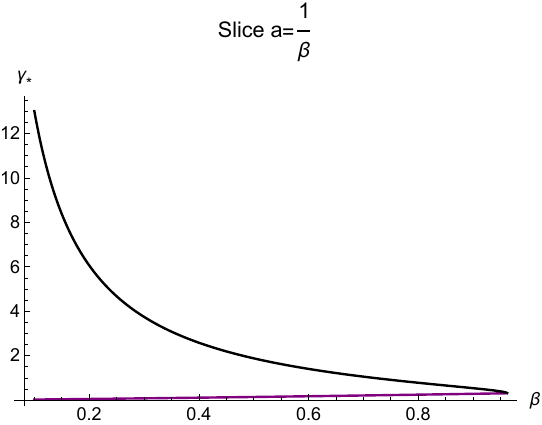}
	\caption{The absolute maximum $\gamma_{\text{\rm peak}}$ (black) and $\gamma_c$ (purple) in the region $\beta \le \beta_c$. Left: Slice $a=1/2$. Center: Slice $a=2$. Right: Schur slice $a=\beta^{-1}$.}
	\label{fig:plotgammamaxA}
\end{figure}\par

Beyond this value, $\beta < \beta_c$, the exponent in \eqref{eq:ZEx2gammaEff} changes sign and we have to look for the maximum of $ \mathcal{S}_{\beta,a} (\gamma)$.
Let us start with the assumption that $\gamma > \gamma_c$. By direct inspection, we find that 
\begin{itemize}
    \item $\mathcal{S}_{\beta,a}$ has an absolute maximum at a value $\gamma_{\text{\rm peak}}>0$;
    \item The saddle $\gamma_{\text{\rm peak}}$ exists for all $\beta \le \beta_c$;
    \item $\gamma_{\text{\rm peak}} > \gamma_c$ for all $\beta < \beta_c$ and $\gamma_{\text{\rm peak}} \big\rvert_{\beta=\beta_c} = \gamma_c \big\rvert_{\beta=\beta_c}$.
\end{itemize}
This situation is plotted in Figure \ref{fig:plotgammamaxA}. From the last point it follows that 
\begin{equation}
\label{eq:calStwophases}
    \mathcal{S}_{\beta=\beta_c,a} (\gamma \le \gamma_c) =0 , \qquad \mathcal{S}_{\beta=\beta_c,a} (\gamma > \gamma_c) <0.
\end{equation}
We thus find a positive saddle point value which is compatible with $\gamma > \gamma_c$ at all temperature above $\beta_c ^{-1}$. This is thus the saddle point $\gamma_{\ast}$ that we are looking for, and we set $\gamma_{\ast}= \gamma_{\text{\rm peak}}$.\par

We also plot $\mathcal{S}_{\beta,a} (\gamma)$ as a function of $\gamma$ at different values of $\beta$ in Figure \ref{fig:plotSbetagamma}, and a zoom-in close to the transition point is in shown Figure \ref{fig:plotSbetaZoom}. The behavior shown is agreement with the analytic calculation: the maximal contribution to $\mathcal{S}_{\beta,a} $ is from $\gamma_{\ast}=0$ if $\beta > \beta_c$ (low temperature phase), whilst $\mathcal{S}_{\beta,a} $ has a global maximum at which is is strictly positive if $\beta < \beta_c$ (low temperature phase).\par

\begin{figure}[tb]
	\centering
	\includegraphics[width=0.32\textwidth]{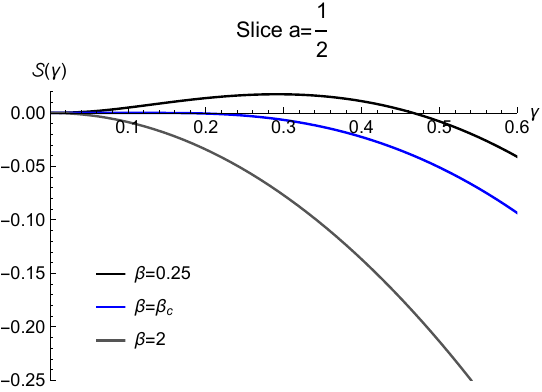}%
    \hspace{0.01\textwidth}%
    \includegraphics[width=0.32\textwidth]{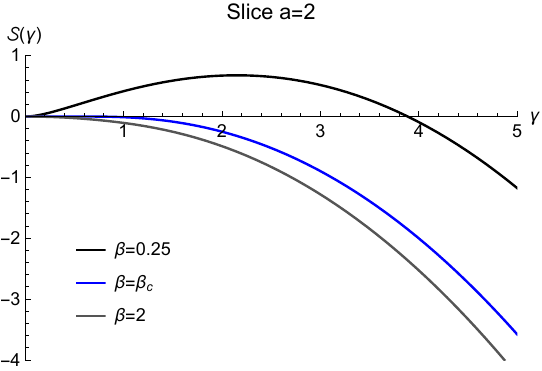}
    \hspace{0.01\textwidth}%
    \includegraphics[width=0.32\textwidth]{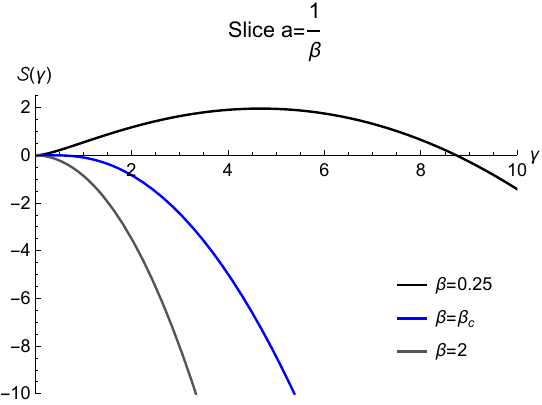}
	\caption{Plot of $\mathcal{S}_{\beta,a} (\gamma)$ as a function of $\gamma$ for $\beta > \beta_c$ (gray), $\beta =\beta_c$ (blue) and $\beta < \beta_c$ (black). Left: Slice $a=1/2$. Center: Slice $a=2$. Right: Schur slice $a=\beta^{-1}$. At high temperature, $\mathcal{S}_{\beta,a}$ has a global maximum at $\gamma_{\ast}>0$ with $\mathcal{S}_{\beta,a} (\gamma_{\ast})>0$ (black curve). Decreasing the temperature until the critical value (blue curve), $\mathcal{S}_{\beta,a} (\gamma_{\ast})=0$. Below that value, $\mathcal{S}_{\beta,a}$ is non-positive definite and vanishes at $\gamma=0$, which is the new global maximum (gray curve).}
	\label{fig:plotSbetagamma}
\end{figure}\par

\begin{figure}[tb]
	\centering
	\includegraphics[width=0.32\textwidth]{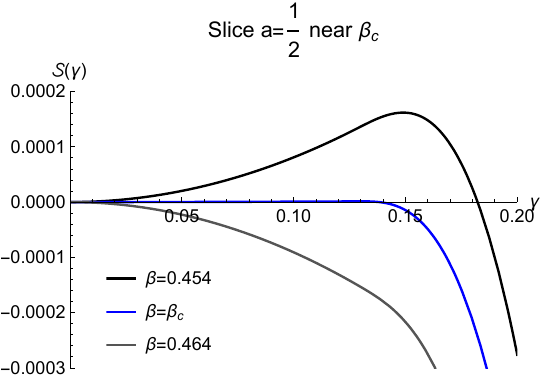}%
    \hspace{0.01\textwidth}%
    \includegraphics[width=0.32\textwidth]{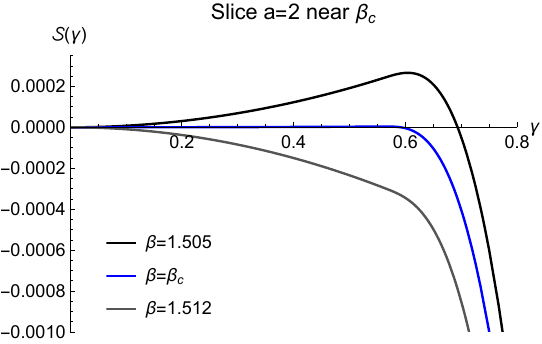}
    \hspace{0.01\textwidth}%
    \includegraphics[width=0.32\textwidth]{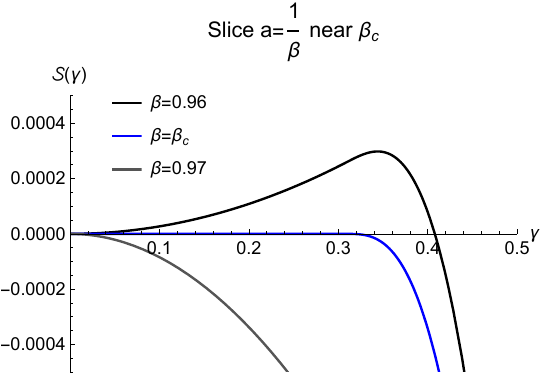}
    	\caption{Plot of $\mathcal{S}_{\beta,a} (\gamma)$ as a function of $\gamma$ for $\beta > \beta_c$ (gray), $\beta =\beta_c$ (blue) and $\beta < \beta_c$ (black). Left: Slice $a=1/2$. Center: Slice $a=2$. Right: Schur slice $a=\beta^{-1}$. These plots are analogous to Figure \ref{fig:plotSbetagamma}, but showing a narrow window around the critical temperature.}
	\label{fig:plotSbetaZoom}
\end{figure}\par
\medskip
To conclude the proof it suffices to notice that, when $\beta< \beta_c$,  $\gamma_{\ast} >0$ and the saddle point value of $\ln \mz_{\text{\rm Ex2}}$ grows with coefficient $N^2$, whilst at $\beta > \beta_c$ we have $\gamma_{\ast}=0$, the $O(N^2)$ growth is cancelled and one is left with the remnant $O(1)$ contributions. The exchange of dominance of saddle points at $\beta=\beta_c$ yields a first order transition at temperature 
\begin{equation}
    T_H = \beta_c ^{-1} .
\end{equation}
This is analogous to what happens in $\mathcal{N}=4$ super-Yang--Mills in four dimensions at finite temperature \cite{Liu:2004vy}.\par

\end{proof}

\subsection{Spectral density of \texorpdfstring{QCD$_2$}{QCD2} with sum over flavor symmetries}
\label{sec:rhoQCD2}

We now apply the technology of Subsection \ref{sec:specAvg} to QCD$_2$ with the sum over flavors. The computation is essentially the same as in the model of Subsection \ref{sec:IOPrho}, and the structure of the branch cuts of the Fourier transformed Wightman functions is very similar. Let us first consider the simpler case \emph{without} summing over $L$, and return to the sum over flavor below.\par

\subsubsection{Spectral density of \texorpdfstring{QCD$_2$}{QCD2}: no sum}
\begin{prop}
    Consider the quantum mechanical system with partition function given by the matrix model $\mz_{\text{\rm QCD}_2}^{(N)}$ in \eqref{eq:QCD2MM}. The associated spectral density has compact, continuous support on $\R$.
\end{prop}
\begin{proof}
On both sides of the third order phase transition we have an explicit formula for $\varrho_{\ast}$ and its support. In the phase in which the constraint is not active, the model behaves exactly as cIOP and IOP. In particular has the same eigenvalue density and, by Lemma \ref{lemma:OmegalargeN}, has the same Fourier transformed Wightman functions $\widetilde{G}_{L,\mathrm{R}} (\omega), \widetilde{G}_{L,\mathrm{R}} (\omega)$. It follows immediately that the spectral density $\rho (\omega)$ is the same as cIOP in that phase.\par
Beyond the third order transition, i.e. $\gamma > \gamma_c$, $\varrho_{\ast}$ is supported on $[0, 1+\gamma^{-1}]$ and is non-constant on $[x_{-}, x_{+}]$, as given in \eqref{eq:xpmCP}. It is clear that the open interval $x_{-}<x< x_{+}$ is non-empty if $\gamma_c <\gamma< \infty$. Applying again Lemma \ref{lemma:OmegalargeN}, we obtain the branch cuts of $\widetilde{G}_{L,\mathrm{R}} (\omega), \widetilde{G}_{L,\mathrm{R}} (\omega)$ along $[x_{-}+\mu , x_{+}+\mu ]$ if $\omega >\mu$, and along $[-x_{+}-\mu, -x_{-}-\mu]$ if $\omega <-\mu$. In contrast, the regions $0<x<x_{-}$ and $x_{+} < x < 1+\gamma^{-1}$ where $\varrho_{\ast}$ is constant contribute to the smooth part of $\widetilde{G}_{L,\mathrm{R}} (\omega), \widetilde{G}_{L,\mathrm{R}} (\omega)$. Using the relation \eqref{eq:rhodiscG} between the spectral density $\rho (\omega) $ and the Fourier transformed Wightman functions, we conclude 
\begin{equation}
    \supp \rho = [-x_{+}-\mu, -x_{-}-\mu] \cup [x_{-}+\mu , x_{+}+\mu ] ,
\end{equation}
explicitly known as a function of $y=e^{-\beta} $ and of the Veneziano parameter $\gamma$, cf. \eqref{eq:xpmCP}.
\end{proof}\par

\subsubsection{Spectral density of \texorpdfstring{QCD$_2$}{QCD2}: sum over flavor symmetries}
Let us now reintroduce the sum over flavors. The argument in the high temperature phase is analogous to the one just shown. We again use 
\begin{equation}
    \rho (\omega) = \lim_{\varepsilon \to 0^{+} } \left[ \widetilde{G}_{\mathrm{R}} (\omega + i \varepsilon ) - \widetilde{G}_{\mathrm{A}} (\omega -  i \varepsilon) \right]
\end{equation}
and the fact that, for $\beta^{-1} > T_H$, $\widetilde{G}_{\mathrm{R}} (\omega), \widetilde{G}_{\mathrm{A}} (\omega)$ have branch cuts along $\pm \supp \varrho_{\ast}$, which is the saddle point eigenvalue density obtained in \eqref{eq:CPansatz} and further evaluated at $\gamma=\gamma_{\ast}$. Lacking an explicit formula for $\gamma_{\ast}$ as a function of $\beta$, we are unable to provide $\rho (\omega)$ in closed form. However, computing $\gamma_{\ast}$ at a given temperature, e.g. numerically, one simply needs to plug that number in \eqref{eq:xpmCP}.\par
The knowledge of the branch cuts in the high temperature phase, as well as the knowledge that the eigenvalues do not condensate in a continuum spectrum in the low temperature phase (both facts shown in Theorem \ref{thm:Amodel}) are sufficient to argue for the structure of $\supp \rho$.
\begin{thm}
    The spectral density $\rho(\omega)$ associated to the matrix model $\mz_{\text{\rm Ex2}}$ has compact, continuous support at $1/\beta >T_H$.
\end{thm}
Combining this statement with the general recipe of Section \ref{sec:vNtot}, we conclude that the QCD$_2$ model summed over the number of flavors with Gaussian weight carries a type III$_1$ factor above $T_H$, and trivially, a type I factor below $T_H$.
\begin{equation*}
\begin{tabular}{c|c|c}
\hspace{8pt} \textsc{Temperature} \hspace{8pt} & \hspace{8pt} $ \ln \mz_{\text{\rm Ex1}} $ \hspace{8pt} & \hspace{8pt} \textsc{Algebra type} \hspace{8pt} \\
\hline 
$\beta^{-1}<T_H$ & $O(1)$ & I\ \\
$\beta^{-1}>T_H$ & $O(N^2)$ & III$_1$ \\
\hline
\end{tabular}
\end{equation*}

\section{Example 3: Conifold Donaldson--Thomas partition function}
\label{sec:conifold}

The next example we consider does not fully fit in our general discussion because it neither undergoes a phase transition, nor has a flavor symmetry indexed by $L$ which we can sum over. It nevertheless fits in the prescription to derive a quantum mechanical system from a matrix model. It turns out that the resulting quantum system has probe correlation functions characterized by a type III$_1$ von Neumann algebra at every temperature in the large $N$ limit.\par
\medskip
Consider the family of matrix models 
\begin{equation}
    \mz^{(N)} (Q, Y) = \oint_{SU(N+1)} [\dd U] ~ \exp \left\{ \tr \ln \left( 1 + Y \otimes U \right) -  \tr \ln \left( 1 + Q Y \otimes U^{\dagger} \right)  \right\} ,
\end{equation}
where  $Q \in \C $ is a scalar parameter, $Y$ is a square matrix, possibly of infinite size, with real eigenvalues $(y_1, y_2, \dots )$, and the trace is over both $U$- and $Y$-indices (i.e. $\tr $ in this formula means the trace in the tensor product).\par
Setting 
\begin{equation}
    Q=1, \qquad y_i = \begin{cases} e^{- \beta/2} & 1 \le i \le L \\ 0 & i>L \end{cases}
\end{equation}
we get a hybrid model between the two previous examples. It can be studied by the methods above and we omit this analysis.\footnote{The phase diagram of matrix models interpolating between the three cases was obtained in \cite{Santilli:2021eon}.}\par 
We consider instead the specialization 
\begin{equation}
\label{eq:Ex3conifoldY}
    0<\lvert Q \rvert<1 , \qquad y_i = q^{i - \frac{1}{2} } \ \forall i \ge 1 
\end{equation}
for a fugacity $q$ with $\lvert q \rvert <1$. The constraint $\lvert Q \rvert \ne 1$ is to avoid singularities and we restrict for concreteness to the interior of the punctured unit disk without substantial loss of generality. We then define 
\begin{equation}
    \mz^{(N)} _{\text{\rm conifold}} (-Q, q)  := \left. \mz^{(N)} (Q, Y) \right\rvert_{(Q, Y)\text{ as in \eqref{eq:Ex3conifoldY}}} .
\end{equation}
The minus sign convention in the argument $-Q$ on the left-hand side is just to reduce clutter in later expressions. The matrix model with these parameters reads 
\begin{equation}
\label{eq:conifoldUMM}
    \mz^{(N)} _{\text{\rm conifold}} (-Q, q) =  \oint_{SU(N+1)} [\dd U] ~\prod_{j =1}^{\infty} \frac{ \det \left( 1 + q^{j- \frac{1}{2} } U \right) }{ \det \left( 1 +Q  q^{j- \frac{1}{2} } U^{\dagger} \right) }  ,
\end{equation}
with the determinant taken over $SU(N+1)$. Up to an overall factor given by $\mathcal{M} (-q)^2$, where $\mathcal{M}$ is the MacMahon function
\begin{equation}
\label{eq:Macmahon}
    \mathcal{M} (-q) := \prod_{j=1}^{\infty} \left( 1-(-q)^{j} \right)^{-j} ,
\end{equation}
the matrix model \eqref{eq:conifoldUMM} is the generating function of Donaldson--Thomas (DT) invariants of the resolved conifold \cite{Ooguri:2010yk,Szabo:2010sd} (here we follow \cite{Szabo:2010sd}). More precisely, \eqref{eq:conifoldUMM} is a truncation of the full generating function of DT invariants, which is recovered at $N \to \infty$. This is the limit we are interested in. For definition and generalities on DT invariants and their generating functions we refer to \cite{Maulik:2003rzb}.\par 
The parameter $Q \in \C$ is such that $-\ln Q $ is the (complexified) K\"ahler parameter of the exceptional rational curve in the resolution of the conifold, and $q=e^{-g_{\text{\tiny string}}}$. Therefore:
\begin{itemize}
    \item The point $Q \to 0$ is usually referred to as large volume point, and it is possible to extract DT invariants from a Taylor expansion of \eqref{eq:conifoldUMM} around this point;
    \item The singularity of the matrix model at $Q =1=q$ signals the singularity at the conifold point of the K\"ahler moduli space, where the rational curve is blown down.
\end{itemize}\par
Explicitly, let $\mz^{\text{\rm DT}} _{\text{\rm CY3}} (Q, q)$ denote the generating function of DT invariants of any Calabi--Yau threefold. The contribution from degree 0 sub-schemes (D0-branes in the string theory parlance) of a Calabi--Yau with topological Euler characteristic $\chi (\text{CY3})$ is \cite{Maulik:2003rzb,Li:2006}
\begin{equation}
    \left. \mz^{\text{\rm DT}} _{\text{\rm CY3}} (Q, q) \right\rvert_{\text{degree 0}} = \mathcal{M} (-q)^{\chi (\text{CY3})} ,
\end{equation}
independent of $Q$. A reduced generating function of DT invariants was introduced and studied by Maulik--Nekrasov--Okounkov--Pandharipande \cite{Maulik:2003rzb}.
\begin{defin}
    The \emph{reduced} generating function of DT invariants of a Calabi--Yau threefold is $\mz^{\text{\rm DT}} _{\text{\rm CY3}}$ with the generating function of degree 0 sub-schemes stripped off, 
    \begin{equation}
        \mz^{\text{\rm DT}} _{\text{\rm CY3}} (Q, q) / \left. \mz^{\text{\rm DT}} _{\text{\rm CY3}} (Q, q) \right\rvert_{\text{degree 0}} .
    \end{equation}
\end{defin}
Moreover we have \cite{Szabo:2010sd}
\begin{equation}
    \mz^{\text{\rm DT}} _{\text{\rm conifold}} (Q, q) = \mathcal{M} (-q)^2 \lim_{N \to \infty}  \mz^{(N)} _{\text{\rm conifold}} (-Q, q) .
\end{equation}
Comparing the expressions and using $\chi (\text{conifold})=2$, we conclude that 
\begin{lem}[\cite{Ooguri:2010yk,Szabo:2010sd}]
    $\mz^{(N)} _{\text{\rm conifold}} (-Q, q)$ is a truncation of the reduced generating function of DT invariants of the resolved conifold.
\end{lem}\par
\medskip
To derive a quantum mechanical interpretation along the lines of Part \ref{part1}, we write down the character expansion of \eqref{eq:conifoldUMM} and set $-Q=e^{-\beta}$. That is, the K\"ahler parameter in the resolution of the conifold and the inverse temperature are related through
\begin{equation*}
    \text{inverse temperature $\beta$ } \ \longleftrightarrow \ \text{ K\"ahler parameter $\beta - i \pi $}.
\end{equation*}
\begin{lem}
    With the notation above, the equality 
    \begin{equation}
    \label{eq:Zconifold}
    \mz^{(N)} _{\text{\rm conifold}} (e^{-\beta}, q) = \sum_{R \ : \ \ell (R) \le N } \left( \frac{e^{- \beta}}{q} \right)^{ \lvert R \rvert } ~ \chi_R \left( \mathrm{diag} (q,q^2, q^3, \dots ) \right) \chi_{R^{\top}} \left( \mathrm{diag} (q,q^2, q^3, \dots ) \right)  
    \end{equation}
    holds, where $R^{\top}$ denotes the transpose partition to $R$.
\end{lem}
\begin{proof}
    Applying the Cauchy identity \eqref{eq:Cauchyid} to the denominator of \eqref{eq:conifoldUMM} gives:
    \begin{align}
        \left[ \prod_{j=1}^{\infty}\det \left( 1 +Q  q^{j- \frac{1}{2} } U^{\dagger} \right) \right]^{-1} &= \sum_{\tilde{R} \ : \ \ell (\tilde{R}) \le N } \chi_{\tilde{R}} \left( \mathrm{diag} (-Q q^{\frac{1}{2}},-Q q^{\frac{3}{2}}, -Q q^{\frac{5}{2}}, \dots ) \right) \chi_{\tilde{R}} (U^{\dagger}) \notag \\
        &=  \sum_{\tilde{R} \ : \ \ell (\tilde{R}) \le N } \left( -Q q^{-\frac{1}{2}} \right)^{\lvert \tilde{R} \rvert} \chi_{\tilde{R}} \left( \mathrm{diag} (q, q^2, q^3, \dots ) \right) \chi_{\tilde{R}} (U^{\dagger}) .
    \end{align}
    Applying the dual Cauchy identity \eqref{eq:Cauchydual} to the numerator of \eqref{eq:conifoldUMM} gives:
    \begin{align}
        \prod_{j =1}^{\infty} \det \left( 1 + q^{j- \frac{1}{2} } U \right) &= \sum_{R \ : \ \ell (R^{\top}) \le N } \chi_R \left( \mathrm{diag} (q^{\frac{1}{2}},q^{\frac{3}{2}}, q^{\frac{5}{2}}, \dots ) \right) \chi_{R^{\top}} (U) \notag \\
        &= \sum_{R \ : \ R_1 \le N } \left( q^{- \frac{1}{2}} \right)^{\lvert R \rvert }  \chi_R \left( \mathrm{diag} (q,q^2, q^3, \dots ) \right) \chi_{R^{\top}} (U) .
    \end{align}
    Combining the two expressions, the integration over $SU(N+1)$ yields $\delta_{R^{\top}, \tilde{R}}$. We thus obtain the character expansion 
\begin{equation}
    \mz^{(N)} _{\text{\rm conifold}} (-Q, q) = \sum_{R \ : \ R_1 \le N } \left( - \frac{Q}{q} \right)^{ \lvert R \rvert } ~ \chi_R \left( \mathrm{diag} (q,q^2, q^3, \dots ) \right) \chi_{R^{\top}} \left( \mathrm{diag} (q,q^2, q^3, \dots ) \right)  .
\end{equation}
We relabel $R^{\prime} = R^{\top}$ (and drop the prime) to rewrite the sum over representations restricted to those of length at most $N$. Finally we set $-Q=e^{-\beta}$.
\end{proof}

The model \eqref{eq:Zconifold} belongs to the family of $q$-ensembles and it can be shown that it possesses a single phase in the large $N$ limit. We do not study it explicitly here, but discuss a related, albeit simpler, model in Appendix \ref{app:qCSMM}. The large $N$ limit of $\mz^{(N)} _{\text{\rm conifold}}$ can be analyzed along the same lines.\par
Here, instead, we are interested in reading off a quantum mechanical system from the character expansion \eqref{eq:Zconifold}. The presence of the fugacity $q$ prevents us from interpreting the terms $\chi_R$ as accounting for degeneracy. To overcome this difficulty, we work in the limit $q \to 1$, which in the string theory picture means $g_{\text{\tiny string}} \to 0^{+}$.\par
A few technical comments:
\begin{itemize}
    \item In this simplified regime the MacMahon function is singular, but the \emph{reduced} generating function is well-defined. 
    \item It is an intriguing fact that, despite our character expansion and ensuing construction need not know anything about a string theory origin of the matrix models, to have a quantum mechanical interpretation requires us to send $g_{\text{\tiny string}} \to 0^{+}$.
    \item Being interested in the quantum mechanical interpretation, we stick to the limit $q \to 1 $, which unrefines the DT partition function. This should not be confused with the ``decategorification'' limit $q \to -1$ sometimes considered in enumerative geometry.
\end{itemize}\par
We therefore consider 
\begin{equation}
\label{eq:ZEx3def}
    \mz_{\text{\rm Ex3}} (e^{-\beta}) := \lim_{q \to 1} \mz^{(N)} _{\text{\rm conifold}} (e^{-\beta}, q) = \sum_{R \ : \ \ell (R) \le N } e^{- \beta \lvert R \rvert } ~ \dim (R) \dim (R^{\top}) .
\end{equation}
Let us summarize the enumerative meaning of $\mz_{\text{\rm Ex3}}$.
\begin{center}
\noindent\fbox{%
\parbox{0.98\linewidth}{%
\begin{prop}
    Let $\mz_{\text{\rm Ex3}}$ be as in \eqref{eq:ZEx3def}. The large $N$ limit $\lim_{N \to \infty} \mz_{\text{\rm Ex3}} (e^{- \beta})$ yields the reduced, unrefined generating function of Donaldson--Thomas invariants of the resolved conifold, as a function of the K\"ahler parameter $\beta - i \pi$.
\end{prop}}}\end{center}\par
To emphasize the difference with the previous models, let us remind the reader that the representation $R^{\top}$ is (generically) not isomorphic to $R$ nor to $\overline{R}$. For instance, assume $N=4$ and $R=(3,3,1,0)$, which gives $R^{\top}=(3,2,2,0)$:
\begin{equation}
\label{eq:R331}
 \ytableausetup{centertableaux,smalltableaux}
 R=(3,3,1,0) \quad \begin{ytableau} \ & \ & \  \\ \ & \ & \  \\  \ \end{ytableau}  \qquad \Longrightarrow \qquad R^{\top}=(3,2,2,0) \quad \begin{ytableau} \ & \ & \  \\ \ & \  \\  \ & \ \end{ytableau} 
\end{equation}
By formula \eqref{eq:dimRformula}, the two have different dimensions as $SU(5)$ representations,
\begin{equation}
    \dim (3,3,1,0) = 60 , \qquad \dim (3,2,2,0) = 36 .
\end{equation}
Another well-known example is when $R$ is the rank-$k$ antisymmetric representation, then $R^{\top}$ is the rank-$k$ symmetric representation, 
\begin{equation}
    R=(\underbrace{1, \dots, 1}_k , 0, \dots ,0) \qquad \Longrightarrow \qquad R^{\top}=(k,0, \dots, 0) ,
\end{equation}
and the two are not isomorphic unless $k=1$.\par
\medskip
Let us focus on the quantum mechanical interpretation of \eqref{eq:ZEx3def}, where now all the pieces are consistent with our general analysis in Section \ref{sec:QM}. In this example we have $\phi (R) = R^{\top}$, which requires minor edits to the interaction Hamiltonian $H_{\text{\rm int}}$ compared to the case $\phi (R) = \overline{R}$. We suitably adjust $H_{\text{\rm int}}$ so that the interaction energy $E_J ^{\text{\rm int}}$ is as in the general discussion in Subsection \ref{sec:spectral}, possibly up to $1/N$ terms.\par
\medskip
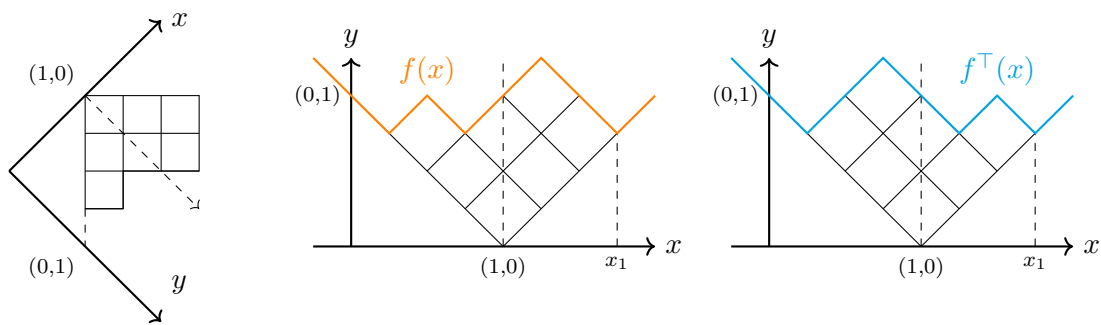
\begin{figure}[ht]
    \centering
    \begin{tikzpicture}
        \draw[-] (-4.5,0) -- (-3,0) -- (-3,-1) -- (-4,-1) -- (-4,-1.5) -- (-4.5,-1.5) -- (-4.5,0);
        \draw (-4.5,-0.5) -- (-3,-0.5);
        \draw (-4.5,-1) -- (-3,-1);
        \draw (-4.5,-1.5) -- (-4,-1.5);
        \draw (-4,0) -- (-4,-1.5);
        \draw (-3.5,0) -- (-3.5,-1);
        \draw[thick,->] (-5.5,-1) -- (-3.5,1);
        \draw[thick,->] (-5.5,-1) -- (-3.5,-3);
        \draw[dashed,->] (-4.5,0) -- (-3,-1.5);
        \draw[dashed] (-4.5,-1.5) -- (-4.5,-2);
        \node[anchor=west] at (-3.5,1) {$x$};
        \node[anchor=west] at (-3.5,-2.5) {$y$};
        \node[anchor=south east] at (-4.5,0) {$\scriptstyle (1,0)$};
        \node[anchor=north east] at (-4.5,-2) {$\scriptstyle (0,1)$};

        \draw[thick,->] (-1.5,-2) -- (3,-2);
        \draw[thick,->] (-1,-2) -- (-1,0.5);
        \node[anchor=west] at (3,-2) {$x$};
        \node[anchor=south] at (-1,0.5) {$y$};
        \draw[] (1,-2) -- (2.5,-0.5);
        \draw[] (1,-2) -- (-1,0);
        \draw[] (2,0) -- (0.5,-1.5);
        \draw[] (0.5,-0.5) -- (0,-1);
        \draw[] (.5,-0.5) -- (1.5,-1.5);
        \draw[] (1,0) -- (2,-1);
        \draw[dashed] (2.5,-2) -- (2.5,-0.5);
        \draw[dashed] (1,-2) -- (1,0.5);
        \node[anchor=north] at (1,-2) {$\scriptstyle (1,0)$};
        \node[anchor=east] at (-1,0) {$\scriptstyle (0,1)$};
        \node[anchor=north] at (2.5,-2) {$\scriptstyle x_1$};
        \node[anchor=south] at (0,0) {$\begin{color}{orange}f(x)\end{color}$};
        \draw[orange,thick] (3,0) -- (2.5,-0.5) -- (1.5,0.5) -- (.5,-0.5) -- (0,0) -- (-.5,-.5) -- (-1.5,0.5);

        \draw[thick,->] (4,-2) -- (8.5,-2);
        \draw[thick,->] (4.5,-2) -- (4.5,0.5);
        \node[anchor=west] at (8.5,-2) {$x$};
        \node[anchor=south] at (4.5,0.5) {$y$};
        \draw[] (6.5,-2) -- (8,-0.5);
        \draw[] (6.5,-2) -- (4.5,0);
        \draw[] (7.5,0) -- (6,-1.5);
        \draw[] (6.5,0) -- (5.5,-1);
        \draw[] (5.5,0) -- (7,-1.5);
        \draw[] (6.5,0) -- (7.5,-1);
        \draw[dashed] (8,-2) -- (8,-0.5);
        \draw[dashed] (6.5,-2) -- (6.5,0.5);
        \node[anchor=north] at (6.5,-2) {$\scriptstyle (1,0)$};
        \node[anchor=east] at (4.5,0) {$\scriptstyle (0,1)$};
        \node[anchor=north] at (8,-2) {$\scriptstyle x_1$};
        \node[anchor=south] at (7.5,0) {$\begin{color}{Cerulean}f^{\top}(x)\end{color}$};
        \draw[Cerulean,thick] (8.5,0) -- (8,-0.5) -- (7.5,0) -- (7,-0.5) -- (6,0.5) -- (5,-.5) -- (4,0.5);

    \end{tikzpicture}
    \caption{Coordinate system for Young diagrams, exemplified for $N=4$ and $R, R^{\top}$ as in \eqref{eq:R331}. Left: Orientation of the $(x,y)$-axes with respect to $R$. Center: The shape function $f(x)$ is shown in orange. Right: The shape function $f^{\top} (x)$ of the transposed diagram $R^{\top}$ is shown in cyan.}
    \label{fig:coordR}
\end{figure}
To study the large $N$ limit of $ \mz_{\text{\rm Ex3}} $ we adopt the method of IOP \cite{Iizuka:2008eb}. After appropriate rescaling, the Young diagram $R$ can be represented by a piece-wise linear function as follows. One sets the $(x,y)$-axes rotated by $-135^{\circ}$ with respect to the diagram $R$ and with the $x$-axis flipped, as shown in Figure \ref{fig:coordR} (left). The origin of the $(x,y)$-plane is chosen so that the top-left corner of $R$ has coordinates $(1,0)$, and the length of the axis is scaled by $1/N$, so that the constraint $\ell (R) \le N$ becomes the requirement $R$ is entirely contained in the positive quadrant. Then the shape of $R$ is the graph of a piece-wise linear function $f(x)$ with 
\begin{equation}
    f(x) \ge \lvert x-1 \rvert , \qquad f(x) = -x+1 \text{ if } x \le 0
\end{equation}
and there exists $x_{1} \ge 1$ such that $f(x) = x-1$ if $x>x_1$. The conventions are chosen to match with \cite[Sec.5]{Iizuka:2008eb}.\par
In this coordinate system, transposition $R \mapsto R^{\top}$ sends $f \mapsto f^{\top}$ and it acts as a reflection about the vertical axis $x=1$. Therefore 
\begin{equation}
    f^{\top} (x)= f(2-x) 
\end{equation}
as shown in Figure \ref{fig:coordR} (right).\par
Writing \eqref{eq:ZEx3def} as 
\begin{equation}
    \mz_{\text{\rm Ex3}} (e^{-\beta}) = \sum_{R \ : \ \ell (R) \le N } \exp \left\{ - \frac{1}{2}  \left[\beta \lvert R \rvert - 2 \ln \dim (R) \right]  - \frac{1}{2}  \left[\beta \lvert R^{\top} \rvert - 2 \ln \dim (R^{\top}) \right] \right\} ,
\end{equation}
expressing the right-hand side in terms of the shape function $f$ and comparing with the setup of \cite[Sec.5]{Iizuka:2008eb}, we have
\begin{equation}
    \mz_{\text{\rm Ex3}} (e^{-\beta}) = \sum_{R \ : \ \ell (R) \le N } \exp \left\{ - \frac{N^2}{2} \left( \mathcal{S}_{\text{\rm IOP}} [f] + \mathcal{S}_{\text{\rm IOP}} [f^{\top}]  \right) \right\} 
\end{equation}
where $\mathcal{S}_{\text{\rm IOP}} [f]$ is the effective action that appears in the IOP model. The saddle point equation (derived analogously to \cite[Sec.5.2]{Iizuka:2008eb}) is:
\begin{equation}
    \beta - 2 \ln (x) + \int_{-\infty}^{\infty} \dd \xi \ln \lvert x-\xi \rvert \left( \frac{ f^{\prime \prime} (\xi) + f^{\prime \prime} (2-\xi) }{2} \right) =0
\end{equation}
for $0 <x<x_1$. While we do not solve this problem explicitly, one checks that the terms with positive and negative signs compete, exactly as in the IOP model, yielding a non-trivial saddle point. We conclude that $\ln \mz_{\text{\rm Ex3}} = O(N^2)$ in the large $N$ limit.\par
An alternative way to argue for the same result without inspecting the details of the saddle point equation is that the typical $SU(N+1)$ representation at large $N$ has all rows of length $O(N)$ and $\ln \dim (R) = O(N^2)$. The only difference between this example and the IOP model of Section \ref{sec:IOP} is the appearance of $\dim (R)\dim (R^{\top})$ instead of $\dim (R)^2$. For typical $R$, $\ln \dim (R^{\top}) = O(N^2)$ as well, and moreover there are no other negative signs that may produce cancellations. The saddle point shape $f_{\ast}$ is thus expected to be similar to the IOP solution.\par
Concretely, in the low temperature limit we approximate 
\begin{equation}
    \tanh \left( \frac{\beta}{4} \right) \approx 1-2 e^{- \beta/2} , \qquad \coth \left( \frac{\beta}{4} \right) \approx 1+2 e^{- \beta/2} ,
\end{equation}
and the IOP limit shape is supported on $[ 1-2 e^{- \beta/2},  1+2 e^{- \beta/2}]$ and symmetric about the vertical axis $x=1$. Therefore at low temperature, $f_{\ast}$ in this example agrees with the IOP limit shape, and it is deformed away from the IOP shape as the temperature is increased ($\beta$ is decreased), but in a continuous way, for which the property $\ln \mz_{\text{\rm Ex3}} = O(N^2)$ persists.\par
Yet another way to show that the saddle point eigenvalue density $\varrho_{\ast}$ of this model is a continuous function, would be to study directly the unitary matrix model. In that case it is easier to work in a 't Hooft limit 
\begin{equation}
    g_{\text{\tiny string}} \to 0^{+}, \qquad N \to \infty , \qquad \text{ with } \lambda_{\text{\tiny string}} := N g_{\text{\tiny string}} \text{ fixed},
\end{equation}
and take the unrefined limit $\lambda_{\text{\tiny string}} \to 0$ at the end of the computation. The procedure is rather cumbersome and we do not present it here. However, one observes that it admits a solution consistent with the large $N$ growth $\ln \mz_{\text{\rm Ex3}} = O(N^2)$ and moreover we have not found any phase transition as a function of $\beta$, supporting the picture advocated for using typical representations.\par
\bigskip
Constructing the Wightman functions as explained in Subsection \ref{sec:spectral}, this observation is enough to deduce that $\rho (\omega)$, the spectral density of the quantum system coupled to a probe, has continuous support. We conclude that these correlation functions are those of a von Neumann algebra of type III$_1$. 
\begin{equation*}
\begin{tabular}{c|c|c}
\hspace{8pt} \textsc{Temperature} \hspace{8pt} & \hspace{8pt} $ \ln \mz_{\text{\rm Ex3}} $ \hspace{8pt} & \hspace{8pt} \textsc{Algebra type} \hspace{8pt} \\
\hline 
$\forall \beta$ & $O(N^2)$ & III$_1$ \\
\hline
\end{tabular}
\end{equation*}

\section{Example 4: Systems with a Casimir Hamiltonian}
\label{sec:YM2Global}

Consider a system with $U(L)$ global symmetry, and assume its finite temperature partition function takes the form 
\begin{equation}
\label{eq:YM2MM}
    \mz_{\text{\rm YM$_2$}} (L, e^{- \beta}) = \sum_{R \in \mathfrak{R}^{U(L)}} (\dim R)^2 ~e^{- \frac{\beta}{b} C_2 (R)} ,
\end{equation}
where the sum is over all isomorphism classes of irreducible $U(L)$ representations, and 
\begin{equation}
    C_2 (R) = \sum_{i=1}^{L} R_i (R_i -2i + L +1)
\end{equation}
is the quadratic Casimir invariant. The number $b^{-1}>0$ is interpreted as a coupling constant.\par
The inspiration to write down this system is twofold:
\begin{itemize}
    \item Take a two-dimensional gravity system with gauge group $U(L)$, dual to an ensemble of boundary theories with global $U(L)$ symmetry. It was shown in \cite{Kapec:2019ecr} that the partition function of the random matrix ensemble with global symmetry decomposes according to \eqref{eq:YM2MM}.
    \item Take any unitary CFT in $d \ge 2$ spacetime dimensions, with $U(L)$ global symmetry, at finite temperature $T$. In \cite{Kang:2022orq} it was found that the probability of a state to be in the irreducible representation $R$, in the high temperature limit, is given by 
    \begin{equation}
    \label{eq:coefProbRKLO}
        \mathrm{Prob} (R) = \left( \frac{ 4 \pi}{b T^{d-1}} \right)^{\frac{L^2}{2}} ~ (\dim R)^2 \exp \left( - \frac{1}{b T^{d-1}} C_2 (R) \right) .
    \end{equation}
    The $R$-independent coefficient is just a normalization. Up to the overall coefficient, summing over all the states we get \eqref{eq:YM2MM}, with identification 
    \begin{equation}
        \beta = \frac{1}{T^{d-1}} .
    \end{equation}
    In particular, for a two-dimensional CFT, we have the identification between the quantum mechanical $\beta$ and the inverse temperature $T^{-1}$ of \cite{Kang:2022orq}.
\end{itemize}
When \eqref{eq:YM2MM} stems from a large $N$ gauge theory, the parameter $b$ encodes the dependence on the gauge rank $N$. Here we need not necessarily assume any concrete $N$-dependence. In order to obtain a nontrivial scaling limit, we will instead scale $b$ with $L$ in the following way:
\begin{defin}\label{def:VenezianoYM2}
    For every $L\in \mathbb{N},\beta>0, b>0$, the \emph{planar limit} of the ensemble \eqref{eq:YM2MM} is the limit $L \to \infty$ with 
    \begin{equation}
        \tilde{\gamma} = \frac{L}{b} \quad \text{fixed} .
    \end{equation}
\end{defin}
Regardless of the origin of \eqref{eq:YM2MM}, the scaling limit with $\frac{L}{b}$ fixed is necessary at the level of the matrix model, to have a non-trivial large $L$ limit.\par

\subsection{Quantum mechanics from ensembles with Casimir Hamiltonian}

Example \eqref{eq:YM2MM} slightly differs from our general discussion in Section \ref{sec:QM}, because we do not start from a unitary matrix model. Besides, and related, the Hamiltonian of the quantum mechanical system is quadratic in $R$, as opposed to the Hamiltonian linear in $R$ considered in Section \ref{sec:QM}. This modification does not spoil the argument, which can be run without changes.\par
We henceforth take \eqref{eq:YM2MM} as our starting point, and interpret it as the partition function of a system at inverse temperature $\beta$. This defines our toy model. The Hilbert space of the system decomposes into 
\begin{equation}
    \mathscr{H}_L = \bigoplus_{R \in \mathfrak{R}^{U(L)}} \mathscr{H}_L (R) \otimes \mathscr{H}_L (R) ,
\end{equation}
and the Hamiltonian $H$ acts diagonally on the representation basis with eigenvalues 
\begin{equation}
    H ~ \lvert R, \ai \rangle \otimes \lvert R, \dot{\ai} \rangle = (C_2 (R)/b) ~ \lvert R, \ai \rangle \otimes \lvert R, \dot{\ai} \rangle \qquad \forall \ai, \dot{\ai} = 1, \dots , \dim R .
\end{equation}
We will assume the system is coupled to a probe as explained in Subsection \ref{sec:probe}. 
\begin{center}
\noindent\fbox{%
\parbox{0.98\linewidth}{%
\begin{thm}\label{thm:YM2VN3}
    Let $L \in \N$ and $\beta >0$, and consider $\mz_{\text{\rm YM$_2$}} (L, e^{- \beta})$ as given in \eqref{eq:YM2MM}, interpreted as the partition function of a quantum mechanical system with global symmetry $U(L)$. Consider the planar limit of Definition \ref{def:VenezianoYM2}. The von Neumann algebra associated to the system is of type III$_1$.
\end{thm}}}\end{center}\par
Crucially, \eqref{eq:YM2MM} equals the partition function of two-dimensional Yang--Mills theory on the sphere \cite{Migdal:1975zg,Rusakov:1990rs}, with $U(L)$ interpreted as a gauge group in that context. The planar limit of Definition \ref{def:VenezianoYM2} is nothing but the standard 't Hooft planar limit of 2d pure Yang--Mills. The large $L$ planar limit of \eqref{eq:YM2MM} has been addressed in \cite{Douglas:1993iia}, which lays the groundwork for Theorem \ref{thm:YM2VN3}.
\begin{lem}[\cite{Douglas:1993iia}]
    Let $\mz_{\text{\rm YM$_2$}}$ be the matrix model \eqref{eq:YM2MM}, at arbitrary $\beta >0$. In the planar limit, $\ln \mz_{\text{\rm YM$_2$}} = O(L^2)$, and it undergoes a third order phase transition at $\tilde{\gamma} = \frac{\pi^2 }{2\beta}$. Moreover, the density of eigenvalues has compact and continuous support on the real axis, which enhances as $\beta \to 0$. 
\end{lem}
\begin{proof}
    This lemma is part of the classical result of \cite{Douglas:1993iia}. We briefly sketch the main ideas for completeness, and refer to \cite{Douglas:1993iia} (and subsequent work) for the details.\par
    The starting point is to rewrite \eqref{eq:YM2MM} in a form akin to \eqref{eq:genericdiscreteMM}. The change of variables 
    \begin{equation}
        h_i = R_i - i + \frac{L+1}{2} 
    \end{equation}
    recasts the ensemble of representations into 
    \begin{equation}
         \mz_{\text{\rm YM$_2$}} (L, e^{- \beta}) = \frac{e^{\frac{\beta}{12b} L(L^2-1)} }{G(L+1)^2 } \sum_{\substack{(h_1, \dots , h_L) \in \Z^{L} \\ h_1 > h_2 > \dots > h_L}} ~ e^{-\frac{\beta}{b} \sum_{i=1}^{L} h_i ^2} ~ \prod_{1 \le i < j \le L} (h_i - h_j)^2 
    \end{equation}
    where we have used the properties of the $U(L)$ representations. Similar to the previous examples, the matrix model effective action is 
    \begin{equation}
        S (h_1, \dots , h_L) = \frac{\beta}{b} \sum_{i=1}^{L} h_i ^2 - \sum_{i \ne j} \ln \lvert h_i - h_j \rvert .
    \end{equation}
    Inserting the scaled variable $x$, defined through $h_i = L^{\eta} x_i$ for some $\eta >0$ to be determined momentarily, we define the density of eigenvalues $\varrho (x)$ exactly as above. Using the definition of the parameter $\tilde{\gamma}$, we arrive at 
     \begin{equation}
        S (h_1, \dots , h_L) = L^2 \int \dd x \varrho (x) \left[  L^{2\eta -2} \beta \tilde{\gamma} x^2 - \mathrm{P}\!\!\!\int \dd x^{\prime} \varrho (x^{\prime}) \ln \lvert x-x^{\prime} \rvert \right]  .
    \end{equation}
    In the large $L$ limit, this action admits a non-trivial saddle point if $\eta =1$. We arrive at the saddle point equation 
    \begin{equation}
    \label{eq:DKSPE}
        \mathrm{P}\!\!\!\int \dd x^{\prime} \frac{\varrho_{\ast} (x^{\prime}) }{x-x^{\prime}} = \beta \tilde{\gamma} x ,
    \end{equation}
    where the saddle point eigenvalue density $\varrho_{\ast}$ is to be looked for in the restricted functional space subject to the conditions \cite{Douglas:1993iia}
    \begin{equation}
        \int_0 ^{\infty} \dd x \varrho (x) =1 , \qquad 0 \le  \varrho (x) \le 1 .
    \end{equation}
    So far, the derivation is analogous to Subsection \ref{sec:IOPlargeN}, except for the quadratic dependence on $x$ in the action. The solution to \eqref{eq:DKSPE} is given by the Wigner semicircle law 
    \begin{equation}
        \varrho_{\ast} (x)= \frac{ \beta \tilde{\gamma}}{\pi} \sqrt{\frac{2}{\beta \tilde{\gamma}} - x^2} , \qquad \qquad \supp \varrho_{\ast} = \left[ - \sqrt{\frac{2}{\beta \tilde{\gamma}}} , \sqrt{\frac{2}{\beta \tilde{\gamma}}} \right] .
    \end{equation}
    The solution satisfies $\varrho_{\ast} (x) \le 1$ on the entire support if $2 \beta \tilde{\gamma} < \pi^2$. We interpret this bound as a critical value for $\tilde{\gamma}$ at arbitrary $\beta >0$. Raising $\tilde{\gamma}$ above the threshold, one must replace the Wigner semicircle with a ``capped'' solution of the form \cite{Douglas:1993iia}
    \begin{equation}
        \varrho (x) = \begin{cases} 1 & \hspace{8pt} 0 \le x < x_{-} \\ \hat{\varrho} (x) \  & x_{-} \le x \le x_{+} \\ 0 &  \hspace{8pt} x_{+} < x \end{cases}
    \end{equation}
    with $\hat{\varrho}$ satisfying the continuity conditions $\hat{\varrho} (x_{-}) =1 , \hat{\varrho} (x_{+}) =0 $. The solution in this phase is more involved, and can be found in \cite{Douglas:1993iia}.\par
    Focusing on the high temperature regime, i.e. the Cardy limit $\beta \to 0$, the critical value for $\tilde{\gamma}$ is moved to large positive values. The system remains in the first phase for $\tilde{\gamma}$ fixed and $\beta \to 0$, and $\supp \varrho_{\ast}$ has width proportional to $\beta^{-1/2}$, thus spreads on the whole real axis in the high temperature limit.
\end{proof}
\begin{proof}[Proof of Theorem \ref{thm:YM2VN3}]
    Given the saddle point eigenvalue density of \cite{Douglas:1993iia}, the result follows from it and the general result of Subsection \ref{sec:spectral}.\par
    It is important to note that, although the Hamiltonian $H$ has changed, and the eigenvalues are now $C_2 (R)$, we are still using the interaction Hamiltonian \eqref{eq:Hint}. The eigenvalues of $H$ determine the eigenvalue density $\varrho_{\ast}$, whereas the Hamiltonian $H_{\mathrm{int}}$ enters in the Fourier transform of the correlation functions. For this reason, the expressions for the Wightman functions --- and hence for the spectral density $\rho (\omega)$ --- derived in Subsection \ref{sec:spectral} remain valid, as can be immediately checked walking through the same steps. What changes is the form of $\varrho_{\ast}$ used for the evaluation of $\rho (\omega)$ at large $N$.
\end{proof}
As a concluding remark, we stress once again that we are not making claims about the nature of the von Neumann algebra of the holographic CFTs considered in \cite{Kang:2022orq}. We find that, given the toy quantum mechanical model built out of \eqref{eq:YM2MM}, its correlation functions are captured by a type III$_1$ von Neumann algebra. We \emph{do not} claim any implication for the holographic systems of \cite{Kapec:2019ecr,Kang:2022orq}. While the ideas developed here might prove useful, to rigorously establish the type of von Neumann algebra for those holographic systems is beyond the scope of the present work.

\subsection{Ensembles with a sum over flavor symmetries}

We now consider the extension of the matrix model \eqref{eq:YM2MM}, and of the corresponding quantum system, by introducing the sum over the rank $L$ of the global symmetry:
\begin{equation}
\label{eq:YM2sumL}
    \mz_{\text{\rm Ex4}} (\mathfrak{q}, e^{- \beta}) = \sum_{L=0} ^{\infty} \mathfrak{q}^{L^2} \sum_{R \in \mathfrak{R}^{U(L)}} (\dim R)^2 ~e^{- \frac{\beta}{b} C_2 (R)} .
\end{equation}
In this situation, the dependence on the parameter $b$ replaces the dependence on the gauge rank $N$, simply because it is the only other parameter we have at hand besides $\beta$. 
\begin{prop}
    In the limit $b \to \infty$, the matrix model \eqref{eq:YM2sumL} behaves as 
    \begin{equation}
        \ln \mz_{\text{\rm Ex4}}  = O(b^2) 
    \end{equation}
    for every $\beta >0$.
\end{prop}
\begin{proof}
    As before, we write $\mathfrak{q}= e^{-1/(2a)}$. Besides, we consider two choices of slice in the parameter space: (i) the fixed-$a$ slice, in which $a$ is a given number independent on the other parameters; and (ii) the Schur slice $a= \beta^{-1}$.\par
    To begin with, we rewrite 
    \begin{equation}
    \label{eq:largeNZYMsumL}
        \ln \mz_{\text{\rm Ex4}} \approx \ln  \int_0 ^{\infty} \dd \tilde{\gamma} ~\exp \left[ - b^2 \left( \frac{\tilde{\gamma}^2}{2a}  - \tilde{\gamma}^2 \mf_{\text{YM$_2$}} (\beta \tilde{\gamma}) \right) \right]
    \end{equation}
    where 
    \begin{equation}
        \mf_{\text{YM$_2$}} = \lim_{L \to \infty} \frac{1}{L^2} \ln  \mz_{\text{\rm YM$_2$}} 
    \end{equation}
    which, by definition, only depends on the product $\beta \tilde{\gamma}$. In the limit $b \to \infty$ we ought to look for the saddle points $\tilde{\gamma}_{\ast}$ of \eqref{eq:largeNZYMsumL}. Assuming that the system is in the first phase, we have \cite{Douglas:1993iia}
    \begin{equation}
        \mf_{\text{YM$_2$}} \big\rvert_{\tilde{\gamma}< \frac{\pi^2}{2 \beta } } = \frac{\beta \tilde{\gamma}}{12} - \frac{1}{2} \ln (2 \beta \tilde{\gamma}) .
    \end{equation}
    In the constant-$a$ slice, in which $a$ is taken independent of $\beta$, the saddle point $\tilde{\gamma}_{\ast}$of the integrand in \eqref{eq:largeNZYMsumL} is 
    \begin{equation}
    \label{eq:gammaastYM2}
        \beta \tilde{\gamma}_{\ast} = - 4 \mathsf{ProductLog} \left( - \frac{e^{- \frac{1}{2} - \frac{1}{a} }}{8} \right) ,
    \end{equation}
    where $\mathsf{ProductLog} (z)$ denotes the function that gives the principal solution for $w$ to the equation $z=w e^{w}$. The saddle point is plotted in Figure \ref{fig:ProdLog}. The most important point is that it is a positive, monotone function of $a>0$ which satisfies 
    \begin{equation}
        \lim_{a \to 0} \beta \tilde{\gamma}_{\ast}  =0 , \qquad \lim_{a \to \infty} \beta \tilde{\gamma}_{\ast}  \approx 0.3293 .
    \end{equation}
    That is to say, 
    \begin{equation}
        0 <\beta \tilde{\gamma}_{\ast}  < \frac{\pi^2}{2} \qquad \forall 0<a<\infty .
    \end{equation}
    On the other hand, expanding $\mf_{\text{YM$_2$}}$ in the phase $ \tilde{\gamma} > \frac{\pi^2}{2 \beta}$ near the critical point, there is no saddle point. Evaluating \eqref{eq:largeNZYMsumL} at $\tilde{\gamma}_{\ast}$ we obtain $\ln \mz_{\text{\rm Ex4}} =O(b^2)$ at all temperatures.\par
    \begin{figure}
        \centering
        \includegraphics[width=0.5\textwidth]{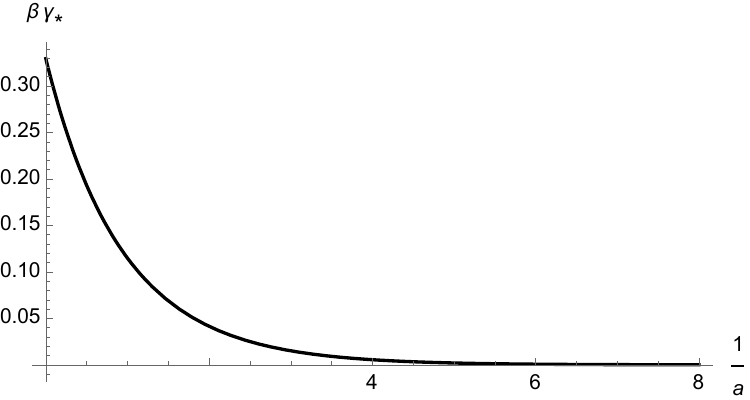}
        \caption{Plot of the saddle point value $\beta \tilde{\gamma}_{\ast}$ in the constant-$a$ slice, shown as a function of $a^{-1}$.}
        \label{fig:ProdLog}
    \end{figure}\par
    Let us now consider the slice $a= \beta^{-1}$. Due to the simple way in which the dependence on $a$ appears, it turns out that the saddle point is simply given by \eqref{eq:gammaastYM2} with $a^{-1} = \beta$. Once again, this saddle point is positive and valid at all temperatures, thus yielding $\ln \mz_{\text{\rm Ex4}} =O(b^2)$ at all temperatures.
\end{proof}
\begin{cor}
    The von Neumann algebra associated to quantum system constructed from the ensemble \eqref{eq:YM2sumL} with a sum over the rank of the global symmetry is of type III$_1$ for every $\beta >0$.
\end{cor}

\section{Conclusions and outlook}
\label{sec:discussion}

\begin{figure}[ht]
    \centering
    \begin{tikzpicture}
        \node[anchor=south] at (-5,0) {\textsc{Keyword}};
        \node[anchor=south] at (0,0) {\textsc{Result}};
        \node[anchor=south] at (5,0) {\textsc{Section}};
        \draw (-6.5,0) --  (6.5,0);
        \draw (-3.5,0.5) --  (-3.5,-5.5);
        \draw (3.5,0.5) --  (3.5,-5.5);

        \node[align=center, anchor=north] (al) at (-5,-0.5) {von Neumann\\ algebra};
        \node[align=center, anchor=north] (bl) at (-5,-2.25) {Quantum\\ mechanics};
        \node[align=center, anchor=north] (cl) at (-5,-4) {Hagedorn\\ transition};

        \node[align=center, anchor=north] (ac) at (0,-0.5) {Construction and type of\\ large $N$ von Neumann algebra\\ from $\mathrm{supp} \rho $};
        \node[align=center, anchor=north] (bc) at (0,-2.25) {Quantum mechanical system\\ with type III$_1$ algebra at large $N$};
        \node[align=center, anchor=north] (cc) at (0,-4) {Extended systems with first order\\ type I to type III$_1$ transition};

        \node[align=center, anchor=north] (ar) at (5,-0.5) {Section \ref{sec:vNtot}\\ \small (Eq. \eqref{eq:typefromrho})};
        \node[align=center, anchor=north] (br) at (5,-2.25) {Section \ref{sec:QM}\\ \small (Thm. \ref{thm:istype3Rigor})};
        \node[align=center, anchor=north] (cr) at (5,-4) {Section \ref{sec:Fermi}\\ \small (Thm. \ref{thm:ItoIII})};
    \end{tikzpicture}
    \caption{Main results in a nutshell}
    \label{fig:mainresults}
\end{figure}\par
In this work we obtained a general construction of large $N$ von Neumann algebras applicable to observables satisfying large $N$ factorization, and constructed quantum systems with an emergent type III$_1$ von Neumann algebra at large $N$. These quantum systems, which vastly generalize the IOP model introduced in \cite{Iizuka:2008eb}, are able to implement some kind of gauge constraint, and have a partition function that is expressible in terms of a matrix model. Their Hilbert space can be explicitly constructed from the character expansion of the matrix integral.\par 
We showed that when a heavy probe is coupled to our systems, the K\"all\'en--Lehmann spectral density at finite temperature becomes continuously supported, which is a hallmark of type III$_1$ structure. Furthermore, upon introducing an extended Hilbert space with sectors carrying different flavor symmetries, we showed that our systems are generically promoted to systems with a Hagedorn transition. Upon definition of an appropriate notion of probe, we showed that below the Hagedorn temperature, the partition function is convergent, while above the Hagedorn transition, the large $N$ algebras once again satisfy large $N$ factorization and have type III$_1$.\par
Our main findings in Part \ref{part1} are summarized in Figure \ref{fig:mainresults}.\par
\medskip
Our construction can be applied to a large class of examples, as was demonstrated very explicitly in the second part of this work. In particular, we introduced examples inspired from the IOP model \cite{Iizuka:2008eb}, as well as from a toy model of QCD$_2$ \cite{Hallin:1998km}. We also analyzed a system constructed from the generating function of DT invariants of the conifold, even though a complete understanding of the possible geometric meaning of the type III$_1$ von Neumann algebra remains an open question. We have then analyzed a matrix model describing holographic systems with global symmetries \cite{Kang:2022orq}.\par
\medskip
The aim of this work was to lay the groundwork to explicitly set up a quantum system given any gauge theory with continuous flavor symmetry, and to assign a large $N$ von Neumann algebra to it. There are several interesting questions that future work could address.\par
The most obvious and most ambitious avenue for future research is to embrace more realistic models of holography, possibly applying the techniques and results herein to the Hilbert series of these models. Maybe without studying full-fledged holographic examples, a desirable feature that the IOP model does not possess \cite{Michel:2016kwn}, nor do our models, would be that they exhibit maximal chaos. 
It would be interesting to see whether one can find maximally chaotic models and study their large $N$ von Neumann algebra by extending the methods developed here to more complicated systems. 
Large $N$ chaos is likely to have a nice interpretation in the language of modular flow \cite{DeBoer:2019kdj,Faulkner:2022ada}, and more generally, the algebraic properties of a local bulk spacetime are closely related to mathematical definitions of chaos in von Neumann algebras \cite{Gesteau:2023rrx,Ouseph:2023juq}. \par
A complementary question to ask is whether the examples presented here, which are only known to possess some of the weaker chaotic properties of \cite{Gesteau:2023rrx,Ouseph:2023juq}, can still have a bulk description that is geometric in some sense. In particular, the compactness of the spectral density of the models of this paper seems to make their putative dual description very stringy. It would be interesting to understand whether for such models, one can still make sense of a notion of dual geometry that, in particular, displays some kind of connectedness between the two sides of the thermofield double.
Another feature that we would like to eventually remove is the introduction of an external probe, which would be absent in a fully holographic system, where the relevant correlation functions should correspond to operators directly within the gauge theory.\par
Another question that remains open is a first principles understanding of the sum over flavors. Posing this question from the bulk side of the holographic correspondence, it would be very interesting to establish a connection with other recent proposals entailing averaging procedures in black hole physics.\par
One collateral observation (cf. Appendix \ref{sec:appCardyCS}) is that, for the very special case of $\mN=4$ super-Yang--Mills in the Cardy limit, our sum over $L$ descends to a sum over the Riemann sheets of \cite{Cassani:2021fyv}. It remains to be seen whether there is a general lesson to be learnt from this comment.\par
\medskip
It may also be interesting to study these systems in other states that are not thermal, like, for instance, some microcanonical versions of the thermofield double akin to the ones considered in \cite{Chandrasekaran:2022eqq}. Other types of von Neumann algebras are supposed to appear in that case, and it would be interesting to check the proposal in our matrix model context.\par
A related issue would be to study the crossed product of our large $N$ algebras with their modular automorphism group, as well as perturbative $1/N$ corrections to our calculations. Moreover, including both perturbative and non-perturbative corrections to the large $N$ correlation functions computed here is expected to give a von Neumann algebra of type I. It would be interesting to see explicitly how this happens in examples, deploying the non-perturbative techniques of e.g. \cite{Marino:2008ya,Marino:2012zq,Marino:2007te,Marino:2022rpz}.\par
\medskip
Finally, it is also worthwhile to ask to what extent our criterion sheds light on partial deconfinement. The partial deconfinement proposal \cite{Hanada:2016pwv,Berenstein:2018lrm,Hanada:2019kue} argues for an intermediate coexistence phase in the Hagedorn transition. Passing to the microcanonical ensemble at energies $1 \ll E \ll N^2$, the first order deconfinement transition in the gauge theory is smoothed into a phase in which only a $U(N_{\text{eff}}) \subset U(N)$ is deconfined, with $N_{\text{eff}} ^2 \propto E$. On the bulk side, this supposes the identification of the small black hole phase with the long string phase \cite{Berenstein:2018lrm}. A direct application of our formula would give a type III$_1$ von Neumann algebra also in this intermediate region, in agreement with the proposed picture. A mathematically rigorous treatment of the partially deconfined phase is a problem that we leave for future work.

\vspace{0.6cm}
\subsubsection*{Acknowledgements}
We thank Adam Artymowicz, Kasia Budzik, Carlos Florentino, Shota Komatsu, Ji Hoon Lee, Hong Liu, Matilde Marcolli, Sridip Pal, Kyriakos Papadodimas, Giulio Ruzza, Miguel Tierz, Fengjun Xu and Yixin Xu for discussions, and Hong Liu and Matilde Marcolli for comments on the draft. LS also thanks the Departamento de An\'alisis Matem\'atico y Matem\'atica Aplicada, Universidad Complutense de Madrid for hospitality during the completion of this project.
The work of LS is supported by the Shuimu Scholars program of Tsinghua University, by the Beijing Natural Science Foundation project IS23008 ``Exact results in algebraic geometry from supersymmetric field theory'', and in part by the FCT project PTDC/MAT-PUR/30234/2017 during the early stages of the work. This work was initiated at the Aspen Center for Physics, which is supported by National Science Foundation grant PHY-1607611.

\vspace{0.6cm}
\begin{appendix}

\section{General construction of von Neumann algebras for factorizing systems}
\label{sec:ccrcar}

In this slightly more formal appendix, we introduce a general procedure that allows to construct von Neumann algebras associated to factorizing systems. This procedure heavily relies on the rigorous mathematical results of \cite{Derezinski}. We comment on how this construction can be generically applied to study large $N$ algebras in AdS/CFT.

\subsection{Bosonic case: Canonical commutation relations}
\label{app:CCR}

In order to study bosonic factorizing systems in terms of operator algebras, the right object to introduce is a \textit{representation of the canonical commutation relations} (CCR).

\begin{defin}
Let $Y$ be a real vector space and $\omega$ be an antisymmetric form on $Y$. Let $\mathscr{H}$ be a Hilbert space. A map $y\mapsto W(y)$ is said to be a representation of the CCR over $Y$ in $\mathscr{H}$ if for $y_1,y_2\in Y$, it satisfies the relation \begin{align}W(y_1)W(y_2)=e^{-\frac{i}{2}\omega(y_1,y_2)}W(y_1+y_2).\end{align} 
\end{defin}

In this paper, we are interested in states that factorize in the large $N$ limit. These states are known as \textit{quasi-free states} in the operator-algebraic language. We now define this notion.

\begin{defin}
Let $y\mapsto W(y)$ be a representation of the CCR on a Hilbert space $\mathscr{H}$. A vector $\ket{\Psi}\in\mathscr{H}$ is said to be quasi-free if it is cyclic, and there exists a quadratic form $\eta$ in $\mathscr{H}$ such that for all $y$, \begin{align}\bra{\Psi}W(y)\ket{\Psi}=e^{-\frac{1}{4}\eta(y,y)}.\end{align}
\end{defin}
As the form $\eta$ is quadratic, the fields can be treated as Gaussian, which implies that the correlation functions satisfy Wick's theorem.

An alternative definition of quasi-free states can be formulated thanks to the factorization property of correlation functions. More precisely, we have the following result.

\begin{prop}[\cite{Derezinski}]
Let $\ket{\Psi}$ be a vector in a strongly continuous representation of the CCR $W(y)=e^{i\phi(y)}$. $\ket{\Psi}$ is quasi-free if and only if for all $y_1,...y_n$, $\ket{\Psi}\in \mathrm{Dom}(\phi(y_1)\dots\phi(y_n))$, and
\begin{align}
\bra{\Psi}\phi(y_1)\dots\phi(y_{2m-1})\ket{\Psi}=0,
\end{align}
\begin{align}
\bra{\Psi}\phi(y_1)\dots\phi(y_{2m})\ket{\Psi}=\sum_{\varpi\text{\rm Wick pairing}}\prod_{j=1}^m\bra{\Psi}\phi(y_{\varpi(2j-1)})\phi(y_{\varpi(2j)})\ket{\Psi}.
\end{align}
\end{prop}

The theory of quasi-free states of the CCR is well-studied, and there is a generic procedure that allows for the construction of most quasi-free representations. It is formalized by the notion of Araki--Woods representation, which we now introduce.

Let $Z$ be a Hilbert space, and let $\Gamma$ be the bosonic Fock space over $Z\oplus\bar{Z}$. We equip the space $\mathrm{Re}(Z\oplus\bar{Z}\oplus(\overline{Z\oplus\bar{Z}}))$ with the symplectic form 
\begin{align}
\omega((x,\bar{x}),(y,\bar{y})) :=2\mathrm{Im}(x,y).
\end{align}
Then, there is a canonical representation of the CCR given by \begin{align}W(z_1,\bar{z}_2):=e^{i\phi(z_1,\bar{z}_2)},\end{align} where \begin{align}\phi(z_1,\bar{z}_2)=\frac{1}{\sqrt{2}}(a^\dagger(z_1,\bar{z}_2)+a(z_1,\bar{z}_2)),\end{align}$a$ and $a^\dagger$ being the usual raising and lowering operators. 

The Araki--Woods representations \cite{doi:10.1063/1.1704002} of the CCR are parameterized by an operator $\rho$, defined by \begin{align}\rho:=\gamma(1-\gamma)^{-1},\end{align} where $\gamma$ is a self-adjoint operator satisfying $0\leq\gamma\leq 1$. Then, for $z\in\mathrm{Dom}(\rho^{\frac{1}{2}})$, we can define two unitary operators $W_{\gamma,l}$ and $W_{\gamma,r}$ on $\Gamma$ by 
\begin{equation}
\begin{aligned}
W_{\gamma,l}(z) &:=W((\rho+1)^{\frac{1}{2}}z,\bar{\rho}^{\frac{1}{2}}\bar{z}), \\
W_{\gamma,r}(\bar{z})& :=W(\rho^{\frac{1}{2}}z,(\bar{\rho}+1)^{\frac{1}{2}}\bar{z}).
\end{aligned}
\end{equation}
The von Neumann algebras generated by the $W_{\gamma,l}(z)$ (resp. $W_{\gamma,r}(\bar{z})$) are called the left (resp. right) Araki--Woods algebras associated to the operator $\rho$. Note, in particular, that one can recover the operator $\rho$ entirely from the two-point functions of fields and creation and annihilation operators, for example we have the identity \begin{align}(z_2,\rho z_1)=\bra{\Omega}a^\dagger_{\gamma,l}(z_1)a_{\gamma,l}(z_2)\ket{\Omega},\end{align}
where the creation and annihilation operators $a^\dagger_{\gamma,l}$ and $a_{\gamma,l}$ are defined by 
\begin{equation}
\begin{aligned}
a^\dagger_{\gamma,l}(z) & :=a^\dagger((\rho+1)^{\frac{1}{2}}z,0)+a(0,\bar{\rho}^{\frac{1}{2}}\bar{z}), \\
a_{\gamma,l}(z) & :=a((\rho+1)^{\frac{1}{2}}z,0)+a^\dagger(0,\bar{\rho}^{\frac{1}{2}}\bar{z}),
\end{aligned}
\end{equation}
and $\ket{\Omega}$ is the vacuum of the Fock space of the Araki--Woods representation.

The power of Araki--Woods representations comes from the fact that, under mild assumptions, any quasi-free representation of the CCR is isomorphic to an Araki--Woods representation. In particular, the following result, due to Derezi\'{n}ski, holds:

\begin{thm}[\cite{Derezinski}]
Let $y\mapsto W(y)$, $y\in Y_0$, be a quasi-free representation of the CCR in a Hilbert space $\mathscr{H}$, with a cyclic quasi-free vector $\ket{\Psi}$ satisfying $\bra{\Psi}W(y)\ket{\Psi}=e^{-\frac{1}{4}\eta(y,y)}$, where $\eta$ is a nondegenerate inner product. Let $Y$ be the real Hilbert space completion of $Y_0$ for $\eta$, and $\omega$ be the bounded extension of the antisymmetric form associated to $Y_0$ to $Y$. Assume $\omega$ is nondegenerate on $Y$. Then, $W$ is unitarily equivalent to an Araki--Woods representation of the CCR.
\end{thm}

The interest of knowing that most quasi-free representations of the CCR are isomorphic to Araki--Woods representations is that there exist simple sufficient conditions to determine the type of their bicommutant.

\begin{thm}[\cite{Derezinski}]
Let $M$ be the bicommutant of an Araki--Woods representation of the CCR, and $\gamma$ the associated operator defined as above. If $\gamma$ is trace-class, then $M$ has type I. If $\gamma$ has some continuous spectrum, then $M$ has type III$_1$. 
\end{thm}
It is this statement that, applied to our context in Section \ref{sec:vNtot}, allows us to conclude about the type of the von Neumann algebras on both sides of the Hagedorn phase transitions.

\subsection{Fermionic case: Canonical anticommutation relations}

We can perform an entirely analogous analysis for fermionic oscillators. The starting point is now the algebra of the \textit{canonical anticommutation relations} (CAR).

\begin{defin}
Let $Y$ be a real vector space and $\alpha$ be an positive inner product on $Y$. Let $\mathscr{H}$ be a Hilbert space. A map $y\mapsto \phi(y)$ is said to be a representation of the CAR over $Y$ in $\mathscr{H}$ if its range only contains bounded self-adjoint operators, and for $y_1,y_2\in Y$, it satisfies the relation \begin{align}\{\phi(y_1),\phi(y_2)\}=2\alpha(y_1,y_2).\end{align}
\end{defin}

In the same way as before, we can define the notion of quasi-free state of the CAR. 

\begin{defin}
\label{def:appquasifree}
Let $y\mapsto W(y)$ be a representation of the CAR on a Hilbert space $\mathscr{H}$. A vector $\ket{\Psi}\in\mathscr{H}$ is said to be quasi-free if it is cyclic, and for all $y_i$, \begin{align}\bra{\Psi}\phi(y_1)\dots \phi(y_{2m-1})\ket{\Psi}=0,\end{align}and
\begin{align}\bra{\Psi}\phi(y_1)\dots \phi(y_{2m})\ket{\Psi}=(-1)^{\frac{m(m-1)}{2}}\sum_{\varpi\text{\rm Wick pairing}}\mathrm{sgn} (\varpi) \prod_{j=1}^m\bra{\Psi}\phi(y_{\varpi(j)})\phi(y_{\varpi(j+m)})\ket{\Psi}.\end{align}
\end{defin}

Now, closely following the bosonic case, we introduce a generic procedure to construct a large class of quasi-free representations of the CAR and classify them. These representations are called Araki--Wyss representations \cite{ArakiWyss}.

Let $Z$ be a Hilbert space, and let $\Gamma$ be the fermionic Fock space on  $Z\oplus\bar{Z}$. We see the space $\mathrm{Re}(Z\oplus\bar{Z}\oplus(\overline{Z\oplus\bar{Z}}))$ as a real Hilbert space.
Then, there is a canonical representation of the CAR given by \begin{align}\phi(z_1,\bar{z}_2)=\frac{1}{\sqrt{2}}(a^\dagger(z_1,\bar{z}_2)+a(z_1,\bar{z}_2)),\end{align}$a$ and $a^\dagger$ being the usual fermionic raising and lowering operators. 

Similarly to the Araki--Woods representations of the CCR, the Araki--Wyss representations of the CAR are parameterized by an operator $\rho$, defined by \begin{align}\rho:=\gamma(1-\gamma)^{-1},\end{align} where $\gamma$ is a self-adjoint operator satisfying $0\leq\gamma\leq 1$. Then, for $z\in\mathrm{Dom}(\rho^{\frac{1}{2}})$, we can define the fields $\phi_{\gamma,l}$ and $\phi_{\gamma,r}$ on $\Gamma$ by 
\begin{equation}
\begin{aligned}
\phi_{\gamma,l}(z) &:=\phi((1-\rho)^{\frac{1}{2}}\,z,\bar{\rho}^{\frac{1}{2}}\,\bar{z}),\\
\phi_{\gamma,r}(z)&:=\phi(\rho^\frac{1}{2}\,z,(1-\bar{\rho})^\frac{1}{2}\bar{z}).
\end{aligned}
\end{equation}
The von Neumann algebras generated by the $\phi_{\gamma,l}(z)$ (resp. $\phi_{\gamma,r}(z)$) are called the left (resp. right) Araki--Wyss algebras associated to the operator $\rho$. Note, in particular, that one can recover the operator $\rho$ entirely from the two-point functions of fields and creation and annihilation operators, for example we have the identity \begin{align}(z_2,\rho z_1)=\bra{\Omega}a^\dagger_{\gamma,l}(z_1)a_{\gamma,l}(z_2)\ket{\Omega},\end{align}
where the creation and annihilation operators $a^\dagger_{\gamma,l}$ and $a_{\gamma,l}$ are defined by 
\begin{equation}
\begin{aligned}
a^\dagger_{\gamma,l}(z) &:=e^{\frac{i \pi}{2} \mathsf{N}(\mathsf{N}-1)} ~ \left[ a^\dagger(0,(1-\bar{\rho})^{\frac{1}{2}}\bar{z})+a(\rho^{\frac{1}{2}}z,0) \right] ~ e^{\frac{i \pi}{2} \mathsf{N}(\mathsf{N}-1)},\\
a_{\gamma,l}(z) &:=e^{\frac{i \pi}{2} \mathsf{N}(\mathsf{N}-1)} ~ \left[ a(0,(1-\bar{\rho})^{\frac{1}{2}}\bar{z})+a^\dagger(\rho^{\frac{1}{2}}z,0) \right] ~ e^{\frac{i \pi}{2} \mathsf{N}(\mathsf{N}-1)},
\end{aligned}
\end{equation}
where $\mathsf{N}$ is the number operator. Also in the fermionic case, any reasonable quasi-free representation of the CAR is isomorphic to an Araki--Wyss representation.
\begin{thm}[\cite{Derezinski}]
Let $y\in Z_0\mapsto \phi(y)$ be a quasi-free representation of the CAR on a Hilbert space $\mathscr{H}$, with a cyclic quasi-free vector $\ket{\Psi}$. Let $\omega$ be the antisymmetric form defined by \begin{align}\omega(y_1,y_2):=\frac{1}{i}\bra{\Psi}[\phi(y_1),\phi(y_2)]\ket{\Psi},\end{align} and suppose that the kernel of $\omega$ is even or infinite-dimensional. Then, $W$ is unitarily equivalent to an Araki--Wyss representation of the CAR.
\end{thm}
Sufficient results are also available to determine the type of the bicommutant of an Araki--Wyss representation of the CAR.
\begin{thm}[\cite{Derezinski}]
Let $M$ be the bicommutant of an Araki--Wyss representation of the CAR, and $\gamma$ the associated defined as above. If $\gamma$ is trace-class, then $M$ has type I. If $\gamma$ has some continuous spectrum, then $M$ has type III$_1$. 
\end{thm}

\section{Exotic example: Effective \texorpdfstring{$\mN=4$}{N=4} super-Yang--Mills}
\label{sec:effN4}

The model discussed in this appendix is an effective description of four-dimensional $\mN=4$ super-Yang--Mills theory \cite{Sundborg:1999ue,Aharony:2003sx}. We exploit a result of \cite{Dutta:2007ws} to cast this example in the formalism of discrete matrix models, akin to Subsection \ref{sec:RepToFermi}. $\mN=4$ super-Yang--Mills does not possess an integer $L$, analogous to the number of flavors, to sum over. However, a certain summation is built-in in the formulation of \cite{Dutta:2007ws}. This model will therefore be somewhat exotic, not fully of the type introduced in Section \ref{sec:Fermi}, but it is nonetheless instructive to explore this example. 
We thus mostly reviews old results and rephrases them in our overarching framework.\par
We stress that we content ourselves with discussion on the toy quantum mechanics, and \emph{do not} claim implications for the spectral density of full-fledged $\mN=4$ super-Yang--Mills.\par
\bigskip

Consider four-dimensional $SU(N+1)$ $\mN=4$ super-Yang--Mills theory placed on the compact Euclidean space $\mathbb{S}^3 \times \mathbb{S}^1 _{\bsym}$, with radius of the thermal circle the inverse temperature $\bsym$. Deep in the weak 't Hooft coupling regime, the partition function reduces to \cite{Sundborg:1999ue,Aharony:2003sx}
\begin{equation}
\label{eq:N4SYMMM1}
	\mz_{\mN =4}  (a) = \oint_{SU(N+1)} \dd U ~ \exp \left\{ \frac{a}{2} \tr \left( U \right) \tr \left( U^{-1} \right)   \right\} , 
\end{equation}
with $a=a(\bsym)$ a function of the inverse temperature. Here we have discarded contributions that become irrelevant near the transition point, see Appendix \ref{sec:N4SYMMM} for more details. This matrix model undergoes a first order phase transition at $a=2$ \cite{Liu:2004vy}, reviewed in Appendix \ref{sec:N4SYMMM}.\par
\medskip
An alternative derivation of the first order phase transition in \eqref{eq:N4SYMMM1} was given in \cite{Dutta:2007ws}. The authors of \cite{Dutta:2007ws} started by uncovering the equivalent description of \eqref{eq:N4SYMMM1} in a free fermion formalism. This latter approach rewrites \eqref{eq:N4SYMMM1} as a discrete matrix ensemble, of the type we have considered in Subsection \ref{sec:RepToFermi}.
\begin{lem}[\cite{Dutta:2007ws}]
For every $a>0$, let $\mz_{\mN=4} (a)$ be as in \eqref{eq:N4SYMMM1}. Besides, let $\mathfrak{S}_L$ denote the symmetric group of $L$ elements and $ d_R (\mathfrak{S}_L)$ denote the dimension of the representation $R$ of $\mathfrak{S}_L$. It holds that
\begin{equation}
\label{eq:ZN4SYMfermion}
    \mz_{\mN =4}  (a) = \sum_{L=0} ^{\infty} \frac{a^L}{2^L L!} \sum_{\substack{R \ : \ \lvert R \rvert =L \\ \ell (R) \le N }} d_R (\mathfrak{S}_L) ^2 .
\end{equation}
The inner sum runs over irreducible representations $R$ of $\mathfrak{S}_L$, which are in one-to-one correspondence with Young diagrams of $L$ boxes, restricted to have length at most $N$.
\end{lem}
\begin{proof}The derivation of this identity is in \cite[Sec.3]{Dutta:2007ws}, to which we refer for the details. The proof is conceptually very similar to the character expansion of the other examples, although slightly more involved. It is based on the character expansion of the integrand in \eqref{eq:N4SYMMM1}, and applying the orthogonality of characters to remove $\oint_{SU(N+1)}$.
\end{proof}
The inner sum in \eqref{eq:ZN4SYMfermion} can be rephrased as running over irreducible representations $R$ of $SU(N+1)$, which are in one-to-one correspondence with Young diagrams of length at most $N$. Namely 
\begin{equation}
    R=(R_1, R_2, \dots , R_N) \qquad \qquad \text{with } \quad R_1 \ge R_2 \ge \cdots \ge R_N \ge 0 .
\end{equation}
Importantly, the sum is restricted to diagrams consisting of $\lvert R \rvert = \sum_{i=1} ^{N} R_i = L$ boxes.\par
The quantity $d_R (\mathfrak{S}_L)$ stands for the dimension of $R$ as a representation of $\mathfrak{S}_L$. It differs from the dimension of $R$ viewed as a $SU(N+1)$ representation, customarily denoted by $\dim R$: 
\begin{align}
    \hspace{-1.5cm} R \ : \ \lvert R \rvert = L \text{ and } \ell (R) \le N \quad \Longrightarrow \quad  R \in \left\{ \ SU(N+1) \text{ reps} \ \right\} & \cap  \left\{ \ \mathfrak{S}_L \text{ reps} \ \right\}   \notag \\ 
    \underbrace{\dim R}_{R \text{ is $SU(N+1)$ rep}} & \ne \underbrace{d_R (\mathfrak{S}_L)}_{R \text{ is $\mathfrak{S}_L$ rep}} ~.
\end{align}
Finally, the outer sum in \eqref{eq:ZN4SYMfermion} runs over all the sizes of the symmetric group.\par
With the customary change of variables \eqref{eq:changeRtoH}, and using the Frobenius--Weyl formula for $d_R (\mathfrak{S}_L)$ \cite{Macdonaldbook}, \eqref{eq:ZN4SYMfermion} becomes 
\begin{align}
\label{eq:ZN4SYMGopakumar}
    \mz_{\mN =4}  (a) = \sum_{L=0} ^{\infty} \left(\frac{a}{2 }\right)^L \sum_{h_1 >  \cdots > h_{N} \ge 0} & ~ \frac{ L!}{ \prod_{j=1} ^{N} (h_j !)^2 } \prod_{1 \le i < j \le N} (h_i - h_j)^2  \delta \left( \lvert \vec{h} \rvert - L - \frac{N(N-1)}{2} \right) , 
\end{align}
where we are using a shorthand notation $ \lvert \vec{h} \rvert := \sum_{j=1} ^{L} h_j $.\par
It is also possible to remove the ordering of the eigenvalues $h_j$ and sum over unordered $N$-tuples $\vec{h}\in \mathbb{N}^N$, cancelling the $L!$ in the summand.\par
To complete the analogy with the formalism of Sections \ref{sec:QM}-\ref{sec:Fermi}, we may write
\begin{equation}\label{eq:atobetaN4SYM}
    a = 2 \exp \left( - \beta \right) ,
\end{equation}
where $\beta$ is the inverse temperature of the microscopic system. This \emph{should not} be confused with the inverse temperature $\bsym$ of $\mN=4$ super-Yang--Mills. $\beta = - \ln \frac{a (\bsym)}{2}$ is a monotonically increasing function of $\bsym >0$, as shown in Figure \ref{fig:betavsbsym}.\par
The dependence on $a$ is then interpreted as the Boltzmann factor of the quantum system:
\begin{equation}
    \left( \frac{a}{2} \right)^L = e^{- L \beta} = e^{- \beta \left(\lvert \vec{h} \rvert  -\frac{N(N-1)}{2} \right)} .
\end{equation} 
The partition function in this formalism reads 
\begin{align}
\label{eq:ZN4DGfermion}
     \mz_{\mN =4}  (2e^{- \beta} ) = e^{\frac{\beta}{2} N(N-1)} \sum_{L=0} ^{\infty}  \sum_{h_1 > \cdots > h_{N} \ge 0}  e^{- \beta \lvert \vec{h} \rvert } ~ &  \frac{L !}{ \prod_{j=1} ^{N} (h_j !)^2 } \prod_{1 \le i < j \le N} (h_i - h_j)^2 \\
     & \times \ \delta \left( \lvert \vec{h} \rvert - L - \frac{N(N-1)}{2} \right) . \notag
\end{align}\par
\begin{figure}[t!]
    \centering
    \includegraphics[width=0.4\textwidth]{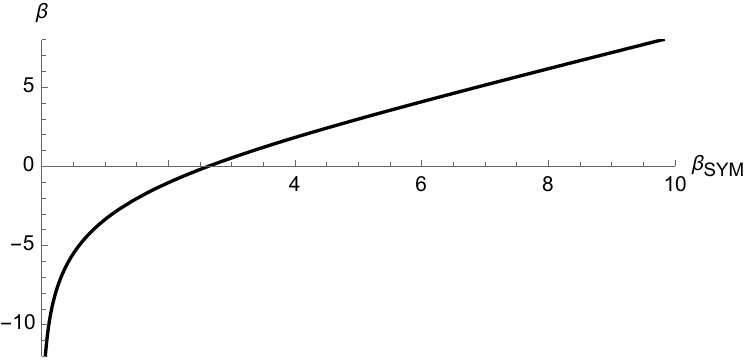}
    \caption{The inverse temperature $\beta$ of the toy model quantum mechanics is a function of the inverse temperature $\bsym$ of $\mN=4$ super-Yang--Mills.}
    \label{fig:betavsbsym}
\end{figure}

We emphasize that the discrete matrix model \eqref{eq:ZN4DGfermion} is closely related to but falls outside of the framework of Section \ref{sec:Fermi}. 
\begin{itemize}
\item The identification \eqref{eq:atobetaN4SYM} is really $a/2=y$, with $y$ a fugacity for the conserved charge $H$ in the simplified quantum system. Let us mention that \eqref{eq:ZN4SYMGopakumar} is derived for $a>0$, and in particular it is continuous. The resulting system is thus defined for arbitrary real inverse temperature $\beta \in \R$, as a consequence of the constraint on the total energy.
\item Here the roles of $N$ and $L$ are partly exchanged: $L$ constrains the allowed configurations $\vec{h}$, and the total number of eigenvalues is $N$. Moreover, the constraint $ \lvert \vec{h} \rvert  - \frac{N(N-1)}{2} =L$ does not directly bound the allowed states, but rather fixes the total energy of the quantum mechanical system. Schematically:
\begin{equation*}
    \begin{tabular}{r|c|c}
       $ \ $ & \textsc{Section \ref{sec:QM}} & \textsc{Appendix \ref{sec:effN4}} \\
       \hline 
        $N$ & constraint on configurations & \# of eigenvalues \\
        $L$ & \# of eigenvalues & constraint on configurations \\
         fixed-$L$ & quantum system & microcanonical ensemble \\
         weight & $\mathfrak{q}^{L^2}$ & $1/L!$ \\
        \hline
    \end{tabular}
\end{equation*}
\end{itemize}
\par
\medskip

\begin{prop}[\cite{Dutta:2007ws}]
\label{prop:N4PT}
For every $a>0$, let $\mz_{\mN =4} (a)$ be as in \eqref{eq:ZN4DGfermion}. In the large $N$ limit, it undergoes a first order phase transition at $a=2$, with
\begin{equation}
    \ln \mz_{\mN =4} (a) = \begin{cases} O(1) & a<2 \\ O(N^2) & a>2 . \end{cases}
\end{equation}
\end{prop}
\begin{proof}
We only review here the salient features of \cite[Sec.4]{Dutta:2007ws} and reformulate some steps in uniformity with the rest of the work. One begins by rewriting the summand in \eqref{eq:ZN4SYMGopakumar} in the form $e^{-S (\vec{h}; L)}$, where 
\begin{equation}
    S (\vec{h}; L) = 2 \sum_{i=1}^{N} \ln (h_i !) - \sum_{i \ne j} \ln \left\lvert h_i - h_j \right\rvert - L \ln \left( \frac{a}{2}\right) - \ln (L!) .
\end{equation}
In the large $N$ limit, it is convenient to interpret the $\delta$ function as constraining the value of $L$, for any given configuration $(h_1 , \dots , h_N)$. Defining 
\begin{equation}
    \ell := \frac{L}{N^2} ,
\end{equation}
we have the constraint 
\begin{equation}
\label{eq:defell}
    \ell = \frac{1}{N^2}\sum_{i=1} ^{N} h_i + \frac{N-1}{2N} .
\end{equation}
One advantage of this approach is that, now, the configurations $\vec{h}$ are unconstrained, and hence the numbers $h_j$ grow linearly in $N$.\par
We can in fact fix the precise scaling with $N$ by requiring the existence of an equilibrium solution at large $N$. Define the scaled variables $x_i$ through 
\begin{equation}
    h_i = N^{\eta} x_i ,
\end{equation}
for some $\eta >0$ and $x_i$ being $O(1)$ at large $N$. Introduce the eigenvalue density 
\begin{equation}
    \varrho (x) = \frac{1}{N} \sum_{i=1} ^{N} \delta (x-x_i) .
\end{equation}
We must look for $\varrho (x)$ satisfying 
\begin{equation}
\label{eq:DGconstraint}
    \int \dd x \varrho (x) =1 \qquad \qquad \varrho (x) \le 1 \qquad \qquad \supp\varrho \subset [0,\infty).
\end{equation}
At large $N$, we arrive at 
\begin{equation}
\label{eq:DGSefflargeN}
    \frac{S}{N^2} = 2 N^{\eta-1} \int x \ln (x) \varrho (x) \dd x  - \int \varrho (x) \dd x ~ \mathrm{P}\!\!\!\int \dd y \varrho (y) \ln \lvert x-y\rvert - \ell \ln \left( \frac{a}{2}\right) - \ell \ln (\ell) + \cdots 
\end{equation}
where $\mathrm{P}\!\!\!\int$ stands for the principal value integral and the dots include constant terms as well as terms that are sub-leading in $N$. We do not need them. We see that, in order to obtain a non-trivial equilibrium configuration, we must have $\eta=1$, so that positive and negative contributions to $S$ are of the same order in $N$ and compete.\par
From here, we must find $\varrho_{\ast} (x)$, subject to \eqref{eq:DGconstraint}, that extremizes $S$. The saddle point equation (i.e. equilibrium equation) reads 
\begin{equation}
    2  \mathrm{P}\!\!\!\int \dd y \frac{ \varrho_{\ast} (y) }{x-y} = 2 \ln(x) - \ln \left( \frac{a}{2} \ell \right) ,
\end{equation}
and the equality \eqref{eq:defell} fixing $\ell$ is 
\begin{equation}
\label{eq:DGelllargeN}
    \ell = \int y \varrho_{\ast} (y) \dd y - \frac{1}{2} .
\end{equation}
The solution was found in \cite{Dutta:2007ws} and reads 
\begin{align}
    \varrho_{\ast} (x) &= \frac{2}{\pi} \aco \left( \frac{x+\xi - \frac{1}{2} }{2 \sqrt{x \xi}} \right) , \label{eq:varrhoN4DG} \\
    \supp \varrho_{\ast} &= [x_{-}, x_{+}] , \qquad \qquad \qquad x_{\pm} := \sqrt{\xi} \pm \frac{1}{\sqrt{2}} . \notag
\end{align}
The parameter $\xi$ is a function of $a$ defined through 
\begin{equation}
    \xi^2 = \frac{a}{2} \ell .
\end{equation}
Because $\ell$ depends on $\varrho $, and hence on $\xi$ and $a$, this condition only fixes $\xi$ implicitly.\par
The solution \eqref{eq:varrhoN4DG} ceases to satisfy $\varrho (x) \le 1$ $\forall x$ at $\xi = \frac{1}{2}$. At this value, the left edge $x_{-} $ hits the hard wall at $x=0$. We then ought to look for a different solution. It was found in \cite{Dutta:2007ws} that the solution in this case reduces to the uniform density on the interval $[0,1]$, leading to a trivial saddle point configuration. Translated in the Young diagram language, the large $N$ limit in this phase is dominated by the trivial representation $R=\emptyset$, and $\ln \mz_{\mN=4}$ remains $O (1)$.
\end{proof}

\subsubsection{Double scaling of the effective \texorpdfstring{$\mN=4$}{N=4} super-Yang--Mills}
A double scaling regime for the matrix model effective description of $\mN=4$ super-Yang--Mills on $\mathbb{S}^3 \times \mathbb{S}^1 _{\bsym}$ was envisioned by Liu \cite{Liu:2004vy}. It consists in 
\begin{equation}
\label{eq:doublescalingLiu}
    N \to \infty , \qquad a \to 2 , \qquad \qquad N^{\frac{4}{3}} \left(\frac{a}{2} -1 \right) \text{ fixed} .
\end{equation}
For the sake of completeness, we sketch here how to derive it from the discrete ensemble. Consider the high temperature phase $a>2$. $\varrho $ is supported on 
\begin{equation}
    x_{-} \le x \le x_{+} \qquad x_{\pm} := \sqrt{\xi} \pm \frac{1}{\sqrt{2}} 
\end{equation}
with the parameter $\xi$ given in \eqref{eq:varrhoN4DG}. We know from above that the critical regime corresponds to $a \to 2$ and $\xi \to \frac{1}{2}$. We thus define the double scaling parameters $a_s$ and $\xi_s$ according to 
\begin{equation}
    \frac{a}{2} -1  = \frac{a_s}{N^{4/3}} , \qquad \qquad \xi - \frac{1}{2} = \frac{\xi_s}{N^{\nu}} ,
\end{equation}
where the exponent of $N$ in $a$ was determined in \eqref{eq:doublescalingLiu}, while the exponent $\nu$ will be fixed momentarily by consistency. Recall that $\xi$ is an implicit function of $a$. Here we will invert the relation and dial $\xi_s$ to go into the double scaling regime, with $a_s$ implicitly fixed by the inverse function theorem (applied locally near $a_s \approx 0$) as a function of $\xi_s$.\par
The edges of the eigenvalue distribution become 
\begin{equation}
    x_{-} \approx \frac{\xi_s}{\sqrt{2} N^{\nu}} , \qquad x_{+} \approx \sqrt{2} .
\end{equation}
In order to explore the critical regime, we zoom in close to the left edge. Stated more formally, we introduce the doubly scaled variable $x_s$ as 
\begin{equation}
    x-x_{-} = \frac{x_{s}}{N^{\nu_2}} ,  
\end{equation}
for some power $\nu_2$ uniquely fixed by $\nu$. With all these substitutions into the eigenvalue density \eqref{eq:varrhoN4DG}, we determine the scaling exponents $2 \nu_2 =  \nu = \frac{2}{3}$. Using this into the condition $\xi^2 = \ell \frac{a}{2}$ and approximating $\ell = \frac{1}{4} + \frac{\ell_s}{N^{\nu_3}} $ we find 
\begin{equation}
    \xi_s = \sqrt{a_s} = \ell_s
\end{equation}
and $\nu_3=\nu$. Altogether we introduce the doubly scaled eigenvalue density 
\begin{equation}
\begin{aligned}
    \varrho_s (x_s)& := \varrho_{\ast} \left(  \sqrt{\frac{a_s}{2}} N^{-2/3} + x_s N^{-1/3} \right) \approx \frac{2}{\pi} \aco \left[ \left( \frac{\xi_s}{2} \right)^{-1/4} x_s \right]  \\
    \supp \varrho_s &= \left[ \frac{a_s}{2} N^{-1/3} , \sqrt{2} N^{1/3} \right] .
\end{aligned}
\end{equation}
$\supp \varrho_{\ast}$ tends to the positive real axis in the double scaling limit.\par
The double scaling limit \eqref{eq:doublescalingLiu} has been translated in the discrete matrix model to a limit probing the left edge of the distribution, where the first order transition is taking place. This is not a genuinely new result, but a consequence of mapping the result in \cite{Liu:2004vy} to the language of \cite{Dutta:2007ws}. Here we have limited ourselves to perform the computation explicitly.

\section{Hagedorn transitions in holographic matrix models}
\label{sec:PTMM}

\subsection{From third to first order phase transitions in unitary matrix models}
\label{sec:UnitaryMM}

We consider a generic unitary one-matrix model, with integration domain $U(N)$:
\begin{equation}
\label{eq:defgenericUMM}
    \mz_{\mathrm{UMM}} (\sigma , \vec{g}) = \oint_{U(N)} [\dd U ] \exp \left\{  \sigma \sum_{n =1} ^{\infty} \frac{g_n}{n} \mathrm{Tr} \left( U^{n} + U^{-n} \right) \right\}
\end{equation}
where $[\dd U] $ is the Haar measure and $\vec{g} = (g_n)_{n \ge 1} $ is an arbitrary collection of interaction coefficients, which satisfy the conditions \cite{Szegoth}
\begin{equation}
    \sum_{n=1} ^{\infty} \frac{g_n}{n} < \infty , \qquad \qquad \sum_{n=1} ^{\infty} \frac{g_n^2 }{n}  < \infty .
\end{equation}
\begin{itemize}
    \item For later convenience, we have factored out an overall coefficient $\sigma \in \R$. Several cases of interest in fact possess a $\mathbb{Z}_2$ symmetry $\sigma \mapsto - \sigma$, which reduces the parameter space to $\sigma \ge 0$. 
    \item Besides, we have restricted our attention to systems that are symmetric under the involution $\mathsf{C} : \ U \mapsto U^{-1}$, which corresponds to charge conjugation in a $U(N)$ QFT. This assumption may be dropped, and the coefficients of $\mathrm{Tr} U^{-n}$ could be generically different from those of $\mathrm{Tr} U^{n}$. This generalization does not add to the central theme of this work, thus we consider models with $\mathsf{C}$-symmetry for clarity.
    \item We have moreover considered integration over $U(N)$, instead of $SU(N+1)$ considered in the main text. This change is to lighten the expressions and, as we are only interested in $N \to \infty$, the difference is negligible.
\end{itemize}
The unitary matrix model \eqref{eq:defgenericUMM} belongs to the family \eqref{eq:genericUMMf} investigated in the main text. The exponential in the integrand is a class function, thus it admits an expansion in characters of $U(N)$. This establishes the relation with the quantum mechanics discussed in Section \ref{sec:QM}.\par
\medskip
Diagonalizing the unitary matrix $U \in U(N)$, \eqref{eq:defgenericUMM} becomes an integral over eigenvalues:
\begin{equation}
\label{eq:eigenvaluegenMM}
     \mz_{\mathrm{UMM}} (\sigma, \vec{g}) = \frac{1}{N!}\int_{[- \pi, \pi]^N}  \prod_{1 \le a < b \le N} \left( 2 \sin \left(\frac{\theta_a - \theta_b}{2} \right) \right)^2 \prod_{a=1} ^{N} e^{2\sigma \sum_{n =1} ^{\infty} \frac{g_n}{n} \cos (n \theta_a) } ~ \frac{\dd \theta_a}{2 \pi} .
\end{equation}
In the \emph{planar} large $N$ limit, i.e. with 't Hooft scaling 
\begin{equation}
    \sigma = N \gamma , \qquad \gamma \text{ fixed},
\end{equation}
the model \eqref{eq:defgenericUMM} has an intricate phase structure as a function of the parameters $\gamma, \vec{g}$. This is more directly visible in the form \eqref{eq:eigenvaluegenMM}. Typically, these phase transitions are third order. We refer to the pertinent literature, see e.g. \cite{Rossi:1996hs} for a review and \cite{Gross:1980he,Wadia:1980cp,Wadia:2012fr,Jurkiewicz:1982iz,Mandal:1989ry,Jain:2013py,Chen:2013vya,Santilli:2019wvq,Russo:2020eif,Sempe:2021qug} for a partial list of works that address phase transitions fitting in our paradigm.\par
The next lemma is an extension of Szeg\H{o}'s theorem to the planar limit.
\begin{lem}[\cite{Szegoth,Santilli:2019wvq}]
\label{lem:Szego0}
Let $ \mz_{\mathrm{UMM}}$ be as in \eqref{eq:eigenvaluegenMM} and 
\begin{equation}
     \mf_{\mathrm{UMM}} (\gamma, \vec{g}) := \lim_{N \to \infty} \frac{1}{N^2} \ln  \mz_{\mathrm{UMM}} \left( N \gamma,  \vec{g} \right) .
\end{equation}
Assume that $\left\{ g_n \right\}_{n \ge 1} $ satisfy 
\begin{equation}
\label{eq:condgrational}
    \sum_{n =1} ^{\infty} g_{n+1} z^{n} = \frac{P_{\vec{g}}(z)}{Q_{\vec{g}}(z)} 
\end{equation}
for some polynomials $P_{\vec{g}}, Q_{\vec{g}}$. Then, there exists $r=r(\vec{g})>0$ such that 
\begin{equation}
    \mf_{\mathrm{UMM}} (\gamma, \vec{g}) = \gamma^2 \sum_{n=1} ^{\infty} \frac{1}{n} g_n ^2 , \qquad \qquad \forall 0 \le \lvert \gamma \lvert < r(\vec{g}) .
\end{equation}
\end{lem}
\begin{proof}This lemma asserts that, for $\gamma$ small enough, \eqref{eq:defgenericUMM} always admits a phase in which $\mf_{\mathrm{UMM}}$ is quadratic in $\gamma$. This was proven by direct computation in \cite{Santilli:2019wvq}, but can also be argued for by analytic continuation of Szeg\H{o}'s theorem \cite{Szegoth}.\par
The condition \eqref{eq:condgrational} guarantees that the functions appearing in the saddle point equation have at most poles. If $z_p$ is a double zero of $Q_{\vec{g}}(z)$, it can be regularized by splitting it into simple poles at $z_p \pm i \varepsilon$, sending $\varepsilon \to 0$ at the end. We can thus assume without loss of generality that at most simple poles appear in the saddle point equation, possibly with the prescribed regularization.
\end{proof}
In the notation of Lemma \ref{lem:Szego0}, $r (\vec{g})$ is the smallest value in the set of critical points. Assume that \eqref{eq:defgenericUMM} has a phase transition at a critical curve 
\begin{equation}
    \gamma = \gamma_c (\vec{g}) .
\end{equation}
For concreteness, we assume the transition is third order, as is typical for unitary matrix models, but the argument works also for second order transitions.\par
Inspired by \cite{Liu:2004vy}, we argue for a way to promote such third order transition to a first order one, which moreover shows Hagedorn behavior. We enforce an average over the coupling $\sigma$, with Gaussian weight of standard deviation $a>0$. The model we consider is:
\begin{align}
    \mz_{\text{holo}} (a, \vec{g}) & = \int_{0} ^{\infty} \dd \sigma  e^{- \frac{\sigma^2 }{2a}} ~  \mz_{\mathrm{UMM}} (\sigma, \vec{g}) \label{eq:defHagMM}  \\
    &= N \int_0 ^{\infty} \dd \gamma ~ \exp \left\{ - N^2 \left[ \frac{\gamma^2}{2a} -   \mf_{\mathrm{UMM}} (\gamma, \vec{g})  \right] \right\} .  \label{eq:HagMMgamma}
\end{align}
If \eqref{eq:defgenericUMM} is not invariant under $\sigma \mapsto - \sigma $, the integration range is from $- \infty$ to $\infty$. The holographic interpretation of unitary matrix models with average over the couplings was advocated in \cite{Murthy:2022ien}, although it was implemented in a different way.\par 
In passing, let us notice that we can refine the additional integration in \eqref{eq:defHagMM} into 
\begin{equation}
    \mz_{\text{holo}} (a, \vec{g}) = \int_{0} ^{\infty} \dd \sigma ~ f(\sigma) \ e^{- \frac{\sigma^2 }{2a}} ~ \mz_{\mathrm{UMM}} (\sigma, \vec{g})
\end{equation}
for an arbitrary function $f(\sigma) $ such that 
\begin{equation}
    \lim_{N \to \infty} \frac{1}{N^2} \ln f(N \gamma) = 0 .
\end{equation}
This condition ensures that the refinement will not affect the phase structure of the model, and hence does not alter the ensuing discussion.\par
The integrand in \eqref{eq:HagMMgamma} at large $N$ is dominated by the saddle points in the variable $\gamma$. Consider first the region $\gamma < \gamma_c$, so that the integrand is Gaussian:
\begin{equation}
\label{eq:expintegrandPI}
   \exp \left\{ - N^2 \gamma^2 \left[ \frac{1}{2a} -  \sum_{n=1} ^{\infty} \frac{1}{n} g_n ^2 \right] \right\} .
\end{equation}
Its maximum is located at $\gamma=0$ if 
\begin{equation}
\label{eq:defacgeneric}
    a< a_c ,\qquad \qquad  \qquad a_c := \frac{1}{2} \left( \sum_{n=1} ^{\infty} \frac{g_n ^2 }{n} \right)^{-1} .
\end{equation}
In this case, the overall growth of \eqref{eq:HagMMgamma} with $N$ is cancelled, and we ought to seek an $O(1)$ solution.\par
However, when $a$ is larger than the threshold $a_c$, the integrand \eqref{eq:expintegrandPI} becomes exponentially large for large values of $\gamma$. We thus need to look for a new saddle point, for strictly positive $\gamma$. In the region $\gamma > \gamma_c$, the saddle point equation reads 
\begin{equation}
\label{eq:genSPEgammaast}
    \frac{\gamma}{a} = \frac{\partial \ }{\partial \gamma } \mf_{\mathrm{UMM}} (\gamma, \vec{g})  .
\end{equation}
If there exists a value $\gamma_{\ast} > \gamma_c$ (with $\gamma_{\ast} = \gamma_{\ast} (a, \vec{g})$ given in terms of the external parameters) that solves the equilibrium condition \eqref{eq:genSPEgammaast}, then 
\begin{equation}
    \ln \mz_{\text{holo}} (a, \vec{g}) \approx  N^2 \left[ \mf_{\mathrm{UMM}} (\gamma_{\ast}, \vec{g})  - \frac{\gamma_{\ast}^2}{2a} \right]  .
\end{equation}
For consistency, one should check that 
\begin{equation}
    \gamma_{\ast} (a, \vec{g})> \gamma_c (\vec{g}) .
\end{equation}
This inequality is understood as defining a region in the $a>0$ parameter space, because $\gamma_{\ast}$ depends (in general) explicitly on $a$ and $\vec{g}$, and $\gamma_c$ only depends on $\vec{g}$. Let us assume the inequality is of the form $a \ge a_{\star}$ for some function $a_{\star} (\vec{g})$. Depending on the values of $a_c$ and $a_{\star}$, three scenarios disclose.
\begin{enumerate}[(i)]
    \item\label{scen1} $a_{\star} < a_c$. The two phases coexists for $a_{\star} < a < a_c$ and a first order phase transition takes place when the dominance of the saddles is exchanged. We think of this phase transition as a function of the control parameter $a$ at fixed couplings $\vec{g}$.
    \item\label{scen2} $a_{\star} = a_c$. There is no coexistence of phases, and a first order transition takes place at $a= a_c( =a_{\star} )$. This situation is sometimes referred to as a \emph{weakly} first order transition.
    \item\label{scen3} $a_{\star} > a_c$. There is a region $a_c < a < a_{\star}$ which admits no saddle point. One should look for the local maxima \eqref{eq:HagMMgamma} and check how the behavior in this intermediate region matches at the junction points $a=a_c$ and $a=a_{\star}$. Typically, the solution to $\ln \mz_{\text{holo}}$ in the intermediate region is $O(N^2)$, thus yielding a first order transition at $a=a_c$. There may or may not be an additional (second or higher order) phase transition at $a=a_{\star}$. The details should be checked on a case-by-case basis.
\end{enumerate}
The dependence may in principle be more general than $a \ge a_{\star}$, for instance, restricting $a$ to a union of intervals. The upshot of the forthcoming analysis is unchanged, although the phase structure would be more involved.

\subsection{Hagedorn phase transition}
In all the scenarios discussed above, the modification of the one-matrix model \eqref{eq:defgenericUMM} into \eqref{eq:defHagMM} led to a first order phase transition, in which the free energy jumps as 
\begin{equation}
\label{eq:Hagedornbehavior}
    \ln \mz_{\text{holo}} = \begin{cases} O(1) & a<a_{H} \\ O(N^2) & a > a_H . \end{cases}
\end{equation}
Here we are generically denoting $a_H$ the value of $a$ at which the transition takes place. It will be a function of the coupling $\vec{g}$ that characterize the model, and oftentimes it is $a_H = a_c$ as defined in \eqref{eq:defacgeneric}.\par

\subsubsection{Other gauge groups}
While we focus on $U(N)$ matrix models, our analysis extends straightforwardly to $SO(N)$ and $Sp(N)$ gauge groups, due to the universality of the large $N$ limit. The extension of Szeg\H{o}'s theorem to these groups was given by Johansson \cite{Johansson:1997C}, from which a version of Lemma \ref{lem:Szego0} is directly worked out. Our algorithm goes through unchanged, except for numerical coefficients in $ \ln \mz_{\text{holo}} $. The result \eqref{eq:Hagedornbehavior} holds. See e.g. \cite[App.S2]{Perez-Garcia:2022geq} for explicit calculations of the large $N$ limit of $\mz_{\mathrm{UMM}}$ with gauge group $SO(N)$ or $Sp(N)$.

\subsubsection{Polyakov loop expectation value}
\label{app:Polyakov}
We have claimed in \eqref{eq:PolyHagedorn} that the expectation value of the Polyakov loop is an order parameter detecting the first order transition. We now prove the claim.\par 
The starting point is that, in the fixed-$L$ matrix model \eqref{eq:ZLNUMM}, the Polyakov loop corresponds to insert $\frac{1}{N}\mathrm{Tr} U$ in the matrix ensemble $\mz_{L}^{(N)}$. Our derivation based on the unitary matrix model is an extension of \cite{Aharony:2003sx,Liu:2004vy}.
\begin{thm}
    Let $\mathcal{P}$ denote the expectation value of a Polyakov loop in the matrix model \eqref{eq:ZFissumLUMM}. Assume the matrix model has the Hagedorn transition described above. Then 
    \begin{equation}
        \begin{cases}  \mathcal{P}= 0 & \text{ if } \beta^{-1}<T_H , \\ \mathcal{P} \ne 0 & \text{ if } \beta^{-1} >T_H , \end{cases}
    \end{equation}
    signalling a first order phase transition at $T_H$.
\end{thm}
\begin{proof}
    Let $U\in SU(N+1)$ and $e^{i \theta_a}$ be its eigenvalues.  Denote by $\dd \sigma_{\ast} (\theta) $ the saddle point measure, normalized to 1. That is, $\dd \sigma_{\ast} (\theta) $ is the equilibrium measure of the matrix model $\mz_{L}^{(N)}$ (that depends on the Veneziano parameter $\gamma$) evaluated at the saddle point $\gamma_{\ast}$. The expectation value of the Polyakov loop is 
    \begin{equation}
        \mathcal{P} = \langle \frac{1}{N}\mathrm{Tr} U \rangle ,
    \end{equation}
    with $\langle \cdot \rangle $ meaning expectation value in the matrix model. A standard large $N$ computation shows that 
    \begin{equation}
        \mathcal{P} = \int_{- \pi} ^{\pi} e^{i \theta } \dd \sigma_{\ast} (\theta) .
    \end{equation}
    It automatically vanishes if the equilibrium measure is uniform on the circle, $ \dd \sigma_{\ast} (\theta) = \frac{\dd \theta }{2 \pi}$ for all $- \pi < \theta \le \pi$, which is true in the phase $\beta^{-1} <T_H$.\par 
    In the phase $\beta^{-1}>T_H$ the unitary matrix model is, by construction, in a phase in which the eigenvalues $e^{i \theta_a}$ are not spread on the whole circle. The non-triviality of $\gamma_{\ast}$ in the phase $\beta^{-1}>T_H$ guarantees that this remains true after the integration over the Veneziano parameter. More precisely, there exist $\theta_{\pm}$, with $- \pi < \theta_{-} < \theta_{+} < \pi$ such that $ \dd \sigma_{\ast} (\theta) \ne  0$ for $\theta_{-} < \theta < \theta_{+}$ and vanishes otherwise. Measures of this type give a non-vanishing expectation value of $\mathrm{Tr} U$, thus $\mathcal{P} \ne 0$.\par
    \medskip 
    Before concluding the proof, we ought to comment on a subtlety in our computation. This remark is meant for experts, and is a generalization of a remark in \cite{Aharony:2003sx}.\par
    In a deconfined phase, these models should possess a $\Z_{N+1}$ one-form symmetry from the center of $SU(N+1)$. We thus should obtain a collection of saddle point configurations permuted under the $\Z_{N+1}$ symmetry. We have only accounted for one of them, namely the one with eigenvalues centered around $0$. Accounting for all the saddle configurations, their contributions sum up to zero. Naively, the Polyakov loop would have trivial expectation value at $T>T_H$ as well. However, this is exactly the same technical issue that one would encounter when trying to compute the order parameter in the Ising model, with vanishing external field. It is well-known that one should instead do the computation of the order parameter in an external field of modulus $\varepsilon$, and send $\varepsilon \to 0$ at the end. In this way, one finds two different limiting values, depending on the temperature. The same resolution applies here. We implicitly assume a small perturbation of the action that breaks the $\Z_{N+1}$ symmetry explicitly and favours one saddle (without loss of generality, we select the one centered at 0). It is straightforward to repeat the computation at finite $\varepsilon >0$ and, turning off the perturbation at the end, we find $\mathcal{P} \ne 0$ if $T>T_H$.
\end{proof}

\subsection{Example 1: Variations on the IOP model}
\label{sec:BQCD2MM}
Recall the unitary matrix model cIOP partition function in \eqref{eq:ZBQCD2equalsZcIOP}, which we rewrite here:
\begin{equation}
    \mz_{\text{\rm cIOP}}^{(N)} (\sigma, y) = \oint_{U(N)} [\dd U ] \left[ \det \left(1-\sqrt{y} U \right)  \det \left(1-\sqrt{y} U^{-1} \right) \right]^{- \sigma } . \label{eq:BosQCD2MMapp} 
\end{equation}
For consistency with the rest of the appendix, we have replaced the discrete number of flavors $L \in \mathbb{N}$ with a continuous parameter $\sigma >0$. Likewise, we introduce its averaged version by integrating over $\sigma$ with Gaussian weight:
\begin{equation}
\label{eq:BosQCD2sum}
    \mz_{\overline{\text{cIOP}}} (y) = \int_0 ^{\infty} \dd \sigma ~ y^{\frac{\sigma^2}{2}}  \mz_{\text{\rm cIOP}}^{(N)} (\sigma, y) .
\end{equation}
For simplicity, we work in what we have denoted as Schur slice, in which $1/\beta = -1/\ln (y)$ plays the role of $a$.\par
\begin{thm}
Let $y=e^{- \beta}$. There exists $T_H >0$ such that, in the large $N$ limit, 
\begin{equation}
    \ln \mz_{\overline{\text{\rm cIOP}}}  =  \begin{cases} O(1) & \frac{1}{\beta}< T_{H} \\ O(N^2) & \frac{1}{\beta}>T_H . \end{cases}
\end{equation}
\end{thm}
\begin{proof}The proof is done in two steps: 
\begin{itemize}
    \item[($i$)] take the large $N$ planar limit of $\mz_{\text{\rm cIOP}}^{(N)} $; 
    \item[($ii$)] insert the result in \eqref{eq:BosQCD2sum} and extremize.
\end{itemize}\par
The Veneziano limit of the unitary matrix model \eqref{eq:BosQCD2MMapp} was addressed in \cite{Baik:2000,Santilli:2020ueh} (see also \cite{Chen:2013vya}). The computation of step ($ii$) is almost identical to the one performed in the next subsection, thus we omit it. Suffice it here to note that \eqref{eq:BosQCD2MMapp} is related to \eqref{eq:QCD2MMapp} through the map
\begin{equation}
    (\sqrt{y},\sigma) \ \mapsto \ (- \sqrt{y}, -\sigma) .
\end{equation}
In particular, it also shows a third order phase transition that can be promoted to a first order one by integrating over $\sigma$. See Appendix \ref{sec:AvgQCD2UMM}.\par
\end{proof}

\subsection{Example 2: Matrix model of \texorpdfstring{QCD$_2$}{QCD2}}
\label{sec:AvgQCD2UMM}
The take-home lesson of this appendix is that, by introducing an extra integration with Gaussian weight, unitary one-matrix models with a third order phase transition are promoted to have a first order one, with Hagedorn behavior near the transition point.\par
We now exemplify this in the toy model for QCD$_2$ \cite{Hallin:1998km} discussed in Section \ref{sec:ExQCD2}. In the notation of Appendix \ref{sec:UnitaryMM} it reads
\begin{align}
    \mz_{\text{\rm QCD$_2$}} (\sigma, y) &= \oint_{U(N)} [\dd U ] \left[ \det \left(1+\sqrt{y} U \right)  \det \left(1+\sqrt{y} U^{-1} \right) \right]^{\sigma} \label{eq:QCD2MMapp} \\
    &= \oint_{U(N)} [\dd U ] \exp \left\{  \sigma \mathrm{Tr} \left[ \ln\left(1+ \sqrt{y} U^{-1} \right) + \ln\left(1+ \sqrt{y} U^{-1} \right) \right] \right\} \notag \\
    &= \oint_{U(N)} [\dd U ] \exp \left\{  \sigma \sum_{n=1}^{\infty} \frac{y^{n/2}}{n} \mathrm{Tr} \left( U^{n} + U^{-n} \right) \right\} . \label{eq:expanndexpQCD2}
\end{align}
From \eqref{eq:expanndexpQCD2}, the model manifestly belongs to the family \eqref{eq:defgenericUMM} with $g_n = \sqrt{y}^{n}$. Here we have relaxed the QCD interpretation, replacing the number of quarks $L \in \mathbb{N}$ with $\sigma >0$.\par
\begin{lem}[\cite{Hallin:1998km}]
\label{lemma:HP}
The quantity
\begin{equation}
    \mathcal{F}_{\text{\rm QCD$_2$}} ( \gamma ,y ) := \lim_{N \to \infty} \frac{1}{N^2} \ln  \mz_{\text{\rm QCD$_2$}} (N \gamma, y) 
\end{equation}
shows a third order phase transition at the critical curve 
\begin{equation}
    \gamma_c (y) = \frac{1-\sqrt{y}}{2\sqrt{y}}  ,
\end{equation}
with
\begin{equation}
	 \mathcal{F}_{\text{\rm QCD$_2$}} ( \gamma , y ) = \begin{cases} - \gamma^2 \ln (1-y) & \gamma < \frac{1-\sqrt{y}}{2\sqrt{y}} \\ - (2 \gamma +1) \ln \left(1 + \sqrt{y}\right) +  \frac{1}{4} \ln y + C(\gamma)  & \gamma > \frac{1-\sqrt{y}}{2\sqrt{y}} . \end{cases}
\end{equation}
In the second phase, $C(\gamma)$ is the $y$-independent term
\begin{equation}
    C(\gamma)= - \gamma ^2 \ln \left( \frac{4 \gamma (\gamma +1)}{(2 \gamma +1)^2}\right)+ \frac{1}{2}\ln \left(1+2 \gamma \right)+(1+2 \gamma) \ln \left(\frac{2 (\gamma +1)}{1+2 \gamma}\right) .
\end{equation}
\end{lem}
\begin{proof}See \cite{Hallin:1998km} for the original proof, or \cite{Santilli:2020ueh} for a uniform treatment of this case and the model appearing in Appendix \ref{sec:BQCD2MM}.
\end{proof}\par
According to the general discussion, we want to integrate over the external field $\sigma$, and show that this produces the Hagedorn behavior \eqref{eq:Hagedornbehavior}. For consistency with the rest of the appendix, here we consider:
\begin{equation}
\begin{aligned}
    \mz_{\overline{\text{\rm QCD$_2$}}} (y) & = \int_0 ^{\infty} \dd \sigma  ~ y^{\frac{\sigma^2}{2}} \mz_{\text{\rm QCD$_2$}} (\sigma, y) \\
    & =  \int_0 ^{\infty} \dd \sigma ~ y^{\frac{\sigma^2}{2}} \oint_{U(N)} [\dd U ] \exp \left\{  \sigma \mathrm{Tr} \left[ \ln\left(1+ \sqrt{y} U^{-1} \right) + \ln\left(1+ \sqrt{y} U^{-1} \right) \right] \right\} . \label{eq:QCD2withsum}
\end{aligned}
\end{equation}
Again we work in the analogue of the Schur slice, in which the Gaussian weight in the measure for $\sigma$ is $y^{\sigma^2/2}$. We might have taken a different Gaussian weight $e^{-\sigma^2/(2a)}$, but the reader can check that the conclusions are unchanged in the more general case. Besides, in the planar limit, one can check numerically that, in this and the previous example, the difference between summing over $L$ or integrating over $\sigma$ is sub-leading in $N$. Hence, $\mz_{\overline{\text{\rm QCD$_2$}}}$ and $\mz_{\text{\rm Ex2}}$ of Section \ref{sec:AvgQCD2FF} will have the same properties.
\begin{thm}
Let $\mz_{\overline{\text{\rm QCD$_2$}}}$ be as in \eqref{eq:QCD2withsum}, with $y=e^{- \beta}$. In the large $N$ limit, it undergoes a first order phase transition at $\frac{1}{\beta} =T_H$, where $ T_H\approx 1.039$. Moreover, 
\begin{equation}
\label{eq:ZAHagedornapp}
    \ln \mz_{\overline{\text{\rm QCD$_2$}}} =  \begin{cases} O(1) & \frac{1}{\beta}< T_{H} \\ O(N^2) & \frac{1}{\beta}>T_H . \end{cases}
\end{equation}
\end{thm}
\begin{proof}
As in the proof of Theorem \ref{thm:Amodel} in Subsection \ref{sec:AvgQCD2FF}, we divide the computation in two steps. Step ($i$) consists in maximizing the inner matrix model $\mz_{\text{\rm QCD$_2$}}$ in the planar large $N$ limit. The solution is in Lemma \ref{lemma:HP}.\par
The discussion in step ($ii$) in the proof of Theorem \ref{thm:Amodel} goes through identically at this stage. Let us recall the situation we have found in Subsection \ref{sec:AvgQCD2FF}: 
\begin{itemize}
    \item There exists a trivial saddle point $\gamma=0$, valid for $\beta > \beta_c$;
    \item There exists a nontrivial saddle point $\gamma_{\ast} >0$, valid for $\beta < \beta_{c}$;
    \item There is no coexistence phase and the two saddles exchange dominance precisely at $\beta=\beta_c$.
\end{itemize} \par
Therefore, this case corresponds to scenario \eqref{scen2}. The direct inspection in the proof of Theorem \ref{thm:Amodel} can be repeated here and tells us that a phase transition kicks in at $T_H = 1/\beta_c$. Putting everything together, we finally arrive at the first order, Hagedorn-like phase transition \eqref{eq:ZAHagedornapp} with Hagedorn temperature
\begin{equation}
    T_{H} = \frac{1}{\beta_c} \approx 1.039 .
\end{equation}\par
\end{proof}

\subsection{Example 3: Effective \texorpdfstring{$\mN=4$}{N=4} super-Yang--Mills}
\label{sec:N4SYMMM}
The reasoning of Appendix \ref{sec:UnitaryMM} was inspired by and extended a result of Liu \cite{Liu:2004vy} on the effective description of $\mN=4$ super-Yang--Mills theory on $\mathbb{S}^3 \times \mathbb{S}^1 _{\bsym}$, with radius of the thermal circle the inverse temperature $\bsym=1/T$. It was shown in \cite{Sundborg:1999ue,Aharony:2003sx} that the partition function in the weak 't Hooft coupling limit reduces to:
\begin{equation}
\label{eq:N4SYMMMfull}
	\hat{\mz}_{\mN =4} (\vec{a}) = \oint_{U(N)} \dd U ~ \exp \left\{ \sum_{n \ge 1} \frac{a_n}{n} \tr \left( U^{n} \right) \tr \left( U^{-n} \right)   \right\} .
\end{equation}
The coefficients $(a_n)_{n \ge 1} \subset \R$ are functions of the inverse temperature $\bsym$. The leading contribution at large $N$ comes from the reduced matrix model  \cite{Aharony:2003sx,Alvarez-Gaume:2005dvb}
\begin{equation}
\label{eq:N4SYMMMapp}
	\mz_{\mN =4}  (a) = \oint_{U(N)} \dd U ~ \exp \left\{ \frac{a}{2} \tr \left( U \right) \tr \left( U^{-1} \right)   \right\} . 
\end{equation}
We have denoted
\begin{equation}
\label{eq:aSYMTemp}
    a:= 2 a_1 = \frac{4e^{- \bsym} (3-e^{- \bsym/2} )}{(1-e^{- \bsym/2})^3} .
\end{equation}
\begin{prop}[\cite{Liu:2004vy}]
The matrix model \eqref{eq:N4SYMMMapp} undergoes a first order phase transition at $\frac{1}{\bsym} = T_H$, where 
\begin{equation}
    T_{H} = \frac{1}{2 \ln \left( 2 + \sqrt{3} \right) }.
\end{equation}
Besides, it has the behavior \eqref{eq:Hagedornbehavior}.
\end{prop}
\begin{proof}Using a Hubbard--Stratonovich transformation \cite{Liu:2004vy}, \eqref{eq:N4SYMMMapp} can be recast in the form \eqref{eq:defHagMM}:
\begin{equation}
\label{eq:N4SYMGWW}
	\mz_{\mN =4}  (a) = \int_0 ^{\infty} \sigma \dd \sigma ~ e^{- \frac{\sigma^2}{2a}} \oint_{U(N)} \dd U ~ \exp \left\{ \frac{\sigma}{2} \tr \left( U+ U^{-1} \right)   \right\} .
\end{equation}
The inner integral 
\begin{equation}
    e^{N^2 \mf_{\mathrm{GWW}} (\sigma/N)} = \oint_{U(N)} \dd U ~ \exp \left\{ \frac{\sigma}{2} \tr \left( U+ U^{-1} \right)   \right\}
\end{equation}
is the famous Gross--Witten--Wadia (GWW) model \cite{Gross:1980he,Wadia:1980cp,Wadia:2012fr}, and belongs to the family \eqref{eq:defgenericUMM} with $g_1 = \frac{1}{2}$, $g_{n >1} =0$. In the planar limit, with scaling $\sigma =N \gamma$, it undergoes a third order phase transition at $\gamma =1$:
\begin{equation}
    \mf_{\mathrm{GWW}} (\gamma) = \begin{cases} \frac{\gamma^2}{4} & \gamma \le 1 \\ \gamma - \frac{1}{2} \ln \gamma - \frac{3}{4} & \gamma >1 . \end{cases}
\end{equation}
Plugging this result back into \eqref{eq:N4SYMGWW}, the large $N$ limit is dominated by the saddle point $\gamma_{\ast}$ that solves 
\begin{equation}
\label{eq:SPEN4GWW}
    \frac{\gamma_{\ast}}{a} =  \mf_{\mathrm{GWW}} ^{\prime} (\gamma_{\ast} ) .
\end{equation}
In the region $\gamma \le 1$, the solution is $\gamma_{\ast} =0$, which is the absolute minimum if 
\begin{equation}
    \frac{1}{2a} - \frac{1}{4} >0 \qquad \Longrightarrow \qquad a >2 .
\end{equation}
Inspecting the region $\gamma >1$, the solution to \eqref{eq:SPEN4GWW} is the saddle point 
\begin{equation}
    \gamma_{\ast} = \frac{1}{2} \left[ a + \sqrt{a(a - 2)} \right] .
\end{equation}
This solution is valid in the region $a \ge 2$. Computing the free energy, one finds a first order phase transition at $a_H=2$.\par
We are in the scenario \eqref{scen2} described above, and the transition manifests the Hagedorn behavior \eqref{eq:Hagedornbehavior}, expected in this effective description of $\mN=4$ super-Yang--Mills \cite{Liu:2004vy}. Using the relation \eqref{eq:aSYMTemp} and $\bsym=1/T$, one in facts finds out the expected behavior
\begin{equation}
    \ln \mz_{\mN =4}  (a(T)) =  \begin{cases} O(1) & T < T_{H} \\ O(N^2) & T>T_H , \end{cases}
\end{equation}
with Hagedorn temperature
\begin{equation}
    T_{H} = \frac{1}{2 \ln \left( 2 + \sqrt{3} \right) } \approx 0.380 .
\end{equation}
\end{proof}

\subsection{Counterexample: \texorpdfstring{$q$}{q}-ensemble and effective Chern--Simons theory}
\label{app:qCSMM}
Our procedure to upgrade a third order phase transition to a first order one relies on Lemma \ref{lem:Szego0}. In practice, it boils down to the existence of a phase in which the eigenvalues $\theta_a$ of the random matrix in $\mz_{\text{UMM}}$ fill the whole interval $[-\pi, \pi]$. For the sake of completeness, here we provide a counterexample to our procedure, based on violating the assumption \eqref{eq:condgrational} in Lemma \ref{lem:Szego0}. We will also explain that the lack of a first order phase transition in this case is expected on physical grounds.\par
\medskip
Let $\mathcal{P} (x) = \prod_{m \ge 1} (1-x^{m})$ be the generating function of partitions and $\vartheta (z;q)$ denote the Jacobi theta function. Consider the following $q$-ensemble \cite{Okuda:2004mb,Romo:2011qp,Szabo:2010sd}:
\begin{equation}
\label{eq:CSthetaMM}
    \mz_{\mathrm{CS}} (q) = \frac{1}{ \mathcal{P} (q^2) N!} \int_{[- \pi, \pi]^N}  \prod_{1 \le a < b \le N} \left( 2 \sin \left(\frac{\theta_a - \theta_b}{2} \right) \right)^2 \prod_{a=1} ^{N} \vartheta (e^{i \theta_a};q)  ~ \frac{\dd \theta_a}{2 \pi} .
\end{equation}
At $q =e^{i 2 \pi /k}$, $k \in \mathbb{Z}$, this unitary one-matrix model computes the partition function of topological $U(N)_k$ Chern--Simons theory on $\mathbb{S}^3$ \cite{Romo:2011qp}, and is understood as the analytic continuation of the Chern--Simons level for $0 < \lvert q \rvert < 1$. Throughout, we will set 
\begin{equation}
\label{eq:defqsigma}
    q= e^{-1/2 \sigma} , \qquad \sigma >0 ,
\end{equation}
which is the usual $q$-parameter in passing from Chern--Simons theory to the associated topological string theory. The Jacobi triple product identity implies that 
\begin{align}
    \frac{\vartheta (z;q) }{\mathcal{P} (q^2) }  & = \prod_{m =0}^{\infty} (1+z q^{2m+1})(1+z^{-1} q^{2m+1}) \notag \\
        &= \exp \left[ \sum_{m =0 }^{\infty} \left( \ln (1+z q^{2m+1}) + \ln (1+z^{-1} q^{2m+1}) \right)\right] \notag \\
        &= \exp \left[ \sum_{n=1} ^{\infty} \frac{(z^{n} + z^{-n})}{n} (-1)^{n+1}\sum_{m =0 }^{\infty}  q^{n(2m+1)} \right] \notag  \\
        &= \exp \left[ \sum_{n=1} ^{\infty} \frac{(z^{n} + z^{-n})}{n} \left( (-1)^{n+1} \frac{q^{n}}{1-q^{2n} } \right) \right] .
\end{align}
Using \eqref{eq:defqsigma} we immediately identify the coefficients 
\begin{equation}
    \hat{g}_n = \frac{(-1)^{n+1}}{2 \sinh \left( \frac{n}{2 \sigma} \right)} , \qquad \forall n \ge 1 .
\end{equation}
Here we are using $\hat{g}_n$ to stress that they differ from $g_n$ in \eqref{eq:defgenericUMM} in that they are not normalized by an overall $\sigma$. We consider the 't Hooft limit $N \to \infty$ with $\sigma = \gamma N$, $\gamma$ fixed. In this regime, in which \eqref{eq:CSthetaMM} matches with the topological string partition function on $T^{\ast} \mathbb{S}^3$ \cite{Marino:2004eq}, we have 
\begin{equation}
    \hat{g}_n \approx  (-1)^{n+1} \frac{ \sigma}{n} ,
\end{equation}
which reduces to the form \eqref{eq:defgenericUMM}. Moreover, 
\begin{equation}
    \sum_{n=1} ^{\infty} \frac{\hat{g}_n}{n} z^n \approx - \sigma \mathrm{Li}_2 (-z) , \qquad  \sum_{n=1} ^{\infty} \frac{\hat{g}_n}{n} \approx \frac{\pi^2}{12} \sigma,  \qquad  \sum_{n=1} ^{\infty} \frac{\hat{g}_n ^2 }{n} \approx \zeta (3)  \sigma ,
\end{equation}
where $\mathrm{Li}_2$ is the polylogarithm of order 2. Matrix models involving polylogarithmic potentials are relevant in holography and have appeared before in \cite{Amado:2016pgy}.\par
The relation 
\begin{equation}
    \frac{\dd \ }{\dd z } \left(  - \sigma \mathrm{Li}_2 (-z) \right) = - \sigma \ln (1+z) \quad \Longrightarrow \quad \sum_{n=1} ^{\infty} \hat{g}_{n+1} z^n \approx - \sigma \ln (1+z) 
\end{equation}
implies the violation of condition \eqref{eq:condgrational}. In practice, $\theta = \pi$ is a branch point for the current model, bringing a branch cut into the saddle point equation. Denoting $\varrho (\theta)$ the density of eigenvalues, a phase in which $\supp \varrho= [- \pi, \pi]$ is ruled out by the presence of the branch cut, invalidating the argument used though the rest of this appendix. The appearance of the logarithm in the saddle point equation is indeed a trademark of Chern--Simons theory on $\mathbb{S}^3$ \cite{Marino:2004eq}.\par
The solution of \eqref{eq:CSthetaMM} in the 't Hooft large $N$ limit is a unitary matrix model version of the computation in \cite{Marino:2004eq}. The upshot is that only one phase exists, with eigenvalue density $\varrho (\theta) $ supported on an arc, $\supp \varrho \subset [- \pi, \pi]$, which shrinks as $\gamma$ is increased. Averaging over $\sigma$ with Gaussian weight, 
\begin{equation}
\label{eq:ZCSavg}
    \mz_{\overline{\mathrm{CS}}} (a) = \int_0 ^{\infty} \dd \sigma ~ e^{- \frac{\sigma^2}{2a}} \mz_{\mathrm{CS}} (e^{-1/2\sigma}) ,
\end{equation}
we always find a non-trivial saddle point $\gamma_{\ast} >0$, which is moved to the right as $a>0$ is increased. This behavior is due to the higher than second order dependence of $\mz_{\text{CS}} := \lim_{N \to \infty} \frac{1}{N^2}\ln \mz_{\text{CS}}$ on $\gamma$, obtained by solving the saddle point equation explicitly adapting \cite{Marino:2004eq}. In conclusion, \eqref{eq:ZCSavg} is always in a phase in which $\ln  \mz_{\overline{\mathrm{CS}}}  = O(N^2)$.\par

\subsubsection{Chern--Simons theory from the Cardy limit of \texorpdfstring{$\mN=4$}{N=4} super-Yang--Mills}
\label{sec:appCardyCS}
The outcome of our analysis agrees with the physics of the problem, as we now explain.\par
The partition function of Chern--Simons theory on $\mathbb{S}^3$ appears in the Cardy limit of the superconformal index of $\mN=4$ super-Yang--Mills \cite{ArabiArdehali:2021nsx}. In particular, for $q$ near a root of unity, one gets $U(N)_{\ell N}$ Chern--Simons theory, with $\ell$ specifying the root of unity \cite{ArabiArdehali:2021nsx}. Upon analytic continuation of the coupling, the factor of $N$ matches precisely with our choice of scaling for $\sigma$, with $\gamma$ playing the role of the Wick-rotated $\ell$. The Cardy limit drives $\mN=4$ super-Yang--Mills deep in the regime in which the black hole entropy is large, thus 
\begin{itemize}
    \item no phase transition is expected, and
    \item the partition function should grow as $e^{N^2}$,
\end{itemize}
consistent with our findings.\par
The dictionary between our computation and \cite{ArabiArdehali:2021nsx} implies that our integral over $\gamma$ should correspond to a sum over $\ell \in \N$ in the $U(N)_{\ell N}$ Chern--Simons theory. This observation hints at an interpretation of the averaging procedure \eqref{eq:ZCSavg} in terms of summing over the Riemann sheets of the superconformal index of $\mN=4$ super-Yang--Mills, in the Cardy limit.\par 
It would be interesting to see whether a neat holographic interpretation emerges, in connection with the work \cite{Cassani:2021fyv}.

\section{Proofs}
\label{app:longproof}

This appendix contains the derivation of the statements in Subsections \ref{sec:spectral}-\ref{sec:MMalgebra}, which build up the spine of our main result about von Neumann algebras: Theorem \ref{thm:istype3Rigor}.\par
Throughout the main text we have worked with representations $R$ whose Young diagram has $\ell (R) \le L$, and interpreted them as $SU(L+1)$ representations. This has the unfortunate drawback of having all the expressions normalized by $L+1$, rather than $L$. To reduce clutter throughout this (already long and technical) appendix, we adopt the shorthand notation 
\begin{equation}
\label{eq:Ltilde}
    \tilde{L} := L+1 .
\end{equation}

\subsection{Explicit form of the Wightman functions}
\label{app:proofthmrho}
This appendix contains the proofs of Theorems \ref{thmrhoisrho} and \ref{thm:rhoandOmega}.\par
A caveat is that we stick to the conventions of \cite{Festuccia:2006sa} for the $\pm i \varepsilon $. In particular, this entails defining the Fourier transform 
\begin{equation} 
\widetilde{f} (\omega) = \int_{-\infty}^{\infty} \dd t e^{i t \omega} f(t) .
\end{equation}
To compare with other references, one might need to redefine $\omega \mapsto - \omega$ and some signs in intermediate expressions will be opposite. The final answers for physical quantities such as $\rho (\omega)$ are of course the same.

\begin{proof}[Proof of Theorem \ref{thmrhoisrho}]
We divide the proof of \eqref{eq:rhoisrho} in four steps:
\begin{enumerate}[(1)]
    \item We decompose the problem into the evaluation of two simpler pieces;
    \item We evaluate the first piece and show that it contributes only if $\omega >0$;
    \item We evaluate the second piece and show that it contributes only if $\omega <0$;
    \item We sum the two pieces, being careful with the probe approximation and the factors of $e^{- \beta \mu}$.
\end{enumerate}
\underline{Step (1).} We plug the definition of $\phi_L$ into \eqref{eq:defGplus}: 
\begin{equation}
    G_{L,+} (t)  = \frac{1}{\mz_L ^{(N)} }  \tr_{\mathscr{H}_L^{(N)} \otimes \Gamma_{\mathrm{probe}}} \left[ e^{- \beta H^{\prime}} ~ \frac{1}{2} \left( e^{i t H^{\prime}}  \mathcal{O}_L ^{\dagger} e^{-it H^{\prime}} \mathcal{O}_L + e^{i t H^{\prime}}  \mathcal{O}_L e^{-it H^{\prime}} \mathcal{O}_L ^{\dagger}  \right) \right] ,
\end{equation}
where we have used the fact that $\mathcal{O}_L (t) \mathcal{O}_L (0) $ has vanishing thermal expectation value. We define 
\begin{equation}
    G_{\mathcal{O}_L } (t)  := \frac{1}{\mz_L ^{(N)} }  \tr_{\mathscr{H}_L^{(N)} \otimes \Gamma_{\mathrm{probe}}} \left[ e^{- \beta H^{\prime}} ~ e^{i t H^{\prime}}  \mathcal{O}_L ^{\dagger} e^{-it H^{\prime}} \mathcal{O}_L  \right] ,
\end{equation}
so that $  G_{\mathcal{O}_L ^{\dagger} } (t) $ is the same but with $\mathcal{O}_L ^{\dagger}$ and $ \mathcal{O}_L$ exchanged, and 
\begin{equation}
    G_{L,+} (t)  = \frac{  G_{\mathcal{O}_L } (t)  +  G_{\mathcal{O}_L ^{\dagger} } (t) }{2} .
\end{equation}
For later convenience we also introduce the Fourier transform of this quantity, 
\begin{equation}
    \widetilde{G}_{L,+} (\omega) = \frac{\widetilde{G}_{\mathcal{O}_L } (\omega) + \widetilde{G}_{\mathcal{O}_L ^{\dagger} } (\omega) }{2} .
\end{equation}
Our goal is to compute 
\begin{equation}
\label{appeq:omegatwoparts}
    \rho (\omega) = \frac{(1- e^{- \beta \omega} ) }{2} \left[ \widetilde{G}_{\mathcal{O}_L } (\omega) + \widetilde{G}_{\mathcal{O}_L ^{\dagger} } (\omega) \right] .
\end{equation}
We thus need to evaluate the two terms on the right-hand side.\par
\medskip
\underline{Step (2).} We now focus on the evaluation of $\widetilde{G}_{\mathcal{O}_L } (\omega) $.\par

We let the probe operators act on the system in the probe approximation $\mu \gg 1$, with $\mu$ the mass term in the probe Hamiltonian, to neglect excited states of the probe. Besides, we omit the subscript from the probe sectors to reduce clutter. We now analyze the right-hand side of \eqref{eq:defGplus}. By direct computation we get 
\begin{align}
      & \sum_{\ai ,\si ,\dot{\ai},\dot{\si}}   \langle  R, \ai ,\si ; \phi(R) , \dot{\ai}, \dot{\si} \rvert \otimes \langle 0,\dots, 0 \rvert ~ e^{- \beta H^{\prime} } \notag \\
      & \times \frac{1}{\tilde{L}}\sum_{p=1} ^{\tilde{L}} \left[ e^{it H^{\prime}} ( 1 \otimes  c_p ) e^{-i t H^{\prime}} ( 1 \otimes c^{\dagger}_{p} ) \right] \lvert R, \ai ,\si ; \phi(R) , \dot{\ai}, \dot{\si} \rangle \otimes \lvert 0, \dots , 0 \rangle \notag \\&=\frac{e^{(- \beta+it) H^{\prime} (R,\phi(R))}}{\tilde{L}}\mathrm{Tr}_{\bigoplus_{J} \mathscr{H}(R\sqcup\Box_J)\otimes \mathscr{H}(\phi(R))}(e^{-it H^\prime})\\
    & = \frac{1}{\tilde{L}} \sum_{\hat{R}= R \sqcup \Box} e^{- \beta H^{\prime} (R,\phi(R)) - it \left[ H^{\prime} (\hat{R},\phi(R)) - H^{\prime} (R,\phi(R)) \right] } ~ \mathfrak{d}_R ^2 ~ \dim (\hat{R} ) \dim (\phi(R)) . \label{eq:probarhosumA}
\end{align}
The sum is over the representations $\hat{R}$ obtained in the Clebsch--Gordan expansion of the tensor product $R \otimes \Box$, whose Young diagrams are obtained by adding a box to that of $R$ in all possible consistent ways. Here and in the following we denote $H^{\prime} (R_1, R_2)$ the eigenvalue of the Hamiltonian acting on a state in $\mathscr{H}_L (R_1) \otimes \mathscr{H}_L (R_2)$, and likewise for $H_{\text{\rm int}}  (R_1, R_2)$.\par
Recall from the construction in Subsection \ref{sec:probe} that the Hamiltonian $H^{\prime}$ acts on the composite system $\mathscr{H}_L ^{\mathrm{tot}}$, i.e. on the quantum mechanics coupled to the external probe. The time evolved operator is $e^{it H^{\prime}}  ( 1 \otimes c_p ) e^{-i t H^{\prime}}$. We have then used that the number operator on the probe Fock space vanishes in the state $\lvert 0, \dots , 0 \rangle $. The probe approximation is important here in discarding the contributions from the excited states of the probe.\par
Let us take a closer look at the sum in the last line of \eqref{eq:probarhosumA}. Due to the ordering restriction $\hat{R}_j \ge \hat{R}_{j+1}$, not all the ways to add a box to $R$ give an allowed Young diagram $\hat{R}$. The set of rows of $R$ to which a box can be appended is specified by the set $\mathscr{J}_R$ in \eqref{eq:defsetJ}. Then, the sum over $\hat{R}=R \sqcup \Box$ is more rigorously expressed as (the symbols $\delta (\cdot )$ are Kronecker delta $\delta_{\cdot , 0 }$)
\begin{equation}
    \sum_{\hat{R}= R \sqcup \Box} = \sum_{J\in \mathscr{J}_R} \sum_{\hat{R} \in \mathfrak{R}^{SU(\tilde{L})}} \delta \left( \hat{R}_J - R_J -1 \right) ~ \prod_{\underset{j \ne J}{j=1}}^{\tilde{L}} \delta \left( \hat{R}_j - R_j \right) .
\end{equation}
We then sum \eqref{eq:probarhosumA} over $R \in \mathfrak{R}_L ^{(N)} $ and, denoting $R \sqcup \Box_J$ the Young diagram obtained by adding a box at the end of the $J^{\text{th}}$ row of $R$, we expand the resulting expression as 
\begin{equation}
    \sum_{R \in \mathfrak{R}_L ^{(N)}} \frac{1}{\tilde{L}} \sum_{J\in \mathscr{J}_R}  e^{- \beta H^{\prime} (R,\phi(R)) - it \left[ H^{\prime}(R \sqcup \Box_J, \phi(R)) - H^{\prime}(R, \phi(R) ) \right] } ~ \mathfrak{d}_R ^2 ~\dim (R \sqcup \Box_J)  \dim (\phi(R)) .
\end{equation}
The $j^{\text{th}}$ row of the diagram $R \sqcup \Box_J$ has length $R_j$ if $j \ne J$ and $R_{J}+1$ if $j=J$, which only makes sense for $J \in \mathscr{J}_R$.\par
Recall from the definition \eqref{eq:defHprime} that we have
\begin{equation}
    H^{\prime} =H \otimes 1 + 1 \otimes H_{0, \mathrm{probe}} + H_{\text{int}} ,
\end{equation}
with, by construction, $e^{- \beta H^{\prime} } \lvert \mathsf{v} \rangle =  e^{-\beta H } ~\lvert \mathsf{v} \rangle$, for every $\lvert \mathsf{v} \rangle \in \mathscr{H}_L (R) \otimes \mathscr{H}_L (\phi(R)) \otimes \lvert 0, \dots, 0 \rangle_{\text{\rm probe}}$. The term $1 \otimes H_{0, \mathrm{probe}}$ is simply a large mass enforcing the probe approximation. On the other hand, we restrict our attention to the interaction \eqref{eq:Hint}, which gives 
\begin{equation}
    H_{\text{int}} (R\sqcup \Box_J,\phi(R)) - H_{\text{int}}(R , \phi(R)) = g (R_J +\tilde{L} - J ) =: E^{\text{int}} _J
\end{equation}
for $\phi (R)=R$ or $\overline{R}$. More general choices of $\phi (R)$ are treated analogously and yield to similar definitions of $E^{\text{int}} _J$. The time dependence only appears through 
\begin{equation}
    e^{- it \left[H^{\prime}(R \sqcup \Box_J, \phi(R)) - H^{\prime}(R, \phi (R) ) \right]} = e^{-it E_J^{\text{\rm int}} -it \mu } , 
\end{equation}
and hence, taking the Fourier transform, the dependence on $\omega$ is through 
\begin{equation}
    \delta \left( \omega - E^{\text{int}} _J - \mu \right) 
\end{equation}
in each summand. Notice that the uniform shift $E^{\text{int}} _J \mapsto E^{\text{int}} _J + \mu $ induced by the probe mass is expected, because the correlation functions we are considering are tailored to probe energies of the scale of the first excitation.\par
Putting all together, we obtain 
\begin{equation}
\label{eq:FourierGOdeltaomega}
    \sum_{R \in \mathfrak{R}_L ^{(N)}} \frac{1}{\tilde{L}} \sum_{J \in \mathscr{J}_R}  e^{- \beta \lvert R \rvert } ~\delta \left( \omega - \mu  - E^{\text{int}} _J \right)  ~\mathfrak{d}_R ^2  \dim (R \sqcup \Box_J) \dim (\phi (R)) .
\end{equation}
The additional quantum numbers $\lambda_s$ are not involved in our definition of probe, and this is the reason we get $\mathfrak{d}_R ^2$ in the right-hand side, instead of, say, $\mathfrak{d}_R \mathfrak{d}_{R\sqcup \Box_J} $.\par 
We arrive at an expression for $\widetilde{G}_{\mathcal{O}_L} (\omega)$ which is precisely \eqref{eq:defoprhsrho}. We also observe that 
\begin{itemize}
    \item \eqref{eq:FourierGOdeltaomega} vanishes if $\omega <0$, because $\mu>0$ and $E^{\text{int}} _J >0$ $\forall J$, and so the $\delta$ is never satisfied.
    \item For $\omega >0$, expression \eqref{eq:FourierGOdeltaomega} depends on $\omega_{\mathrm{r}}$.
\end{itemize}\par
\underline{Step (3).} We repeat the procedure to evaluate $\widetilde{G}_{\mathcal{O}_L ^{\dagger} } (\omega) $.\par
Now the operator $\mathcal{O}_L ^{\dagger}$ acts first, thus annihilating the probe vacuum. The first non-trivial term comes from $\mathcal{O}_L ^{\dagger}$ acting on the first probe excited state. Namely, the lowest order in $e^{- \beta \mu}$ contribution comes from the terms schematically of the form 
\begin{equation}
    \langle \Box; \emptyset \rvert \otimes {}_L\langle R ; \phi (R) \vert e^{- \beta H^{\prime}} e^{i t H^{\prime}}  \mathcal{O}_L e^{-it H^{\prime}} \mathcal{O}_L ^{\dagger} \lvert  R ; \phi (R) \rangle_L \otimes \lvert \Box; \emptyset \rangle .
\end{equation}
A computation analogous to Step (2) gives 
\begin{equation}
    \sum_{R \in \mathfrak{R}_L ^{(N)}} \frac{1}{\tilde{L}} \sum_{J\in \mathscr{J}_R}  e^{- \beta H^{\prime}(R \sqcup \Box_J, \phi(R)) + it \left[ H^{\prime}(R \sqcup \Box_J, \phi(R)) - H^{\prime}(R, \phi(R) ) \right] } ~ \mathfrak{d}_R ^2 ~ \dim (R \sqcup \Box_J) \dim (\phi (R)).
\label{eq:secondtermAppWF}
\end{equation}
The difference with respect to Step (2) lies in the reflection $t \mapsto - t$ and in the fact that $\exp \left( - \beta H^{\prime} \right)$ acting on its left produces $ \exp \left( - \beta H^{\prime}(R \sqcup \Box_J, \phi(R)) \right)$. Compared to the function in Step (2), the action on a probe state different from the vacuum introduces the additional factor 
\begin{equation}
\label{eq:expmupieceGOdag}
    \exp \left( - \beta \mu - \beta E^{\text{int}} _J \right) 
\end{equation}
in each summand. Higher powers of $e^{- \beta \mu}$ are neglected by the definition of probe. Taking the Fourier transform, we are therefore led to 
\begin{equation}
\label{eq:Gtildeomeganeg}
    \sum_{R \in \mathfrak{R}_L ^{(N)}} \frac{1}{\tilde{L}} \sum_{J \in \mathscr{J}_R} e^{- \beta \left( \mu+ E^{\text{int}} _J \right) } ~  e^{- \beta \lvert R \rvert } ~\delta \left( - \omega - \mu - E^{\text{int}} _J \right)  ~\mathfrak{d}_R ^2 \dim (R \sqcup \Box_J) \dim (\phi (R)).
\end{equation}
This is analogous to \eqref{eq:FourierGOdeltaomega}, except for the additional weight \eqref{eq:expmupieceGOdag} in front of each summand and for the reflection $\omega \mapsto - \omega$. For the same reason as above, i.e. that the energy spectrum is positive, the latter expression vanishes if $\omega >0$. Once again, it can be written as a function of $\omega_{\mathrm{r}}$.\par
\underline{Step (4).} We now wish to put the terms together. Let us take a closer look at \eqref{appeq:omegatwoparts}. 
\begin{itemize}
    \item[(4.i)] From \eqref{eq:FourierGOdeltaomega} we know that $\widetilde{G}_{\mathcal{O}_L } (\omega) =0$ unless $\omega - \mu >0 $. By the probe approximation, we have 
\begin{equation}
    (1-e^{- \beta \omega}) \widetilde{G}_{\mathcal{O}_L } (\omega) = \theta (\omega - \mu) \widetilde{G}_{\mathcal{O}_L } (\omega) + O(e^{- \beta \mu}).
\end{equation}
    \item[(4.ii)] Likewise, $\widetilde{G}_{\mathcal{O}_L^{\dagger} } (\omega) =0$ unless $\omega + \mu <0 $, which in the probe approximation gives 
\begin{equation}
    (1-e^{- \beta \omega}) \widetilde{G}_{\mathcal{O}_L^{\dagger} } (\omega) = - \theta (- \omega - \mu) e^{- \beta\omega }\widetilde{G}_{\mathcal{O}_L^{\dagger} } (\omega) + O(e^{- \beta \mu}) ,
\end{equation}
where we have used $e^{- \beta \omega} \gg 1 $ if $\omega <-\mu \ll 0$. By the computation in Step (3), cf. \eqref{eq:Gtildeomeganeg}, we have that the latter expression is non-vanishing precisely when $e^{-\beta\omega }$ cancels the piece \eqref{eq:expmupieceGOdag} in $\widetilde{G}_{\mathcal{O}_L^{\dagger} } (\omega)$. Simplifying, we get 
\begin{equation}
    - e^{-\beta\omega }\widetilde{G}_{\mathcal{O}_L^{\dagger} } (\omega) = - \widetilde{G}_{\mathcal{O}_L } (- \omega) , \qquad \text{ if } \omega \le -\mu .
\end{equation}
    \item[(4.iii)] Writing $\widetilde{G}_{\mathcal{O}_L } (\pm \omega) $ in the form \eqref{eq:defoprhsrho} and summing the two contributions, we prove the statement.
\end{itemize}
\end{proof}

\begin{proof}[Proof of Theorem \ref{thm:rhoandOmega}]
Most of the details are identical to the proof of Theorem \ref{thmrhoisrho}, so we will be sketchy. We divide the proof in various steps:
\begin{enumerate}[(1)]
    \item We decompose the problem into the evaluation of four simpler pieces;
    \item We evaluate each piece;
    \item We take the Fourier transform;
    \item We sum the pieces, being careful with the probe approximation and the factors of $\pm i \varepsilon $.
\end{enumerate}\par
\underline{Step (1).} As in the proof of Theorem \ref{thmrhoisrho}, we start by decomposing $\phi_L(t)$ into the operators $\mathcal{O}_L, \mathcal{O}_L^{\dagger}$. We have 
\begin{align}
    \left[ \phi (t) , \phi (0) \right]  =  \frac{1}{2} \left( \mathcal{Y}_1 (t) + \mathcal{Y}_2 (t) - \mathcal{Y}^\dagger_1 (t) - \mathcal{Y}^\dagger_2 (t) \right) \ + \cdots ,
\end{align}
where the dots include all the terms that vanish when taken inside the correlation functions, and the four pieces are:
\begin{align}
     \mathcal{Y}_1 (t) & := \frac{1}{\tilde{L}} \sum_{p=1}^{\tilde{L}}  e^{i t H^{\prime}} ~ c_p ~ e^{-i t H^{\prime}} ~ c_p^{\dagger} ; \\
     \mathcal{Y}_2 (t) & := \frac{1}{\tilde{L}} \sum_{p=1}^{\tilde{L}}  e^{i t H^{\prime}} ~ c_p^{\dagger} ~ e^{-i t H^{\prime}} ~ c_p ; \\
     \mathcal{Y}^\dagger_1 (t) & := \frac{1}{\tilde{L}} \sum_{p=1}^{\tilde{L}} c_p ~ e^{i t H^{\prime}} ~ c_p^{\dagger} ~ e^{-i t H^{\prime}} ; \\
     \mathcal{Y}^\dagger_2 (t) & := \frac{1}{\tilde{L}} \sum_{p=1}^{\tilde{L}}  c_p^{\dagger} ~ e^{i t H^{\prime}} ~ c_p ~ e^{-i t H^{\prime}} .
\end{align}
(recall that $\tilde{L} = L+1 = \dim \Box$).\par
\underline{Step (2).} We compute the thermal expectation value of the four terms above. Each of them is essentially identical to Steps (2)-(3) in the proof of Theorem \ref{thmrhoisrho}. For $\mathcal{Y}_1$ the answer was given in \eqref{eq:probarhosumA}, and for $\mathcal{Y}_2$ in \eqref{eq:secondtermAppWF}:
\begin{align}
    \langle \Psi_{\beta} \lvert \mathcal{Y}_1 (t) \rvert \Psi_{\beta} \rangle & = \frac{1}{\mz_L ^{(N)}} \sum_{R \in \mathfrak{R}_L ^{(N)}} e^{- \beta \lvert R \rvert }  \mathfrak{d}_R ^2 \dim (R )  \dim (\phi(R)) \left[\sum_{J\in \mathscr{J}_R}  e^{ - it E_J (\mu)  } \frac{\dim (R \sqcup \Box_J)}{ \tilde{L} \dim (R) } \right]  \\
    \langle \Psi_{\beta} \lvert \mathcal{Y}_2 (t) \rvert \Psi_{\beta} \rangle & = \frac{1}{\mz_L ^{(N)}} \sum_{R \in \mathfrak{R}_L ^{(N)}} e^{- \beta \lvert R \rvert }  \mathfrak{d}_R ^2 \dim (R )  \dim (\phi(R)) \left[ \sum_{J\in \mathscr{J}_R}  e^{ (- \beta + it)  E_J (\mu) } \frac{\dim (R \sqcup \Box_J)}{ \tilde{L} \dim (R) } \right]  . 
\end{align}
where $E_J (\mu):= E^{\text{\rm int}}_J + \mu $ is a shorthand notation. The analogous computation for the other two terms gives:
\begin{align}
    \langle \Psi_{\beta} \lvert \mathcal{Y}^\dagger_1 (t) \rvert \Psi_{\beta} \rangle & = \frac{1}{\mz_L ^{(N)}} \sum_{R \in \mathfrak{R}_L ^{(N)}} e^{- \beta \lvert R \rvert }  \mathfrak{d}_R ^2 \dim (R )  \dim (\phi(R)) \left[\sum_{J\in \mathscr{J}_R}  e^{ it E_J (\mu)  } \frac{\dim (R \sqcup \Box_J)}{ \tilde{L} \dim (R) } \right]  \\
    \langle \Psi_{\beta} \lvert \mathcal{Y}^\dagger_2 (t) \rvert \Psi_{\beta} \rangle & = \frac{1}{\mz_L ^{(N)}} \sum_{R \in \mathfrak{R}_L ^{(N)}} e^{- \beta \lvert R \rvert }  \mathfrak{d}_R ^2 \dim (R )  \dim (\phi(R)) \left[ \sum_{J\in \mathscr{J}_R}  e^{ (- \beta - it) E_J (\mu) } \frac{\dim (R \sqcup \Box_J)}{ \tilde{L} \dim (R) } \right]  . 
\end{align}\par

\underline{Step (3).} We now take the Fourier transform of each term. From the definitions \eqref{eq:defGRfunc}-\eqref{eq:defGAfunc} we have
    \begin{align}
        2 G_{L, \mathrm{R}} (t) & = i \theta (t) \langle \Psi_{\beta} \lvert \mathcal{Y}_1 (t) + \mathcal{Y}_2 (t) + \mathcal{Y}^\dagger_1 (t) + \mathcal{Y}^\dagger_2 (t) \rvert \Psi_{\beta} \rangle , \\
        2 G_{L, \mathrm{A}} (t) & = - i \theta (- t) \langle \Psi_{\beta} \lvert \mathcal{Y}_1 (t) + \mathcal{Y}_2 (t) + \mathcal{Y}^\dagger_1 (t) + \mathcal{Y}^\dagger_2 (t) \rvert \Psi_{\beta} \rangle ,
    \end{align}
    so that we have to compute the Fourier transform of $\pm i \theta (\pm t) \mathcal{Y} (t)$ for $\mathcal{Y}\in \left\{ \mathcal{Y}_1, \mathcal{Y}_2 , \mathcal{Y}^\dagger_1 , \mathcal{Y}^\dagger_2 \right\}$, evaluated at $\omega \pm i \varepsilon$, with the sign $\pm$ in front of $i \varepsilon$ being the same as in $\theta (\pm t)$. The only time-dependent part in $\mathcal{Y} (t)$ is the exponential $e^{\mp i t E_J (\mu)} $. We thus get the contributions of $\mathcal{Y}_1, \mathcal{Y}^\dagger_1$ to $2 \widetilde{G}_{L, \mathrm{R}} (\omega + i \varepsilon)$:
    \begin{align}
        i \theta (t) e^{ - i t E_J (\mu)} & \ \mapsto \ - \frac{1}{\omega + i \varepsilon - E_J (\mu) } \\
        - i \theta (t) e^{ i t E_J (\mu)} & \ \mapsto \ \frac{1}{ \omega + i \varepsilon + E_J (\mu) } 
    \end{align}
    and likewise for the contributions of $\mathcal{Y}_1, \mathcal{Y}^\dagger_1$ to $2 \widetilde{G}_{L, \mathrm{A}} (\omega - i \varepsilon)$:
    \begin{align}
        - i \theta (-t) e^{ - i t E_J (\mu)} & \ \mapsto \ - \frac{1}{\omega - i \varepsilon - E_J (\mu) } \\
        + i \theta (-t) e^{  i t E_J (\mu)} & \ \mapsto \ \frac{1}{ \omega - i \varepsilon + E_J (\mu) } . 
    \end{align}
    The same expressions coming from $\mathcal{Y}_1$ (respectively $\mathcal{Y}^\dagger_1$) appear in the Fourier transform of $\mathcal{Y}^\dagger_2$ (respectively $\mathcal{Y}_2$).\par

\underline{Step (4).} Denoting $\Omega (\omega)$ the operator that is diagonal in the representation basis with eigenvalues 
\begin{equation}
    \sum_{J\in \mathscr{J}_R}  \frac{\dim (R \sqcup \Box_J)}{ \tilde{L} \dim (R) } \frac{1}{\omega - E_J (\mu) } 
\end{equation}
as defined in \eqref{eq:defOmega}, we see that $2 \widetilde{G}_{L, \mathrm{R}} (\omega + i \varepsilon)$ receives contributions from 
\begin{equation}
     - \Omega (\omega + i \varepsilon ) - \Omega (-\omega - i \varepsilon ) ,
\end{equation}
while $2 \widetilde{G}_{L, \mathrm{A}} (\omega - i \varepsilon)$ receives contributions from 
\begin{equation}
     - \Omega (\omega - i \varepsilon ) - \Omega (-\omega + i \varepsilon ) .
\end{equation}
Putting all the pieces together we arrive at:
    \begin{align}
         \widetilde{G}_{L, \mathrm{R}} (\omega+i \varepsilon)  & = - \frac{1}{2}\langle \Psi_{\beta} \lvert \Omega (\omega + i \varepsilon ) + \Omega (-\omega - i \varepsilon ) \rvert  \Psi_{\beta} \rangle + O(e^{- \beta \mu} ) \\
         \widetilde{G}_{L, \mathrm{A}} (\omega-i \varepsilon)  & = -  \frac{1}{2}\langle \Psi_{\beta} \lvert \Omega (\omega - i \varepsilon ) + \Omega (-\omega + i \varepsilon ) \rvert  \Psi_{\beta} \rangle + O(e^{- \beta \mu} ) .
    \end{align}
\end{proof}

\subsection{Wightman functions in the Veneziano limit}
\label{app:proofLemmaOmega}
This appendix contains the proof of Lemma \ref{lemma:OmegalargeN}, which evaluates the Wightman functions in the Veneziano limit using a saddle point approximation. We still use the shorthand $\tilde{L}=L+1$ from \eqref{eq:Ltilde}.

\begin{proof}[Proof of Lemma \ref{lemma:OmegalargeN}]
Consider $\widetilde{G}_{L,\mathrm{R}} (\omega)$ as given by Theorem \ref{thm:rhoandOmega}. The change of variables \eqref{eq:changeRtoH} recasts the sum over $R$ into a sum over $\tilde{L}$-tuples $\vec{h}=(h_1, \dots, h_{\tilde{L}-1}, -1)$. The ordering $R_{j} \ge R_{j+1}$ becomes $h_j > h_{j+1}$, but the appearance of the Vandermonde squared factor and the total symmetry of the Hamiltonian allow us to remove this restriction. The set of indices $\mathscr{J}_R$ becomes 
\begin{equation}
    \mathscr{J}_{\vec{h}} = \left\{ J \in \left\{ 1, \dots , \tilde{L} \right\} \ : \ h_J +1 \ne h_j , \ \forall j \ne J \right\} .
\end{equation}
Notice that the presence of a sum over a restricted set of indices does not spoil the total symmetry of the summand in the variables $h_j$, because, for those indices $j \notin \mathscr{J}_{\vec{h}}$, the Vandermonde determinant after adding a box to $R$ vanishes. One can thus safely include them (add finitely many zeros) and sum over all indices.\par
With the change of variables \eqref{eq:changeRtoH}, 
\begin{equation}
    E_J^{\text{\rm int}} = g (h_J +1) 
\end{equation}
and, subdividing 
\begin{equation}
    \prod_{1 \le i < j \le \tilde{L}} = \prod_{\underset{i \ne J, \ j \ne J}{1 \le i < j \le \tilde{L}}} \cdot \prod_{\underset{j=J}{1\le i \le J-1}} \cdot \prod_{\underset{i=J}{J+1\le j \le \tilde{L}}}
\end{equation}
in formula \eqref{eq:dimRVdm}, we express the ratio 
\begin{equation}
\label{eq:ratiodimsexact}
\frac{\dim (R \sqcup \Box_J) }{  \dim R} = \prod_{i=1}^{J-1} \left(  1 - \frac{1}{h_i-h_J} \right) \cdot \prod_{j=J+1}^{\tilde{L}} \left(  1 - \frac{1}{h_j-h_J} \right) = \prod_{j \ne J}  \left(  1 - \frac{1}{h_j-h_J} \right) .
\end{equation}
The right-hand side of \eqref{eq:defOmega} evaluated at $\omega \pm i \varepsilon$ reads
\begin{equation}
\label{eq:rhsOmegainH}
     \frac{1}{\tilde{L}} \sum_{J \in \mathscr{J}_R } \frac{1}{\omega  \pm i \varepsilon - E_J^{\text{\rm int}} - \mu } ~ \frac{\dim (R \sqcup \Box_J) }{  \dim R} = \frac{1}{\tilde{L}} \sum_{J=1}^{\tilde{L}} \frac{1}{\omega  \pm i \varepsilon - g (h_J +1) - \mu } \prod_{j \ne J}  \left(  1 - \frac{1}{h_j-h_J} \right) .
\end{equation}
Let us repeat that, although the sum should be restricted to $J\in \mathscr{J}_{\vec{h}}$, the right-hand side of \eqref{eq:rhsOmegainH} vanishes for $j$ such that $h_j = h_J+1$, so that the extra contributions that we include in the sum are trivial.\par
Let us introduce some notation. First, we redefine 
\begin{equation}
    \tilde{\omega}_{\pm} = \omega \pm i \varepsilon - \mu - g .
\end{equation}
We are eventually interested in the planar limit $N \to \infty$ with $g= \lambda/N$ and $\lambda$ fixed, so that the shift by $g$ drops out, whereas $\Im \tilde{\omega}_{\pm} = \pm \varepsilon$. We denote the expression in the right-hand side of \eqref{eq:rhsOmegainH} $\Omega_{\vec{h}} (\omega \pm \varepsilon)$ for short.\par
For every $\tilde{L}$-tuple $\vec{h} \in \mathbb{N}^{\tilde{L}}$ (the ones of interest to us have $h_{\tilde{L}} \equiv -1$), we define the auxiliary function 
\begin{equation}
    \Phi_{\vec{h}} (\xi) := \prod_{j=1}^{\tilde{L}} \left( 1 - \frac{1}{h_j-\xi } \right) , \qquad \qquad \xi \in \mathbb{C} .
\end{equation}
The residue theorem instructs us that 
\begin{equation}
\label{eq:OmegaCauchy}
     \Omega_{\vec{h}} (\omega\pm i\varepsilon) = \frac{1}{\tilde{L}} \oint_{\mathcal{C}} \frac{ \dd \xi}{2 \pi i }  \frac{\Phi_{\vec{h}} (\xi) }{\tilde{\omega}_{\pm} - g \xi } ,
\end{equation}
with the integration contour $\mathcal{C} = \mathcal{C} (\vec{h})$ encircling the points $h_{J}$ and leaving outside the point $\tilde{\omega}_{\pm}/g$. This is exemplified in Figure \ref{fig:contourC}. Let us remark that, had we restricted the sum on $J \in \mathscr{J}_{\vec{h}}$, we could take a contour that leaves outside the points $h_J$ for $J \notin \mathscr{J}_{\vec{h}}$. Then, deforming such contour into $\mathcal{C}$, we pick the poles with vanishing residues, obtaining \eqref{eq:OmegaCauchy}.\par
\begin{figure}[th]
    \centering
    \includegraphics[width=0.67\textwidth]{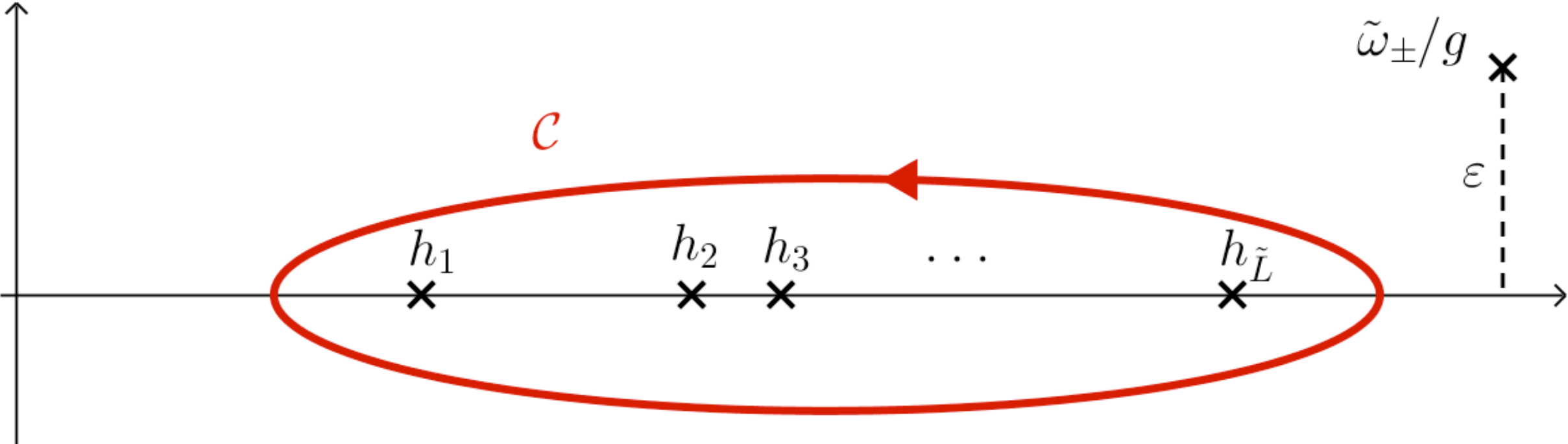}
    \caption{Integration contour $\mathcal{C}$.}
    \label{fig:contourC}
\end{figure}\par
Deforming the contour $\mathcal{C}$ to infinity, the Cauchy integral \eqref{eq:OmegaCauchy} picks up the pole at $\xi= \tilde{\omega}_{\pm}/g$, and one is left with a residual integration over the circle at infinity. Using that $\Phi_{\vec{h}} (\xi) \to 1 $ as $\xi \to \infty$, the residual integration evaluates to 
\begin{equation}
    \frac{1}{\tilde{L}}\oint_{\lvert \xi \rvert \gg 1 } \frac{\dd \xi}{2 \pi i} \frac{\Phi_{\vec{h}} (\xi)}{(-g \xi)} = - \frac{1}{\tilde{L} g} .
\end{equation}
Conveniently expressing \eqref{eq:OmegaCauchy} as 
\begin{equation}
    \tilde{L} g ~ \Omega_{\vec{h}} (\omega\pm i \varepsilon ) = \oint_{\mathcal{C}} \frac{ \dd \xi}{2 \pi i }  \frac{\Phi_{\vec{h}} (\xi) }{\frac{\tilde{\omega}_{\pm}}{g} - \xi } ,
\end{equation}
we arrive at 
\begin{equation}
\label{eq:OmomPhifinite}
     \tilde{L} g ~ \Omega_{\vec{h}} (\omega) = \Phi_{\vec{h}} (\tilde{\omega}_{\pm}/g) -1 .
\end{equation}\par
We write\footnote{In this appendix and Section \ref{sec:QM} we define the 't Hooft coupling $\lambda = g L$. The procedure for 't Hooft coupling $g N$ is identical up to rescaling $\tilde{\omega}_{\pm}$ appropriately.} $\frac{\tilde{\omega}_{\pm}}{g} = \frac{\tilde{L}\tilde{\omega}_{\pm}}{ \lambda} $ and 
\begin{align}
    \Phi_{\vec{h}} (\tilde{\omega}_{\pm}/g) & = \exp \left[ \sum_{j=1}^{\tilde{L}} \ln \left( 1- \frac{1}{h_j - \frac{\tilde{L}\tilde{\omega}_{\pm}}{\lambda}} \right) \right] \notag \\
    & = \exp \left[ \frac{1}{\tilde{L}}\sum_{j=1}^{\tilde{L}}  \tilde{L} \ln \left( 1+ \frac{1}{\tilde{L}} \cdot \frac{1}{\frac{\tilde{\omega}_{\pm}}{ \lambda} - \frac{h_j}{\tilde{L}}} \right) \right] \notag \\
    & = \exp \left[ \frac{1}{\tilde{L}}\sum_{j=1}^{\tilde{L}} \frac{1}{\frac{\tilde{\omega}_{\pm}}{ \lambda} - \frac{h_j}{\tilde{L}}} + O (1/\tilde{L})\right] \notag \\
    \approx \exp \left[ \int \dd x \varrho (x) \frac{1}{\frac{\tilde{\omega}}{ \lambda} - x } \right]
\end{align}
with the last step holding in the planar limit. In the second-to-last step, we have expanded the logarithm in power series in $\tilde{L}^{-1}$, and in the last step we have discarded the sub-leading terms. At the same time, we have plugged the eigenvalue density $\varrho (x)$ for the matrix model \eqref{eq:genericdiscreteMM}, defined in the standard way 
\begin{equation}
        \varrho (x) := \frac{1}{\tilde{L}} \sum_{i=1}^{\tilde{L}} \delta \left( x - \frac{h_i}{\tilde{L}} \right) , \qquad x \in \mathbb{R} .
    \end{equation}
No assumption on $\varrho (x)$ is made at this stage, and its definition uses the fact that $h_j/\tilde{L}$ attains a finite value in the planar limit, $\forall j$. The linear growth of $h_j$ with $L$ is a standard result on discrete matrix models, or equivalently on the combinatorics of the distribution of partitions of length $L$; it is a known result that transcends the models of interest to us. The fact that this is the leading contribution in the Veneziano limit, and does not trivialize, is encoded in the assumption of the existence of a non-trivial saddle point. Therefore, in the planar limit we find that \eqref{eq:OmomPhifinite} becomes 
\begin{equation}
     \lambda \Omega_{\vec{h}} (\omega \pm i \varepsilon) \approx -1+ \exp \left[ \int \dd x \varrho (x) \frac{1}{\frac{\tilde{\omega}_{\pm}}{ \lambda} - x } \right] .
\end{equation}\par
At leading order, we perform a saddle point approximation and evaluate this expression at the saddle point eigenvalue density $\varrho_{\ast} (x)$. We thus arrive at 
\begin{equation}
     \left. \Omega_{\vec{h}} (\omega \pm i \varepsilon ) \right\rvert_{\text{saddle point}} =\frac{1}{\lambda } \left[ - 1 +  \exp \left( \int \dd x \frac{ \varrho_{\ast} (x)}{\frac{\omega \pm i \varepsilon - \mu}{ \lambda} - x } \right) \right] ,
\end{equation}
where we have reinserted the appropriate variable $\tilde{\omega}_{\pm} = \omega \pm i \varepsilon - \mu + O(1/N)$ in the planar limit. Notice that this expression only holds under the assumption that the saddle point is non-trivial. This hypothesis is crucial, but no further assumptions are made. By the standard large $N$ argument, under the mentioned hypothesis, the eigenvalues coalesce and form a continuum, and we get
\begin{equation}
      {}_L\langle \Psi_{\beta} \lvert  \Omega (\omega)  \rvert \Psi_{\beta} \rangle_L \approx  \left. \Omega_{\vec{h}} (\omega) \right\rvert_{\text{saddle point}} ,
\end{equation}
which concludes the proof.
\end{proof}

\subsection{Large \texorpdfstring{$N$}{N} factorization lemma}
\label{app:LargeNproof}
This appendix contains a justification for the factorization Lemma \ref{lem:factor}. We will very explicitly show that the four-point function factorizes, and then give a general argument that extends more easily to the general case.\footnote{An elegant proof of large $N$ factorization for permutation invariant observables in Hermitian matrix models with a symmetric group symmetry was given in \cite{Barnes:2021tjp}, based on combinatorial techniques. While our method is based on a direct evaluation of the terms, it seems plausible that the techniques of \cite{Barnes:2021tjp,Barnes:2022qli} can be adapted and applied to the present context to establish a proof that only relies on the symmetries of the operators, and not their explicit form.} The two key hypotheses are:
\begin{itemize}
    \item[(H1)] The partition function admits a non-trivial saddle point for large $N$;
    \item[(H2)] $H_{\rm int}$ respects the flavor symmetry.
\end{itemize}
We will work explicitly with $H_{\rm int}$ given in \eqref{eq:Hint}, but a direct generalization of the argument holds as long as the second assumption holds.

\subsubsection{Step 0: Lighten the notation}
The approach in this appendix is based on explicit direct calculation, which will produce cumbersome equations. To lighten the expressions, we introduce shorthand notations used throughout this subsection.\par
Recall from Subsection \ref{sec:QMflavor} that the representation basis is given by the vectors 
\begin{equation}
    \lvert R , \ai ,\si ;\ \phi (R) , \dot{\ai},\dot{\si} \rangle , \quad \ai =1, \dots, \dim R , ~  \dot{\ai}  =1, \dots, \dim \phi (R) 
\end{equation}
and $\si, \dot{\si}$ running over possible additional degeneracies. The indices $\si, \dot{\si}$ do not play any role in the ensuing discussion, thus we will mute them. Likewise, by construction the probe interacts with the $R$-sector, and $\phi (R)$ only appears through the interaction Hamiltonian. Since the state $\lvert \phi (R), \dot{\ai} \rangle $ is not affected by the action of the probe operators, we will omit it from the notation as well. Every sum over $\ai$ will implicitly be accompanied by a sum over $\dot{\ai}$.\par
Furthermore, we will further contract the notation for tensor product with a probe state and write:
\begin{equation}
\begin{aligned}
    \lvert R, \ai \rrrang & \text{ for } \lvert R , \ai ,\si ;\ \phi (R) , \dot{\ai},\dot{\si} \rangle \otimes \lvert 0, \dots, 0 \rangle_{\text{\rm probe}} \\
    \lvert R, \ai; j \rrrang & \text{ for } \lvert R , \ai ,\si ;\ \phi (R) , \dot{\ai},\dot{\si} \rangle \otimes \lvert 0, \dots, \underbrace{1}_{j^{\text{th}}}, \dots,  0 \rangle_{\text{\rm probe}} 
\end{aligned}
\end{equation}
and so on.\par
We recall that $H^{\prime}$ is the total Hamiltonian, $H^{\prime} = H + H_{\rm probe} + H_{\rm int}$. As in the main text, we will denote by $R \sqcup \Box_J$ the Young diagram obtained by appending a box at the end of the $J^{\text{th}}$ row, and 
\begin{equation}
    E_J (\mu) := H^{\prime} (R \sqcup \Box _J, \phi (R)) - H^{\prime} (R , \phi (R)) .    
\end{equation}\par
We use indices $\ai,\mathsf{b}$ to run over the generators of $R$, hatted indices $\hat{\ai},\hat{\mathsf{b}}$ to run over the generators of $R \sqcup \Box$ ($R$ with one box added), and checked indices $\check{\mathsf{c}}$ to run over the generators of $R \sqcup \Box\sqcup \Box$ ($R$ with two boxes added).\par
We will use the symbol $\approx$ to indicate that the two sides are equal up to terms that vanish in the planar limit, and also sometimes abbreviate $t_{jk} := t_j - t_k$. Finally, we continue using the abbreviation $\tilde{L}=L+1$ from \eqref{eq:Ltilde}.

\subsubsection{Step 1: Identify the non-trivial contributions}
Recall that the goal of this appendix is to compute four-point functions 
\begin{equation}
\label{eq:phi4ptapp}
    {}_L\langle \Psi_{\beta} \lvert \phi_L (t_3) \phi_L (t_2) \phi_L (t_1) \phi_L (t_0) \rvert \Psi_{\beta}\rangle_L 
\end{equation}
in the planar limit, and show that \eqref{eq:phi4ptapp} equals 
\begin{equation}
\begin{aligned}
     & {}_L\langle \Psi_{\beta} \lvert \phi_L (t_3) \phi_L (t_2) \rvert \Psi_{\beta}\rangle_L  \cdot {}_L\langle \Psi_{\beta} \lvert \phi_L (t_1) \phi_L (t_0) \rvert \Psi_{\beta}\rangle_L  \\
     + & {}_L\langle \Psi_{\beta} \lvert \phi_L (t_3) \phi_L (t_1) \rvert \Psi_{\beta}\rangle_L  \cdot {}_L\langle \Psi_{\beta} \lvert \phi_L (t_2) \phi_L (t_0) \rvert \Psi_{\beta}\rangle_L \\
     + & {}_L\langle \Psi_{\beta} \lvert \phi_L (t_3) \phi_L (t_0) \rvert \Psi_{\beta}\rangle_L  \cdot {}_L\langle \Psi_{\beta} \lvert \phi_L (t_2) \phi_L (t_1) \rvert \Psi_{\beta}\rangle_L .
\end{aligned}
\label{eq:phi4ptfactin2}
\end{equation}
We expand \eqref{eq:phi4ptapp} using the definition of $\phi_L (t)$ in terms of the operators $\mathcal{O}_L (t), \mathcal{O}_L^{\dagger} (t)$. After elementary calculations, we 
\begin{itemize}
    \item[(i)] discard all terms except those with two $\mathcal{O}_L$ and two $\mathcal{O}_L^{\dagger}$; and then 
    \item[(ii)] discard all the terms in which the rightmost operator in an annihilation operator or the leftmost operator is a creation operator.
\end{itemize}
The contributions discarded in (i) vanish exactly, while those discarded in (ii) are $O(e^{- \beta \mu})$, thus negligible in the probe approximation. We are therefore left with 
\begin{align}
{}_L\langle \Psi_{\beta} \lvert \phi_L (t_3) \phi_L (t_2) \phi_L (t_1) \phi_L (t_0) \rvert \Psi_{\beta}\rangle_L & = \frac{1}{4} {}_L\langle \Psi_{\beta} \lvert \mathcal{O}_L^{\dagger}(t_3) \mathcal{O}_L(t_2) \mathcal{O}_L^{\dagger}(t_1) \mathcal{O}_L (t_0) \rvert \Psi_{\beta}\rangle_L \label{eq:phi4ptintoO4pt}  \\ & + \frac{1}{4} {}_L\langle \Psi_{\beta} \lvert \mathcal{O}_L^{\dagger}(t_3) \mathcal{O}_L^{\dagger}(t_2) \mathcal{O}_L(t_1) \mathcal{O}_L(t_0) \rvert \Psi_{\beta}\rangle_L .  \notag 
\end{align}
We claim
\begin{prop}\label{lemmaAppFactorization}
    The first term on the right-hand side of \eqref{eq:phi4ptintoO4pt} yields the first line in \eqref{eq:phi4ptfactin2}, and the second term on the right-hand side of \eqref{eq:phi4ptintoO4pt} yields the second and third lines in \eqref{eq:phi4ptfactin2}.
\end{prop}
The rest of this appendix is devoted to show this statement.\par
We plug the definition of $\mathcal{O}_L$ as a sum of probe creation operators in the right-hand side of \eqref{eq:phi4ptintoO4pt}. We also set $t_0=0$ without loss of generality. We are then interested in four-point functions of the form 
\begin{align}
G^{(1)}_{ijkl}(t_1,t_2,t_3):={}_L\langle \Psi_{\beta} \lvert c_i(t_3)c^\dagger_j(t_2)c_k(t_1)c^\dagger_l(0) \rvert \Psi_{\beta}\rangle_L, \\
G^{(2)}_{ijkl}(t_1,t_2,t_3):={}_L\langle \Psi_{\beta} \lvert c_i(t_3)c_j(t_2)c^\dagger_k(t_1)c^\dagger_l(0) \rvert \Psi_{\beta}\rangle_L,
\end{align}
and then sum over the flavor indices $i,j,k,l \in \left\{ 1, \dots, \tilde{L} \right\}$ with the overall normalization by $1/\tilde{L}^2$. The idea is to evaluate these correlation functions in a particular sector $(R,\phi(R))$ for a fixed irreducible representation $R$ appearing in the thermal ensemble, and then to apply a saddle point argument.

\subsubsection{Intermezzo: Useful formulas}
Throughout we will make extensive use of the Clebsch--Gordan isomorphism 
\begin{equation}
    \mathscr{H} (R \otimes \Box) \cong \bigoplus_{J \in \mathscr{J}_R} \mathscr{H} (R \sqcup \Box_J) .
\end{equation}
The normalization of the two bases are related as follows. Consider the identities
\begin{equation}
    \sum_{j=1}^{\tilde{L}} \sum_{\ai=1}^{\dim R} \lllang R, \ai ; j \vert R, \ai ; j \rrrang = \tilde{L} \dim R = \sum_{J \in \mathscr{J}} \sum_{\hat{\ai}=1}^{\dim (R \sqcup \Box_J)} \langle R \sqcup \Box_J, \hat{\ai} \vert R \sqcup \Box_J, \hat{\ai} \rangle
\end{equation}
and insert the resolution of the identity on the left-most part, expressed in the representation basis. We have the normalization:
\begin{align}
   \sum_{j=1}^{\tilde{L}} \sum_{\ai=1}^{\dim R} \lvert \lllang R \sqcup \Box_J, \hat{\ai} \vert  R, \ai ; j \rrrang \rvert^2 &= 1 = \sum_{J \in \mathscr{J}} \sum_{\hat{\ai}=1}^{\dim (R \sqcup \Box_J)}   \lvert\langle R \sqcup \Box_J, \hat{\ai} \vert  R, \ai ; j \rrrang \rvert^2. \label{eq:normcoefsrequired}
\end{align}\par
To write the result of the calculations in factorized form, we will use the properties:
\begin{align}
    H_{\rm int} (R \sqcup \Box_J \sqcup \Box_{K}) - H_{\rm int} (R\sqcup \Box_J) & \approx H_{\rm int} (R  \sqcup \Box_{K}) - H_{\rm int} (R ) \label{eq:HintshiftJ} \\
    \frac{ \dim   (R \sqcup \Box_J \sqcup \Box_{K}) }{ \dim   (R \sqcup \Box_J) } & \approx \frac{ \dim   (R \sqcup \Box_{K}) }{ \dim R } .  \label{eq:approxrationdimsJ}
\end{align}

\subsubsection{Step 2: First correlation function}

Let us start with $G^{(1)}_{ijkl}(t_1,t_2,t_3)$, and consider the contribution to the thermal expectation value for a given representation $R$. By the flavor symmetry of the interaction between the probe and the system (and ignoring potential antisymmetric terms that vanish when summed over) we have 
\begin{align}
\sum_{\ai }\lllang R, \ai \lvert c_i(t_3)c^\dagger_j(t_2)c_k(t_1)c^\dagger_l(0) \rvert R, \ai \rrrang & = \delta_{ij}\delta_{kl}\sum_{\ai }\lllang R, \ai \lvert c_i(t_3)c^\dagger_i(t_2)c_k(t_1)c^\dagger_k(0) \rvert R, \ai \rrrang \\
& +\delta_{il}\delta_{jk}\sum_{\ai }\lllang R, \ai \lvert c_i(t_3)c^\dagger_j(t_2)c_j(t_1)c^\dagger_i(0) \rvert R, \ai \rrrang  \notag \\
& -\delta_{ijkl}\sum_{\ai }\lllang R, \ai \lvert c_i(t_3)c^\dagger_i(t_2)c_i(t_1)c^\dagger_i(0) \rvert  R, \ai \rrrang.  \notag 
\end{align}
Averaging over $i,j,k,l$, in the large $L$ limit only the third term is suppressed and we get 
\begin{align}
\frac{1}{\tilde{L}^2}\sum_{ijkl}\sum_{\ai }\lllang R, \ai \lvert c_i(t_3)c_j^\dagger(t_2)c_k(t_1)c^\dagger_l(0) \rvert R, \ai \rrrang & \approx \frac{1}{\tilde{L}^2}\sum_{ik \ai}\lllang R, \ai \lvert c_i(t_3)c^\dagger_i(t_2)c_k(t_1)c^\dagger_k(0) \rvert R, \ai \rrrang \label{eq:appDG1split} \\
& +\frac{1}{\tilde{L}^2}\sum_{ij \ai}\lllang R, \ai \lvert c_i(t_3)c^\dagger_j(t_2)c_j(t_1)c^\dagger_i(0) \rvert R, \ai \rrrang . \notag
\end{align}
We now proceed to show the following:
\begin{lem}
    The first term in the right-hand side of \eqref{eq:appDG1split} factorizes into the product of two-point functions, and the second term is sub-leading.
\end{lem}
\begin{proof}
Inserting a resolution of the identity in the representation basis, the first term is rewritten as 
\begin{align}
&\frac{1}{\tilde{L}^2}\sum_{ik \ai } \lllang R, \ai \lvert c_i(t_3)c^\dagger_i(t_2)c_k(t_1)c^\dagger_k(0) \rvert R, \ai \rrrang 
= \frac{1}{\tilde{L}^2}\sum_{ik  \ai \mathsf{b} } \lllang R, \ai \lvert c_i(t_3)c^\dagger_i(t_2)\rvert R,\mathsf{b} \rrrang \lllang R,\mathsf{b} \lvert c_k(t_1)c^\dagger_k(0) \rvert R, \ai \rrrang \notag \\ 
& \qquad = \sum_{\ai \mathsf{b} } \left( \lllang R, \ai \lvert \frac{1}{\tilde{L}} \sum_i   c_i(t_3)c^\dagger_i(t_2)\rvert R,\mathsf{b} \rrrang \right) \left(\lllang R,\mathsf{b} \lvert  \frac{1}{\tilde{L}} \sum_k  c_k(t_1)c^\dagger_k(0) \rvert R, \ai \rrrang \right) \label{eq:AppDintermediateG1long}
\end{align}
The inner products in the last line eventually produce $\delta_{\ai \mathsf{b}}$, and the flavor symmetry further implies that $ \lllang R, \ai \lvert \frac{1}{\tilde{L}} \sum_i   c_i(t+ \delta t)c^\dagger_i(t)\rvert R,\ai \rrrang$ is actually independent of $\ai$. The last line of \eqref{eq:AppDintermediateG1long} is thus equal to
\begin{align}
& (\dim R) \left( \frac{1}{\tilde{L} \dim R} \tr_{R \otimes \Box} \left( e^{- i t_{32} H^{\prime}} \right) \right) \left( \frac{1}{\tilde{L} \dim R} \tr_{R \otimes \Box} \left( e^{- i t_{1} H^{\prime}} \right) \right) \notag \\
= & (\dim R) \left( \frac{1}{\tilde{L} \dim R} \sum_{J\in \mathscr{J}_R} \tr_{R \sqcup \Box_J} \left( e^{- i t_{32} H^{\prime}} \right) \right) \left( \frac{1}{\tilde{L} \dim R} \sum_{K\in \mathscr{J}_R} \tr_{R \sqcup \Box_K}  \left( e^{- i t_{1} H^{\prime}} \right) \right) \notag \\
= & (\dim R) \left(\sum_{J\in \mathscr{J}_R}  e^{ - i(t_3-t_2) E_J (\mu)  } \frac{\dim (R \sqcup \Box_J)}{ \tilde{L} \dim (R) }\right)\left(\sum_{K\in \mathscr{J}_R} e^{ - it_1 E_K (\mu)  } \frac{\dim (R \sqcup \Box_K)}{ L \dim (R) }\right) .  \label{eq:G1appcorrel}
\end{align}
Additionally, when summing over the spectator index $\dot{\ai}$ the expression will acquire a factor $\dim \phi (R)$. Altogether it manifestly factorizes into the desired product of two-point functions.\par

The second term in \eqref{eq:appDG1split} is suppressed by $1/L$. Indeed, inserting successive resolutions of the identity, we obtain:
\begin{align}
& \frac{1}{\tilde{L}^2}\sum_{ij \ai}\lllang R, \ai \lvert c_i(t_3)c^\dagger_j(t_2)c_j(t_1)c^\dagger_i(0) \rvert R, \ai \rrrang \notag \\
= & \frac{1}{\tilde{L}^2}\sum_{ij \ai}\sum_{J,K\in\mathscr{J}_R} \sum_{\hat{\ai} \hat{\mathsf{b}} } \exp \left\{  it_3H^{\prime}(R,\phi(R)) - i t_{32} H^{\prime}(R\sqcup\Box_K,\phi(R))-it_1H^{\prime}(R\sqcup\Box_J,\phi(R) \right\} \notag   \\
\times & \lllang R \sqcup \Box_J , \hat{\ai}\vert R, \ai ; i \rrrang \lllang R, \ai;i\vert R\sqcup\Box_K,\hat{\mathsf{b}} \rrrang \lllang R\sqcup\Box_K,\hat{\mathsf{b}} \lvert c_j^\dagger e^{-it_{21} H^{\prime}}c_j\rvert R\sqcup\Box_J,\hat{\ai} \rrrang. 
\end{align}
In this expression $\hat{\ai} =1, \dots , \dim (R \sqcup \Box_J)$ and $\hat{\mathsf{b}} =1, \dots , \dim (R \sqcup \Box_K)$. In the last line, the first two inner products together impose $\delta_{JK}\delta_{\hat{\ai} \hat{\mathsf{b}}}$. We can then use the normalization formula \eqref{eq:normcoefsrequired} and simplifying, we obtain the expression
\begin{align}
\frac{1}{\tilde{L}^2}\sum_{K} \sum_{j \hat{\mathsf{b}}} e^{i(t_3-t_2+t_1) \left[ H^{\prime}(R,\phi(R)) - H^{\prime}(R\sqcup\Box_K,\phi(R)) \right] } \lllang R\sqcup\Box_K,\hat{\mathsf{b}} \lvert c_j^\dagger c_j\rvert R\sqcup\Box_K,\hat{\mathsf{b}} \rrrang .  
\end{align}
The innermost sum evaluates to $\dim (R\sqcup \Box_K)$, and the overall factor is $1/\tilde{L}^2$, which means that the contribution is suppressed compared to a two-point function that has an overall factor $1/\tilde{L}$. This is consistent with the fact that the two-point function $\langle c^\dagger_j(t_2)c_j(t_1)\rangle$ is zero in the probe approximation.
\end{proof}
\begin{cor}\label{corol:firstpieceAppD}
   The first term on the right-hand side of \eqref{eq:phi4ptintoO4pt} equals in the planar limit the first line in \eqref{eq:phi4ptfactin2}.
\end{cor}
\begin{proof}
    By hypothesis, the correlator $G^{(1)}_{ijkl}(t_1,t_2,t_3)$ localizes on a large saddle in the large $L$ limit, so we can use the previous computation to conclude that it factorizes. After summing terms that are identically zero, we get that the first term on the right-hand side of \eqref{eq:phi4ptintoO4pt} reduces to $G^{(1)}_{ijkl}(t_1,t_2,t_3)$ which, by formula \eqref{eq:G1appcorrel} in the previous lemma, yields in the planar limit the first line in \eqref{eq:phi4ptfactin2}, as claimed.
\end{proof}

\subsubsection{Intermezzo: Bosonic statistic and symmetric representation}
\label{app:explainsymD3}
By definition, the two-particle probe state lives in the symmetric representation, 
\begin{equation}
    c_i^{\dagger} c_j^{\dagger} \lvert 0 \rangle_{\text{\rm probe}} = \frac{\lvert i,j \rangle_{\text{\rm probe}} + \lvert j,i \rangle_{\text{\rm probe}} }{\sqrt{2}} \in \begin{ytableau} \ & \ \end{ytableau}   .
\end{equation}
For simplicity of notation, here we assume $i \ne j$, and the neglected term is suppressed by $1/L$, thus safely discarded at later steps. The state $ \rvert R, \ai \rrrang$ is tensored with the latter probe state, not with $\lvert j,i \rangle_{\text{\rm probe}} \in \Box \otimes \Box$.\par
Hence, we will denote 
\begin{equation}
    R \sqcup \symbox_{(JK)} \subset (R \sqcup \Box_J ) \sqcup \Box_K \oplus (R \sqcup \Box_K ) \sqcup \Box_J 
\end{equation}
the symmetrization of the ways of appending two boxes, one at the end of the $J^{\text{th}}$ row and the other at the end of the $K^{\text{th}}$ row, to the Young diagram for $R$. This notation is to insist on the fact that the bosonic statistics of the probe forces the resulting state to live in the Hilbert space sector $\mathscr{H}(R \otimes \symbox)$, rather than generically in $\mathscr{H} (R \otimes \Box \otimes \Box )$. At the level of Hilbert spaces we have a generic state 
\begin{align}
        \lvert R\sqcup \symbox_{(JK)},\check{\mathsf{c}} \rrrang &= \frac{\xi}{\sqrt{2}} \left[\lvert  (R \sqcup \Box_J) \sqcup \Box_K ,\check{\mathsf{c}} \rrrang + \lvert\left(  (R \sqcup \Box_K ) \sqcup \Box_J \right),\check{\mathsf{c}} \rrrang\right] , \label{eq:defsymrepH} 
\end{align}
for all $\check{\mathsf{c}} =1, \dots, \dim (R \sqcup \Box_J \sqcup \Box_K)$, and $\xi$ is a phase, see around Equation \eqref{eq:isotropy} for a more detailed justification. Note that the dimension is manifestly insensitive to the order in which the boxes are appended to the rows, thus the definition is well posed. Likewise, the Hamiltonian is insensitive to the ordering, so we can simply write $H^{\prime} (R \sqcup \Box_J \sqcup \Box_K)$.\par
In the calculation of the second correlation function, we will repeatedly use the notation on the left-hand side of \eqref{eq:defsymrepH} to signify the result of tensoring with the symmetric probe state.

\subsubsection{Step 3: Second correlation function}

We now turn to the second expression $G^{(2)}_{ijkl}(t_1,t_2,t_3)$. Similarly, for any given representation $R$ (ignoring potential antisymmetric terms), 
\begin{align}
\sum_{\ai }\lllang R, \ai \lvert c_i(t_3)c_j(t_2)c^\dagger_k(t_1)c^\dagger_l(0) \rvert R, \ai \rrrang &=\delta_{ik}\delta_{jl}\sum_{\ai }\lllang R, \ai \lvert c_i(t_3)c_j(t_2)c^\dagger_i(t_1)c^\dagger_j(0) \rvert R, \ai \rrrang\\
& +\delta_{il}\delta_{jk}\sum_{\ai }\lllang R, \ai \lvert c_i(t_3)c_j(t_2)c^\dagger_j(t_1)c^\dagger_i(0) \rvert R, \ai \rrrang \notag \\
& -\delta_{ijkl}\sum_{\ai }\lllang R, \ai \lvert c_i(t_3)c_i(t_2)c^\dagger_i(t_1)c^\dagger_i(0) \rvert  R, \ai \rrrang.  \notag
\end{align}
Just like before, averaging over $i,j,k,l$ in the large $L$ limit only the first two terms are not suppressed and we obtain 
\begin{align}
\frac{1}{L^2}\sum_{ijkl}\sum_{\ai }\lllang R, \ai \lvert c_i(t_3)c_j(t_2)c^\dagger_k(t_1)c^\dagger_l(0) \rvert R, \ai \rrrang & \approx\frac{1}{\tilde{L}^2}\sum_{ij \ai}\lllang R, \ai \lvert c_i(t_3)c_j(t_2)c^\dagger_i(t_1)c^\dagger_j(0) \rvert R, \ai \rrrang  \label{eq:AppDG2split}\\
& +\frac{1}{\tilde{L}^2}\sum_{ij \ai}\lllang R, \ai \lvert c_i(t_3)c_j(t_2)c^\dagger_j(t_1)c^\dagger_i(0) \rvert R, \ai \rrrang.  \notag
\end{align}
Unlike the previous case, both terms give a nonzero contribution in the probe approximation.\par
\begin{lem}
    Both terms in the right-hand side of \eqref{eq:AppDG2split} factorize into the product two two-point functions.
\end{lem}
\begin{proof}[Proof of the first term in \eqref{eq:AppDG2split}]
Inserting successive resolutions of the identity, the first term on the right-hand side of \eqref{eq:AppDG2split} gives
\begin{align}
& \frac{1}{\tilde{L}^2}\sum_{ij \ai}\lllang R, \ai \lvert c_i(t_3)c_j(t_2)c^\dagger_i(t_1)c^\dagger_j(0) \rvert R, \ai \rrrang \notag \\
=& \sum_{J,K\in\mathscr{J}_R} \sum_{(MP)\in\mathscr{J}^{(2)}_R} e^{ it_3H^{\prime}(R,\phi(R))-it_{32} H^{\prime}(R\sqcup \Box_K,\phi(R))-it_1 H^{\prime}(R\sqcup\Box_J)- it_{21} H^{\prime}(R\sqcup\Box_M\sqcup\Box_P) } \notag \\
\times & \frac{1}{\tilde{L}^2}\sum_{ij \ai} \sum_{ \hat{\ai } \hat{\mathsf{b}} }\lllang R, \ai \lvert c_i \rvert R\sqcup\Box_K,\hat{\mathsf{b}} \rrrang \lllang R \sqcup\Box_J , \hat{\ai} \lvert c_j^\dagger \rvert R, \ai \rrrang \notag \\
& \qquad \times \sum_{ \check{\mathsf{c}}} \lllang R\sqcup\Box_K,\hat{\mathsf{b}}  \lvert c_j \rvert R\sqcup  \symbox_{(MP)},\check{\mathsf{c}} \rrrang \lllang R\sqcup \symbox_{(MP)}, \check{\mathsf{c}} \lvert c_i^\dagger \rvert  R\sqcup\Box_J,\hat{\ai} \rrrang ,  \label{eq:apDLongG2}
\end{align}
where the indices are $\hat{\ai } =1, \dots , \dim (R \sqcup \Box_J) $, $ \hat{\mathsf{b} } =1, \dots , \dim (R \sqcup \Box_K) $ and $\check{\mathsf{c}} =1, \dots, \dim (R \sqcup \Box_M \sqcup \Box_P)$. In this expression, $(MP)\in\mathscr{J}^{(2)}_R$ stands for the symmetrization of the ways of appending two boxes at the end of the $M^{\text{th}}$ and $P^{\text{th}}$ rows, as explained above in Subsection \ref{app:explainsymD3} --- cf. the prescription \eqref{eq:defsymrepH}.\par

Letting $c_j$ act on the right and $c_i ^{\dagger}$ act on the left, we note that the last line of \eqref{eq:apDLongG2} vanishes unless $K \in \left\{ M,P \right\}$ and $J \in \left\{M, P \right\} $. Moreover, as we now show in more detail, the leading contribution comes from the case in which $K=M$ and $J=P$ or $J=M$ and $K=P$, while the other cases, which require $J=K$, are sub-leading in the planar limit.\par
To formalize the previous sentence, let us fix $J,K,M,P$ and introduce the projections onto $R\sqcup\Box_J$ and $R\sqcup\Box_K$, denoted $\Pi_J$ and $\Pi_K$ respectively. The last two lines of \eqref{eq:apDLongG2} can be recombined in such a way to remove the sums over $\ai, \hat{\ai},\hat{\mathsf{b}}$, and we get
\begin{align}
\frac{1}{\tilde{L}^2}\sum_{ij }  \sum_{ \check{\mathsf{c}}} & \lllang R\sqcup  \symbox_{(MP)},\check{\mathsf{c}} \lvert c_i^\dagger\Pi_J c_j^\dagger c_i \Pi_K c_j\rvert R\sqcup \symbox_{(MP)},\check{\mathsf{c}} \rrrang \notag \\
=\frac{1}{\tilde{L}^2}\sum_{ij}  \sum_{ \check{\mathsf{c}}} & \left[ \lllang R\sqcup  \symbox_{(MP)},\check{\mathsf{c}} \lvert c_i^\dagger\Pi_J c_i c_j^\dagger\Pi_K c_j\rvert R\sqcup \symbox_{(MP)},\check{\mathsf{c}} \rrrang \right. \\
-& \left. \delta_{ij} \lllang R\sqcup \symbox_{(MP)},\check{\mathsf{c}} \lvert c_i^\dagger\Pi_J \Pi_K c_j\rvert R\sqcup  \symbox_{(MP)},\check{\mathsf{c}} \rrrang \right] .  \notag
\end{align}
If $J=K$, these two terms cancel at leading order in the large $L$ limit and the contribution is suppressed. If $J\neq K$, then the second term vanishes, and the first term is nonzero if and only if the sets $\left\{ M,P\right\} $ and $\left\{ J,K \right\}$ are equal. In that case, we insert a resolution of the identity for the basis of $R \otimes\begin{ytableau} \ & \ \end{ytableau}$ between $c_i$ and $c_j ^{\dagger}$, obtaining 
\begin{align}
& \sum_{ \check{\mathsf{c}}} \frac{1}{\tilde{L}^2}\sum_{ij}  \lllang R\sqcup \symbox_{(JK)},\check{\mathsf{c}} \lvert c_i^\dagger\Pi_J c_i c_j^\dagger\Pi_K c_j\rvert R\sqcup \symbox_{(JK)},\check{\mathsf{c}} \rrrang \notag \\
& = \sum_{ \check{\mathsf{c}}} \sum_{(M^{\prime} P^{\prime})\in\mathscr{J}^{(2)}_R}  \sum_{ \check{\mathsf{d}}}  \lllang R\sqcup \symbox_{(JK)},\check{\mathsf{c}} \lvert \frac{1}{\tilde{L}} \sum_i c_i^\dagger\Pi_J c_i \rvert R\sqcup \symbox_{(M^{\prime} P^{\prime})},\check{\mathsf{d}} \rrrang \\
& \qquad \times \lllang R\sqcup \symbox_{(M^{\prime}P^{\prime})},\check{\mathsf{d}}  \vert  \frac{1}{\tilde{L}} \sum_j c_j^\dagger\Pi_K c_j\rvert R\sqcup  \symbox_{(JK)},\check{\mathsf{c}} \rrrang , \notag 
\end{align}
with symmetrization over the added boxes, as explained around \eqref{eq:defsymrepH}.\par
Similar to the computation in step 2, the non-vanishing contributions come from $\left\{ M^{\prime}, P^{\prime} \right\} =\left\{ J,K \right\} $ and $\check{\mathsf{d}}=\check{\mathsf{c}}$, and all the summands have the same value. 
We can therefore rewrite the expression as 
\begin{align}
\frac{1}{\dim(R\sqcup\Box_J\sqcup\Box_K)} \sum_{ \check{\mathsf{c}}} &  \lllang R\sqcup \symbox_{(JK)},\check{\mathsf{c}} \lvert \frac{1}{\tilde{L}} \sum_i c_i^\dagger\Pi_J c_i \rvert R\sqcup \symbox_{(JK)},\check{\mathsf{c}} \rrrang \notag  \\
\times \sum_{ \check{\mathsf{d}}} & \lllang R\sqcup \symbox_{(JK)},\check{\mathsf{d}}  \vert  \frac{1}{\tilde{L}} \sum_j c_j^\dagger\Pi_K c_j\rvert R\sqcup \symbox_{(JK)},\check{\mathsf{d}} \rrrang \notag \\ 
= \frac{1}{\tilde{L}^2\dim(R\sqcup\Box_J\sqcup\Box_K)}& \left(\sum_{i,\hat{\ai},\check{\mathsf{c}}}\vert\lllang R\sqcup \symbox_{(JK)},\check{\mathsf{c}} \lvert c_i^\dagger \rvert  R\sqcup\Box_J,\hat{\ai} \rrrang\vert^2\right)  \notag  \\ 
\times & \left(\sum_{j,\hat{\mathsf{b}},\check{\mathsf{d}}}\vert\lllang R\sqcup \symbox_{(JK)},\check{\mathsf{d}} \lvert c_j^\dagger \rvert  R\sqcup\Box_K,\hat{\mathsf{b}} \rrrang\vert^2\right).  \label{eq:leap}
\end{align}
Acting with the creation operator $c_i^{\dagger}$ in the first bracket we get the state $\lvert  R\sqcup\Box_J,\hat{\ai} ; i \rrrang$, to be contracted with $\lllang R\sqcup \symbox_{(JK)},\check{\mathsf{c}} \lvert$. Could we forget about the symmetrization and replace $R \sqcup \symbox_{(JK)}$ with $R \sqcup \Box_J \sqcup \Box_K$, then we would simply apply the left-hand side of \eqref{eq:normcoefsrequired}, with $R$ there replaced by $R \sqcup \Box_J$ and labels $(j,J)$ there replaced by $(i,K)$. Doing the residual sum over $\check{\mathsf{c}}$, we would conclude that the first bracket contributes $\dim(R\sqcup\Box_J\sqcup\Box_K)$. The same applies to $c_j^\dagger \rvert  R\sqcup\Box_K,\hat{\mathsf{b}} \rrrang $.\par
Let us now track more carefully the effect of the symmetrization due to the bosonic nature of the probe.
\begin{enumerate}[(i)]
    \item Throughout this analysis, we neglect the probe states in which the two particles are created in the same index, i.e. states $\lvert k,l \rangle$ are assumed to have $k \ne l$. We are thus neglecting $O(1/L)$ contributions, and our claims are valid at large $L$.
    \item The state in the bra (both in the first and second line of \eqref{eq:leap}) comes from the Clebsch--Gordan expansion of states of the form $ \rvert R, \ai \rrrang \otimes \frac{\lvert k,l \rangle_{\text{\rm probe}} + \lvert l,k \rangle_{\text{\rm probe}} }{\sqrt{2}}$. 
    \item We can now expand the state by tensoring with each of the two summands. From \eqref{eq:defsymrepH}, we write schematically 
    \begin{align}
        \lvert R\sqcup \symbox_{(JK)},\check{\mathsf{c}} \rrrang &= \frac{\xi}{\sqrt{2}} \left[\lvert  (R \sqcup \Box_J) \sqcup \Box_K ,\check{\mathsf{c}} \rrrang + \lvert\left(  (R \sqcup \Box_K ) \sqcup \Box_J \right),\check{\mathsf{c}} \rrrang\right] , \label{eq:RboxboxtoBoxJK}
\end{align}
    with $\xi$ a phase. Here we have traded the symmetrization over the probe indices to an exchange of the order of the projectors on the boxes --- see around Equation \eqref{eq:isotropy} for a rigorous justification.\par
    We observe that the projection onto the symmetric representation $\symbox$ that appears in the Clebsch--Gordan decomposition of $\Box \otimes \Box$ induces a projector onto $\mathscr{H}(R \otimes \symbox )$ from $\mathscr{H}(R \otimes \Box \otimes \Box)$, for all $R$. In the following, we denote by $\Pi_{\text{\rm sym}}$ this projection acting on the Hilbert space. 
    \item The $1/\sqrt{2}$ becomes a weight $1/2$ in the sum due to the absolute value square. After contraction, it cancels against a factor of $2$ produced from the ket, yielding $1$. 
    \item The state in the ket of the first line of \eqref{eq:leap} is in  $(R \sqcup \Box_{J} ) \otimes \Box$, indexed by the pair $(\hat{\ai},i)$; likewise, the state in the ket of the second line of \eqref{eq:leap} is in  $(R \sqcup \Box_{K} ) \otimes \Box$, indexed by the pair $(\hat{\mathsf{b}},j)$.
    \item Contracting with the bra, using the projectors on the ket, and the orthogonality relations for the basis of $R \otimes \Box \otimes \Box$, we conclude that only one of the states labelled by $(\hat{\ai},i)$ (respectively $(\hat{\mathsf{b}},j)$) in the first (respectively second) line of \eqref{eq:leap} contributes non-trivially, with weight 1.
\end{enumerate}
The conclusion of this argument is reliable in the saddle point approximation, in which the sum over $R$ localizes onto a large Young diagram.\par
In this way, we obtain a total contribution of 
\begin{align}
\frac{1}{\tilde{L}^2}\dim(R\sqcup\Box_J\sqcup\Box_K) &= \frac{1}{\tilde{L}^2} \cdot \frac{\dim(R\sqcup\Box_J\sqcup\Box_K)}{\dim (R\sqcup \Box_K )} \cdot \frac{\dim(R\sqcup\Box_K)}{\dim (R)} \cdot \dim (R) \notag \\
& \approx  \dim R \left( \frac{\dim (R\sqcup \Box_J )}{\tilde{L} \dim (R) } \right)  \left( \frac{\dim (R\sqcup \Box_K )}{\tilde{L} \dim (R) } \right).  \label{eq:appleapresult}
\end{align}
In the second line we have used the planar approximation \eqref{eq:approxrationdimsJ} for the ratio of dimensions. The latter formula implies that \eqref{eq:apDLongG2} asymptotes to
\begin{align}
    \sum_{J,K\in\mathscr{J}_R} & e^{ it_3H^{\prime}(R,\phi(R))-it_{32} H^{\prime}(R\sqcup \Box_K,\phi(R))-it_1 H^{\prime}(R\sqcup\Box_J)- it_{21} H^{\prime}(R\sqcup\Box_J\sqcup\Box_K) } \notag \\
    & \times \dim R \left( \frac{\dim (R\sqcup \Box_J )}{\tilde{L} \dim (R) } \right)  \left( \frac{\dim (R\sqcup \Box_K )}{\tilde{L} \dim (R) } \right) .
\end{align}\par
Using \eqref{eq:HintshiftJ} and reinserting the $\phi (R)$-dependence, that produces $\dim \phi (R)$, we obtain a term that factorizes into the product of two-point functions in the large $L$ limit (under the assumption of a non-trivial saddle point for the ensemble of representations).
\end{proof}\par
\medskip
\begin{proof}[Proof of the second term in \eqref{eq:AppDG2split}]
Now we treat the second term in a similar way. In this case, we obtain an analogous formula by inserting successive resolutions of the identity:
\begin{align}
& \frac{1}{\tilde{L}^2}\sum_{ij \ai}\lllang R, \ai \lvert c_i(t_3)c_j(t_2)c^\dagger_j(t_1)c^\dagger_i(0) \rvert R, \ai \rrrang  \notag \\
=& \sum_{J,K\in\mathscr{J}_R} \sum_{(MP)\in\mathscr{J}^{(2)}_R} e^{it_3H^{\prime}(R,\phi(R))-it_{32} H^{\prime}(R\sqcup \Box_K,\phi(R))-it_1 H^{\prime}(R\sqcup\Box_J)-it_{21}H^{\prime} (R\sqcup\Box_M\sqcup\Box_P)} \notag \\
\times & \frac{1}{\tilde{L}^2}\sum_{ij \ai} \sum_{ \hat{\ai } \hat{\mathsf{b}} } \lllang R, \ai \lvert c_i \rvert R\sqcup\Box_K,\hat{\mathsf{b}} \rrrang \lllang R \sqcup \Box_J, \hat{\ai} \lvert c_i^\dagger \rvert R, \ai \rrrang \notag \\
& \qquad \times \sum_{ \check{\mathsf{c}}}  \lllang R\sqcup\Box_K,\hat{\mathsf{b}}  \lvert c_j \rvert R\sqcup \symbox_{(MP)},\check{\mathsf{c}} \rrrang \lllang R\sqcup \symbox_{(MP)}\check{\mathsf{c}} \lvert c_j^\dagger \rvert R\sqcup\Box_J,\hat{\ai} \rrrang.  \label{eq:leapbis}
\end{align}
This is only nonzero if $K=J$. In this case, the inner sum further simplifies as 
\begin{align}
\label{eq:twoboxes}
    \sum_{j,\hat{\ai},\check{\mathsf{c}}}\vert\lllang R\sqcup\symbox_{(MP)},\check{\mathsf{c}} \lvert c_j^\dagger \rvert  R\sqcup\Box_J,\hat{\ai} \rrrang\vert^2,
\end{align}
which is nonzero only if $M=J$ or $P=J$ (without loss of generality one can suppose $M=J$). By the same argument used after \eqref{eq:leap}, the piece \eqref{eq:twoboxes} can be approximated by $\dim(R\sqcup\Box_J\sqcup\Box_P$). From here together with \eqref{eq:approxrationdimsJ} we readily find that expression \eqref{eq:leapbis} asymptotes to
\begin{align}
    \frac{1}{\tilde{L}^2} \dim (R\sqcup \Box_J ) \left( \frac{\dim (R\sqcup \Box_J\sqcup\Box_P) }{\dim (R\sqcup \Box_J )}  \right)\approx \dim R \left( \frac{\dim (R\sqcup \Box_J) }{\tilde{L} \dim (R)}  \right) \left( \frac{\dim (R\sqcup \Box_P) }{\tilde{L} \dim (R)}  \right) .
\end{align}
Similarly, reinserting into the sum over $J,K$ with the various Kronecker $\delta_{MJ}$ and so on, we obtain a term that factorizes into the product of two-point functions in the planar limit.
\end{proof}\par
\begin{cor}\label{corol:secondpieceAppD}
    The second line on the right-hand side of \eqref{eq:phi4ptintoO4pt} reproduce in the planar limit the second and third line on the right-hand side of \eqref{eq:phi4ptfactin2}. 
\end{cor}
\begin{proof}
    Upon adding to $G^{(2)}_{ijkl}(t_1,t_2,t_3)$ terms that are identically zero (correlation functions of mismatching number of creation and annihilation operators), we get the second piece on the right-hand side of \eqref{eq:phi4ptintoO4pt}. On the other hand, the explicit computation in the previous lemma shows that $G^{(2)}_{ijkl}(t_1,t_2,t_3)$ reduces in the planar limit to \eqref{eq:AppDG2split}.\par
    Comparing the expressions we see that the two lines on the right-hand side of \eqref{eq:AppDG2split} reproduce in the planar limit the second and third line on the right-hand side of \eqref{eq:phi4ptfactin2}. 
\end{proof}\par

\subsubsection{Final step: Summing the contributions}

Corollary \ref{corol:firstpieceAppD} equates the first term on the right-hand side of \eqref{eq:phi4ptintoO4pt} with the first line in \eqref{eq:phi4ptfactin2}, and Corollary \ref{corol:secondpieceAppD} equates the second term on the right-hand side of \eqref{eq:phi4ptintoO4pt} with the second and third line on the right-hand side of \eqref{eq:phi4ptfactin2}.\par 
Summing the two concludes the derivation of Proposition \ref{lemmaAppFactorization}.

\subsubsection{Complete argument for the large \texorpdfstring{$N$}{N} factorization}
\label{app:completeAppD}

We now show that our models have factorizing correlation functions at large $N$ for any $2n$-point function, using a different method than extends to higher point functions more directly. For concreteness we study the four-point function here, and it should be clear how this method straightforwardly generalizes to higher point functions.\par
Let us focus on the term
\begin{align}
& \frac{1}{L^2}\sum_{ijkl \ai}\lllang R, \ai \lvert c_i(t_3)c_j(t_2)c^\dagger_k(t_1)c^\dagger_l(0) \rvert R, \ai \rrrang  \notag \\
=& \frac{1}{L^2}\sum_{\substack{J\in\mathscr{J}_R\\K\in\mathscr{J}_R }} \sum_{(MP)\in\mathscr{J}^{(2)}_R} e^{it_3H(R,\phi(R))+(it_2-it_3)H(R\sqcup \Box_K,\phi(R)) -it_1 H(R\sqcup\Box_J)+(it_1-it_2)H(R\sqcup\Box_M\sqcup\Box_P)} \notag \\
\times&\sum_{\substack{  \ai ijkl}}\lllang R, \ai \lvert c_i \Pi_K c_j \Pi_{(MP)}c_k^\dagger \Pi_J c_l^\dagger \rvert R, \ai \rrrang.  \label{eq:appD4genarg}
\end{align}

Here, just like above, the various $\Pi_X$ denote the Young projections onto the representation $R$ with the boxes $X$ adjoined to it.
By flavor invariance, each of the inner products in the innermost sum can be nonzero only if the representations in the projectors differ by exactly one box. This is only possible if either $K=J=M$, or if $J=M$ and $K=P$. In the first case, the inner sum on the third line of \eqref{eq:appD4genarg} can be rewritten as 
\begin{align}
\label{eq:D4geninnersum}
\sum_{\substack{  \ai ijkl}}\lllang R, \ai \lvert c_i \Pi_K c_j \Pi_{(KP)} c_k^\dagger \Pi_{K} c_l^\dagger \rvert R, \ai \rrrang.
\end{align}
For $i$ a fundamental index, let $\alpha_i^\dagger$ be the operator that tensors a representation with a copy of the fundamental in the state $i$. On the $n$-particle bosonic Hilbert space, $c_i^\dagger$ acts as $\sqrt{n}\Pi_{\mathrm{sym}}\alpha_i^\dagger$, where $\Pi_{\mathrm{sym}}$ is the symmetrization projector, explicitly realized in accordance with \eqref{eq:RboxboxtoBoxJK}. With this notation, \eqref{eq:D4geninnersum} can be rewritten as
\begin{align}
2\sum_{\substack{  \ai ijkl}}\lllang R, \ai \lvert \alpha_i \Pi_K \alpha_j \Pi_{\mathrm{sym}}\Pi_{(KP)}\Pi_{\mathrm{sym}} \alpha_k^\dagger \Pi_{K} \alpha_l^\dagger \rvert R, \ai \rrrang.
\end{align}
Now, $\Pi_{\mathrm{sym}}\Pi_{(KP)}\Pi_{\mathrm{sym}}=\Pi_{(KP)}$, so we can rewrite the expression as 
\begin{align}
2\sum_{\substack{  \ai ijkl}}\lllang R, \ai \lvert \alpha_i \Pi_K \alpha_j\Pi_{(KP)}\alpha_k^\dagger \Pi_{K} \alpha_l^\dagger \rvert R, \ai \rrrang.
\label{eq:projsym}
\end{align}\par
Here, it is useful to pause and understand the meaning of $\Pi_K$ and $\Pi_{(KP)}$ in more detail. They are projectors on the irreducible representations $R\sqcup\Box_K$ and $R\sqcup\symbox_{(KP)}$, respectively. The representation $R\sqcup\symbox_{(KP)}$ is to be understood as the (unique) copy of the representation $R\sqcup\symbox_{(KP)}$ in the tensor product $R\otimes(\symbox)$. It is a linear combination of the form \begin{align}
R\sqcup\symbox_{(KP)}= \gamma_1 (R\sqcup\Box_K\sqcup \Box_P) \oplus \gamma_2 (R\sqcup\Box_P\sqcup\Box_K),
\end{align}
with $\lvert\gamma_1\rvert^2+ \lvert\gamma_2\rvert^2=1$. $(R\sqcup\Box_K\sqcup \Box_P)$ and $(R\sqcup\Box_P\sqcup \Box_K)$ are the respective images of $R\otimes\Box\otimes\Box$ by the Young projectors with the two last indices at position $K$ then $P$ or $P$ then $K$. Note that in general, Young projectors are not mutually orthogonal when the tableau shape is the same \cite{Stembridge}, however here we have a small number of added boxes compared to the size of the tableau, so most of them actually are (for example here they always are when only two boxes are added at different nonadjacent rows to a given Young diagram), and we can neglect this subtlety.\par

\begin{lem}
    $\gamma_1=\gamma_2$, hence $\lvert \gamma_i \rvert^2= \frac{1}{2}$ for $i=1,2$.
\end{lem}
\begin{proof}
We now show that the $\gamma_i$ are equal. Let us consider a vector of the form 
\begin{equation}
    \frac{1}{\sqrt{2}} \left( \lvert R, \ai ; ij \rrrang + \lvert R, \ai ; ji \rrrang \right) .
\end{equation}
We can safely assume that the boxes are added to different rows, as in this way the neglected contributions are suppressed in the large $N$ limit. Then, the representations $(R\sqcup \Box_K)\sqcup \Box_P$ and $(R\sqcup \Box_P)\sqcup \Box_K$ with $P \ne K$ are orthogonal to each other, and we can write:
\begin{align}
\label{eq:isotropy}
\lllang R\sqcup \symbox_{(KP)},\check{\mathsf{c}} \vert \frac{1}{\sqrt{2}} \left( \lvert R, \ai ; ij \rrrang + \lvert R, \ai ; ji \rrrang \right) & = \frac{1}{\sqrt{2}} \lllang (R\sqcup \Box_K)\sqcup \Box_P,\check{\mathsf{c}}\vert \left( \lvert R, \ai ; ij \rrrang + \lvert R, \ai ; ji \rrrang \right) \\
&+\frac{1}{\sqrt{2}}\lllang (R\sqcup \Box_P)\sqcup \Box_K,\check{\mathsf{c}}\vert \left( \lvert R, \ai ; ij \rrrang + \lvert R, \ai ; ji \rrrang \right). \notag 
\end{align}
We want to show the two terms on the right-hand side are equal. It is therefore useful to compute their difference. Multiplying by $\sqrt{2}$ and expanding it we obtain:
\begin{align}
 & \lllang (R\sqcup \Box_K)\sqcup \Box_P,\check{\mathsf{c}}\vert \left( \lvert R, \ai ; ij \rrrang + \lvert R, \ai ; ji \rrrang \right) -  \lllang (R\sqcup \Box_P)\sqcup \Box_K,\check{\mathsf{c}}\vert \left( \lvert R, \ai ; ij \rrrang + \lvert R, \ai ; ji \rrrang \right) \notag \\
 =& \lllang (R\sqcup \Box_K)\sqcup \Box_P,\check{\mathsf{c}} \lvert R, \ai ; ij \rrrang - \lllang (R\sqcup \Box_P)\sqcup \Box_K,\check{\mathsf{c}}\lvert R, \ai ; ji \rrrang \notag \\
 +& \lllang (R\sqcup \Box_K)\sqcup \Box_P,\check{\mathsf{c}} \lvert R, \ai ; ji \rrrang -\lllang (R\sqcup \Box_P)\sqcup \Box_K,\check{\mathsf{c}}\lvert R, \ai ; ij \rrrang.  
\end{align}\par
Going back to \eqref{eq:appD4genarg}, the two terms in the second line are equal by definition of the representations $(R\sqcup \Box_P)\sqcup \Box_K$ and $(R\sqcup \Box_K)\sqcup \Box_P$ respectively, and the same goes for the two terms in the third line. Therefore, we obtain that any element of $R\otimes(\symbox)$ has equal overlap with  $(R\sqcup \Box_P)\sqcup \Box_K$ and $(R\sqcup \Box_K)\sqcup \Box_P$, up to corrections negligible in the large $N$ limit. In particular $\gamma_1=\gamma_2$, and $\lvert \gamma_i\rvert^2$ is constant equal to $1/2$.
\end{proof}\par
When $\alpha_j^\dagger$ acts on $R\sqcup\Box_K$ in \eqref{eq:projsym}, the only nonzero term after acting with $\Pi_{(KP)}$ comes from its $R\sqcup\Box_K\sqcup \Box_P$ component (where $K$ and $P$ are not symmetrized). 
By the above, the previous expression then becomes
\begin{align}
\frac{2}{2}\sum_{\substack{  \ai ijkl}}\lllang R, \ai \lvert \alpha_i \Pi_K \alpha_j\Pi_{KP}\alpha_k^\dagger \Pi_{K} \alpha_l^\dagger \rvert R, \ai \rrrang.
\end{align}
One can then further remove the intermediate projections (since we are using Young projectors \cite{Stembridge}) and find the large $N$ contribution
\begin{align}
\sum_{\substack{  \ai ijkl}}\lllang R, \ai \lvert \alpha_i \alpha_j\Pi_{KP}\alpha_k^\dagger \alpha_l^\dagger \rvert R, \ai \rrrang\,{\approx}\,\mathrm{Tr}_{R\otimes\Box\otimes\Box}(\Pi_{KP})=\dim(R\sqcup\Box_K\sqcup\Box_P).
\end{align}
This is consistent with the factorization shown above. The second term ($J=M,K=P$) works in a similar way: we obtain
\begin{align}
\sum_{\substack{  \ai ijkl\check{\mathsf{c}}}}\lllang R, \ai \lvert \alpha_i \alpha_j\vert R\sqcup\Box_J\sqcup\Box_K,\,\check{\mathsf{c}}\rrrang\lllang R\sqcup\Box_K\sqcup\Box_J,\,\check{\mathsf{c}}\vert\alpha_k^\dagger \alpha_l^\dagger \rvert R, \ai \rrrang\\=\sum_{\substack{  \ai ijkl\check{\mathsf{c}}}}\lllang R, \ai \lvert \alpha_i \alpha_j\vert R\sqcup\Box_J\sqcup\Box_K,\,\check{\mathsf{c}}\rrrang\lllang R\sqcup\Box_J\sqcup\Box_K,\,\check{\mathsf{c}}\vert\alpha_l^\dagger \alpha_k^\dagger \rvert R, \ai \rrrang.
\end{align}
We then find the contribution at large $N$
\begin{align}
\sum_{\substack{  \ai ijkl}}\lllang R, \ai \lvert \alpha_i \alpha_j\Pi_{KP}\alpha_l^\dagger \alpha_k^\dagger \rvert R, \ai \rrrang{\,\approx}\,\mathrm{Tr}_{R\otimes\Box\otimes\Box}(\Pi_{KP})=\dim(R\sqcup\Box_K\sqcup\Box_P).
\end{align}

It is a tedious but straightforward exercise to check that this method extends to higher point functions and different orders of $c_i$ and $c_i^\dagger$, yielding large $N$ factorization.

\end{appendix}

\clearpage 
{\small
\bibliography{tfdmm}
}
\end{document}